%% file: thesis.tex
\documentclass[a4paper,12pt,customfont,numbered,print,index, PageStyleII]{Classes/PhDThesisPSnPDF}

\input{Preamble/preamble}

\input{thesis-info}

\ifdefineAbstract
\fi

\input{Auxiliaries/packages}

\input{Auxiliaries/commands}
\makeindex

\begin{document}

\frontmatter

\maketitle

\newpage

\gutachter

\include{Dedication/dedication}
\include{Acknowledgement/acknowledgement}

\include{Abstract/abstract}


\tableofcontents

\listoffigures

\listoftables

\clearpage
\chapter{Acronyms}
\input{Auxiliaries/acronyms}

\renewcommand{\nomname}{Important Symbols}
\printnomenclature[3cm]

\mainmatter
\input{Auxiliaries/notation}

\include{Auxiliaries/mynomenclature}
\include{Introduction/introduction}

\include{Chapter1/chapter1}
\include{Chapter2/chapter2}
\include{Chapter3/chapter3}
\include{Chapter4/chapter4}
\include{Chapter5/chapter5}
\include{Chapter6/chapter6}
\include{Conclusions/conclusions}

\begin{appendices}
\include{Appendix1/appendix1}

\include{Appendix2/appendix2}
\end{appendices}

\backmatter



\bibliographystyle{ieeetr}
\renewcommand{\bibname}{References} 
\bibliography{sty/IEEEabrv,References/references} 


\printthesisindex 

\include{Curriculum_Vitae/cv}

\end{document}

%% file: thesis-info.tex
\title{Fountain Codes under Maximum Likelihood Decoding} 

\university{Vom Promotionsausschuss der  Technischen~Universit\"at~Hamburg-Harburg}

\author{\href{mailto:francisco.lazaroblasco@dlr.de}{Francisco L\'azaro Blasco}}

\degreetitle{Doktor-Ingenieur (Dr.-Ing.)}

\subject{Fountain Codes under Maximum Likelihood Decoding} \keywords{{Fountain codes} {LT codes} {Raptor codes} {Maximum likelihood (ML) decoding} {inactivation decoding}}

%% file: Auxiliaries/packages.tex
\usepackage{epsfig,amsfonts}
\usepackage{mathrsfs}
\usepackage[printonlyused]{acronym}
\usepackage{tikz}
\usepackage{caption}
\usepackage[section]{placeins}
\usepackage{anyfontsize}
\usepackage{amsthm}
\usepackage{relsize}
\usepackage{mathtools}
\mathtoolsset{showonlyrefs}
\usepackage{subcaption}

\usepackage{ifthen}
\usetikzlibrary{shapes,arrows}

\usepackage{lmodern}
\usepackage{xcolor}
\usepackage{needspace}
\usepackage{enumitem}
\usepackage{tabularx}
\usepackage{epigraph}

%% file: Auxiliaries/commands.tex
\newcommand{\msr}[1]{{\mathscr{#1}}}

\newcommand{\fran}[1]{{\textcolor{black}{#1}}}

\newcommand{\figw}{0.6\columnwidth}

\newcommand{\figwbigger}{0.7\columnwidth}
\newcommand{\figwBigger}{0.75\columnwidth}

\newenvironment{absolutelynopagebreak}
  {\par\nobreak\vfil\penalty0\vfilneg
   \vtop\bgroup}
  {\par\xdef\tpd{\the\prevdepth}\egroup
   \prevdepth=\tpd}

%% file: Dedication/dedication.tex

\begin{dedication} 

To my beloved wife, parents and sister.

\end{dedication}



%% file: Acknowledgement/acknowledgement.tex

\begin{acknowledgements}      

The journey towards this dissertation was fascinating but also long and challenging. It was only thanks to the technical advise, encouragement and love of those around me that I was able to finish it.

First, I would like to express my heartfelt gratitude to my advisor Prof. Gerhard Bauch for his guidance and valuable suggestions. I would also like to thank him explicitly for his support clearing the bureaucratic hurdles of the PhD. The next thank you goes to Prof. Amin Shokrollahi for accepting to review the thesis and for the helpful discussion during the PhD defense. I would also like to thank Prof. Christian Schuster for his efficiency handling the examination process.

Words cannot express my gratitude to my mentor at {DLR}, Gianluigi Liva, for his continuous support and guidance, and for teaching me all I know about channel coding.  His passion and dedication to research have been an inspiration to me during all these years. I am also deeply in debt with Enrico Paolini for his advice and  scientific rigour.

I would like to extend my gratitude also to my colleagues at {DLR} for creating a wonderful working atmosphere. Especially, I would like to thank Bal\'azs, Federico, Giuliano and Javi for the technical and non-technical discussions. I would also like to thank Sandro Scalise and Simon Plass for supporting my research.

Last, but not least, I would like to thank my loving and caring family. I am infinitely thankful to my parents for giving me the most valuable present, a good education. A great thanks goes to my sister, for taking care of me like a mother. I would also like to thank my Italian family for making me feel welcome from the very first moment. I cannot put in words my gratitude to Paola for her love, patience, constant support and for making me feel at home.
The very last ``thank you'' goes to little Francesco for pushing me to finalize this dissertation.

~

\hfill Munich, May 2017.

\end{acknowledgements}



%% file: Abstract/abstract.tex

\begin{abstract}        

This dissertation focuses on fountain codes under \acf*{ML} decoding.
\fran{Fountain codes are a class of erasure correcting codes that can generate an endless amount of coded symbols and were conceived to deliver data files over data networks to a potentially large population of users.}
First \acf{LT} codes  are considered\fran{, which represent the first class of practical fountain codes. Concretely, the focus is on \ac{LT} codes} under inactivation decoding, an efficient \acs*{ML} decoding algorithm that is  widely used in practical systems.
\fran{More precisely}, the decoding complexity of \ac{LT} codes under inactivation decoding is analyzed in terms of the expected number of inactivations. The proposed analysis is based on a dynamical programming approach. This analysis is then extended to provide the probability distribution of the number of inactivations. Additionally a lower complexity approximate analysis is introduced and a code design example is presented that illustrates how these analysis techniques can be used to design \acs{LT} codes.
Next Raptor codes under \acs*{ML} decoding are considered. An upper bound to the probability of decoding failure of $q$-ary Raptor codes is developed, considering the weight enumerator of the outer code (precode).  The bound is shown to be tight, specially in the error floor region, by means of simulations. This bound shows how Raptor codes can be analyzed similarly to a traditional serial concatenation of (fixed-rate) block codes. Next, a heuristic method is presented that yields an approximate analysis of Raptor codes under inactivation decoding.  It is also shown by means of an example how the results in this thesis can be used to design Raptor codes.
Raptor codes are next analyzed in a fixed-rate setting. Concretely, a Raptor code ensemble with an outer code picked from the linear random ensemble is considered. For this ensemble, the average weight enumerator and its growth rate are provided. Furthermore, sufficient and necessary conditions for the ensemble to have a minimum distance growing linearly with the block length are presented. The ensemble analyzed resembles standard Raptor codes, as it is shown by means of simulations.
Finally a new class of fountain codes is introduced, that consists of a parallel concatenation of a block code with a \acf{LRFC}. This scheme is specially interesting when the block code is a \ac{MDS} code.  In this case, the scheme can provide failure probabilities lower than those of \acs{LRFC} codes by several orders of magnitude, provided that the erasure probability of the channel is not too high.

\end{abstract}



%% file: Auxiliaries/acronyms.tex
\begin{acronym}
\acro{ARQ}{automatic retransmission query}
\acro{AWGN}{additive white Gaussian noise}
\acro{BCH}{Bose Chaudhuri Hocquenghem}
\acro{BEC}{binary erasure channel}
\acro{BSC}{binary symmetric channel}
\acro{BP}{belief propagation}
\acro{CER}{codeword error rate}
\acro{CO-WEF}{conditional output-weight enumerator function}
\acro{CRC}{cyclic redundancy check}
\acro{FEC}{forward error correction}
\acro{GE}{Gaussian elimination}
\acro{GRS}{generalized Reed-Solomon}
\acro{HDPC}{high-density parity-check}
\acro{ID}{inactivation Decoding}
\acro{i.i.d.}{independent and identically distributed}
\acro{IO-WEF}{input output-weight enumerator function}
\acro{MDS}{maximum distance separable}
\acro{ML}{maximum likelihood}
\acro{MWI}{maximum weight inactivation}
\acro{LDGM}{low-density generator matrix}
\acro{LDPC}{low-density parity-check}

\acro{LRFC}{linear random fountain code}
\acro{LT}{Luby transform}

\acro{PMF}{probability mass function}
\acro{QEC}{$q$-ary erasure channel}
\acro{RI}{random inactivation}
\acro{RS}{Reed-Solomon}
\acro{RSD}{robust soliton distribution}
\acro{SA}{simulated annealing}
\acro{SPC}{single parity-check}
\acro{TCP}{Transmission Control Protocol}
\acro{WE}{weight enumerator}
\acro{WEF}{weight enumerator function}
\end{acronym}

%% file: Auxiliaries/notation.tex
\newcommand{\erasprob}{\varepsilon}
\newcommand{\inputalpha}{\mathcal{X}}
\newcommand{\outputalpha}{\mathcal{Y}}
\newcommand{\inputrv}{X}
\newcommand{\outputrv}{Y}

\newcommand{\codensemble}{\msr{C}}
\newcommand{\ensemble}{\msr{C}}
\newcommand{\oensemble}{\msr{C}_{\mathsf o}}
\newcommand{\code}{\mathcal{C}}
\newcommand{\precodegeneric}{\msr{C}^{*}}

\newcommand{\outercodeensemble}{\msr{C}^\text{o}}

\newcommand{\codef}{\mathcal{C}_f}

\newcommand{\Exp}{\mathbb{E}}

\newcommand{\len}{L}
\newcommand{\osymb}{c} 
\newcommand{\rosymb}{y} 

\newcommand{\nuncovered}{\phi}
\newcommand{\fracuncovered}{\zeta}

\newcommand{\reducedsyst}{\alpha}

\newcommand{\reloverhead}{\epsilon}
\newcommand{\absoverhead}{\delta}

\newcommand{\isd } {\Omega^{\mathrm{ISD}}}
\newcommand{\avgisd } {\bar{\Omega}^{\mathrm{ISD}}}
\newcommand{\rsd } { \Omega^{\mathrm{RSD}} }
\newcommand{\avgrsd } {\bar{\Omega}^{\mathrm{RSD}}}
\newcommand{\rsdt } { \Omega^{\mathrm{RSD}'} }
\newcommand{\Pf } {{\mathsf{P}}_{\mathsf{F}}}
\newcommand{\Pe} { \mathsf{P}_{\mathsf{E}}}
\newcommand{\Pse} { \mathsf{P}_{\mathsf{e}}}
\newcommand{\Pflow } {\protect\underline{\mathsf{P}}_{\mathsf{F}}}

\newcommand{\Pflowisd } {P_{\text{ns}}}
\newcommand{\barPf } { \bar {\mathsf{P}}_{\mathsf{F}}}

\newcommand{\parone } {\psi}
\newcommand{\partwo} {\varsigma}

\newcommand{\Amatrix}{\mathbf{A}}

\ifthenelse{\isundefined{\Bmatrix}}{
  \newcommand{\Bmatrix}{\mathbf{B}}
}{
  \renewcommand{\Bmatrix}{\mathbf{B}}
}
\renewcommand{\Bmatrix}{\mathbf{B}}
\newcommand{\Cmatrix}{\mathbf{C}}
\newcommand{\Dmatrix}{\mathbf{D}}

\newcommand{\myop}[1]{%
  \mathchoice{\raisebox{8pt}{$\displaystyle #1$}}
             {\raisebox{8pt}{$#1$}}
             {\raisebox{4pt}{$\scriptstyle #1$}}
             {\raisebox{1.6pt}{$\scriptscriptstyle #1$}}}

\newcommand{\rippleset}{\mathscr{R}}
\newcommand{\Ripple}{\mathtt{R}}
\renewcommand{\r}{\mathtt{r}}
\newcommand{\ripple}[1]{ \msr{R}_{#1}}
\newcommand{\ru}{\r_u}
\newcommand{\Ru}{\Ripple_u}
\newcommand{\R}[1]{\Ripp_{#1}}

\newcommand{\cloudset}{\mathscr{W}}
\newcommand{\Cloud}{\mathtt{W}}
\renewcommand{\c}{\mathtt{w}}
\newcommand{\cloud}[1]{\msr{W}_{#1}}
\newcommand{\Cu}{\Cloud_u}
\newcommand{\cu}{\c_u}

\renewcommand{\S}[1]{\mathtt{S}_{#1}}
\ifthenelse{\isundefined{\C}}{
    \newcommand{\C}[1]{\mathtt{C}_{#1}}
    }{
    \renewcommand{\C}[1]{\mathtt{C}_{#1}}
    }

\ifthenelse{\isundefined{\N}}{
    \newcommand{\N}[1]{\mathsf{N}_{#1}}
    }{
    \renewcommand{\N}[1]{\mathsf{N}_{#1}}
    }

\newcommand{\Erv}{\mathtt{A}}
\newcommand{\erv}{\mathtt{a}}
\renewcommand{\b}{\mathtt{b}}
\newcommand{\n}{\mathtt{n}}
\newcommand{\Nu}{\mathtt{N}_u}
\renewcommand{\nu}{\n_u}

\newcommand{\bu}{\b_u}

\newcommand{\Y}{\mathtt{N}}
\newcommand{\y}{\mathtt{n}}

\newcommand{\dmax}{ d_{\max}}

\newcommand{\x}{\mathtt{x}}

\newcommand{\Rijsets}[1]{\mathscr{Z}^{(#1)}}
\newcommand{\Rijset}[2]{\mathscr{Z}^{(#2)}_{#1}}
\newcommand{\Rijcard}[2]{\mathtt{z}^{(#2)}_{#1}}
\newcommand{\Rijcards}[1]{\mathtt{z}^{(#1)}}
\newcommand{\Rijreals}[1]{ \mathtt{Z}^{(#1)}}
\newcommand{\Rijreal}[2]{ \mathtt{Z}^{(#2)}_{#1}}
\newcommand{\Rij}[2]{ \hat {\mathtt{Z}}^{(#2)}_{#1}}
\newcommand{\rij}[2]{\xi^{(#2)}_{#1}}
\newcommand{\probtx}[2]{p^{(#2)}_{#1}}
\newcommand{\ntx}[2]{B^{(#2)}_{#1}}
\newcommand{\Rijrealv}[1]{ \pmb{\mathtt{Z}}_{#1}}
\newcommand{\Rijcardv}[1]{ \pmb{\mathtt{z}}_{#1}}
\newcommand{\driftv} {\pmb{\mathtt{w}}}
\ifthenelse{\isundefined{\m}}{
    \newcommand{\m}{m}
    }{
    \renewcommand{\m}{m}
    }

\newcommand{\Raptorinput}{u}
\newcommand{\vecu}{\mathbf{\Raptorinput}}

\newcommand{\Rintermsymbol}{v}
\newcommand{\vecv}{\mathbf{\Rintermsymbol}}

\newcommand{\Rosymb}{c}
\newcommand{\vecout}{\mathbf{\Rosymb}}

\newcommand{\rate}{\mathsf{r}}
\newcommand{\ro}{\rate_{\mathsf o}}
\newcommand{\ri}{\rate_{\mathsf i}}

\newcommand{\srten}{s_{\text{R}10} }
\newcommand{\hrten}{h_{\text{R}10} }

\newcommand{\constmatrix}{\mathbf{M}}
\newcommand{\hmatrixpre}{\mathbf{H}_{\text{o}}}
\newcommand{\GrxLT}{\Grx_{LT}}

\newcommand{\zeros}{\mathbf{z}}
\newcommand{\Rrosymb}{\rosymb}

\newcommand{\we}{\mathcal{A}}
\newcommand{\wesingle}{A}
\newcommand{\weo}{A^{\mathsf o}}
\newcommand{\wei}{A^{\mathsf i}}
\newcommand{\geo}{G^{\mathsf o}}
\ifthenelse{\isundefined{\A}}{
    \newcommand{\A}{\we}
    }{
    \renewcommand{\A}{\we}
    }
\newcommand{\wef}{\mathscr {A}}

\newcommand{\p}{p}
\newcommand{\pl}{\p_{\l}}
\newcommand{\pjl}{ \p_{j,\l} }
\newcommand{\pnl}{ \p_{\nl} }
\newcommand{\npnl}{ \np_{\nl} }
\newcommand{\pjnl}{ \p_{j,\nl} }
\newcommand{\pnlo}{ \p_{\nlo} }
\newcommand{\plo}{ \p_{\lo} }
\newcommand{\np}{\varrho}

\newcommand{\nd}{\varpi}
\newcommand{\dmin}{d_{\mathrm{min}}}
\renewcommand{\d}{d}
\newcommand{\ds}{d^{\star}}
\renewcommand{\l}{l}
\newcommand{\nl}{\lambda}
\newcommand{\nls}{\lambda^\star}
\newcommand{\Hb}{\mathsf H_{\mathsf b}}
\newcommand{\tw}{\tilde w}
\newcommand{\two}{\tilde w_0}
\newcommand{\nlo}{\lambda_0}

\newcommand{\fmax}{\mathsf f_{\textrm{\textnormal{max}}}}
\newcommand{\f}{\mathsf f}
\newcommand{\gmax}{\mathsf g_{\textrm{\textnormal{max}}}}

\ifthenelse{\isundefined{\g}}{
    \newcommand{\g}{\mathsf g}
    }{
    \renewcommand{\g}{\mathsf g}
    }

\newcommand{\lo}{\l_0}

\newcommand{\inner}{\msr{I}}
\renewcommand{\outer}{\msr{O}}
\newcommand{\outerprime}{\msr{O}'}
\newcommand{\region}{\msr{P}}
\newcommand{\Q}{\mathsf Q}
\newcommand{\lambert}{\mathsf W}

\newcommand{\de}{\mathrm{d}}

\newcommand{\vecU}{\mathbf{U}}
\newcommand{\vecV}{\mathbf{V}}
\newcommand{\vecX}{\mathbf{X}}

\newcommand{\hw}{w_{\mathsf H}}
\renewcommand{\deg}{\mathrm{deg}}

\ifthenelse{\isundefined{\G}}{
    \newcommand{\G}{\mathbf{G}}
    }{
    \renewcommand{\G}{\mathbf{G}}
    }

\newcommand{\dmint}{\nd^{\star}}
\newcommand{\dmintt}{d_{\mathrm{min}}^{\star}}
\newcommand{\D}{D}
\newcommand{\Z}{Z}

\newcommand{\avgd}{\bar{\Omega}}

\newcommand{\Omegaeq}{{\hat \Omega}}
\newcommand{\meq}{{\hat m}}

\newcommand{\Omegarten}{ \Omega^{\mathrm{R10}} }
\newcommand{\Omegatwo}{ \Omega^{\dagger} }

\newcommand{\neighbour}{\mathcal{N} }

\newcommand{\neighbourof}[1]{\mathcal{N}\left( #1 \right) }

\newtheorem{algo}{Algorithm}
\newtheorem{mydef}{Definition}
\newtheorem{prop}{Proposition}
\newtheorem{theorem}{Theorem}
\newtheorem{remark}{Remark}
\newtheorem{example}{Example}
\newtheorem{lemma}{Lemma}
\newtheorem{corollary}{Corollary}

\newcommand{\GrxRaptor}{\Grx_{R}}

\newcommand{\GLTa}{\G}
\newcommand{\GrxR}{\G_\text{R}}
\newcommand{\Grx}{\tilde{\mathbf{G}}}

\newcommand{\Gp}{\G_{\text{o}}}
\newcommand{\GLT}{\G_{\text{LT}}}
\newcommand{\weoensemble}{\mathsf{A}}
\newcommand{\pil}{\pi_{\l}}
\newcommand{\krawt}{\mathcal{K}}
\newcommand{\rank}{\mathsf{rank}}
\newcommand{\pifroml}{\vartheta_{i,l,j} }
\newcommand{\qi}{\varphi_i }


\newcommand{\isymbc}{v}
\newcommand{\osymbpre}{\osymb'}
\newcommand{\osymblrfc}{\osymb''}
\newcommand{\osymbconcat}{\osymb}
\newcommand{\Gpre}{\mathbf{G}'}
\newcommand{\Glrfc}{\mathbf{G}''}

\newcommand{\Grxpre}{\Grx'}
\newcommand{\Grxlrfc}{\Grx''}

\newcommand{\nc}{h}
\newcommand{\lc}{n}
\newcommand{\fevent}{F}
\newcommand{\hc}{l}

\newcommand{\probi}{p_i}

\ifthenelse{\isundefined{\argmax}}{
\newcommand{\argmax}{{\arg\max}}
}{
\renewcommand{\argmax}{{\arg\max}}
}

%% file: Auxiliaries/mynomenclature.tex
\nomenclature[aa]{$A_d$}{ number of codewords of weight $d$}
\nomenclature[ac]{$\msr{C}$}{ code ensemble}

\nomenclature[ak]{$k$}{ number of input symbols}
\nomenclature[ah]{$h$}{ number of intermediate symbols}
\nomenclature[ahb]{$\mathsf H_{\mathsf b}$}{ binary entropy function}
\nomenclature[am]{$m$}{ number of output symbols collected by the receiver}
\nomenclature[an]{$n$}{ number of output symbols generated by the encoder}

\nomenclature[ar]{$\mathsf{r}$}{overall rate of a fixed-rate Raptor code}
\nomenclature[aro]{$\mathsf{r}_{\mathsf o}$}{outer code rate}
\nomenclature[ari]{$\mathsf{r}_{\mathsf i}$}{inner fixed-rate LT code rate}
\nomenclature[aw]{$w_{\mathsf H}$}{Hamming weight}
\nomenclature[apf]{${\mathsf{P}}_{\mathsf{F}}$}{probability of decoding failure}

\nomenclature[ge]{$\epsilon$}{relative receiver overhead, $\epsilon=m/k-1$}                
\nomenclature[ge]{$\varepsilon$}{erasure probability of the channel}  
\nomenclature[gd]{$\delta$}{absolute receiver overhead $\delta=m-k$}
\nomenclature[gw]{$\varpi$}{normalized output weight $\varpi =w/n$}
\nomenclature[gD]{$\Delta$}{Transmitter overhead}
\nomenclature[gl]{$\lambda$}{normalized Hamming weight of the intermediate word $\lambda=l/h$}
\nomenclature[go]{$\Omega$}{output degree distribution of an LT code}
\nomenclature[goa]{$\bar \Omega$}{average output degree of an LT code}

\nomenclature[rv]{$\mathbf{v}$}{row vector of input (source) symbols}
\nomenclature[rc]{$\mathbf{c}$}{row vector of output symbols}
\nomenclature[ry]{$\mathbf{y}$}{row vector of received output symbols}
\nomenclature[rG]{$\mathbf{G}$}{generator matrix}
\nomenclature[rH]{$\tilde{\mathbf{G}}$}{matrix corresponding to the non-erased positions of $\mathbf{G}$}

\nomenclature[su]{$\mathbf{u}$}{row vector of input (source) symbols}
\nomenclature[sv]{$\mathbf{v}$}{row vector of intermediate symbols}
\nomenclature[sc]{$\mathbf{c}$}{row vector of output symbols}
\nomenclature[sy]{$\mathbf{y}$}{row vector of received output symbols}
\nomenclature[sG]{$\mathbf{G}_{\text{LT}}$}{generator matrix of the inner LT code}
\nomenclature[sGt]{$\tilde{ \mathbf{G}}_{\text{LT}}$}{matrix corresponding to the non-erased positions of $\mathbf{G}_{\text{LT}}$ }
\nomenclature[sH]{$\mathbf{H}_{\text{p}}$}{parity check matrix of the outer block code (precode)}
\nomenclature[sM]{$\mathbf{M}$}{constraint matrix}

\nomenclature[xw]{$\mathscr{W}$}{ cloud}
\nomenclature[xn]{$\mathcal{N}$}{neighbourhood of a node in a graph}
\nomenclature[xo]{$\mathcal{O}$}{Landau big O notation}
\nomenclature[xos]{$\mathrm o$}{Landau small $\mathrm o$ notation}
\nomenclature[xr]{$\mathscr{R}$}{ripple}
\nomenclature[xp]{$\msr{P}$}{positive growth rate region of a fixed-rate Raptor code}
\nomenclature[xf]{$\mathbb {F}_{q}$}{Galois field of order $q$}
\nomenclature[xkra]{$\mathcal{K}_k(x;n,q)$}{Krawtchouk polynomial of degree $j$ with parameters $h$ and $q$.}

\nomenclature[za]{$\mathrm{deg}(c)$}{degree of output symbol $c$}
\nomenclature[zab]{$\mathrm{deg}_r(c)$}{reduced degree of output symbol $c$}
\nomenclature[zb]{$\binom{k}{i}$}{binomial coefficient}
\nomenclature[zc]{$\lceil x \rceil$}{smallest integer larger than or equal to x}
\nomenclature[zf]{$\lfloor x \rfloor$}{largest integer smaller than or equal to x}
\nomenclature[zap]{$\lfloor x \rceil$}{closest integer to x}
\nomenclature[zrank]{$\text{rank}(\mathbf{G})$}{rank of matrix $\mathbf{G}$} 

%% file: Introduction/introduction.tex
\chapter{Introduction}
\ifpdf
    \graphicspath{{Introduction/IntroductionFigs/PNG/}{Introduction/IntroductionFigs/PDF/}{Introduction/IntroductionFigs/}}
\else
    \graphicspath{{Introduction/IntroductionFigs/EPS/}{Introduction/IntroductionFigs/}}
\fi
\epigraph{I just wondered how things were put together.}{Claude E. Shannon}

In the early years of communication systems it was not known whether error free communication was possible over a communication channel that introduced errors using a rate \fran{that} was not vanishingly small. It was C.E Shannon, who in his landmark paper from 1948 \cite{Shannon1948} proved that error free communication is possible if one communicates at a rate lower than the channel capacity. This milestone gave birth to the Information Age in which we live nowadays.

\fran{Initially the research community focused on the communication channels that arise in the physical layer of a communication system.}
At the physical layer of a communication system the thermal noise generated by thermal vibrations of the atoms in conductors can be accurately modeled as  \ac{AWGN}, giving rise to the \ac{AWGN} channel. The \ac{AWGN} channel was one of the first models to be studied. Another simpler model of the physical layer is the \ac{BSC} channel that was also widely studied during the early days of the Information Age. The \ac{BSC} can be seen as a degradation of the \ac{AWGN} when the input to the channel is constrained to be binary and symmetric and the receiver applies hard decision detection.

After the publication of Shannon's work a humongous amount of research has been carried out in the field of channel coding. The dominant motivation in the research community was getting closer and closer to Shannon's capacity with an affordable complexity. In the early decades of channel coding, algebraic codes were the main focus of research. The most prominent fruits of this research were Hamming, Golay, Reed Muller, \acs{BCH} and \ac{RS} codes \cite{hamming:1950,golay:1949,muller:1954,reed:1954,hocquenghem:1959,bose:1960,reed:RS}. Algebraic coding usually aims at finding codes with good distance properties, usually by maximizing the minimum distance of a code. Due to their good distance properties, algebraic codes tend to exhibit a low probability of error under optimal (\acl{ML}) decoding. The main disadvantage of algebraic codes, is that \fran{in general soft decoding tends to be complex, specially for large block lengths.}
%

The first paradigm change in coding was shifting the focus towards \emph{probabilistic codes} where the aim is at improving the average performance of a code with constraints on the encoding and decoding complexity \cite{Costello:history}. At this stage, the research community had realized that the structure of the codes needed to be tailored to simplify the implementation in practical systems. Convolutional codes, introduced by Elias in \cite{Elias55:2noisy} are generally considered to be the first class of probabilistic codes \cite{Costello:history}. Optimal decoding algorithms for convolutional codes were first derived by Viterbi \cite{viterbi1967error} and then by Bahl, Cocke, Jelinek and Raviv \cite{bahl1974optimal}. Another important milestone in coding was the introduction of concatenated codes by Forney \cite{forney1966concatenated}, which involve a serial cascade of  two linear block codes, usually denoted as inner and outer code. The main advantage of concatenated codes is that the inner and outer codes can be short and easy to decode. Hence, it is possible to decode concatenated codes using so called 2 stage decoders (decoding first the inner and then the outer coder). This decoder is suboptimal but it still shows a very good performance. In fact, the serial concatenation of \ac{RS} and convolutional codes developed by NASA \cite{ccsds:bluebookarticle}, and inspired in Forney's concatenated codes, was for many years one of the best performing coding schemes known and was widely used in practice.

The second paradigm change came with turbo codes, introduced in 1993 \cite{berrou1996near}. Thanks to iterative soft decoding algorithms both turbo and \ac{LDPC} codes were able to approach the Shannon limit in \ac{AWGN} channels with a modest complexity. \ac{LDPC} codes had been proposed and studied by Gallager in his doctoral thesis in 1963 \cite{Gallager63} but later they had been largely forgotten because their potential for long block lengths was not recognised. Shortly after the introduction of turbo codes, \ac{LDPC} codes were rediscovered in \cite{MacKayldpc}, \fran{where it was observed that their performance was better than that of convolutional and concatenated codes, and similar to that of turbo codes}. Nowadays, the majority of practical wireless communication systems use turbo or \ac{LDPC} codes since these codes allow to close largely the gap to capacity in most cases.

In the meantime digital communications have become ubiquitous and channel coding problems are no longer exclusive to the physical layer of communications systems. In this thesis we deal exclusively with erasure channels which are generally not typical from the physical layer. The \ac{BEC} was introduced by Elias in \cite{Elias55:2noisy}. In this channel the transmitters sends one bit (either zero or one) and the receiver either receives this bit error free or receives an \emph{erasure}. The \ac{BEC} was originally regarded as a purely theoretical channel. However, this changed with the emergence of the Internet. It was soon realized that erasure channels are a very good abstraction model for the transmission of data over the Internet, where packets get lost due to, for example, buffer overflows at intermediate routers. Erasure channels also find applications in wireless and satellite channels where deep fading events can cause the loss of one or several packets.

Reliable communication in data networks can be achieved by using an \ac{ARQ} mechanism in which the receiver requests the retransmissions of the information they have not been able to decode successfully.  However, \ac{ARQ} mechanisms present some limitations. The first is that they rely intensively on feedback. The second limitation enters into play in a reliable multicasting application, where one single transmitter wants to send an object (a file) to a set of receivers. In this scenario different receivers suffer different losses. If the number of receivers is large, the transmitter needs to process a large number of feedback messages and it also needs to perform a large number of retransmissions. For such applications, one would desire to have an alternative to \ac{ARQ} that does not rely so much on feedback and whose efficiency scales better with the number of users.

Probably, one of the first works proposing erasure coding as an alternative to \ac{ARQ} mechanisms is \cite{Metzer84:retransmission}, where an algorithm is proposed for the transmission of a file to multiple receivers. Instead of retransmitting lost packets, the transmitter sends redundancy packets until all receivers acknowledge the reception of the file. In that work \acl{RS} codes and linear random codes were considered, which become impractical due to their complexity  for medium-large block lengths\fran{, i.e., for block lengths exceeding the few thousands.}

Tornado codes were proposed for transmission over erasure channels \cite{luby97:PracticalLossRes, luby2001efficient}. Tornado codes have linear encoding and decoding complexity (under iterative decoding). However, the encoding and decoding complexity is proportional to their block lengths and not their dimension, \cite{luby97:PracticalLossRes}. Hence, they are not suitable for low rate applications such as reliable multicasting, where the transmitter needs to adapt its code rate to the user with the worst channel (highest erasure probability). Another family of codes with good performance over erasure channels are \ac{LDPC} codes. Several works have considered \ac{LDPC} codes over erasure channels \cite{oswald2002capacity,miller04:bec,paolini2012maximum} and they have been proved to be practical in several scenarios even under \acf{ML} decoding. \fran{For example, in \cite{paolini2012maximum} a decoding speed of up to 1.5 Gbps was reported for a $(2048,1024)$ \ac{LDPC} using \ac{ML} decoding.} However, \fran{for a fixed code dimension, the decoding complexity of \ac{LDPC} codes increases with the block length. Thus, as the erasure rate of the channel increases one is forced to increase the block length (i.e., decrease the rate), and the decoding complexity increases}.

Although solutions based on linear block codes usually outperform \ac{ARQ} mechanisms in the reliable multicasting setting, they still present some limitations. The first limitation is that the rate, and hence the block length, needs to be fixed a-priori. In the chosen rate turns out not to be low enough, it can happen that some users are unable to recover the original file. Furthermore, block codes usually need to be carefully designed taking  into account the information and block lengths. Thus, if one decides to change these parameters one usually needs to carry out a new code design.

The concept of a digital fountain was introduced in \cite{byers98:fountain} as an ideal solution to the problem of distributing data to a large number of users. A fountain code is basically an erasure code that is  able to generate a potentially endless amount of encoded symbols.  As such, fountain codes find application in contexts where the channel erasure rate is not known a priori. The first class of practical fountain codes, \ac{LT} codes, was introduced in \cite{luby02:LT}. \ac{LT} codes admit a sparse graph representation and can be decoded efficiently by means of iterative decoding when the code dimension (or number of input symbols, usually denoted by $k$) is large.  The main drawback of \ac{LT} codes is that in order to have a low probability of unsuccessful decoding, the encoding and iterative decoding cost per output/input\footnote{The encoding cost is defined as the encoding complexity in terms of operations normalized by the number of output symbols and the decoding cost as the decoding complexity normalized by the number of input symbols.} symbol has to grow at least logarithmically with the dimension of the code, $k$. Thus, \ac{LT} codes have a scalability problem. On the one hand we need the number of input symbols $k$ to be very large so that iterative decoding succeeds with high probability. On the other hand, by making $k$ large the encoding and iterative decoding cost increase.

Raptor codes were introduced in \cite{shokrollahi2001raptor} and published in \cite{shokrollahi04:raptor},\cite{shokrollahi06:raptor} as an evolution of \ac{LT} codes. They were also independently proposed in \cite{maymounkov2002online}, where they are referred to as online codes. Raptor codes consist of a serial concatenation of an outer block code, commonly referred to as precode, with an inner \ac{LT} code. The basic idea behind Raptor codes is relaxing the \ac{LT} code design, thus, requiring only the recovery of a fraction $1-\gamma$ of the input symbols, where $\gamma$ is usually small. This can be achieved with linear complexity, both in encoding and (iterative) decoding. The outer code is responsible for recovering the remaining fraction of input symbols, $\gamma$. If the precode is linear-time encodable, then the Raptor code has linear encoding complexity on the number of input symbols $k$, and therefore the overall encoding cost per output symbol is constant with respect to the number of input symbols $k$. If iterative decoding is used and the outer code can be decoded iteratively with linear complexity (in the number of input symbols $k$), the decoding complexity is also linear which results in a constant decoding cost per symbol. Furthermore, in \cite{shokrollahi06:raptor} it was shown that Raptor codes under iterative decoding are universally capacity-achieving on the binary erasure channel. This means that  a Raptor code can achieve the capacity of \emph{all} \acp{BEC}, no matter which value the erasure probability takes.
Thus, they can be used for transmission over an erasure channel whose erasure probability is unknown and they are still guaranteed to achieve capacity.

Both \ac{LT} and Raptor codes have been analyzed in depth under the assumption of iterative decoding and very large input blocks (at least in the order of a few tens of thousands symbols). However, often much smaller input block lengths are used due to different reasons.  For example, the decoders have sometimes limited memory resources allocated, the files to be delivered are often of smaller size, and sometimes a short latency is desired. This leads to the need of efficient short fountain codes. \fran{This is the reason why, for the Raptor codes standardized in 3GPP Multimedia Broadcast Multicast Service (MBMS) and IETF it is recommend to use between $1024$ and $8192$ input symbols (see \cite{MBMS12:raptor} and \cite{luby2007rfc} for more details)}. For these input block lengths, the performance under iterative decoding degrades considerably. In fact, these codes are decoded using an efficient \ac{ML} decoding algorithm known as inactivation decoding \cite{shokrollahi2005systems}.

The focus of this doctoral thesis is on the analysis and design of fountain codes under \ac{ML} decoding inspired by practical applications. Major parts of the results in this dissertation have been published in \cite{lazaro:ITW, lazaro:scc2015, lazaro:Allerton2015, lazaro2011concatenation, lazaro2013parallel, lazaro:ISIT2015, lazaro:JSAC, lazaro:Globecom2016, garrammone2013fragmentation}.

The remaining of this thesis is organized as follows. Chapter~\ref{chap:prelim} provides some preliminaries on erasure channels, block codes and fountain codes. The two main classes of fountain codes, \ac{LT} and Raptor codes are introduced in Chapter~\ref{chap:basics}.
In Chapter~\ref{chap:LT} \ac{LT} codes under inactivation decoding are considered. The main contribution of this chapter is an analysis of the decoding complexity of \ac{LT} codes under inactivation decoding using a dynamical programming approach.
Chapter~\ref{chap:Raptor} focuses on Raptor codes under inactivation decoding. First, an upper bound on the probability of decoding failure of Raptor codes under \ac{ML} decoding is presented. Then, a heuristic analysis of inactivation decoding is presented that provides an approximation of the number of inactivations.
Chapter~\ref{chap:Raptor_fixed_rate} contains several results related to the distance spectrum of an ensemble of fixed-rate Raptor codes.
In Chapter~\ref{chap:parallel} a novel fountain coding scheme is presented that consists of a parallel concatenation of a linear block code with a \acf{LRFC}. This scheme is particularly interesting when the outer code is a \acf{MDS} code.
Some concluding remarks are presented in Chapter~\ref{chap:conclusions}.
Appendix~\ref{app:practical} contains a comparison of the performance of the different inactivation techniques used in practice.
Finally, Appendix~\ref{app:proofs} contains some proofs that were omitted from Chapters~\ref{chap:Raptor} and \ref{chap:Raptor_fixed_rate}.

%% file: Chapter1/chapter1.tex

\chapter{Background} \label{chap:prelim}
\ifpdf
    \graphicspath{{Chapter1/Chapter1Figs/PNG/}{Chapter1/Chapter1Figs/PDF/}{Chapter1/Chapter1Figs/}}
\else
    \graphicspath{{Chapter1/Chapter1Figs/EPS/}{Chapter1/Chapter1Figs/}}
\fi
\epigraph{Everything should be made as simple as possible, but not simpler.}{Albert Einstein}

In this chapter we briefly introduce the communication channels that are considered in this thesis. Concretely, we present three different channels, the \acf*{BEC},  the \acf*{QEC} and the packet erasure channel. We then present some basic concepts related to block codes and fountain codes. Finally, the notation used in the thesis is described.

\section{Channel Models}

\subsection{The Memoryless Binary Erasure Channel}
The memoryless \acf*{BEC} \cite{Elias55:2noisy} is a communication channel with a binary input alphabet $\inputalpha = \{0,1\}$ and a ternary output alphabet  $\outputalpha = \{0,1, E\}$, as depicted in Figure~\ref{fig:BEC}. The symbol ``$E$'' denotes an erasure. Let $\inputrv \in \inputalpha$ be the random variable associated to the input of the channel
and $\outputrv \in \outputalpha$ be the random variable associated with the output of the channel. The transition probabilities of the channel are:
\begin{align*}
&\Pr(\outputrv = y | \inputrv = x) = 1 - \erasprob ,  \qquad \text {if } y = x, \\
&\Pr(\outputrv = E | \inputrv = x) = \erasprob.
\end{align*}
When the symbols ``$0$'' or ``$1$'' are received there is no uncertainty about the symbol transmitted. However, when symbol ``$E$'' is received the receiver does not know which symbol was transmitted.
\begin{figure}[t]
\begin{center}
\includegraphics[width=0.6\textwidth]{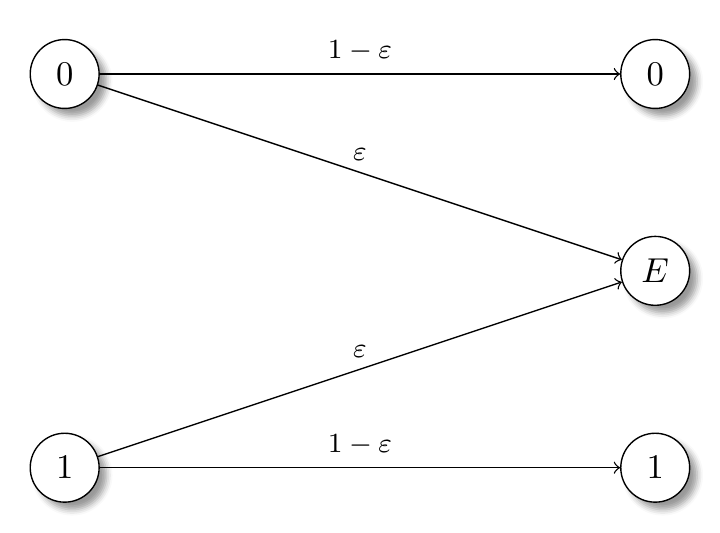}
\centering \caption[Binary erasure channel model]{The \acf*{BEC}.}
\label{fig:BEC}
\end{center}
\end{figure}

The capacity of the \ac{BEC} is
\[
C = 1 - \erasprob \,\,\,[\text{bits}/\text{channel use}],
\]
and it is attained with $\inputrv$ uniformly distributed.

\subsection{The $q$-ary Erasure Channel}
The \ac{QEC} is a communication channel with a q-ary input alphabet $\inputalpha = \{0,1,...,q-1\}$ and an output alphabet of cardinality $q+1$, ${\outputalpha = \{0,1,...,q-1, E\}}$, as depicted in Figure~\ref{fig:QEC}. Again, symbol ``$E$'' denotes an erasure. Let $\inputrv \in \inputalpha$ be the random variable associated to the input of the channel and $\outputrv \in \outputalpha$ be the random variable associated to the output of the channel. The transition probabilities of the channel are:
\begin{align*}
&\Pr(\outputrv = y | \inputrv = x) = 1 - \erasprob ,  \qquad \text {if } y = x, \\
&\Pr(\outputrv = E | \inputrv = x) = \erasprob.
\end{align*}
\begin{figure}[t]
\begin{center}
\includegraphics[width=0.6\textwidth]{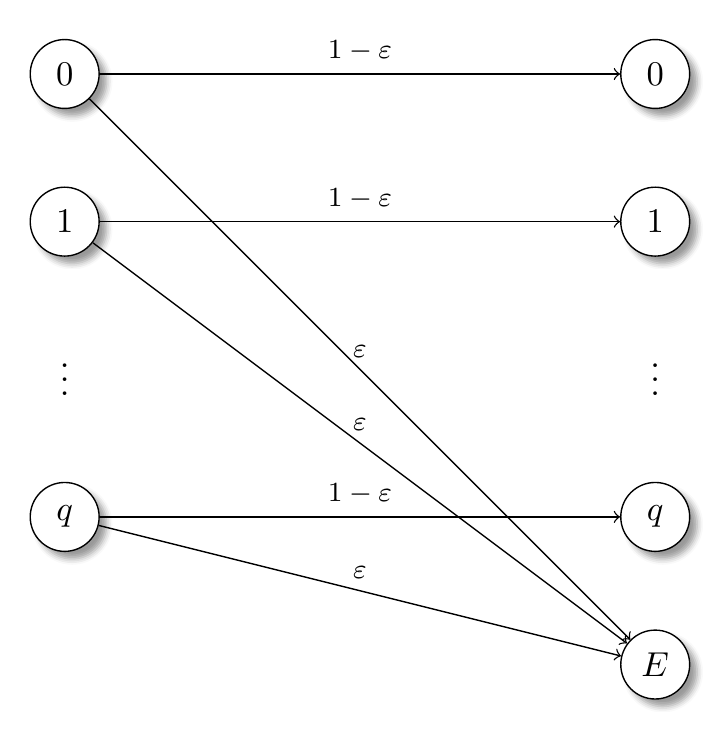}
\centering \caption[\acl{QEC} model]{The \acf*{QEC}.}
\label{fig:QEC}
\end{center}
\end{figure}
\noindent
The capacity of the \ac{QEC} is
\[
C = 1 - \erasprob \,\,\,[\text{symbols}/\text{channel use}],
\]
and
\[
C_b = \log_2(q) \, C \,\,\,[\text{bits}/\text{channel use}].
\]
The capacity is attained with $\inputrv$ uniformly distributed in $\inputalpha$.
\subsection{The Packet Erasure Channel}

The packet erasure channel is a communication channel in which the input is a packet, that is, an array of $\len$ symbols belonging to the alphabet $\{0,1\}$, i.e. $\inputalpha = \{0,1\}^\len$. Similarly to the \ac{BEC} and \ac{QEC}, in the packet erasure channel at the output the input is received error free with probability $1-\erasprob$, and an erasure is received with probability $\erasprob$.

The packet erasure channel can be seen as $\len$ parallel, fully correlated \acp{BEC} \cite{massey:81}. Thus, the capacity of the packet erasure channel is
\[
C = 1 - \erasprob \,\,\,[\text{packets}/\text{channel use}] ,
\]
and
\[
C_b = \len C \,\,\,[\text{bits}/\text{channel use}].
\]
Furthermore, all coding methods and performance bounds from the \ac{BEC} can be applied to the packet erasure channel with slight modifications.

The packet erasure channel has a great practical importance. For example, let us consider a satellite or terrestrial communication link.
 The data to be transmitted is usually split into packets and each of these packets is transmitted using a channel code at the physical layer. At the receiver side, channel decoding is performed at the physical layer in order to correct the errors introduced by the (physical) channel. After channel decoding  some residual errors might still be present. At this stage error detection is carried out and packets containing errors are marked as erased (discarded). It is easy to see how, \fran{under the assumption of perfect error detection}, the upper layers can abstract the behavior of the lower layers as a packet erasure channel.

The packet erasure channel can also be used to abstract the behavior of a computer data network such as the Internet. In this case, generally, the packets need to be forwarded through different intermediate nodes before reaching their destination. In this case, packet losses can occur due to, for example, a buffer overflow in some intermediate node. Additionally, during transmission bit errors can occur. Protocols (i.e. IP protocol) usually add a \ac{CRC} to each packet, that is used to detect and discard erroneous packets. All in all, the behavior of the data network can be abstracted by the upper layers as a packet erasure channel between the encoder and decoder.

Figure~\ref{fig:packet_eras} shows the block diagram of a typical digital communication system that makes use of erasure coding in a single link communication. At upper layers, a packet erasure channel encoder is used which accepts at its input $k$ source packets and generates $n$ output packets. Before transmission, each frame is protected by an erasure code.  At the receiver side channel decoding is performed at the physical layer in order to correct the errors introduced by the (physical layer) channel. After channel decoding  some residual errors might be present. At this stage error detection is carried out and packets containing errors are marked as erased (discarded). Next, this packets are passed on to the packet erasure channel decoder which then recovers the $k$ original source packets.

\begin{figure}[!ht]
\begin{center}
\includegraphics[width=1\textwidth]{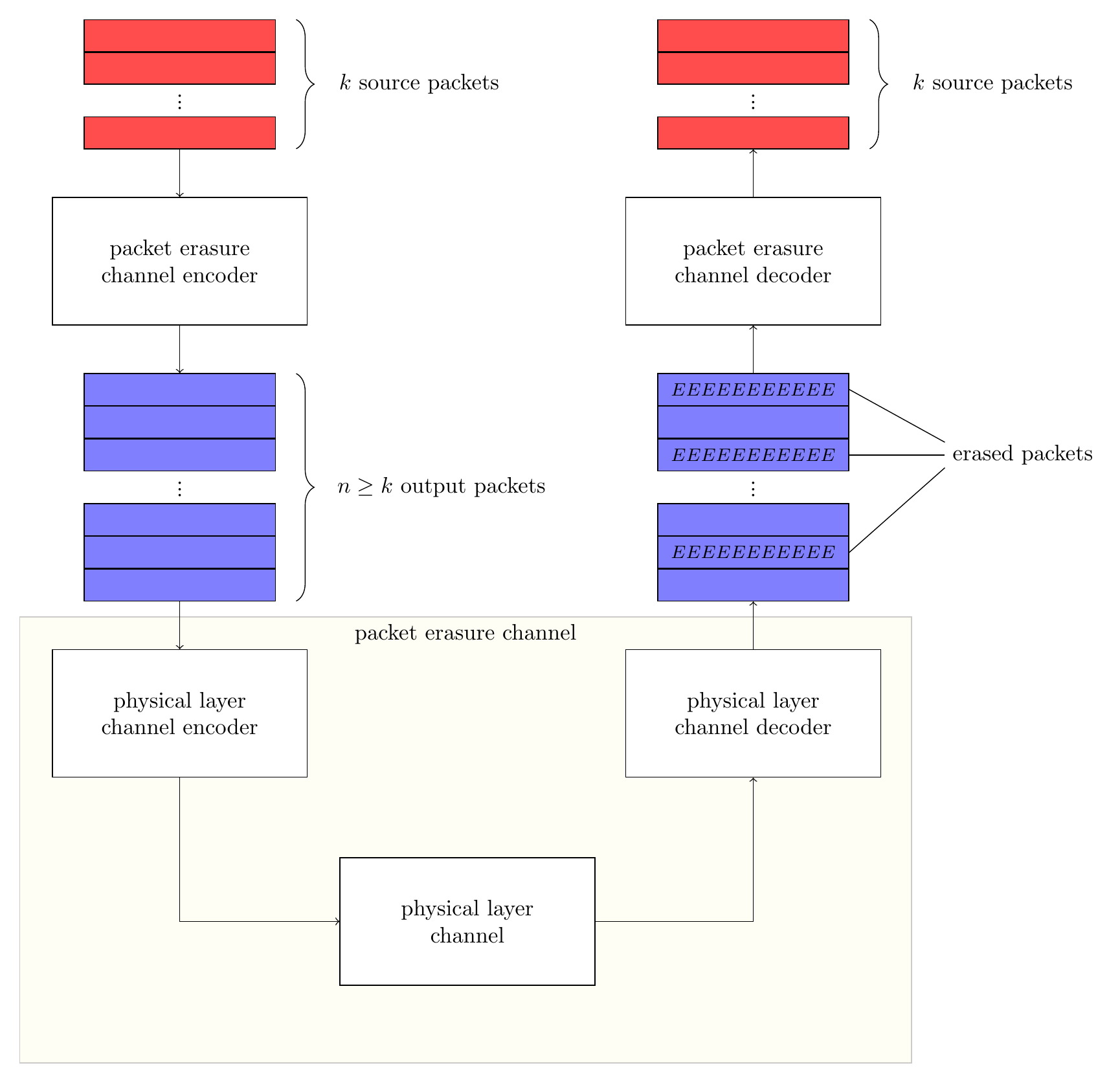}
\centering \caption[Packet erasure coding diagram]{Simplified diagram of a communication system that makes use of packet erasure coding.}
\label{fig:packet_eras}
\end{center}
\end{figure}

Due to  the easy mapping of the packet erasure channel to the \ac{BEC} and \ac{QEC}, for ease of exposition all the results in this thesis will be stated in the \ac{BEC}/\ac{QEC} setting, being the extension to the packet erasure channel straightforward. This approach is quite widespread in the recent literature of coding for erasure channels.

\section{Block Codes: Basics and Performance Bounds}

Consider the transmission over the \ac{BEC} with a $(n,k)$ binary linear block code $\code$. It is possible to show that the block error probability, $P_B$ satisfies the following inequality
\[
P_B (\code) \geq P^{(\mathsf S)}_{B},
\]
where $P^{(\mathsf S)}_{B}$ is the Singleton bound \cite{singleton1964maximum},
\begin{align}
P^{(\mathsf S)}_{B}(n,k,\erasprob)= \sum_{e=n-k+1}^n {n \choose e} \erasprob^e (1-\erasprob)^{n-e}.
\label{eq:bound_S}
\end{align}
In this bound, equality is achieved only if $\code$ is a $(n,k)$ \acf*{MDS} code, i.e., if the code minimum distance is:
\[
\dmin = n-k+1.
\]

Berlekamp derived an upper bound on the \emph{average} block error probability of random binary linear block codes \cite{berlekamp:bound}
\begin{align}\label{eq:bound_Berlekamp}
{P^{(\mathsf B)}_{B}}&{= \sum_{e=n-k+1}^n {n \choose e} \erasprob^e (1-\erasprob)^{n-e}} \nonumber \\
&{+\sum_{e=1}^{n-k} {n \choose e} \erasprob^e (1-\erasprob)^{n-e} \,2^{-(n-k-e)}}.
\end{align}
If we compare \eqref{eq:bound_S} and \eqref{eq:bound_Berlekamp} we can see how the Berlekamp bound is composed of the Singleton bound plus a correction term.

Let us denote by $\code^*$ the best code among all $(n,k)$ binary linear block codes, where by best we mean the one with the minimum block error probability over a \ac{BEC}. We have that:
\[
P^{(\mathsf S)}_{B} \leq  P_B (\code^*) < P^{(\mathsf B)}_{B}.
\]
That is, the Singleton and the Berlekamp bounds provide lower and upper bounds to the block error probability of the best binary linear block code with parameters $(n,k)$.

The block error probability of a linear block code not only depends on its minimum distance, $\dmin$, but also on its weight enumerator, $\wesingle_w$, that corresponds to the number of codewords of Hamming weight $w$. Unfortunately, \fran{when dealing with modern (turbo/ \ac{LDPC}) codes,} deriving the exact weight enumerator of a code is a very challenging problem \cite{berlekamp78:intractability}. For this reason it is convenient to work with code ensembles since it is usually easier to derive average results for the ensemble.

A code ensemble $\codensemble$ is a set of codes $\codensemble = \{ \code_1, \code_2, \hdots, \code_m \} $ together with a probability distribution that gives the probability of the occurrence of each of the codes in the ensemble. We will illustrate the concept of code ensemble by means of an example.
\begin{example}
The $(n,k)$ binary linear random ensemble  is given by all possible codes obtained by generating at random a $(n-k) \times n$ parity check matrix $\mathbf{H}$ in which each element of the parity check matrix takes value one with probability 1/2. This ensemble contains all $(n,k')$ codes with $k'\leq k$, \fran{since the rank of $\mathbf{H}$ can be smaller than $n-k$}.
\end{example}
Let us consider a binary linear block code ensemble $\codensemble = \{ \code_1, \code_2, \hdots, \code_m \} $. The ensemble average weight enumerator $\we_w$ is defined as
\[
\we_w =\Exp_{\code \in \codensemble} \left[ \wesingle_w (\code)\right],
\]
where $\Exp_{\code \in \codensemble}[\cdot]$ denotes expectation over all the codes $\code$ in the ensemble $\codensemble$, and $\wesingle_w (\code)$ is the weight enumerator of code $\code$.

Consider a binary linear block code ensemble $\codensemble$ with average weight enumerator $\we_w$.  The average block error probability for codes in the ensemble,
$P_B (\codensemble)$, can be upper bounded as \cite{CDi2001:Finite}
\begin{align}\label{eq:bound_Gavg}
P_B (\codensemble) &= \Exp_{\code \in \codensemble}  \left[P_B(\code)\right]\leq
P^{(\mathsf S)}_{B}(n,k,\erasprob) \nonumber \\ & + \sum_{e=1}^{n-k} {n \choose e} \erasprob^e (1-\erasprob)^{n-e} \min \left\{1, \sum_{w=1}^e {e \choose w} \frac{\A _w}{{n \choose w}}\right\}.
\end{align}

\section{Fountain Codes: Basics and Performance Bounds}\label{sec:fountain_bounds}
\fran{Consider a fountain code $\codef$ of dimension $k$. The fountain encoder }receives at its input $k$ input symbols (also called source symbols) out of which it generates $n$ output symbols (also called coded symbols).
The key property of a fountain code is that the number of output symbols $n$ does not need to be fixed a-priori. Additional output symbols can be generated \emph{on the fly} in an on-demand fashion. For this reason, fountain codes are said to be rateless.

We consider the transmission over an erasure channel with a fountain code with $k$ input symbols. In this setting, the output symbols generated by the fountain encoder are transmitted through an erasure channel where they are erased with probability $\erasprob$. We denote by $m$ the number of output symbols that are not erased by the channel at a given receiver. We define the absolute (receiver) overhead as:
\[
\absoverhead := m-k.
\]
We also define the relative overhead $\reloverhead$ as the absolute overhead normalized by the number of input symbols, formally:
\[
\reloverhead := \frac{\absoverhead}{k} = m/k -1.
\]

Given the fact that fountain codes are rateless ($n$ not fixed) it is useful to define the performance bounds of fountain codes in terms of the absolute receiver overhead. More concretely, we are interested in bounds to the probability of decoding failure as a function of the absolute receiver overhead, $\Pf(\absoverhead)$.

A lower bound to the performance of fountain codes is obtained assuming an \emph{ideal} fountain code that allows the receiver to decode successfully whenever $m\geq k$ output symbols are received, \fran{i.e., whenever $\absoverhead \geq 0$}. The performance on an ideal fountain code is, hence, given by:
\[
\Pf^{\text{I}}(\absoverhead) =
\begin{cases}
1 & \absoverhead < 0 \\
0 & \absoverhead \geq 0
\end{cases} .
\]
Thus, for any given fountain code $\codef$ its decoding failure probability can be lower bounded as
\[
\Pf(\codef, \absoverhead) \geq \Pf^{\text{I}} (\absoverhead)
\]

Let us consider a \acf{LRFC}\footnote{\acfp{LRFC} are defined in Section~\ref{sec:lfrc}.} on a finite field of order $q$. In \cite{Liva2013} it was shown how the probability of decoding failure of an \ac{LRFC} can be upper bounded as
\begin{equation}
\Pf(\absoverhead) < \frac{1}{q-1} q^{-\absoverhead}, \qquad \absoverhead \geq 0
\label{eq:lrfc_upper}
\end{equation}

Let us now denote by $\codef^*$ the best code among all $q$-ary fountain codes with $k$ input symbols, where by best we mean the one with the minimum block error probability over a \ac{QEC}. We have that:
\[
\Pf^{\text{I}} (\absoverhead) \leq \Pf(\codef^*, \absoverhead) < \frac{1}{q-1} q^{-\absoverhead}, \qquad \absoverhead \geq 0
\]
That is, the performance of an ideal fountain code and the bound in \eqref{eq:lrfc_upper} provide lower and upper bounds to the probability of decoding failure of the best $q$-ary fountain code with $k$ input symbols, \fran{when used to transmit over a $q$-ary erasure channel.}

\section{ Notation}


\fran{In this section we introduce several definitions which will be used throughout the thesis.
}

\begin{mydef}[$\mathcal{O}$-notation]Let $f$ and $g$ be two real functions. We write:
\[
f(n) = \mathcal{O} \left( g(n) \right), \, \text{ as } n \rightarrow \infty,
\]
if for \fran{sufficiently large values of} $n$, there exists a constant $s$ so that
\[
|f(n)|  \leq s |g(n)|.
\]
\end{mydef}
For example, if a function $f$ is $\mathcal{O}(\log(n))$, given $n$, we can find a value $s$ such that $f$ is upper bounded by $s \log(n) $ for sufficiently large $n$.
This notation is also known as Landau notation and it is employed to characterize the behaviour of a function when its argument tends to infinity \cite{Graham:1994}.

Another useful asymptotic notation is the small $\mathrm o$-notation whose formal definition is introduced next.
\begin{mydef}[$\mathrm o$-notation]
Let $f$ and $g$ be two real functions. We write:
\[
f(n) = \mathrm o \left( g(n) \right), \, \text{ as } n \rightarrow \infty,
\]
if and only if for any constant $s>0$ and sufficiently large $n$
\[
|f(n)|  \leq s |g(n)|.
\]
\end{mydef}

\fran{Note that although the definitions of $\mathcal{O}$-notation and $\mathrm o$-notation are similar, they are not equivalent. For example, consider $f(n) = n^2$. We can say that $n^2$ is $\mathcal{O}(n^2)$ but this would not be true for little $\mathrm o$-notation.}

\begin{mydef}[Exponential equivalence]
Two real-valued positive  sequences $a(n)$ and $b(n)$ are said to be exponentially equivalent \cite{CoverThomasBook}, writing $a(n)~\doteq~b(n)$, when
\begin{equation}\label{eq:asymp_eq}
\lim_{n \to \infty} \frac{1}{n} \log_2 \frac{a(n)}{b(n)}=0.
\end{equation}
\end{mydef}
If $a(n)$ and $b(n)$ are exponentially equivalent, then \begin{align}
\lim_{n \to \infty} \frac{1}{n} \log_2 a(n) = \lim_{n \to \infty} \frac{1}{n}  \log_2 b(n).
\end{align}

%




%% file: Chapter2/chapter2.tex
\chapter{Linear Random Fountain Codes, \ac{LT} and Raptor Codes }\label{chap:basics}
\ifpdf
    \graphicspath{{Chapter2/Chapter2Figs/PNG/}{Chapter2/Chapter2Figs/PDF/}{Chapter2/Chapter2Figs/}}
\else
    \graphicspath{{Chapter2/Chapter2Figs/EPS/}{Chapter2/Chapter2Figs/}}
\fi

Within this chapter we present  three fountain code constructions that can be found in literature. First we introduce \acfp*{LRFC}, which are probably the conceptually simplest fountain code one can think of. We then introduce \ac{LT} codes, and describe their encoding and decoding procedures. Finally, we introduce Raptor codes, which are arguably the best performing fountain coding scheme known.

\section{Linear Random Fountain Codes}\label{sec:lfrc}

\fran{For the sake of completeness, let us start by formally defining a Galois Field
\begin{mydef}[Galois Field]
We denote by $\mathbb {F}_q$ a Galois field or finite field of order $q$. A Galois Field $\mathbb {F}_q$ is a set of $q$ elements on which the addition and multiplication operations fulfil the following properties:
\begin{enumerate}[label=(\alph*)]
  \item $\mathbb {F}_q$ is an Abelian group under addition with identity element denoted by $0$.
  \item $\mathbb {F}_q \backslash \{0\}$ is a multiplicative group with identity element denoted by $1$.
  \item multiplication is distributive over addition
\end{enumerate}
\end{mydef}
}

\fran{A $q$-ary \acf{LRFC} is a fountain code that accepts at its input a set of $k$ input (or source) symbols, ${\mathbf{v}=(v_1,~v_2,~\ldots, v_k)}$, where $v_i \in \mathbb {F}_q$. At its output, the \acl{LRFC} encoder can generate an unlimited amount of output symbols (also known as coded symbols)  $\mathbf{\osymb}=(\osymb_1, \osymb_2, \ldots, \osymb_n)$, where $n$ can grow indefinitely and $\osymb_i \in \mathbb {F}_q$.}  The $i$-th output symbol $\osymb_i$ is generated as:
\[
\osymb_i = \sum_{j=1}^k g_{j,i} v_j,
\]
where the coefficients $g_{j,i}$ are picked from $\mathbb {F}_q$ with uniform probability.
\fran{If we assume $n$ to be  fixed, \ac{LRFC} encoding can be seen as a vector matrix multiplication:
\[
\mathbf{\osymb} = \mathbf{v} \GLTa,
\]
where $\GLTa$ is an $k \times n$ with elements $g_{j,i}$ picked uniformly at random from $\mathbb {F}_q$.}

\fran{Let us now assume that the output symbols produced by the \ac{LRFC} encoder are transmitted over a $q$-ary erasure channel, and let us also assume that out of the $n$ output symbols generated by the \ac{LRFC} encoder, the receiver collects $m=k+\absoverhead$, denoted by $\mathbf{\rosymb}=(\rosymb_1, \rosymb_2, \ldots, \rosymb_m)$.  Denoting by $\mathcal{I} = \{i_1, i_2, \hdots, i_m \}$ the set of indices corresponding to the $m$ non-erased symbols, we have
\[
\rosymb_j = \osymb_{i_j}.
\]
We can now cast the received output symbols as
\begin{align}
\mathbf{\rosymb} = \mathbf{v} \Grx
\label{eq:lrfc_ml_eq_sys}
\end{align}
with $\Grx$ given by the $m$ columns of $\G$ with indices in $\mathcal{I}$.
}

\fran{\ac{LRFC} decoding is performed by solving the system of equations in \eqref{eq:lrfc_ml_eq_sys}. Note that matrix $\Grx$ is dense, since its elements are picked uniformly at random in $\mathbb {F}_q$. Due to the high density of $\Grx$ \ac{LRFC} decoding is quite complex; hence, \acp{LRFC} are only practical for small values of $k$ (at most in the order of the hundreds).}

The performance of these codes is remarkably good and follows a relatively simple model.
Under \ac{ML} decoding,  the decoding failure probability of a binary \ac{LRFC} {\cite{shokrollahi06:raptor,Medard08:ARQ}} can be accurately modeled as $\Pf\sim 2^{-\absoverhead}$ for $\absoverhead \geq 0$. Actually, $\Pf$ can be upper bounded by $2^{-\absoverhead}$ {\cite{berlekamp:bound,shokrollahi06:raptor,Medard08:ARQ}}.

In \cite{Liva10:fountain}, \ac{LRFC} on finite fields of order {equal or larger} than $2$ ($\mathbb {F}_q$, $q\geq2$) were analyzed. It was shown that for an \ac{LRFC} over $\mathbb {F}_q$, the failure probability under \ac{ML} decoding is bounded as \cite{Liva10:fountain}
\begin{align}\label{eq:tightbounds}
q^{-\absoverhead-1}\leq \Pf (\absoverhead,q) < \frac{1}{q-1}q^{-\absoverhead}
\end{align}
where both bounds are tight already for $q=2$, and become tighter for increasing $q$.

\section{LT codes}\label{sec:LT_intro}
\acf{LT} codes were introduced in \cite{luby02:LT} as the first practical implementation of a fountain code. They were originally introduced together with an iterative decoding algorithm that will be explained in detail in Section~\ref{chap:LT_iterative}.

An \ac{LT} code accepts at its input a set of $k$ symbols, ${\mathbf{v}=(v_1,~v_2,~\ldots, v_k)}$, that are commonly referred to as input symbols (or source) symbols. At its output, the \ac{LT} encoder can generate an unlimited amount of output symbols (also known as coded symbols)  $\mathbf{\osymb}=(\osymb_1, \osymb_2, \ldots, \osymb_n)$, where $n$ can grow indefinitely.
\fran{A key concept when dealing with \ac{LT} codes is the degree of an output symbol or output degree, which is defined as the number of input symbols that were used to generate the output symbol under consideration.}
An \ac{LT} code is defined by an output degree distribution ${\Omega= (\Omega_1,\Omega_2,\Omega_3,\hdots,\Omega_{d_{\mathrm{max}}})}$, where $\Omega_d$ corresponds to the probability that an output symbol of degree $d$ is generated, and  $d_{\mathrm{max}}$ is the maximum output degree.

In order to generate one output symbol the \ac{LT} encoder performs the following steps:
\begin{itemize}
\item {Randomly choose a degree $d$ according to the degree distribution $\Omega$.}
\item {Choose uniformly at random $d$ distinct input symbols.}
\item {Compute the output symbol as a xor of the $d$ selected input symbols.}
\end{itemize}

If we assume for a moment that the number of output symbols $n$ is fixed, the \ac{LT} encoding operation can be seen as a vector matrix multiplication:
\[
\mathbf{\osymb} = \mathbf{v} \GLTa,
\]
where $\GLTa$ is an $k \times n$ binary\footnote{Unless otherwise stated we will always consider binary \ac{LT} codes.} matrix which defines the relation between the input and the output symbols. The element $g_{i,j}$ of $\GLTa$ is set to one only if input symbol $i$ was used to generate output symbol $j$. Otherwise, element $g_{i,j}$ is set to zero. \fran{From this description it is easy to see how binary \acp{LRFC} can be considered a particular type of \ac{LT} code in which the output degree distribution corresponds to a binomial distribution with parameters $k$ and $1/2$.}

\ac{LT} codes admit a bipartite graph representation. In the bipartite graph of an \ac{LT} code there are two different types of nodes, corresponding to input and output symbols. Let us introduce the notation $\deg(\osymb)$ to refer to the degree of an output symbol $\osymb$. An output symbol of degree $d$ will have $d$ neighbors in the bipartite graph.  We will use the notation $\neighbourof{\cdot}$ to denote the set of neighbours, i.e. the  neighbourhood of a node.

The bipartite graph of an \ac{LT} code is related to its matrix representation, and can be derived from $\GLTa$. We will illustrate this by means of an example. Figure~\ref{fig:ltgraph} shows the bipartite graph representation of an \ac{LT} code with $k= 5$ input symbols and $n=8$ output symbols. In the figure,
input symbol are represented by red circles and output symbol using blue squares. The generator matrix of the \ac{LT} code represented in the figure corresponds to
\[
\GLTa=
\begin{pmatrix}
  1 & 1 & 1 & 1 & 0 & 0 & 0 & 0 &\\
  0 & 1 & 1 & 0 & 0 & 0 & 0 & 0 &\\
  0 & 1 & 1 & 0 & 0 & 0 & 0 & 1 &\\
  1 & 0 & 0 & 1 & 1 & 1 & 1 & 0 &\\
  0 & 1 & 0 & 0 & 0 & 0 & 1 & 1 &\\
\end{pmatrix}.
\]

\begin{figure}
\begin{center}
\includegraphics[width=1\textwidth]{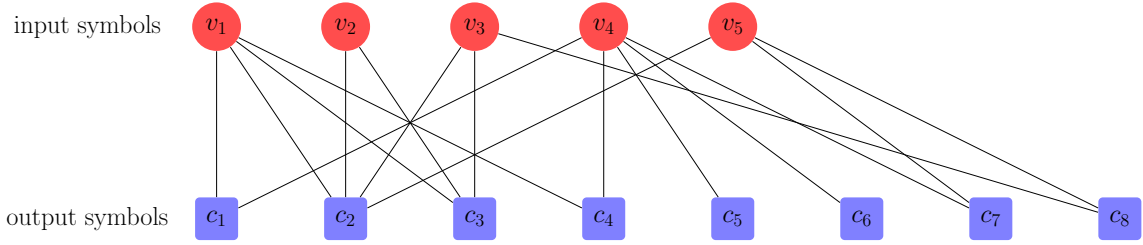}
\centering \caption[Bipartite graph of an \acs{LT} code]{Bipartite graph of an \ac{LT} code.}
\label{fig:ltgraph}
\end{center}
\end{figure}

An important parameter of an \ac{LT} code is its average output degree $\bar \Omega$, that is given by
\[
\bar \Omega := \sum_{i=1}^{\dmax} i \Omega_i.
\]
In \ac{LT} literature, degree distributions are commonly represented in polynomial form. Given a degree distribution $\Omega$, its polynomial representation $\Omega(\x)$ is given by
\[
\Omega(\x) :=  \sum_{i=1}^{\dmax}  \Omega_i \x^i.
\]
This representation can be used to derive moments of the degree distribution (that is a probability mass function) in a very compact form. For example, the average output degree can be expressed as \fran{the first derivative of $\Omega(\x)$ evaluated at $1$,}
\[
\bar \Omega = \Omega'(1).
\]

\subsection{Iterative Decoding}\label{chap:LT_iterative}

\ac{LT} codes were introduced in \cite{luby02:LT} together with a suboptimal, low complexity decoding algorithm.  Although a more proper name for it would be that of peeling decoder, this decoder is usually referred to as iterative decoder. In this thesis we will use the terms iterative decoding and peeling decoder interchangeably.

Iterative decoding of \ac{LT} codes is best described using a bipartite graph. Let us assume that the receiver has collected $m=k+\absoverhead$ output symbols that we will denote by $\mathbf{\rosymb}=(\rosymb_1, \rosymb_2, \ldots, \rosymb_m)$. We will consider a bipartite graph containing the $m$ collected output symbols, $\mathbf{\rosymb}$, and the $k$ input symbols $\mathbf{v}$.

\begin{algo}[Iterative decoding]
{~}
\begin{enumerate}
 \item Search for an output symbol of degree one.
 \begin{enumerate}
    \item If such an output symbol $\rosymb$ exists move to step 2.
    \item If no output symbols of degree one exist, iterative decoding exits and decoding fails.
 \end{enumerate}
 \item Output symbol $\rosymb$ has degree one. Thus, denoting its only neighbour as $v$, the value of $v$ is recovered by setting $v= \rosymb$.
 \item Denoting by $\neighbourof{v}$ the set of neighbours of $v$. For each $\rosymb \in \neighbourof{v}$:
  \begin{enumerate}
  \item Update the value of $\rosymb$ as:  ${\rosymb = \rosymb  + v}$, \fran{where $+$ denotes addition over $\mathbb {F}_2$}.
  \item Remove input symbol $v$ and all its attached edges from the graph.
 \end{enumerate}
\item If all $k$ input symbols have been recovered, decoding is successful and iterative decoding ends. Otherwise, go to step 1.
\end{enumerate}
\label{alg:BP}
\end{algo}

In order to illustrate iterative decoding we will provide a small example. Figure~\ref{fig:lt_iter_step_0} shows the bipartite graph before iterative decoding starts. We can see that the number of source symbols is $k=4$ and the number of output symbols collected by the receiver (not erased by the channel) is $n=5$.

\begin{figure}[h!]
\begin{center}
\includegraphics[width=0.81\textwidth]{ltgraph_iterative_0}
\centering \caption{Iterative decoding example, step 0.}
\label{fig:lt_iter_step_0}
\end{center}
\end{figure}
\FloatBarrier

Iterative decoding starts by searching for a degree one output symbol. In Figure~\ref{fig:lt_iter_step_1} we can see that output symbol  $\rosymb_3$ is the only output symbol with degree one. Using $\rosymb_3$ the decoder recovers $v_2$. Afterwards, the decoder performs the xor (addition over $\mathbb {F}_2$) of $v_2$ with all its neighbors. After doing so all edges attached to $v_2$ are erased.

\begin{figure}[!ht]
\begin{center}
\includegraphics[width=0.9\textwidth]{ltgraph_iterative_1}
\centering \caption{Iterative decoding example, step 1.}
\label{fig:lt_iter_step_1}
\end{center}
\end{figure}
\FloatBarrier

The second run of iterative decoding is shown in Figure~\ref{fig:lt_iter_step_2}. The decoder finds the only degree one output symbol $\rosymb_1$, and uses it to recover $v_1$. Next, the decoder performs the xor (addition over $\mathbb {F}_2$) of $v_1$ with its other neighbor, $\rosymb_4$, and erases the edges attached to $v_1$.

\begin{figure}[!ht]
\begin{center}
\includegraphics[width=0.81\textwidth]{ltgraph_iterative_2}
\centering \caption{Iterative decoding example, step 2.}
\label{fig:lt_iter_step_2}
\end{center}
\end{figure}
\FloatBarrier

Figure~\ref{fig:lt_iter_step_3} depicts the third iteration. We can see how the only degree one output symbol $\rosymb_4$ is used to solve $v_4$. Then the decoder performs the xor of $v_4$ to its other neighbor, $\rosymb_2$ and the edges are removed from the graph.

\begin{figure}[!ht]
\begin{center}
\hspace{7mm}
\includegraphics[width=0.87\textwidth]{ltgraph_iterative_3}
\centering \caption{Iterative decoding example, step 3.}
\label{fig:lt_iter_step_3}
\end{center}
\end{figure}
\FloatBarrier

Finally, the last iteration is shown in Figure~\ref{fig:lt_iter_step_4}. Now there are two degree one output symbols, $\rosymb_2$ and $\rosymb_5$. In this case we assume the decoder chooses at random  $\rosymb_2$ to recover the last input symbol $v_3$.
\begin{figure}[!ht]
\begin{center}
\includegraphics[width=0.81\textwidth]{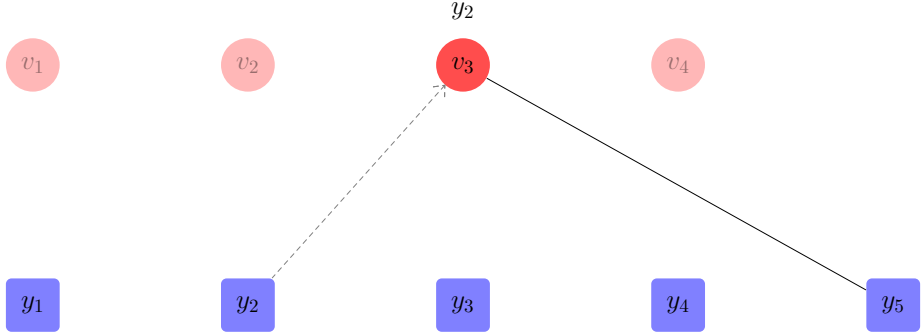}
\centering \caption{Iterative decoding example, step 4.}
\label{fig:lt_iter_step_4}
\end{center}
\end{figure}
\FloatBarrier

The following proposition (\cite{luby02:LT}) provides a necessary condition for decoding to be successful with high probability.
\begin{prop}\label{prop:lt_logk}
A necessary condition for decoding to be successful with high probability is $\bar \Omega=\mathcal{O} \left(\log (k)\right)$.
\end{prop}
\begin{proof}
The proof uses the ``balls into bins'' argument that was presented in \cite{luby02:LT}.
Let us first assume that $k$ and $m$ are very large and let us assume that at encoding each output symbol chooses its neighbors with replacement\footnote{This means that an output symbol will be allowed to choose multiple times the same neighbor. However, this will happen with a negligible probability for large enough values of $k$.}.
Let us consider a randomly chosen input symbol $v$ and an output symbol $\rosymb$ of degree $d$. The probability that $v$ is not in the neighborhood of $\rosymb$ corresponds to:
\[
\Pr \{ v \notin \neighbourof{y} | \deg(y) = d \} = \left( \frac{k-1}{k} \right)^d.
\]
Let us denote by $P_e$ the probability that $v$ does not have any edges to the $m$ received symbols. This probability corresponds to the probability of $v$ not belonging to the union of the neighborhoods of the $m$ received output symbols. Under the replacement assumption we have that
\[
P_e = \Pr \left\{ v \notin \bigcup_{i=1}^{m} \neighbourof{y_i} \mathlarger{\mathlarger{|}} \sum_{i=1}^m \deg(y_i)  \right\}  =  \left( \frac{k-1}{k} \right)^{\sum_{i=1}^m \deg(y_i)} 
\]
\fran{If we now let $k$ tend to infinity, we have
\[
\lim_{k \to \infty } P_e =  e^{-\bar \Omega (1+\reloverhead) }
\]
where we have made use of the relationship
\[
\lim_{k \to \infty } \left( \frac{k-1}{k} \right)^k = e^{-1}.
\]}
Let us denote by $\nuncovered$ the expected number of input symbols not covered by any output symbol,
\[
\nuncovered= k e^{-\bar \Omega (1+\reloverhead) }.
\]
A necessary condition for successful decoding with high probability is that the $\nuncovered$ is vanishingly small. If we relax this condition and let $\nuncovered$ simply be a small positive number, we have
\[
\bar \Omega= \frac{ \log( k / \nuncovered) } {1+\reloverhead}.
\]
This leads us to the statement in the proposition.
\end{proof}
Note that the condition in Proposition~\ref{prop:lt_logk} is valid for any decoding algorithm and not only for iterative decoding.

The performance of \ac{LT} codes under iterative decoding has been object of study in several works and is well understood, \cite{Karp2004, maneva2006new, Maatouk:2012, shokrollahi2009theoryraptor}.
\fran{Iterative decoding of \ac{LT} codes can be seen as an iterative pruning of the bipartite graph of the \ac{LT} code. If we take an instance of decoding in which iterative decoding is successful, we have that initially  all input symbols are unresolved (not yet decoded). At every iteration exactly one input symbol is resolved and all edges attached to the resolved input symbol are erased from the graph. Decoding continues until all input symbols are resolved, which is the case after $k$ iterations. Let us consider the iterative decoder at some intermediate step in which $u$ input symbols are yet unresolved and $k-u$ symbols have already been resolved. Following \cite{Karp2004} we shall introduce some definitions that provide an insight into the iterative decoding process.}
\begin{mydef}[Reduced degree] We define the reduced degree of an output symbol as the degree of the output symbol in a reduced bipartite graph in which only unresolved input symbols are present. \label{def:reduced}
\end{mydef}
Thus, at the initial stage of iterative decoding, when all input symbols are unresolved, the reduced degree of a symbol is equal to its actual degree. However, as iterative decoding progresses the reduced degree of an output symbol decreases if his neighbors get resolved.
\begin{mydef}[Output ripple] We define the output ripple or simply ripple as the set of output symbols of reduced degree 1 and we denote it by $\rippleset$. \label{def:ripple}
\end{mydef}
\begin{mydef}[Cloud] We define the cloud as the set of output symbols of reduced degree $d\geq 2$ and we denote it by $\cloudset$.\label{def:cloud}
\end{mydef}
Figure~\ref{fig:example_cloud_ripple} shows the bipartite graph of an \ac{LT} code in which $4$ input symbols are unresolved. It can be observed how output symbols $\rosymb_1$ and $\rosymb_4$ belong to the ripple since they have reduced degree one and output symbols $\rosymb_2$ and $\rosymb_3$ belong to the cloud since their degree is 2 or larger.
\begin{figure}[t]
\begin{center}
\includegraphics[width=0.6\columnwidth]{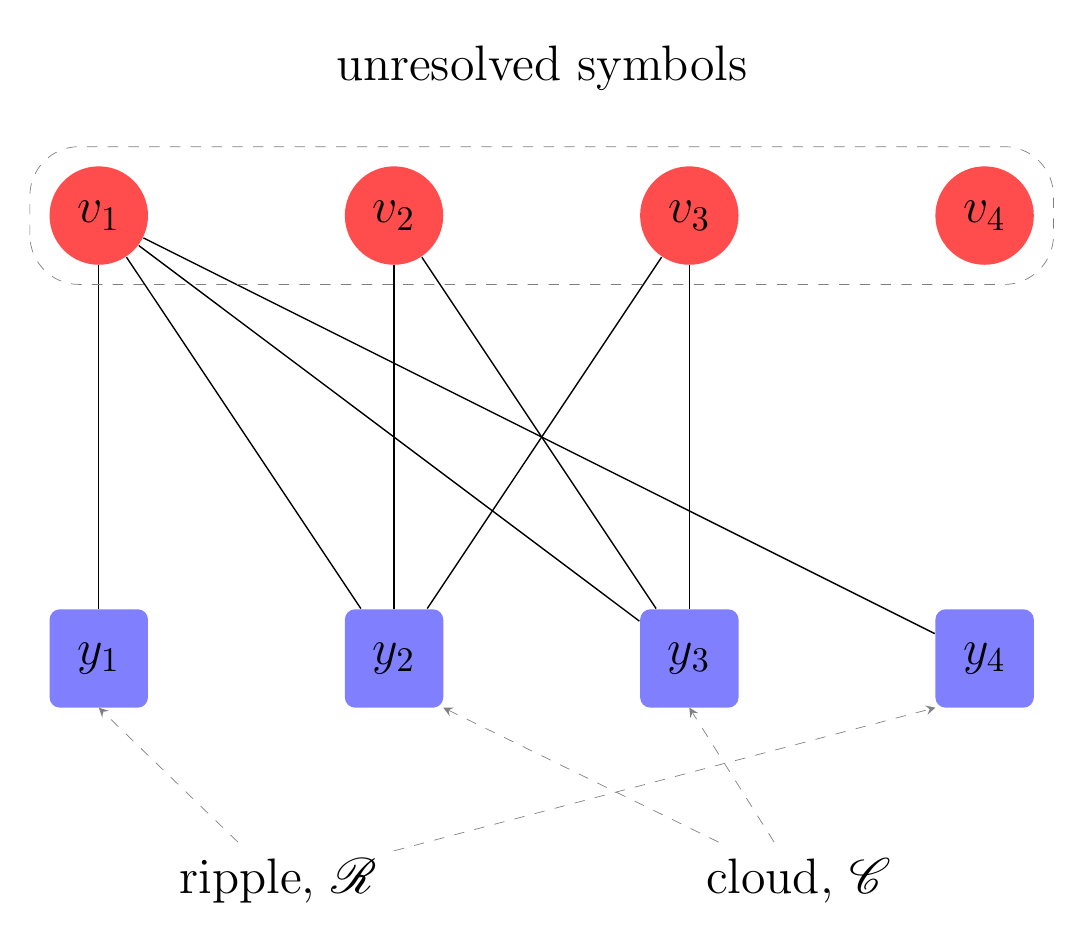}
\centering \caption{Example of ripple and cloud in the bipartite graph of an \ac{LT} code.\label{fig:example_cloud_ripple}}
\end{center}
\end{figure}

It is easy to see how during the iterative decoding process, after every iteration at least one symbol leaves the ripple (assuming decoding is successful). Moreover, at each iteration some output symbols might leave the cloud and enter the ripple if their reduced degree decreases from two to one. Note also that iterative decoding fails if the ripple becomes empty before $k$ iterations. Thus, if one is able to track the size of ripple it is possible to derive the performance of \ac{LT} codes under iterative decoding.  In \cite{Karp2004} a finite length analysis of \ac{LT} codes is proposed that models the iterative decoder as a finite state machine, based on a dynamic programming approach. The full proof of the analysis in \cite{Karp2004}, that was published only in abstract form, can be found in \cite{shokrollahi2009theoryraptor}. This analysis can be used to derive the error probability of the iterative decoder and it also allows to compute the first order moments of the ripple and the cloud. This analysis was extended in \cite{Maatouk:2012}, where the second moment of the ripple size was analyzed. In \cite{maneva2006new} another analysis of \ac{LT} codes under iterative decoding is proposed that has lower complexity and is based on the assumption that the number of output symbol collected by the receiver follows a Poisson distribution.

\subsubsection{Degree Distributions}

In this section we present the two best well known degree distributions, the ideal soliton distribution and the robust soliton distribution. Both distributions were designed for iterative decoding.

\vspace{0.8cm}
\textbf{Ideal Soliton Distribution}
\newline
The first distribution we will present is known as ideal soliton distribution \cite{luby02:LT} and is based on these two design principles:
\begin{itemize}
\item{The expected number of output symbols in the ripple at the start of iterative decoding is one.}
\item{The expected number of output symbols leaving the cloud and entering the ripple is one at every iteration.}
\end{itemize}
Thus, the expected ripple size is $1$ during the whole decoding process.
The ideal soliton distribution, which  we denote by $\isd$, has the following expression.
\begin{align}
\isd_d=
\begin{cases}
\frac{1}{k} & d=1 \\
\frac{1}{d \left( d-1\right)} & 1 < d \leq k.
\end{cases}
\end{align}
\fran{Note that the distribution varies with the number of input symbols $k$.} The average output degree of the ideal soliton distribution is \cite{luby02:LT}
\[
\avgisd = H(k)
\]
where $H(k)$ is the harmonic sum up to $k$:
\[
H(k) = \sum_{i=1}^{k} \frac{1}{i}.
\]
Since, the harmonic sum can be approximated as ${H(k) \approx \log(k)}$, we can approximate the average output degree of  $\isd$ as
\[
\avgisd \approx \log (k),
\]

For illustration we provide a plot of the ideal soliton distribution for $k= 50$ in Figure~\ref{fig:isd}.
\begin{figure}[t]
\begin{center}
\includegraphics[width=\figwbigger]{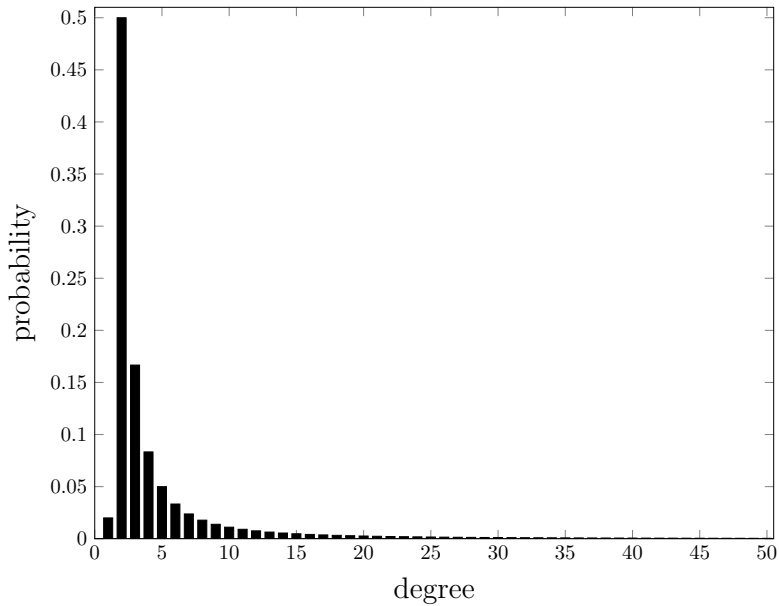}
\centering \caption[Ideal soliton distribution example, $k=50$]{Ideal soliton distribution, $\isd$, for $k=50$.}
\label{fig:isd}
\end{center}
\end{figure}

In practice the ideal soliton distribution does not show a good performance. The reason behind this poor performance is that its design only takes into account the \emph{expected} value of symbols entering the ripple. In practice, however, there are statistical variations in the iterative decoding process that make the ideal soliton distribution fail with  high probability.

Let us denote the probability of decoding failure by $\Pf$. A lower bound to $\Pf$ is the probability that the decoding cannot start at all because the ripple is empty (no degree one output symbols), we shall denote this probability by $\Pflowisd$. This probability corresponds to
\[
\Pflowisd= \binom{m}{0} {\isd_1}^0 \left( 1-\isd_1 \right)^{m} = \left( 1-\frac{1}{k} \right)^{k (1+\reloverhead)}.
\]
If we now let $k$ (and $m$) tend to infinity keeping the relative receiver overhead $\reloverhead$ constant, this expression simplifies to:
\[
\lim_{k \rightarrow \infty } \Pflowisd= e^{-(1+\reloverhead)}.
\]
This implies the probability of decoding failure is in practice very high, since one usually wants to operate at low $\reloverhead$ (the overhead should ideally be small).

\vspace{0.8cm}
\textbf{Robust Soliton Distribution}\label{sec:rsd}
\newline
The robust soliton distribution was introduced in the original \ac{LT} paper from Luby, \cite{luby02:LT}. This distribution is an improvement of the ideal soliton. In fact, the design goal of the robust soliton distribution is ensuring that the expected ripple size is large enough at each point of the decoding \fran{with high probability}. This ensures that iterative decoding does not get stuck in the middle of the decoding process.

The robust soliton distribution is actually a family of parametric distributions that depend on two parameters $\parone$ and $\partwo$. Let $R=\partwo \log (k/\parone) \sqrt{k}$. The robust soliton distribution is obtained as:
\begin{equation}
\rsd_d = \frac{\isd_d + \tau_d}{\beta}, \label{eq:dist_rsd}
\end{equation}
where $\tau_d$ and $\beta$ are given by

\[\tau_d=
\begin{cases}
\frac{R}{d~k}                &  1 \leq d \leq \frac{k}{R-1} \\
R ~\log \left( R / \parone \right)/k & d=\frac{k}{R-1} \\
0                   &   d > \frac{k}{R-1}
\end{cases},
\]
and
\[
\beta = \sum_{d=i}^k \isd_d + \tau_d.
\]
Therefore, the robust soliton distribution is obtained as a mixture of the ideal soliton distribution with a correction term $\tau$ .  The average output degree for this distribution can be upper bounded by \cite{luby02:LT} :
\[
\avgrsd\leq H(k) + 1 + \log(R/\parone).
\]

\begin{figure}[t]
\begin{center}
\includegraphics[width=\figwbigger]{RSD}
\centering \caption[Robust soliton distribution example, $k=50$]{Robust soliton distribution, $\rsd$, for $k=50$, $\parone= 0.2$ and $\partwo=0.05$.}
\label{fig:rsd}
\end{center}
\end{figure}
\begin{figure}[t]
\begin{center}
\includegraphics[width=\figwbigger]{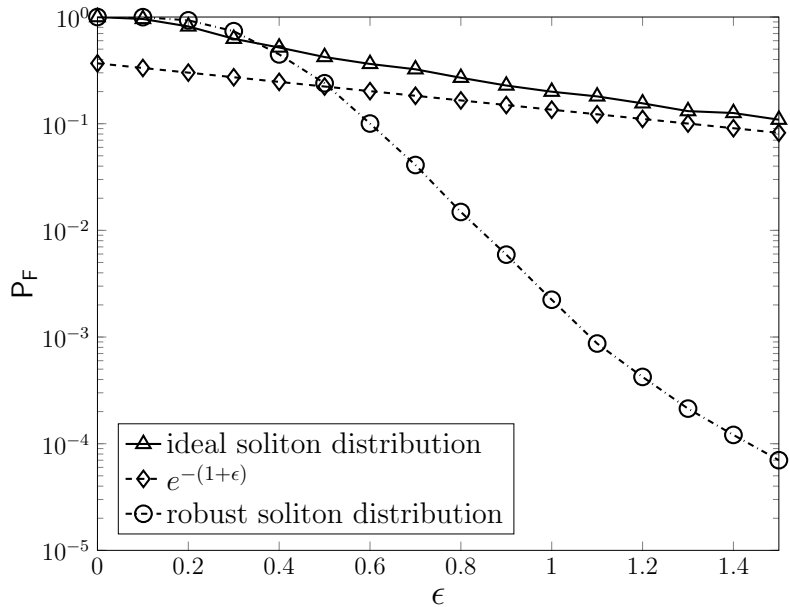}
\centering \caption[$\Pf$ vs.\ $\reloverhead$ for the ideal and robust soliton distribution, $k=100$]{Probability of decoding failure $\Pf$ vs.\ relative receiver overhead $\reloverhead$ for the ideal and robust soliton distribution with $\parone= 0.33$ and $\partwo=0.234$ for $k=100$.}
\label{fig:isd_vs_rsd}
\end{center}
\end{figure}

For illustration we provide a plot of a robust soliton distribution in Figure~\ref{fig:rsd}. We can observe how the probability of degree one output symbols is increased with respect to the ideal soliton distribution. Moreover, a \emph{spike} appears in the distribution at $d={k}/(R-1)$.

In Figure~\ref{fig:isd_vs_rsd} we provide a performance comparison for the ideal and robust soliton distribution for $k=100$. More concretely we show the probability of decoding failure under iterative decoding, $\Pf$, vs.\ the relative receiver overhead $\reloverhead$. It can be observed how the asymptotic lower bound to $\Pf$ for the ideal soliton distribution holds and is actually tight for high $\reloverhead$. Moreover, we can observe how the probability of decoding failure of the robust soliton distribution is much lower than that of the ideal soliton distribution.

\FloatBarrier
\subsection{Maximum Likelihood Decoding}

As we saw in Section~\ref{sec:LT_intro}, for fixed $n$ the relation between source symbols $\mathbf{v}$ and output symbols $\mathbf{\osymb}$ can be expressed by a system of linear equations:

\[
 \mathbf{\osymb} = \mathbf{v} \GLTa
\]
where we recall, that $\GLTa$ was the generator matrix of the fixed-rate \ac{LT} code. That is, under the assumption that the number of output symbols $n$ is fixed.

Let us assume that out of the $n$ output symbols generated by the \ac{LT} encoder the receiver collects $m=k+\absoverhead$, that we denote by $\mathbf{\rosymb}=(\rosymb_1, \rosymb_2, \ldots, \rosymb_m)$.  Denoting by $\mathcal{I} = \{i_1, i_2, \hdots, i_m \}$ the set of indices corresponding to the $m$ non-erased symbols, we have
\[
\rosymb_j = \osymb_{i_j}.
\]
 The dependence of the received output symbols on  the source symbols can be expressed as:
\begin{align}
\mathbf{\rosymb} = \mathbf{v} \Grx
\label{eq:ml_eq_sys}
\end{align}
with $\Grx$ given by the $m$ columns of $\G$ with indices in $\mathcal{I}$.


\ac{LT} decoding consists in finding the solution to the system of linear equations in \eqref{eq:ml_eq_sys}. The solution will be unique only if  $\Grx$ has full rank, that is, if its rank is $k$. If $\Grx$ is rank deficient the system of equations does not have a unique solution and the receiver is not able to recover all source symbols\footnote{In this thesis we focus on problems in which it is necessary to recover all source symbols, therefore, we declare a decoding failure whenever one or several source symbols cannot be recovered.}.

Iterative decoding is a suboptimal algorithm, it is not always able to find the solution when $\Grx$ has full rank.  For example, if $\Grx$ has full rank but does not have any row with Hamming weight one (degree one output symbol), iterative decoding is unable to find the solution.

A \acf{ML} decoding algorithm is an \emph{optimal} decoding algorithm, in the sense that it always finds the solution to the system of linear equations whenever $\Grx$ has full rank. Therefore the performance of any \ac{ML} decoding algorithm depends only on the rank properties of $\Grx$ and, more concretely, on the probability of $\Grx$ having full rank. In  \cite{schotsch:2013} the performance of  \ac{LT} codes under \ac{ML} decoding was studied and a lower bound to the probability of decoding failure $\Pf$ was derived:
\begin{align}
\Pflow  = \sum_{i=1}^k (-1)^{i+1} \binom{k}{i} \left( \sum_{d=1}^k \Omega_d \frac{\binom{k-i}{d}}{\binom{k}{d}}\right)^{k(1+\reloverhead)}.
\label{eq:low_bound_schotsch}
\end{align}
The lower bound is very tight for reception overhead slightly larger than $\reloverhead=0$.

In practice, different \ac{ML} decoding algorithms can be used to solve a system of equations and they all provide the same solution, that is unique when $\Grx$ is full rank. However, different \ac{ML} decoding algorithms have different decoding complexity, and some algorithms are more suitable than others for practical use.

\subsection{Complexity Considerations} \label{chap:LT_complex}

So far, the only performance metric we have dealt with is the probability of decoding failure. The other important metric when dealing with any  coding scheme is its complexity both in encoding and decoding. Let us define complexity as the total number of operations (xor or symbol copy) needed for encoding / decoding.
\fran{Since we consider binary \ac{LT} codes, we only perform xor operations, which correspond to additions over $\mathbb {F}_2$. Note that decoding also requires copying the content of output symbols into input symbols. For the sake of completeness, we shall also count symbol copy as one operation.}
Let us also define the encoding cost as the encoding complexity normalized by the number of output symbols and the decoding cost as the decoding complexity normalized by the number of input symbols.

\subsubsection{Encoding Complexity}
Let us first consider encoding complexity. Generating an output symbol of degree $d$ requires $d$ operations. Thus, given a degree distribution $\Omega$, the encoding cost will be given by the average output degree $\bar \Omega$.  In proposition~\ref{prop:lt_logk} we have shown how a necessary condition for decoding to be successful with high probability is ${\bar \Omega(k) = \mathcal{O} \left( \log(k) \right)}$. This implies that the encoding cost will need to be at least $\mathcal{O} \left( \log(k) \right)$.

\subsubsection{Iterative Decoding Complexity}

We consider now the complexity of \ac{LT} iterative decoding. Let us assume a generic degree distribution $\Omega$, with average output degree $\bar \Omega$ that requires a relative receiver overhead $\reloverhead^*$ for decoding to be successful with high probability. If we think of a bipartite representation of our \ac{LT} code, we can think of  encoding as drawing the edges in the graph, where every edge implies performing one operation (xor or symbol copy). Similarly, iterative decoding starts operating on a bipartite graph containing the $m=(1+\reloverhead^*) k$ received output symbols and $k$ input symbols. During iterative decoding edges are erased from the graph, being each edge again associated to one operation. At the end of iterative decoding all edges are erased from the graph\footnote{Actually, at the last iteration of iterative decoding some edges might still be present in the graph since we might have more than one output symbol in the ripple. We neglect this effect for the sake of simplicity.}. Thus the decoding cost under iterative decoding corresponds to ${(1+ \reloverhead^*) \bar \Omega}$.

In proposition~\ref{prop:lt_logk} we have shown how a necessary condition for decoding to be successful with high probability is ${\bar \Omega(k) = \mathcal{O} \left( \log(k) \right)}$. This implies that the iterative decoding cost will need to be at least $\mathcal{O} \left( \log(k) \right)$.

\subsubsection{Maximum Likelihood Decoding Complexity}
Many different \ac{ML} decoding algorithm exists that can be used to solve a linear system of equations. All \ac{ML} algorithms lead to the same solution, that in our case is unique when matrix $\Grx$ is full rank. The \ac{ML} decoding complexity will vary depending on which \ac{ML} decoding algorithm is used.

The best known algorithm is probably Gaussian elimination. This algorithm has a decoding complexity of $\mathcal{O} \left( k^3 \right)$ and is generally not practical for values of $k$ beyond the hundreds. The problem of solving systems of linear equations is a well known problem that appears not only in erasure correction. Several algorithms exist that have a lower (asymptotic) complexity than Gaussian elimination. For example, the Wiedemann algorithm \cite{wiedemann1986solving} can be used to solve sparse systems of linear equations with a complexity of $\mathcal{O} \left( k^2  \log^2(k) \right)$.  In \cite{lamacchia91:solving} different algorithms are studied to solve large systems of sparse linear equations over finite fields. In this work, the running times of different decoding algorithms are compared for systems of equations arising from integer factorization and the computation of discrete logarithms. The main finding of the paper is that if the system of equations is sparse, there exists a class of algorithms that in practice requires shorter running times than the Wiedemann algorithm when $k$ is below $10^5$. This class of algorithms is usually known as \emph{structured} or \emph{intelligent} Gaussian elimination. They consist of reducing the system of equations to a much smaller one than can be solved using other methods (Gaussian elimination, for example). Let us assume that Gaussian elimination is used to solve the reduced system of equations, and let us also assume that our intelligent Gaussian elimination algorithm is able to reduce the size of the system of equations from $k$ to $k/f$, where $f>1$. Since the complexity of Gaussian elimination is $\mathcal{O} \left( k^3 \right)$, for large enough $k$, the intelligent Gaussian elimination algorithm will reduce complexity at least by a factor $f^3$. Despite having a higher asymptotic complexity (the complexity is still ${\mathcal{O} \left( k^3 \right)}$ ) these algorithms \fran{have shorter running times} than other algorithms, such as the Wiedemann algorithm (provided that $f$ is large enough and $k$ not too large).

The \ac{ML} decoding algorithm used in practice for fountain codes is usually referred to as inactivation decoding. This algorithm belongs to the family of \emph{structured} or \emph{intelligent} Gaussian elimination algorithms and it will be explained in detail in Chapter~\ref{chap:LT_inact}.

\subsection{Systematic \acs{LT} Codes}

In practical applications it is desirable that fountain codes are systematic, that is, the first $k$ output symbols should correspond to the $k$ input symbols. Thus, if the quality of the transmission channel is good and no erasures occur, the receiver does not need to carry out decoding. A straightforward way of making a fountain code systematic is simply transmitting the first $k$ input symbols and afterwards start transmitting output symbols from the fountain code. We will refer to this construction as trivially systematic \ac{LT} code. This construction shows a poor performance since the receiver overhead needed to decode successfully increases substantially \cite{shokrollahi2003systematic}.

Figure~\ref{fig:RSD_systematic} shows the probability of decoding failure for a \acf{RSD} for $k=100$ with parameters with $\parone= 0.33$ and $\partwo=0.234$ under \ac{ML} decoding. In particular, two codes are considered, a standard \ac{LT} code and a trivially systematic \ac{LT} code  over a \ac{BEC} with erasure probability $\erasprob=0.1$. It can be observed the trivial systematic code performs much worse than the standard non systematic \ac{LT} code.

\begin{figure}[t]
        \begin{center}
        \includegraphics[width=\figwbigger]{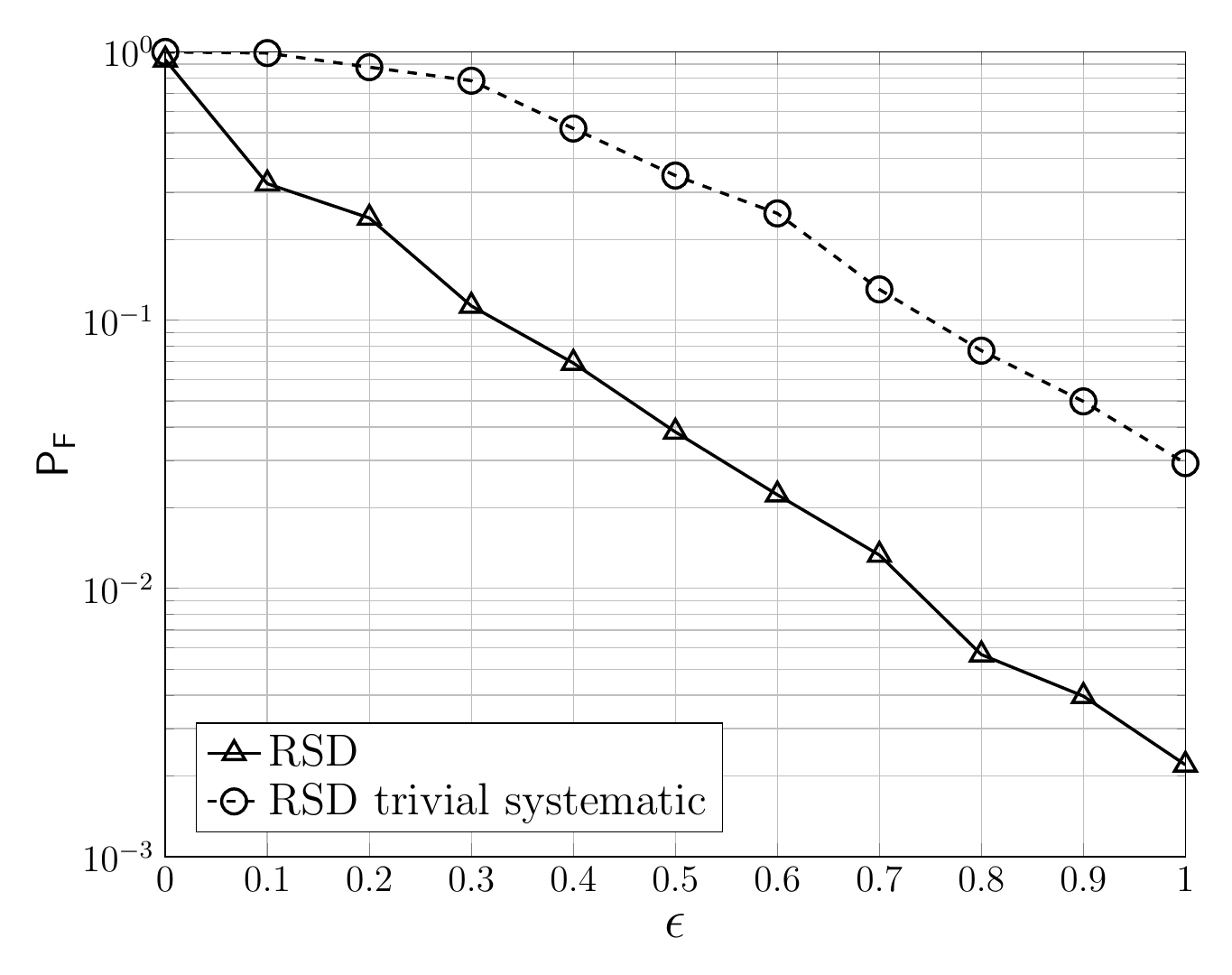}
        \centering
        \caption[$\Pf$ vs.\ $\reloverhead$ for a standard and a trivially systematic \acs{LT} under \acs{ML} decoding.]{$\Pf$ vs.\ $\reloverhead$ for a robust soliton distribution with $\parone= 0.33$ and $\partwo=0.234$ for $k=100$ under \ac{ML} decoding. The solid line with triangle markers represents the probability of failure of a standard \ac{LT} code. The dashed line with round markers represents the probability of failure of a trivially systematic \ac{LT} code over a \ac{BEC} with erasure probability $\erasprob=0.1$.}
        \label{fig:RSD_systematic}
        \end{center}
\end{figure}

The bad performance of trivially systematic \ac{LT} codes might seem surprising at first. The intuition behind this bad performance is the following. Assume that a substantial fraction of systematic symbols are received, for example, let us assume the decoder has received $\mu$ of the systematic symbols and that the remaining fraction $1-\mu$ have been erased. In order to be able to decode, the receiver will need to receive output symbols with neighbors within the yet unrecovered input symbols. Moreover, any output symbol having neighbors only within the received systematic symbols will be useless for decoding. Let us now assume that an output symbol of degree $d$ is received. The probability that all its neighbors are within the received systematic symbols is
\[
\frac{\binom{\mu k}{d}}{\binom{k}{d}}.
\]
\fran{Under the assumption that $k$ is large, $k \gg d$, and that output symbols choose their neighbours with replacement, a simplified expression for this probability can be obtained. Under these assumptions, we have that the probability that one of the neighbors of an output symbol is within the received systematic symbols is $\mu$. Hence, the probability that all $d$ neighbors are within the received systematic symbols is $\mu^d$.}
Thus, when the fraction of received systematic symbols $\mu$ is close to one, and $d$ is not too large, most of the received output symbols will not help at all in decoding. A more detailed analysis of this effect can be found in \cite{shokrollahi2009theoryraptor}.

In practice a different systematic construction is used that was patented in \cite{shokrollahi2003systematic} and that will be presented next.

Let us recall that (for fixed $n$) \ac{LT} encoding can be seen as a vector-matrix multiplication:
\[
\mathbf{\osymb} = \mathbf{v} \mathbf{G},
\]
where $\mathbf{v}$ is  the row vector of $k$ input (source) symbols, $\mathbf{\osymb}$ is the row vector of $n$ output symbols, and $\G$ is an $k \times n$ binary matrix which defines the relation between the input and the output symbols (generator matrix).
To construct a systematic \ac{LT} code we start with an \ac{LT} code with generator matrix in the shape
\[
\G =  \left[ \G_1 | \G_2\right],
\]
where $\G_1$ is a full-rank $k \times k$ matrix  that corresponds to the first $k$ output symbols and $\G_2$ is a $k \times (m-k)$ matrix.
First, one needs to compute the inverse matrix of $\G_1$, ${\G_1}^{-1}$.
The next step is computing:
\[
\mathbf{w} = \mathbf{v}   {\G_1}^{-1}.
\]
Vector $\mathbf{w}$ is then used as input to the \ac{LT} encoder. Thus, the output of the \ac{LT} encoder will be:
\[
\mathbf{\osymb} = \mathbf{w} \mathbf{G}= \mathbf{v}   {\G_1}^{-1} \left[ \G_1 | \G_2\right] = \mathbf{v}  \left[ \mathbf{I} | {\G_1}^{-1}  \G_2\right],
\]
where $\mathbf{I}$ is the $k \times k$ identity matrix. Hence, the first $k$ output symbols correspond to the input symbols  $\mathbf{v}$. For illustration Figure~\ref{fig:syst_LT} shows
a graph representation of a systematic \ac{LT} code.
\begin{figure}[t]
        \begin{center}
        \includegraphics[width=\textwidth]{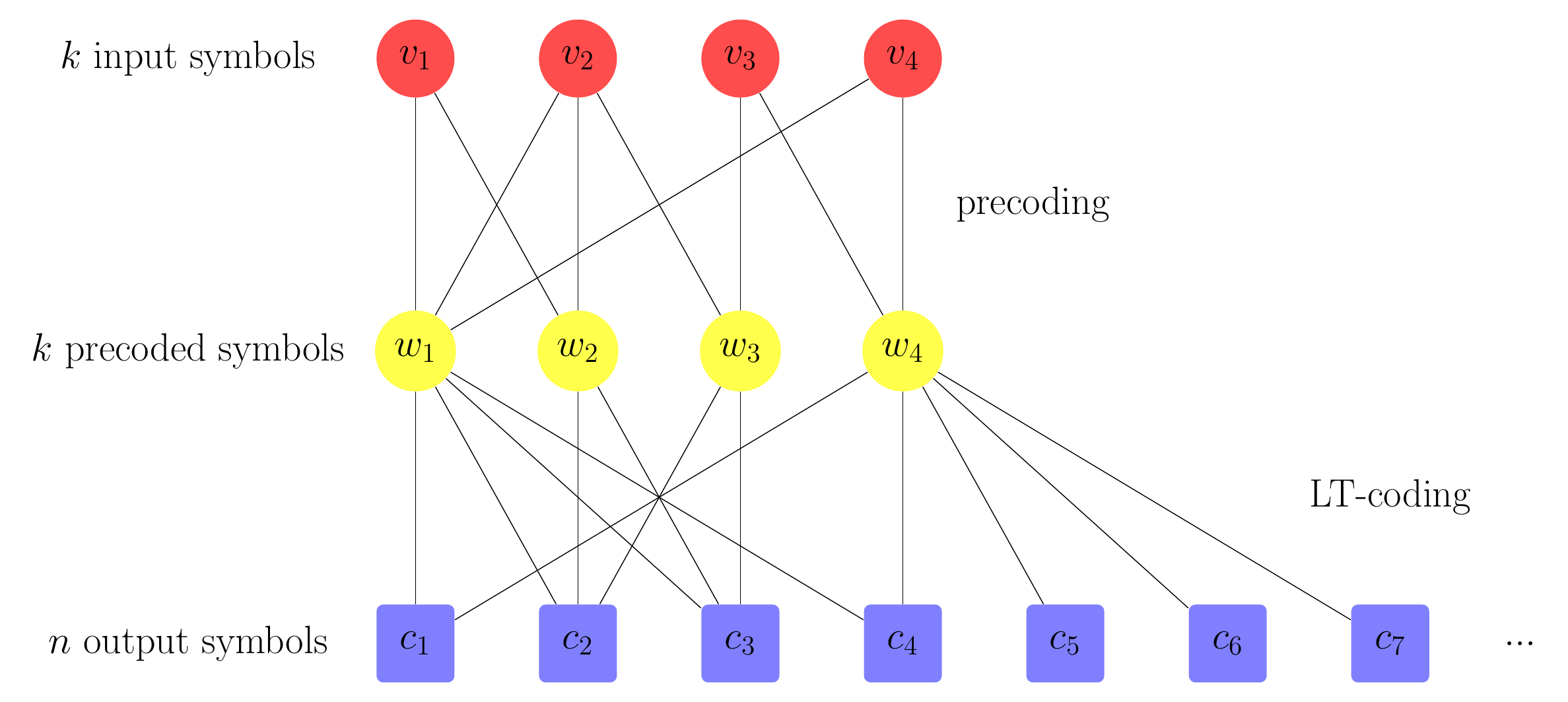}
        \centering \caption{Systematic \ac{LT} code.\label{fig:syst_LT}}
        \end{center}
\end{figure}

At the decoder side two different scenarios can be considered. In case none of the first $k$ output symbols of our systematic \ac{LT} code are erased, there is obviously no need to carry out decoding. In case some erasures do occur, decoding can be done in two steps. First, standard \ac{LT} decoding can be carried out to recover $\mathbf{w}$. This consists of solving the system of equations
\[
\mathbf{\rosymb} = \mathbf{w} \Grx,
\]
where $\Grx$ is a $k \times m$ matrix that corresponds to the $m$ columns of $\mathbf{G}$ associated to the output symbols that were not erased by the channel and $\mathbf{\rosymb}$ are the received output symbols. This system of equations can be solved in several ways, for example using iterative decoding or inactivation decoding. Finally, the input symbols can be recovered computing
\[
\vecv = \mathbf{w} \G_1.
\]

Note that this last step corresponds to LT encoding (since by construction  $\G_1$ is sparse, this last step is actually less complex than a standard vector matrix multiplication).

\fran{The main advantage of this construction is that its performance in terms of probability of decoding failure is similar to that of non-systematic \ac{LT} codes \cite{shokrollahi2011raptor}. However, this comes at some cost in decoding complexity, since an additional \ac{LT} encoding needs to be carried out at the decoder.}


\FloatBarrier
\section{Raptor Codes}

Raptor codes were originally patented in \cite{shokrollahi2001raptor} and published in \cite{shokrollahi04:raptor,shokrollahi06:raptor}. They were also independently proposed in \cite{maymounkov2002online}, where they are referred to as online codes. Raptor codes are an evolution of \ac{LT} codes. More concretely, Raptor codes are a serial concatenation of an outer (fixed-rate) block code $\code$ (usually called precode) with an inner \ac{LT} code.

At the input we have a vector of $k$ input (or source) symbols, ${\vecu=(\Raptorinput_1,~\Raptorinput_2,~\ldots, \Raptorinput_k)}$. Out of the input symbols, the outer code generates a vector of $h$ intermediate symbols ${\vecv=(\Rintermsymbol_1,~\Rintermsymbol_2,~\ldots, \Rintermsymbol_h)}$, \fran{where $h > k$}.
Denoting by $\Gp$ the employed generator matrix of the outer code, of dimension $(k \times h)$, the intermediate symbols can be expressed as
\[
\vecv = \vecu \Gp.
\]
By definition,  ${\vecv=(\Rintermsymbol_1,~\Rintermsymbol_2,~\ldots, \Rintermsymbol_h)}  \in \code$, i.e., the intermediate word is a codeword of the outer code $\code$.

The intermediate symbols serve as input to an \ac{LT} code that can generate an unlimited number of output symbols,
$\mathbf{\Rosymb}=(\Rosymb_1, \Rosymb_2, \ldots, \Rosymb_n)$, where $n$ can grow unbounded. Hence, Raptor codes inherit the rateless property of \ac{LT} codes. For any $n$ the output symbols can be expressed as
\begin{equation}
\mathbf{\Rosymb} = \vecv \GLT = \vecu \Gp \GLT
\end{equation}
where $\GLT$ is the generator matrix of the (fixed-rate) \ac{LT} code. Hence, $\GLT$ is an $(h \times n)$ binary matrix,  each column of $\GLT$ being associated with a received output symbol as seen in Section~\ref{sec:LT_intro}.

Figure~\ref{fig:Raptorgraph} shows a graph representation of a Raptor code, where the input symbols are represented as green diamond-shaped nodes, the intermediate symbols as red circular nodes and the output symbols as blue squared nodes.

\begin{figure}[t]
\begin{center}
\includegraphics[width=0.99\textwidth]{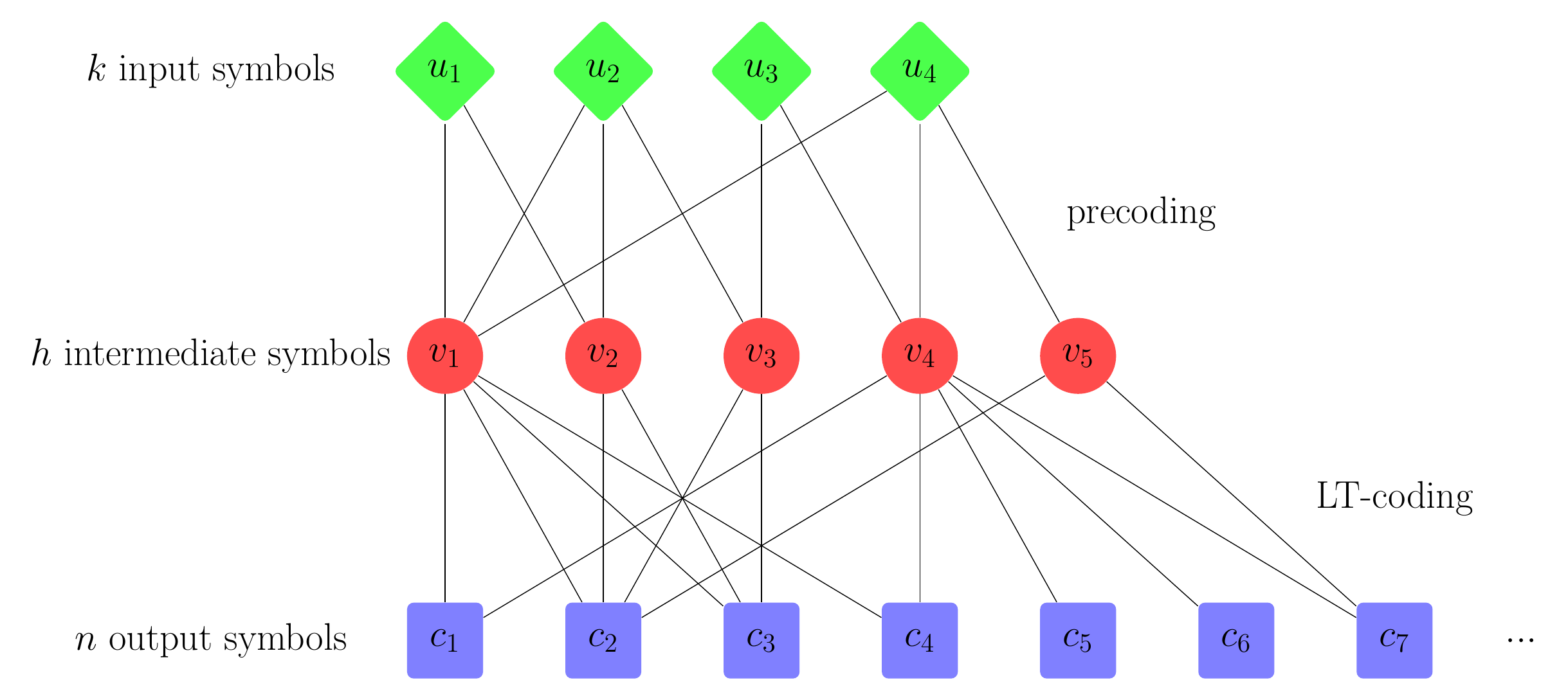}
\centering \caption{Graph representation of a Raptor code.}
\label{fig:Raptorgraph}
\end{center}
\end{figure}

The design principle of Raptor codes can be intuitively explained as follows. In Chapter~\ref{sec:LT_intro} we saw that a necessary condition for \ac{LT} codes to be successfully decoded with high probability is that the average output degree is $\mathcal{O} \left(\log (k)\right)$. This implies an encoding cost of $\mathcal{O} \left(\log (k)\right)$ and a decoding cost of $\mathcal{O} \left(\log (k)\right)$ as well (under iterative decoding).
The main idea behind Raptor codes is relaxing the requirements on the \ac{LT} code. Instead of requiring that the \ac{LT} recovers \emph{all} its input symbols, the inner \ac{LT} code of a Raptor code is only required to recover with high probability a constant fraction $1-\fracuncovered$ of the intermediate symbols. This can be achieved with a constant average output degree.
Let us assume that $k$ is large and that the receiver has collected $m$ output symbols. From the proof of Prop.~\ref{prop:lt_logk}, we have that for asymptotically large $h$ the fraction of intermediate symbols with no edges attached (uncovered) will correspond to:
\[
\fracuncovered= e^{-\bar \Omega \frac{m}{h} }.
\]
Let us assume that all the covered intermediate symbols can be recovered by the \ac{LT} code. The uncovered intermediate symbols can be considered as erasures by the outer code. If the outer code is an erasure correcting code that can recover with high probability from a fraction of $\fracuncovered$ erasures, we will be able to recover all input symbols with high probability.

If the precode $\code$ is linear-time encodable, then the Raptor code has a linear encoding complexity, $\mathcal O\left( k \right)$ since the \ac{LT} code has constant average output degree \fran{(i.e., the average output degree does not increase with $k$)}. Therefore, the overall encoding cost per output symbol is constant with respect to $k$. If the precode also accepts a linear time decoding algorithm (iterative decoding), and the \ac{LT} code is decoded using iterative decoding, the decoding complexity is also linear. Hence, the decoding cost per symbol is constant. Furthermore, already in the original Raptor code paper \cite{shokrollahi06:raptor}, Shokrollahi showed  that Raptor codes under iterative decoding are universally capacity-achieving on the binary erasure channel. Hence, they achieve the capacity of any erasure channel no matter which erasure probability the channel has.

\subsection{Raptor Decoding}\label{chap:raptor_under_ML}

The output symbols $\mathbf{\Rosymb}$ generated by the Raptor encoder are transmitted over a \ac{BEC} at the output of which each transmitted symbol is either correctly received or erased. Let us denote by $m$ the number of output symbols collected by the receiver of interest, where  $m=k+\absoverhead$, being $\absoverhead$ the absolute receiver overhead. We denote by ${\mathbf{\Rrosymb}=(\Rrosymb_1, \Rrosymb_2, \ldots, \Rrosymb_m)}$ the $m$ received output symbols. Denoting by $\mathcal{I} = \{i_1, i_2, \hdots, i_m \}$ the set of indices corresponding to the $m$ non-erased symbols, we have
\[
\Rrosymb_j = \Rosymb_{i_j}.
\]

The relation between the received output symbols and the input symbols can be expressed as:
\begin{equation}
\mathbf{\Rrosymb} = \vecv \GrxR
\label{eq:raptor_G}
\end{equation}
where
\begin{align}
\GrxR = \Gp \GrxLT
\label{eq:sys_eq}
\end{align}
with $\GrxLT$ given by the $m$ columns of $\GLT$ with indices in $\mathcal{I}$.

Raptor decoding consist of recovering the input symbols $\vecv$ given the received output symbols $\mathbf{\Rrosymb}$. Although it is possible to perform Raptor decoding by solving the linear system of equations in \eqref{eq:raptor_G}, this is not done in practice for complexity reasons. The decoding algorithms employed in practice, iterative decoding or inactivation decoding, require that the system of equations is sparse in order to show good performance and matrix $\GrxR$ is not sparse in general.

In practice, instead of the generator matrix of the Raptor code, another matrix representation is used that is usually referred to as constraint matrix, \fran{since it is an alternative representation of the coding constraints of the outer and inner code}.
The constraint matrix of a Raptor code is defined as:
\begin{align}
\constmatrix =
\begin{bmatrix}
 \hmatrixpre \\
 \GrxLT^T
\end{bmatrix},
\label{eq:constmatrix}
\end{align}
where $\hmatrixpre$ is the parity check matrix of the outer code (precode) with size $\left(   \left(h-k \right) \times h  \right)$. Thus, $\constmatrix$ is a $\left( \left(h-k+m \right) \times h\right)$ binary matrix.

By definition, the intermediate word of a Raptor code is a codeword of the precode, ${\vecv=(\Rintermsymbol_1,~\Rintermsymbol_2,~\ldots, \Rintermsymbol_h) \in  \code}$. Hence, one can write
\begin{equation}
\hmatrixpre ~ \vecv^T = \zeros
\label{eq:h_v_z}
\end{equation}
where $\zeros$ is a zero column vector of size $\left( (h-k) \times 1\right)$. Similarly, one can express the vector of received output symbols $\mathbf{\Rrosymb}$ as:
\begin{equation}
\GrxLT^T ~ \vecv^T = \mathbf{\Rrosymb}^T.
\label{eq:g_v_y}
\end{equation}
Putting together \eqref{eq:h_v_z} and \eqref{eq:g_v_y}, we have
\begin{align}
\constmatrix ~ \vecv^T =
\begin{bmatrix}
\zeros \\
\mathbf{\Rrosymb}^T
\end{bmatrix}.
\label{eq:raptor_sys_eq}
\end{align}
In practical Raptor decoders \eqref{eq:raptor_sys_eq} is used for decoding. The main advantage of the constraint matrix is that it preserves the sparsity of the generator matrix of the \ac{LT} code. Moreover, it also preserves the sparsity of the parity check matrix of the precode, in case it is sparse.

The system of equations in \eqref{eq:raptor_sys_eq} can be solved using different techniques, such as iterative decoding, standard Gaussian elimination or inactivation decoding.
Similarly to \ac{LT} codes, most works on Raptor codes consider large input blocks ($k$ at least in the order of a few tens of thousands symbols) and iterative decoding. However, in practice smaller blocks are used, usually due to memory limitations at the decoders. For example, in the most widespread binary Raptor codes, R10 (release 10), values of $k$ ranging from $1024$ to $8192$ are recommended (see Section~\ref{chap:raptor_r10}). For these input block lengths, the performance of iterative decoding suffers a considerable degradation. Therefore, instead of iterative decoding, \ac{ML} decoding is used (inactivation decoding).


\subsection{R10 Raptor Codes}\label{chap:raptor_r10}
The state of the art binary Raptor code is the R10 (release 10) Raptor code. This code is systematic and was designed to support a number of input symbols ranging from $k=4$ to $k=8192$ \cite{shokrollahi2011raptor}. The maximum supported number of output symbols is $n=65,536$. The probability of decoding failure $\Pf$  shows an error floor lower than $10^{-6}$ for all values of $k$ \cite{shokrollahi2011raptor}.

The precode used by R10 Raptor codes is a serial concatenation of two systematic erasure correcting codes.
\begin{itemize}
\item The outer code is a systematic \acf{LDPC} code that introduces  $\srten$ redundant symbols. The number of \ac{LDPC} redundant symbols is a function of $k$ and can be approximated as \cite{shokrollahi2011raptor}
\[
\srten \approx 0.01 k + \sqrt{2 k}.
\]
Its parity check matrix is composed of $\lceil k/\srten\rceil$ degree 3 circulant matrices  plus a $(\srten \times \srten)$ identity matrix. The Hamming weight of each of the rows of the parity check matrix of the \ac{LDPC} code is approximately $ 3 \lfloor k/\srten\rceil +1$, where $\lfloor x \rceil$ denotes the closest integer to $x$ (the last circulant matrix might not be complete).

\item The inner code is a systematic \acf{HDPC} code that introduces $\hrten$ redundant symbols. The number of \ac{HDPC} redundant symbols depends on $k$ approximately as \cite{shokrollahi2011raptor}
\[
\hrten \approx \log_2 (1.01 k + \sqrt{2 k}).
\]
Its parity check matrix is composed of a dense part and a $(\hrten \times \hrten)$ identity matrix. The dense part is obtained from a binary reflected Gray code and has the property that the normalized Hamming weight of every row is approximately $1/2$. Therefore, roughly half of the elements in the dense part of the parity check matrix are set to $1$.
\end{itemize}
Thus, the total number of intermediate symbols $h$ corresponds to
\[
h= k + \srten + \hrten.
\]

Figure.~\ref{fig:raptor_r10_s_and_h} shows the number of redundant \ac{LDPC} and \ac{HDPC} symbols,  $\srten$ and $\hrten$, as a function of $k$. It can be observed how for all values of $k$ except for very small values ($k=4$) the number of \ac{LDPC} redundant symbols is higher than that of \ac{HDPC} redundant symbols.
Therefore, the rate of the \ac{LDPC} code is lower than that of the \ac{HDPC} code, as it can be observed in Figure.~\ref{fig:raptor_r10_rates}. In this last figure it can also be observed how the outer code rate increases with $k$, although it does not do it monotonically.

\begin{figure}
\begin{center}
\includegraphics[width=\figwbigger]{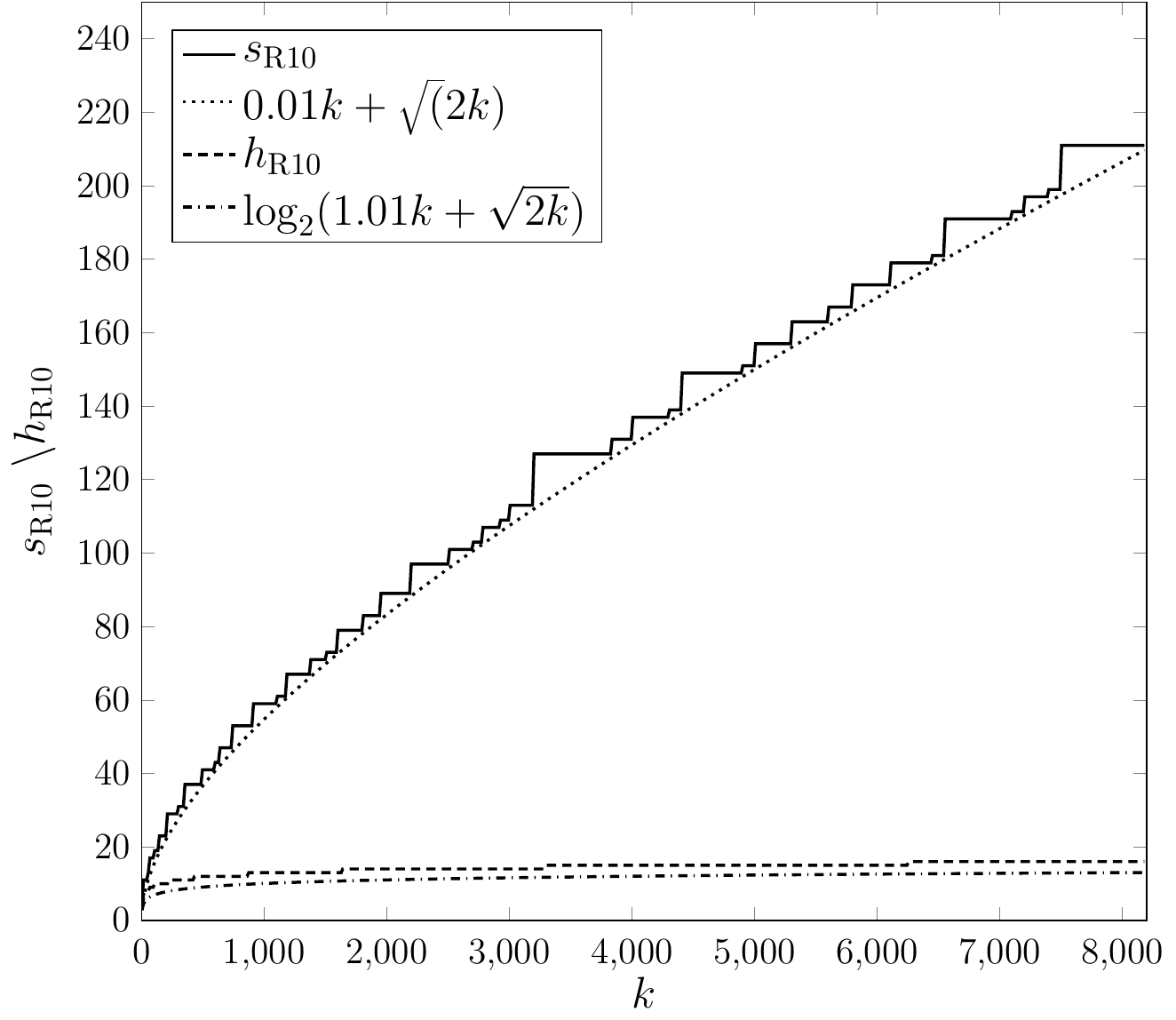}
\centering \caption[Number of \acs{LDPC} and \acs{HDPC} redundant symbols in R10 Raptor codes vs.\ $k$]{Number of \ac{LDPC} ($\srten$) and \ac{HDPC} ($\hrten$) redundant symbols in R10 Raptor codes vs.\ $k$ and their approximate values.}
\label{fig:raptor_r10_s_and_h}
\end{center}
\end{figure}
\begin{figure}[t]
\begin{center}
\includegraphics[width=\figwbigger]{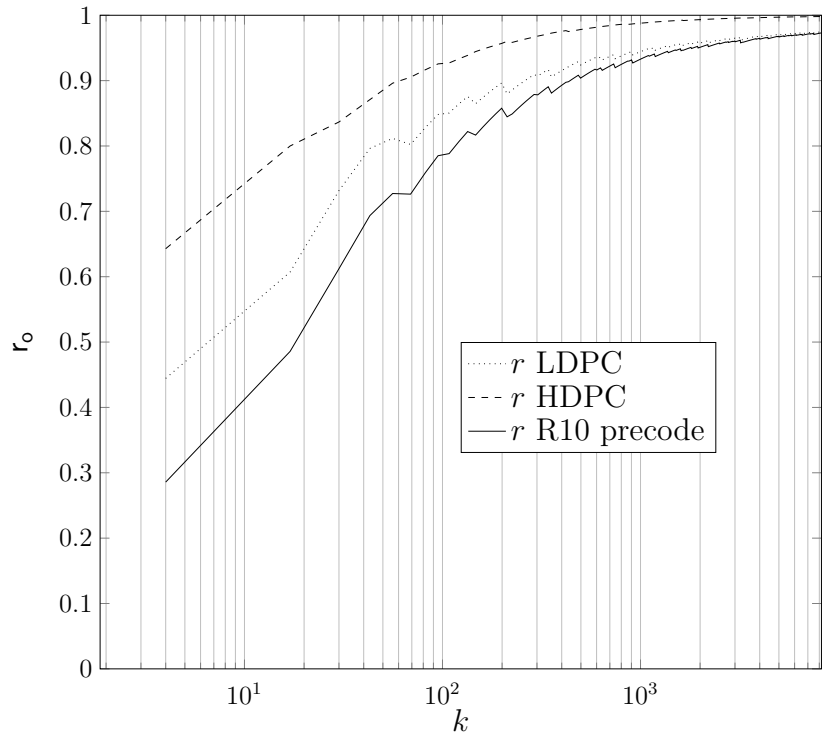}
\centering \caption[Precode rate $\ro$  vs.\ $k$ for R10 Raptor codes]{Precode rate $\ro$  vs.\ $k$ for R10 Raptor codes. The rate of the \ac{LDPC} and \ac{HDPC} codes is also shown.}
\label{fig:raptor_r10_rates}
\end{center}
\end{figure}
The precode of the R10 Raptor code was designed to behave similarly to a uniform random matrix in terms of rank properties but admitting a fast matrix vector multiplication algorithm \cite{shokrollahi2011raptor}.

The degree distribution of the \ac{LT} code is given by:
\begin{align}
\Omegarten(\x) &= 0.0098\x + 0.4590\x^2+ 0.2110\x^3+0.1134\x^4 \\
&+ 0.1113\x^{10} + 0.0799\x^{11} + 0.0156\x^{40}.
\label{eq:dist_mbms}
\end{align}
Its average output degree is $\bar \Omega = 4.631$.


For illustration in Figure.~\ref{fig:raptor_r10_matrix} we provide the constraint matrix for a R10 Raptor code for $k=20$ and $m=30$. In this case $\srten=11$ and $\hrten=7$. In the upper part, highlighted in blue, the parity check matrix of the \ac{LDPC} code can be distinguished. This submatrix is composed of two circulant matrices and an identity matrix. All rows of this submatrix  have Hamming weight $6$ or $7$. Below it the parity check matrix of the \ac{HDPC} code can be observed highlighted in green. These rows have a normalized Hamming weight of around $1/2$ and are the densest in the constraint matrix. The lower part of the constraint matrix, highlighted in red, corresponds to the \ac{LT} symbols and is sparse.

\begin{figure}[t]
\begin{center}
\includegraphics[width=0.55\columnwidth]{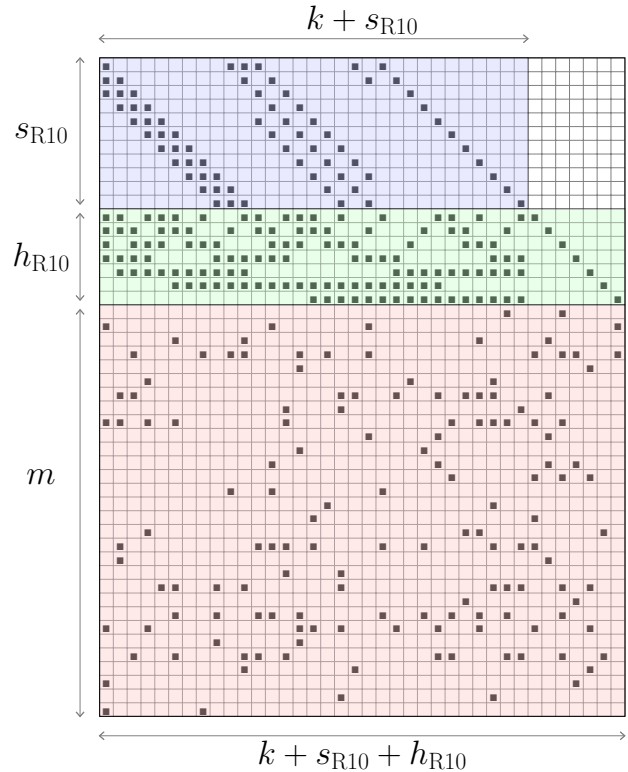}
\centering \caption[Constraint matrix of a R10 Raptor code for $k=20$ and $m=30$]{Constraint matrix of a R10 Raptor code for $k=20$ and $m=30$. The blue and red sub-matrices represents respectively the parity check matrices of the \ac{LDPC} and \ac{HDPC} codes. The red sub-matrix represents the transposed generator matrix of the \ac{LT} code. \fran{The entries filled with a square in the matrix represent the matrix elements set to one, and the empty entries those elements set to zero}.  }
\label{fig:raptor_r10_matrix}
\end{center}
\end{figure}

R10 Raptor codes are the state-of-the-art binary Raptor codes and are widely used in practice. In fact, they are part of several standards \cite{shokrollahi2011raptor}:
\begin{itemize}
    \item 3GPP Multimedia Broadcast Service, \cite{MBMS12:raptor}. Raptor codes are used in a terrestrial cellular network for file delivery and streaming applications.
    \item Internet Engineering Task Force (IETF) RFC 5053, \cite{RFC5053:raptor}. The R10 Raptor code is used for file delivery over data networks (i.e. the Internet). \fran{For example, this standard can be used to deliver data to a user via unicasting, or several user via multicasting, using, for instance, the User Datagram Protocol (UDP).}
    \item Digital Video Broadcasting. The R10 Raptor code is used in several DVB standards:
    \begin{itemize}
        \item In DVB-SH the R10 Raptor code is used as upper layer \acf{FEC} in order to overcome long deep fading events, \cite{DVB-SH:raptor}. \fran{Concretely, R10 Raptor codes can be used to protect Multi Protocol Encapsulation (MPE) fragments at the link layer, or User Datgram Protocol (UDP) datagrams at the transport layer.} This is specially appealing for mobile satellite systems in which terminals suffer frequently of very deep fades in the received signal due to blockage by a building or tunnel, for example.
        \item In DVB-H, \cite{DVB-H:raptor}, R10 Raptor codes are used for similar purposes as in DVB-SH.
        \item DVB has also standardized the R10 Raptor codes for streaming services over IP networks, \cite{DVB-IPBased:raptor}.
        \item The R10 Raptor code is also part of a standard for broadcast/multicast data delivery, \cite{DVB:raptor_broadcast}.
    \end{itemize}
    \item International Telecommunication Union (ITU) IPTV services \cite{ITU:raptor}. The R10 Raptor code is used for streaming applications.
\end{itemize}

R10 Raptor codes are not truly rateless, since the number of output symbols is limited to $n=65,536$. However, for most practical scenarios they can be considered to be rateless.
In spite of their (almost) rateless capability, R10 Raptor codes represent an excellent solution also for fixed-rate communication schemes requiring powerful erasure correction capabilities with low decoding complexity. In fact, in some of cases they are actually used in a fixed-rate setting (see, e.g., \cite{DVB-SH:raptor}).

\subsection{Systematic Raptor Codes}
Systematic Raptor codes are obtained similarly to systematic \ac{LT} codes \cite{shokrollahi2003systematic}. If we recall, Raptor code output symbols are obtained in two stages. First a vector of intermediate symbols $\vecv$ is obtained from the vector of input symbols $\vecu$ using an outer block code (precode):
\[
\vecv = \vecu \Gp,
\]
where $\Gp$ is $(k \times h)$ generator matrix of the precode. The output symbols $\mathbf{\Rosymb}$ are then obtained through an \ac{LT} encoding of the intermediate symbols:
\[
\mathbf{\Rosymb} = \vecv \GLT,
\]
where  $\GLT$ is the generator matrix of the LT code. Thus, the relationship between input and output symbols can be expressed as
\[
\mathbf{\Rosymb} = \vecu \Gp \GLT.
\]

A systematic Raptor code  can be obtained starting with an \ac{LT} generator matrix in the shape:
\[
\GLT =  \left[ \GLT\,\!_{,1} | \GLT\,\!_{,2} \right],
\]
where $\GLT\,\!_{,1}$ has dimension $h \times k$ and $\GLT\,\!_{,2}$ has dimension $h \times (n-k)$. Furthermore, matrix $\GLT\,\!_{,1}$ must be chosen so that matrix
\[
\mathbf{F} = \Gp \GLT\,\!_{,1}
\]
is full rank (has rank $k$). A systematic Raptor code is obtained by computing
\[
\mathbf{w} = \mathbf{v}   \mathbf{F}^{-1},
\]
and using vector $\mathbf{w}$ as input to the Raptor code. Hence, the intermediate symbols of the systematic Raptor code will correspond to:
\[
\vecv = \mathbf{w} \Gp = \mathbf{v}   \mathbf{F}^{-1} \Gp,
\]
and its output will correspond to
\[
\mathbf{\Rosymb} = \mathbf{v} \mathbf{F}^{-1} \Gp \GLT =  \mathbf{v} \mathbf{F}^{-1} \Gp  \left[ \GLT\,\!_{,1} | \GLT\,\!_{,2} \right] =
\mathbf{v} \left[\mathbf{I} | \mathbf{F}^{-1} \Gp \GLT\,\!_{,2} \right],
\]
where $\mathbf{I}$ is the identity matrix of dimension $k$.

For illustration we provide a graph representation of a systematic Raptor code in Figure~\ref{fig:syst_Raptor}.
\begin{figure}[t]
        \begin{center}
        \includegraphics[width=\textwidth]{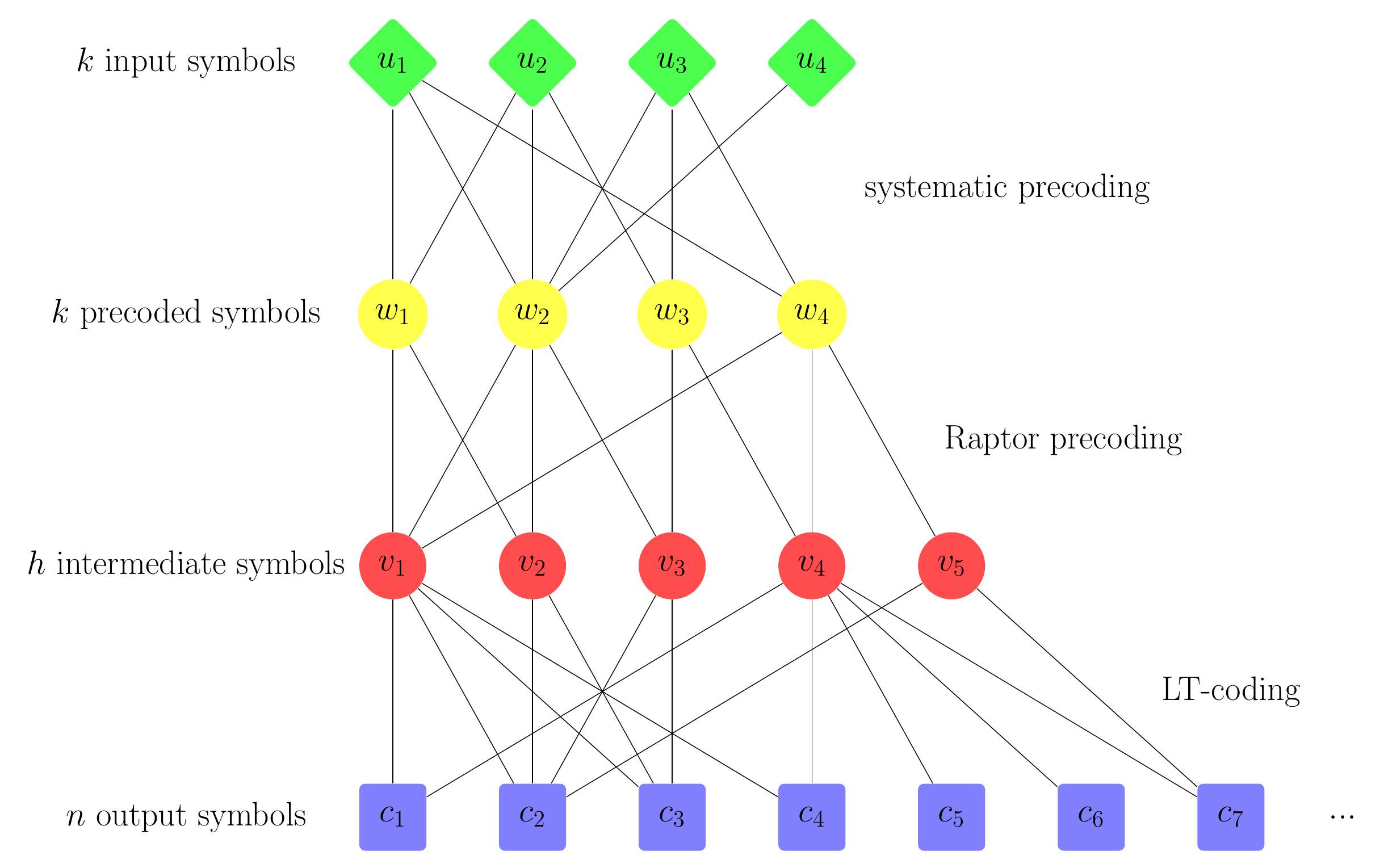}
        \centering \caption{Systematic Raptor code.\label{fig:syst_Raptor}}
        \end{center}
\end{figure}

When decoding a systematic Raptor code, again two cases can be differentiated. In the first none of the first $k$ output symbols are erased and decoding does not need to be carried out. In the second case some of the first $k$ output symbols are erased. In this case the input symbols can be recovered in two stages. First standard Raptor decoding is used to recover $\mathbf{w}$. Then, the input symbols can be recovered as:
\[
\mathbf{v} = \mathbf{w}   \mathbf{F} = \mathbf{w}   \Gp \GLT\,\!_{,1}.
\]
This second stage corresponds to Raptor encoding. First the output of the precode is computed and then an \ac{LT} code is applied.

\fran{This systematic construction of Raptor codes achieves a similar performance to that of non-systematic Raptor codes \cite{shokrollahi2011raptor}.}

\FloatBarrier


%% file: Chapter3/chapter3.tex

\chapter{LT Codes under Inactivation Decoding} \label{chap:LT}
\ifpdf
    \graphicspath{{Chapter3/Chapter3Figs/PNG/}{Chapter3/Chapter3Figs/PDF/}{Chapter3/Chapter3Figs/}}
\else
    \graphicspath{{Chapter3/Chapter3Figs/EPS/}{Chapter3/Chapter3Figs/}}
\fi

In this chapter we consider \ac{LT} codes under inactivation decoding. In Section~\ref{chap:LT_inact} we explain in detail inactivation decoding.  Section~\ref{chap:inact_analysis} focuses on analyzing inactivation decoding of \ac{LT} codes. Concretely Section~\ref{chap:inact_first_order} presents an analysis based on a dynamic programming approach that provides  the first moment of the number of inactivations. This analysis is then extended in Section~\ref{chap:inact_distribution} to obtain the probability distribution of the number of inactivations. In Section~\ref{chap:inact_low_complex} a low complexity approximate analysis of \ac{LT} codes under inactivation decoding is presented. Section~\ref{sec:lt_code_design} shows how the results in this chapter can be used in order to design \ac{LT} codes by means of an example. Finally, Section~\ref{chap:lt_summary} presents a summary of the chapter.

\section{Inactivation Decoding}\label{chap:LT_inact}

Inactivation decoding is a \ac{ML} decoding algorithm that is characterized by a manageable decoding complexity and  is widely used in practice \cite{shokrollahi2005systems}, \cite{MBMS12:raptor}. This algorithm belongs to the family of \emph{structured} or \emph{intelligent} Gaussian elimination algorithms since it aims at reducing the size of the system of equations that needs to be solved.

We will describe inactivation decoding in more detail by means of an example. As explained in Section~\ref{sec:LT_intro}, \ac{LT} decoding consists of solving the system of equations  given in \eqref{eq:ml_eq_sys}, which we replicate here for the sake of completeness,
\[
\mathbf{\rosymb} = \mathbf{v} \Grx,
\]
where we recall, $\mathbf{\rosymb}$ is the $(1 \times m)$ vector of received output symbols, $\mathbf{v}$ is the $(1 \times k)$ vector of input symbols, and $\Grx$ is a $k \times m$ matrix that corresponds to the $m$ columns of $\mathbf{G}$ associated to the output symbols that were not erased by the channel. We will consider an example with $k=50$, $m=60$ and with $\Grx$ as shown in  Figure~\ref{fig:inactivation_0}. In the figure the squares inside a cell represent the elements of  $\Grx$ that are set to $1$ and the empty cells the elements that are set to $0$. As it can be observed in the figure, matrix $\Grx$ is sparse.

\begin{figure}[t]
\centering
    \hspace{0.3cm}
    \includegraphics[width=0.55\columnwidth,draft=false]{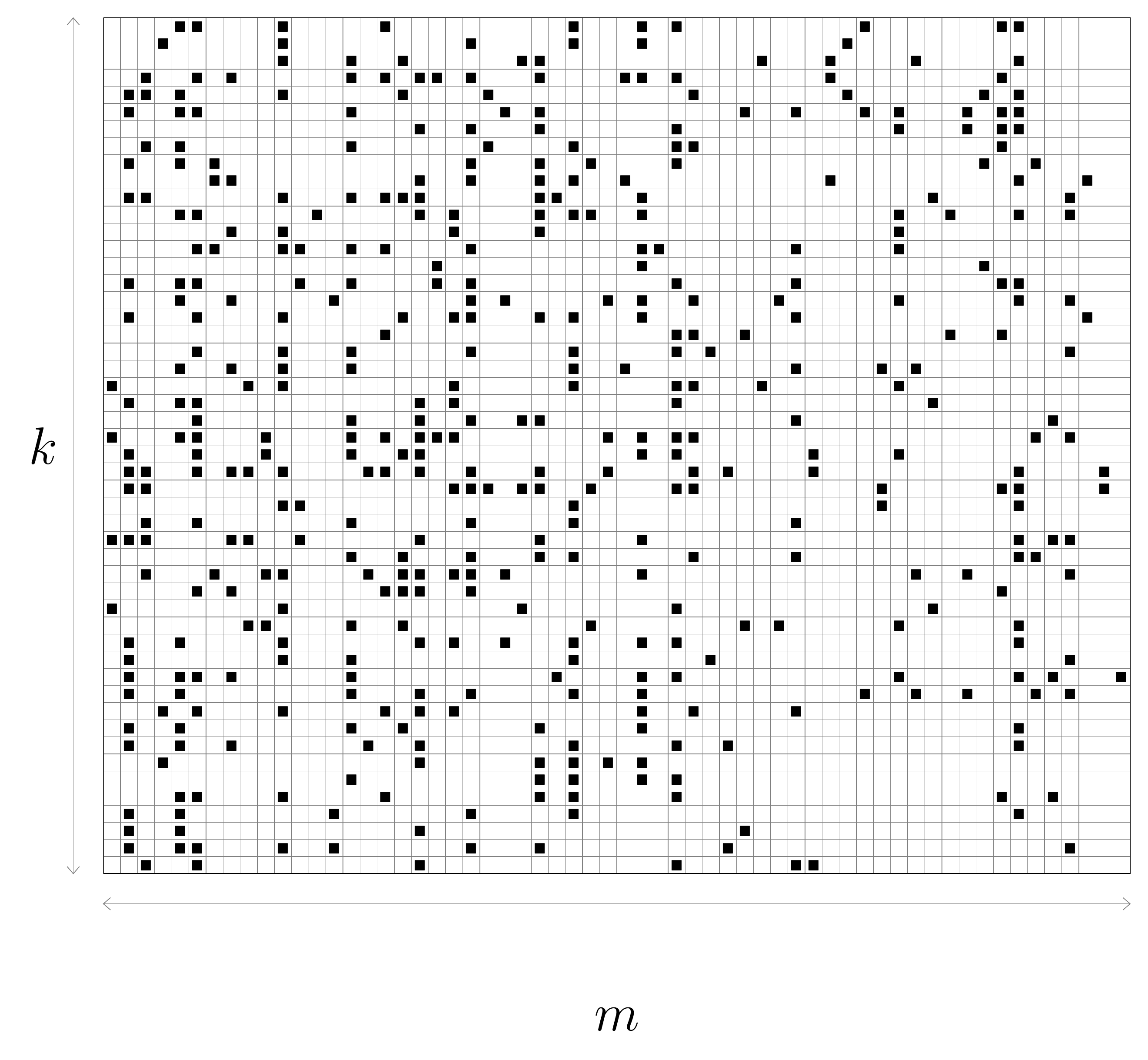}
    \caption{Structure of $\Grx$ before inactivation decoding starts.\label{fig:inactivation_0}}
\end{figure}
\FloatBarrier

Inactivation decoding consists of 4 steps:

\begin{enumerate}
    \item {\emph{Triangulation.} $\Grx$ is put in an approximate upper triangular form by means of column and row permutations. Since no operation is performed on the rows or columns of $\Grx$, the overall density $\Grx$ does not change. At the end of this process we can distinguish 4 sub-matrices in $\Grx$. As shown in Figure~\ref{fig:inactivation_1} in the left upper part of $\Grx$ we have matrix $\Amatrix$ that is an upper triangular matrix of size $(k-\reducedsyst) \times (k-\reducedsyst)$.
        In the upper right part we have matrix $\Bmatrix$ that has size $(k-\reducedsyst) \times (m-k-\reducedsyst)$. Finally at the lower left and right parts we have matrices $\Cmatrix$, and $\Dmatrix$ of respective sizes $\reducedsyst \times (m-k-\reducedsyst)$ and $\reducedsyst \times (k-\reducedsyst)$. The $\reducedsyst$ last (lowest) rows of $\Grx$ corresponding to matrices $\Cmatrix$ and $\Dmatrix$ are usually referred to as inactive rows.
        }
    \item {\emph{Zero matrix procedure.} The matrix $\Amatrix$ is put in a diagonal form and matrix $\Bmatrix$ is zeroed out through column sums. When one performs row/column additions on a sparse matrix, the density of the matrix tends to increase. In our case, matrices $\Cmatrix$ and $\Dmatrix$ become denser as shown in  Figure~\ref{fig:inactivation_2}.

        }
    \item {\emph{\ac{GE}}. \ac{GE} is applied to solve the systems of equations $\tilde{\mathbf{\rosymb}} = \tilde{\mathbf{v}} \Cmatrix'  $, where the symbols in $\tilde{\mathbf{v}}$ are called \emph{inactive variables} and are associated with the rows of the matrix $\Cmatrix'$ in Figure~\ref{fig:inactivation_2} and  $\tilde{\mathbf{\rosymb}}$ are known terms associated with the columns of the matrix $\Cmatrix'$ in Figure~\ref{fig:inactivation_2}. This step drives the cost of inactivation decoding since its complexity is $\mathcal{O} \left( \reducedsyst^3 \right)$, cubic in the number of inactive rows $\reducedsyst$. At the end of the \ac{GE} step, matrix $\Grx$ has the structure shown in Figure~\ref{fig:inactivation_3}.

        }

    \item {\emph{Back-substitution.} Once the values of the inactive variables have been determined, back-substitution is applied to compute the values of the remaining variables in $\mathbf{v}$. This corresponds to setting to zero all elements of matrix $\Dmatrix'$ in Figure~\ref{fig:inactivation_3}. After back-substitution ends all source symbols have been recovered and, therefore, $\Grx$ is in reduced echelon form as shown in Figure~\ref{fig:inactivation_4}.
        }

\end{enumerate}

A unique solution to the system of equations only exists if matrix $\Grx$ has full rank. If we look at Figure~\ref{fig:inactivation_2} it is easy to see how matrix $\Amatrix'$ has full rank since it is an identity matrix. Hence, matrix $\Grx$ has full rank only if submatrix $\Cmatrix'$ has full rank.

Among the 4 steps of inactivation decoding, the one having the highest (asymptotic) complexity is \ac{GE}. However, one has to consider that the size of the system of equations that needs to be solved by means of \ac{GE} is determined by the triangulation step. Therefore, if we want to analyze the complexity of inactivation decoding we need to have a closer look at triangulation.

In order to get a better understanding of triangulation we will use a bipartite graph representation of the \ac{LT} code, similar to the one we used for iterative decoding.  The triangulation step can be represented by an iterative pruning of the bipartite graph of the \ac{LT} code. At each iteration, a reduced graph is obtained that corresponds to a sub-graph of the original LT code graph. This sub-graph involves only a subset of the input symbols (that we call \emph{active} input symbols) and their neighbors. In this context, we use the term \emph{reduced} degree of output symbol $\osymb$ to refer to the  degree of output symbol $\osymb$ in the reduced graph, and we will denote it by $\deg_r(\osymb)$. Therefore, the reduced degree of a node (symbol) is less or equal to its (original) degree. Note that the reduced degree has been defined in the context of iterative decoding (definition~\ref{def:reduced}). The difference is that for iterative decoding the unresolved input symbols where considered, and now we consider the active input symbols.

\begin{figure}
            \centering
            \includegraphics[width=0.55\columnwidth,draft=false]{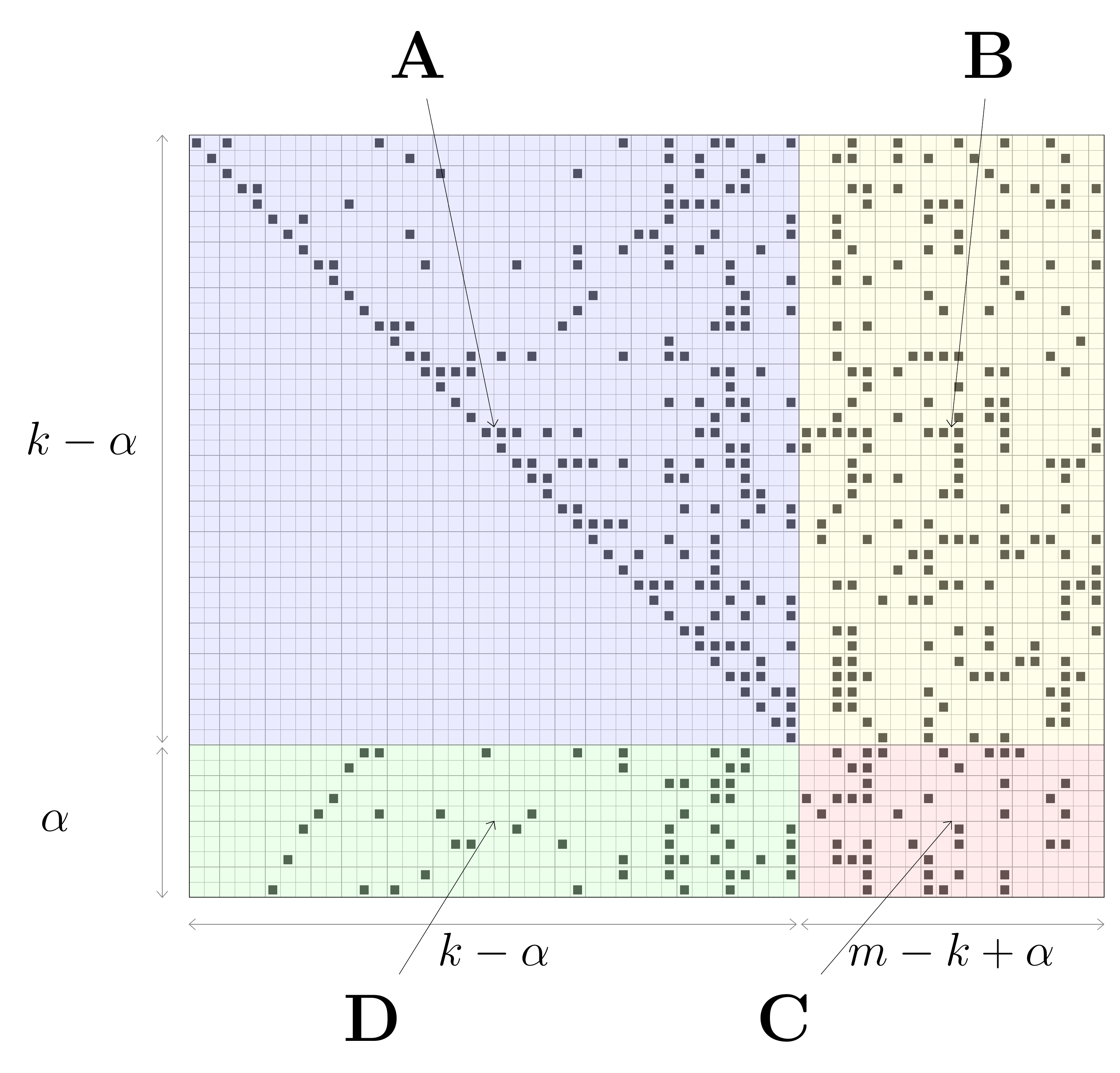}
            \caption{Structure of $\Grx$ after the triangulation process.\label{fig:inactivation_1}}
        \end{figure}
      \begin{figure}
        \centering
        \includegraphics[width=0.55\columnwidth,draft=false]{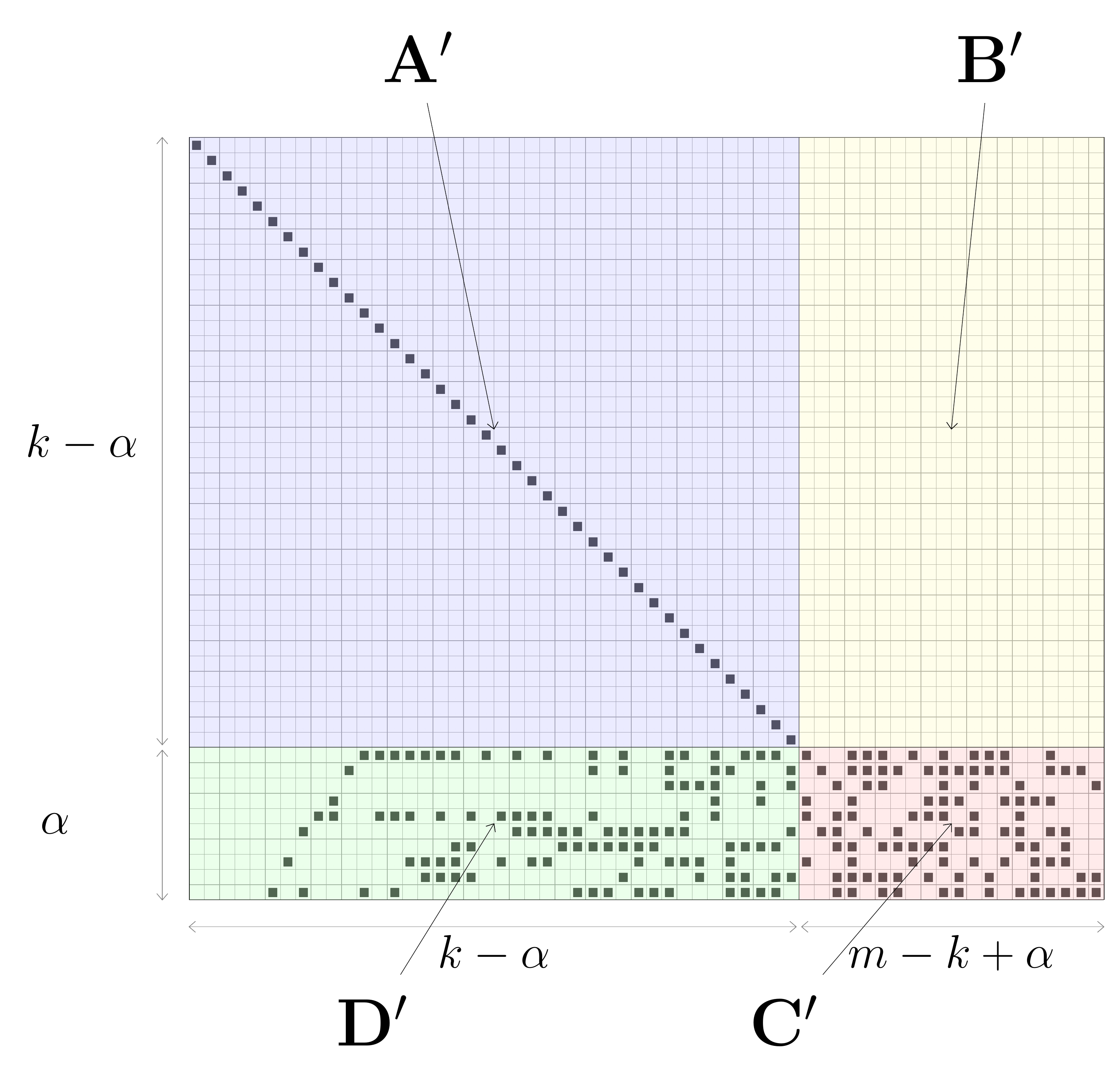}
        \caption{Structure of $\Grx$ after the zero matrix procedure.\label{fig:inactivation_2}}
        \end{figure}
        \begin{figure}
        \centering
        \includegraphics[width=0.55\columnwidth,draft=false]{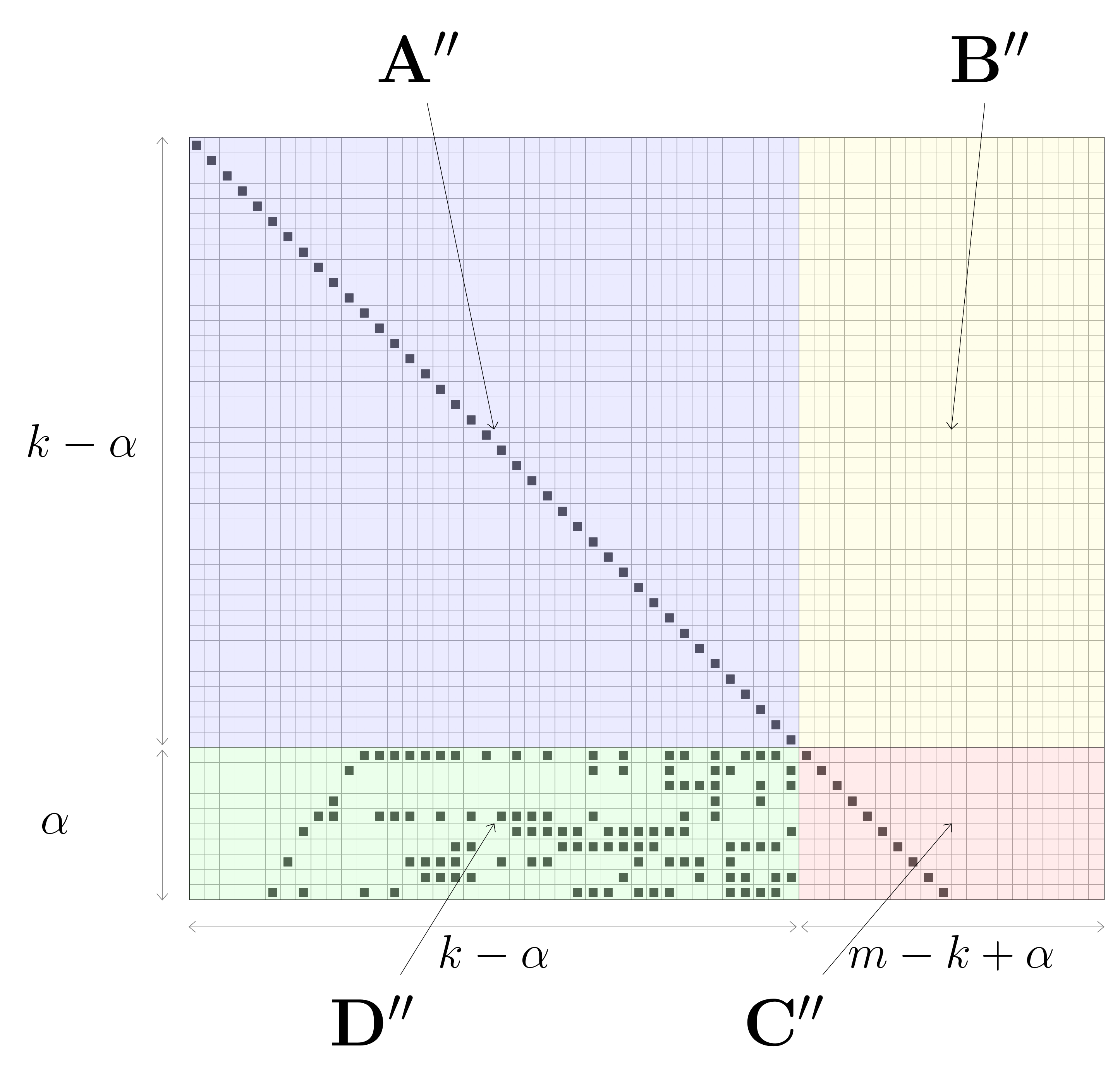}
        \caption{Structure of $\Grx$ after \acl{GE}.\label{fig:inactivation_3}}
        \end{figure}
         \begin{figure}
        \centering
        \includegraphics[width=0.55\columnwidth,draft=false]{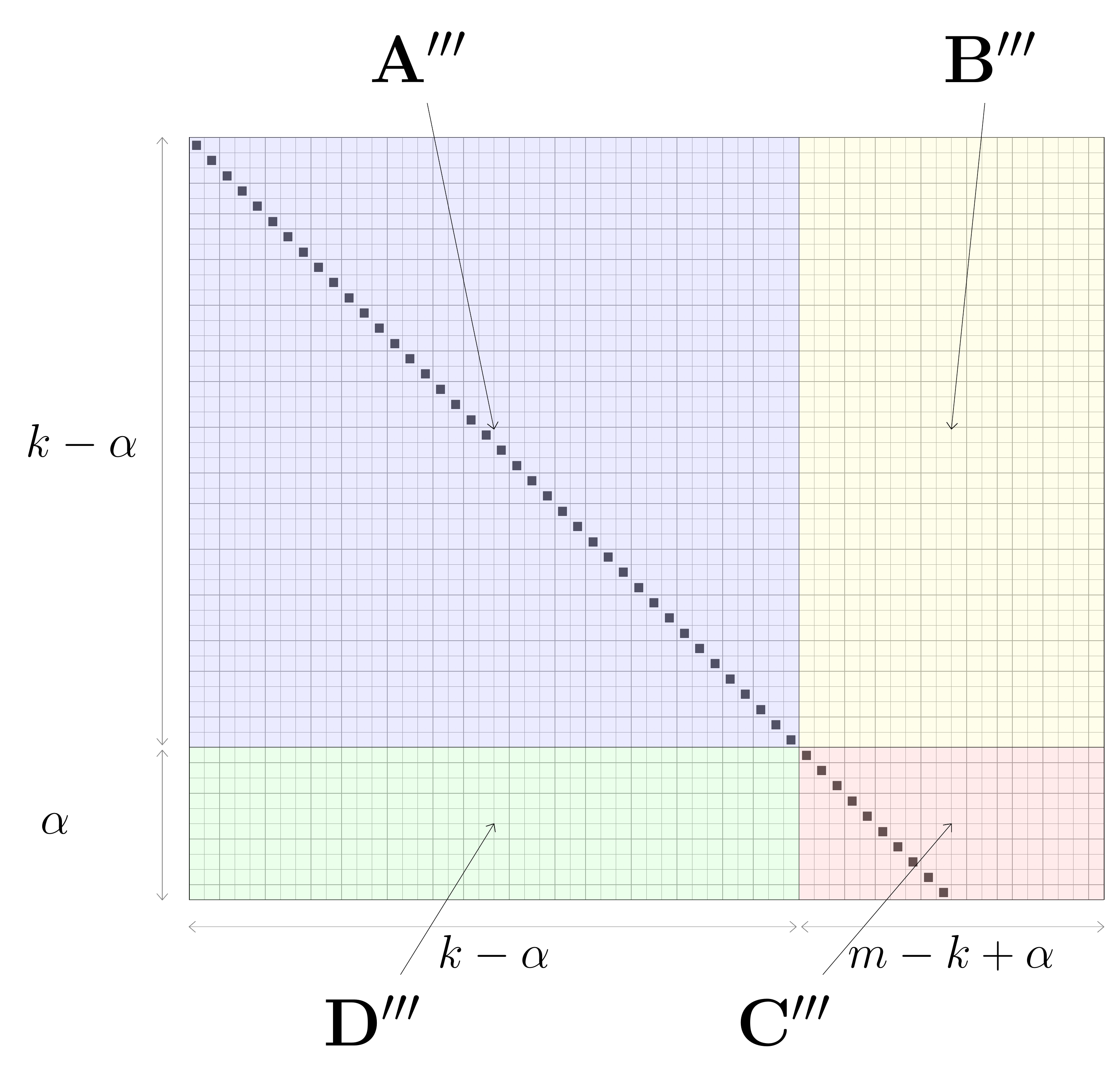}
        \caption{Structure of $\Grx$ after back-substitution.\label{fig:inactivation_4}}
        \end{figure}
\FloatBarrier
The triangulation process and iterative decoding are related to each other since both algorithms consist of an iterative pruning of the bipartite graph. In fact, one can think of
inactivation decoding as an extension of iterative decoding. When iterative decoding is used, at every decoding step one needs to have at least one output symbol in the ripple so that decoding can go on. That is, there always needs to be at least one output symbol of reduced degree one. If the ripple becomes empty at some step, iterative decoding fails. What inactivation decoding does is restarting the iterative decoding process whenever it gets blocked (empty ripple), and does so by marking one of the active input symbols as inactive. The rational behind an inactivation is that hopefully some output symbol of reduced degree two will get its degree reduced and become of reduced degree one (it will enter the ripple), so that iterative decoding can continue.
Similarly as for iterative decoding, the concepts of ripple and cloud are fundamental to understand the triangulation process. Let us recall that the ripple, $\rippleset$, is defined as the set of output symbols of reduced degree $1$, while the cloud, $\cloudset$, corresponds to the set of output symbols of reduced degree $d\geq 2$ (see Definitions \ref{def:ripple} and \ref{def:cloud}). Let us introduce some notation related to the cardinality of the ripple and cloud.  The cardinality of the ripple will be denoted by $\r$ and the corresponding random variable as $\Ripple$. The cardinality of the cloud will be denoted by $\c$ and the corresponding random variable as $\Cloud$.

Triangulation starts operating on the complete bipartite graph of the \ac{LT} code. Thus, before triangulation starts all source symbols are marked as active. At every step of the process, triangulation marks exactly one active source symbol as either \emph{resolvable} or \emph{inactive} and the symbol leaves the reduced graph. After $k$ steps the reduced graph will correspond to an empty graph. In the following, in order to keep track of the steps of the triangulation procedure
we will add a temporal dimension through the subscript $u$.  This subscript $u$ corresponds  to the number of active input symbols in the graph. Given the fact that the number of active symbols decreases by 1 at each step, triangulation will start with $u=k$ active symbols and it will end after $k$ steps with $u=0$. Therefore the subscript decreases as the triangulation procedure progresses.

The following algorithm describes the triangulation procedure at step $u$ (i.e.,  in the transition  from $u$  to $u-1$ active symbols):
\begin{absolutelynopagebreak}
\begin{algo}[Triangulation with random inactivations]\label{alg:triang}
{~}
\begin{itemize}
 \item {If the ripple $ \ripple{u}$ is not empty $(\ru>0)$}
    \begin{itemize}

        \item[] {The decoder selects an output symbol $\rosymb  \in \ripple{u}$ uniformly at random.
        The only  neighbor of $\rosymb$, i.e. the input symbol $v$, is marked as resolvable and leaves the reduced graph. The edges attached to $v$ are removed.}
    \end{itemize}
  \item {If the ripple $\ripple{u}$ is empty $(\ru=0)$}
    \begin{itemize}
        \item[] An inactivation takes place. One of the active input symbols, $v$, is chosen uniformly at random. This input symbol is marked as inactive and leaves the reduced graph.  The edges attached to $v$ are removed.
    \end{itemize}
\end{itemize}
\end{algo}
\end{absolutelynopagebreak}

\fran{Note that choosing the input symbol to be inactivated at random is certainly not the only possible inactivation strategy. However, this strategy makes the analysis trackable. For an overview of the different inactivation strategies we refer the reader to Section~\ref{sec:inact_strat}.}

At the end of the procedure, the source symbols which are marked as resolvable correspond to the rows of matrices $\Amatrix$ and $\Bmatrix$ in Figure~\ref{fig:inactivation_1}. Similarly, the source symbols marked as inactive correspond to the rows of matrices $\Cmatrix$ and $\Dmatrix$.

In order to illustrate Algorithm~\ref{alg:triang} (triangulation) we provide an example for an LT code with $k=4$ source symbols and $m=4$  output symbols. Before triangulation starts all source symbols are active. Figure~\ref{fig:example_4} shows the bipartite graph of our \ac{LT} code before triangulation starts. In the graph we can see how all $4$ source symbols are active. If we now look at the output symbols we can see how the ripple and the cloud are composed of two elements each,  $\rippleset=\{\rosymb_1, \rosymb_4\}$ and $\cloudset=\{\rosymb_2, \rosymb_3\}$.

\begin{figure}
\begin{center}
\includegraphics[width=0.45\columnwidth]{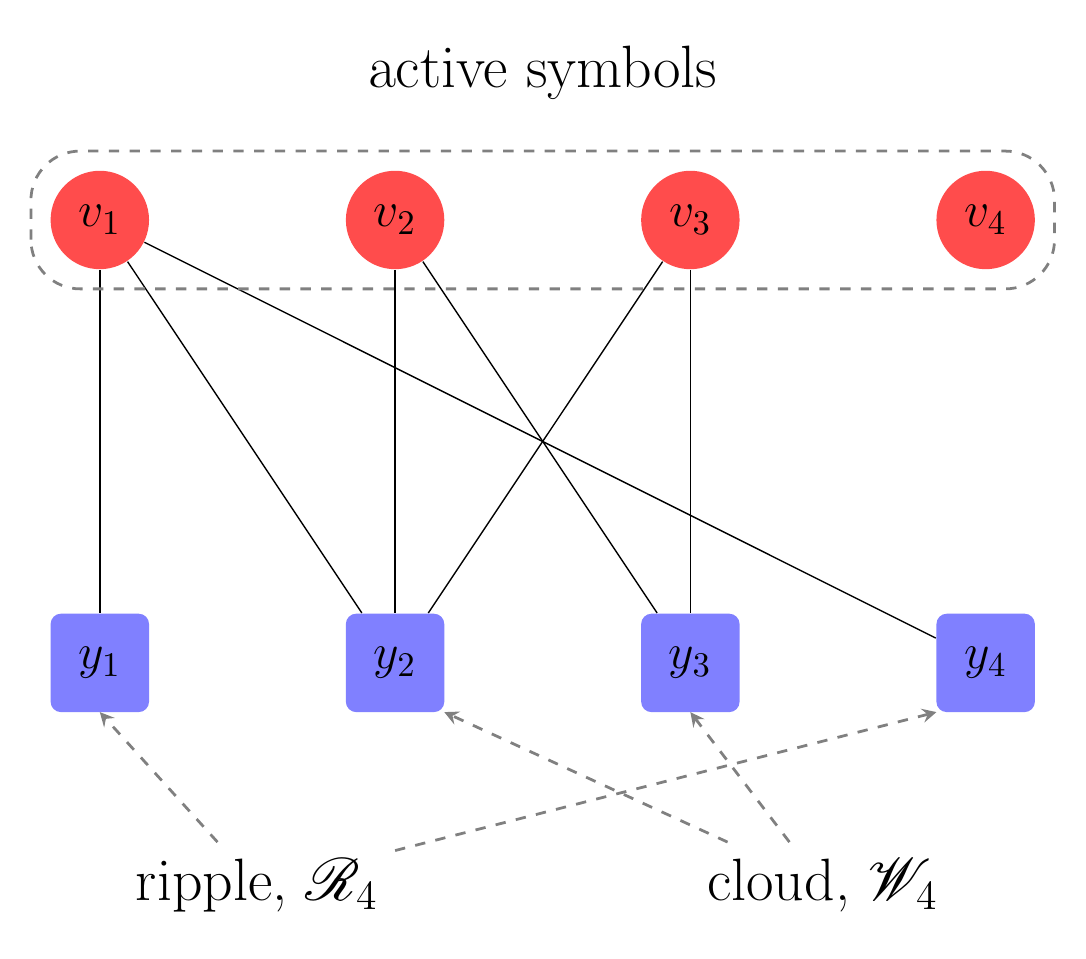}
\centering \caption{Triangulation procedure example, $u=4$.\label{fig:example_4}}
\end{center}
\end{figure}

Triangulation operates as follows:

\FloatBarrier

\begin{enumerate}
    \item {Transition from $u=4$ to $u=3$. At the initial step $u=4$, there are two output symbols in the ripple, $\r_{4}=2$ (see Figure~\ref{fig:example_4}). Hence, in the transition to $u=3$ one of the source symbols ($v_1$) is marked as  resolvable, it leaves the graph and all its attached edges are removed. The graph obtained after the transition from $u=4$ to $u=3$ is shown in  Figure~\ref{fig:example_3}. We can see how nodes $\rosymb_1$ and $\rosymb_4$ have left the graph since their reduced degree became zero. }
        \begin{figure}[h]
        \begin{center}
        \includegraphics[width=0.45\columnwidth]{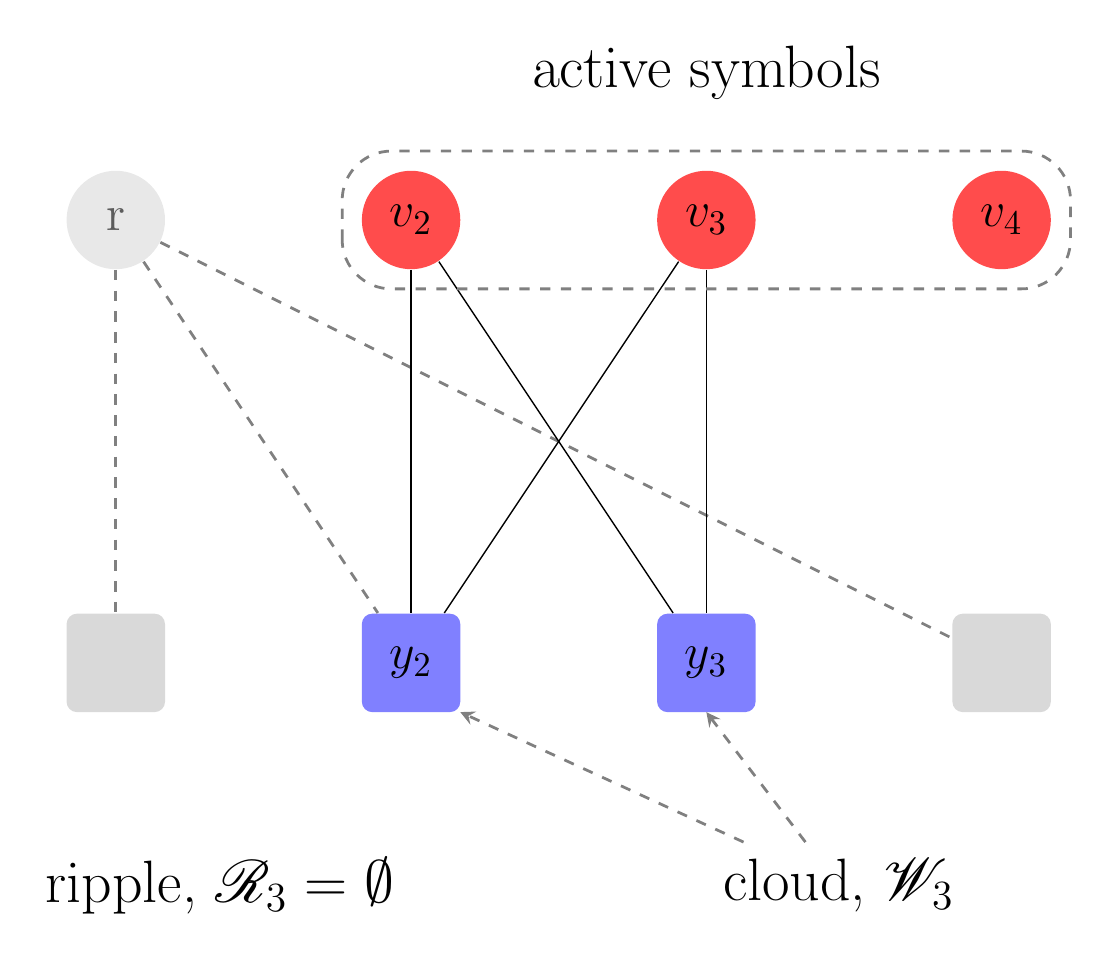}
        \centering \caption{Triangulation procedure example, $u=3$.\label{fig:example_3}}
        \end{center}
        \end{figure}
\FloatBarrier
    \item{Transition from $u=3$ to $u=2$. In Figure~\ref{fig:example_3} we can see how now the ripple is empty, $\r_{3}=0$. Therefore, in the transition to $u=2$ an inactivation takes place. Node $v_2$ is chosen at random and is marked as inactive. All edges attached to $v_2$ are removed from the graph. As a consequence the nodes $\rosymb_2$ and $\rosymb_3$ that were in the cloud $\cloudset_3$ become of reduced degree 1 and  enter the ripple $\rippleset_2$. This can be observed in  Figure~\ref{fig:example_2}.}
        \begin{figure}[h]
        \begin{center}
        \includegraphics[width=0.45\columnwidth]{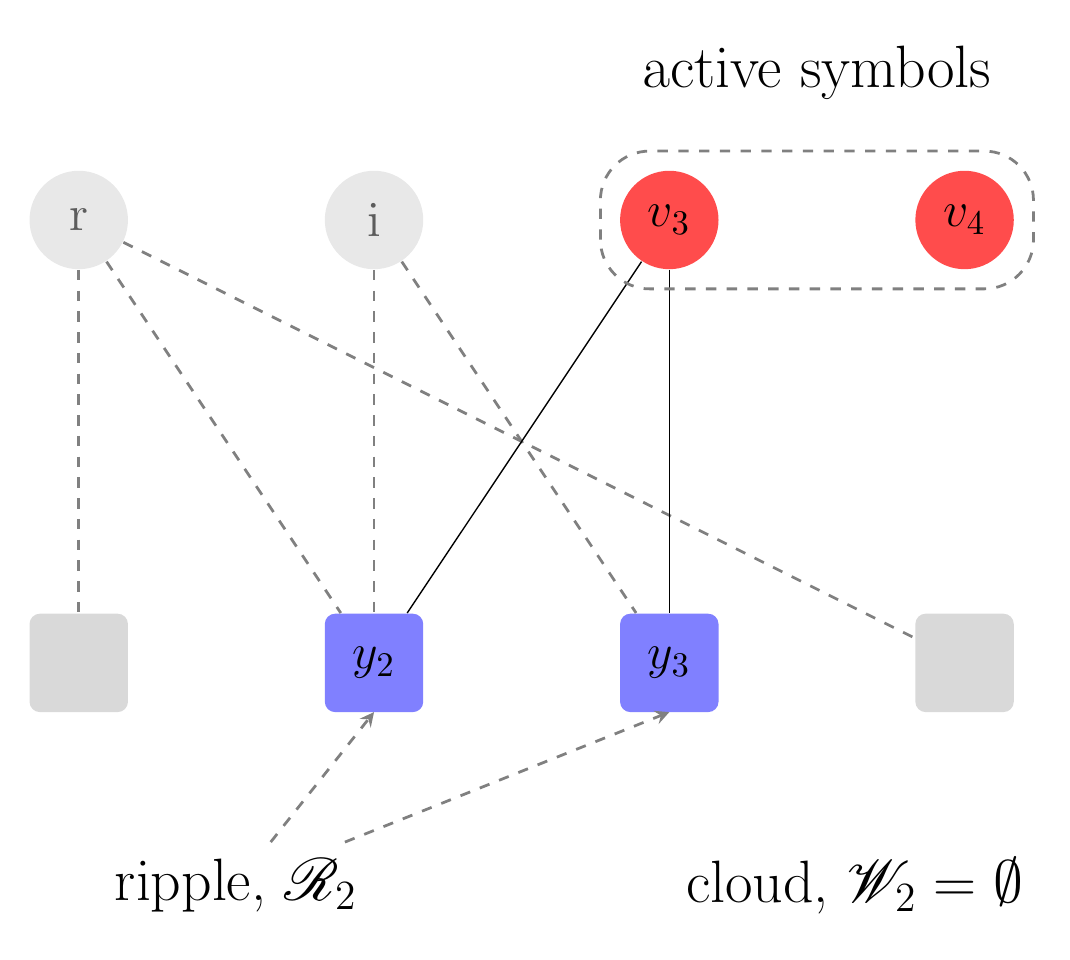}
        \centering \caption{Triangulation procedure example, $u=2$.\label{fig:example_2}}
        \end{center}
        \end{figure}
\FloatBarrier
    \item{Transition from $u=2$ to $u=1$. We can see in Figure~\ref{fig:example_2} how the ripple is not empty, in fact, $\r_{2}=2$. Source symbol $v_3$ is marked as resolvable and all its attached edges are removed. Nodes $\rosymb_2$ and $\rosymb_3$ leave the graph because their reduced degree becomes zero (see Figure~\ref{fig:example_1}).}
        \begin{figure}[h]
        \begin{center}
        \hspace*{0.7 cm}\includegraphics[width=0.51\columnwidth]{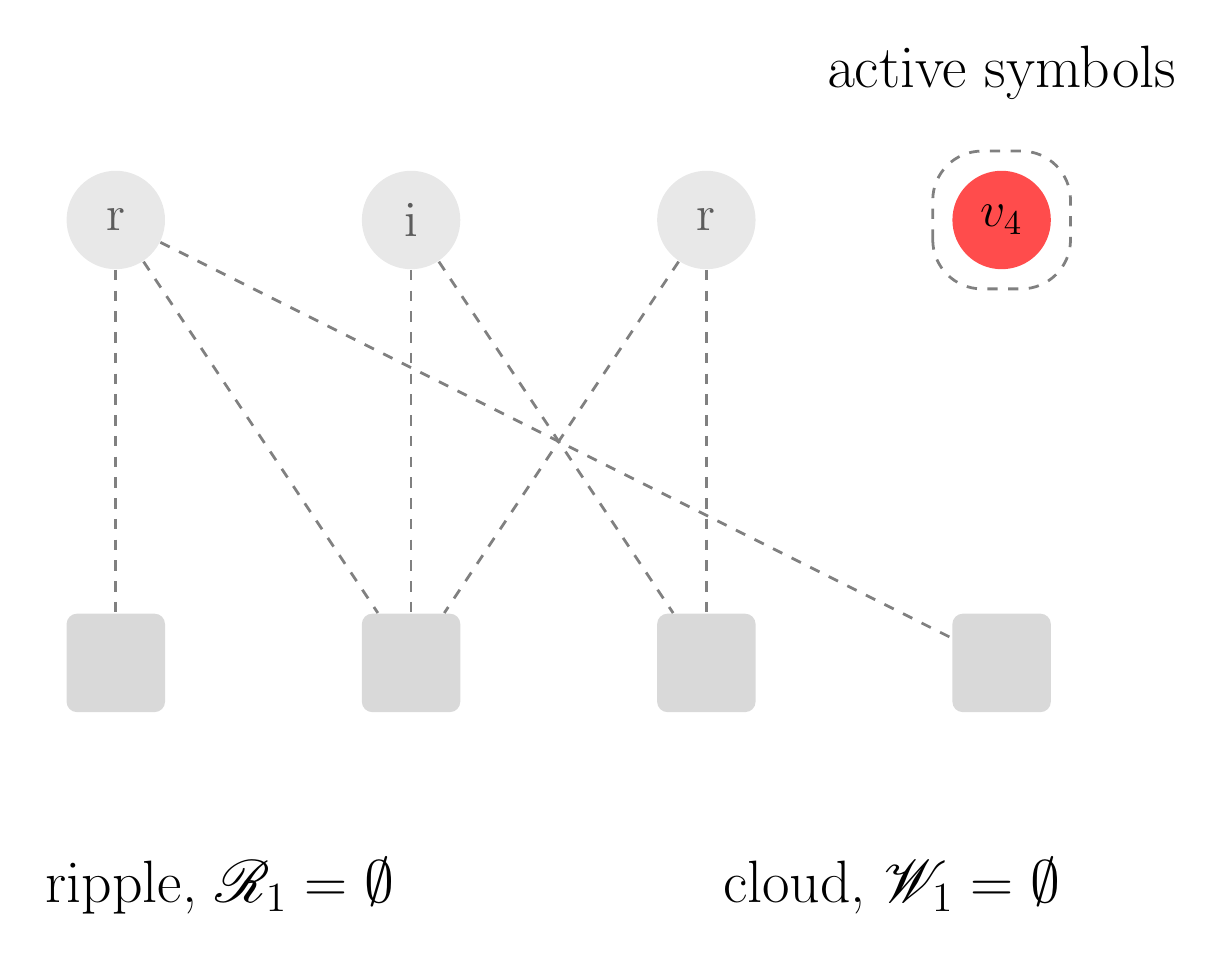}
        \centering \caption{Triangulation procedure example, $u=1$.\label{fig:example_1}}
        \end{center}
        \end{figure}
\FloatBarrier
    \item{Transition from $u=1$ to $u=0$. In Figure~\ref{fig:example_1} we can see how the ripple and cloud are now empty. Hence, an inactivation takes place: node $v_4$ is marked as inactive and the triangulation procedure ends.}
        \begin{figure}[h]
        \begin{center}
        \includegraphics[width=0.45\columnwidth]{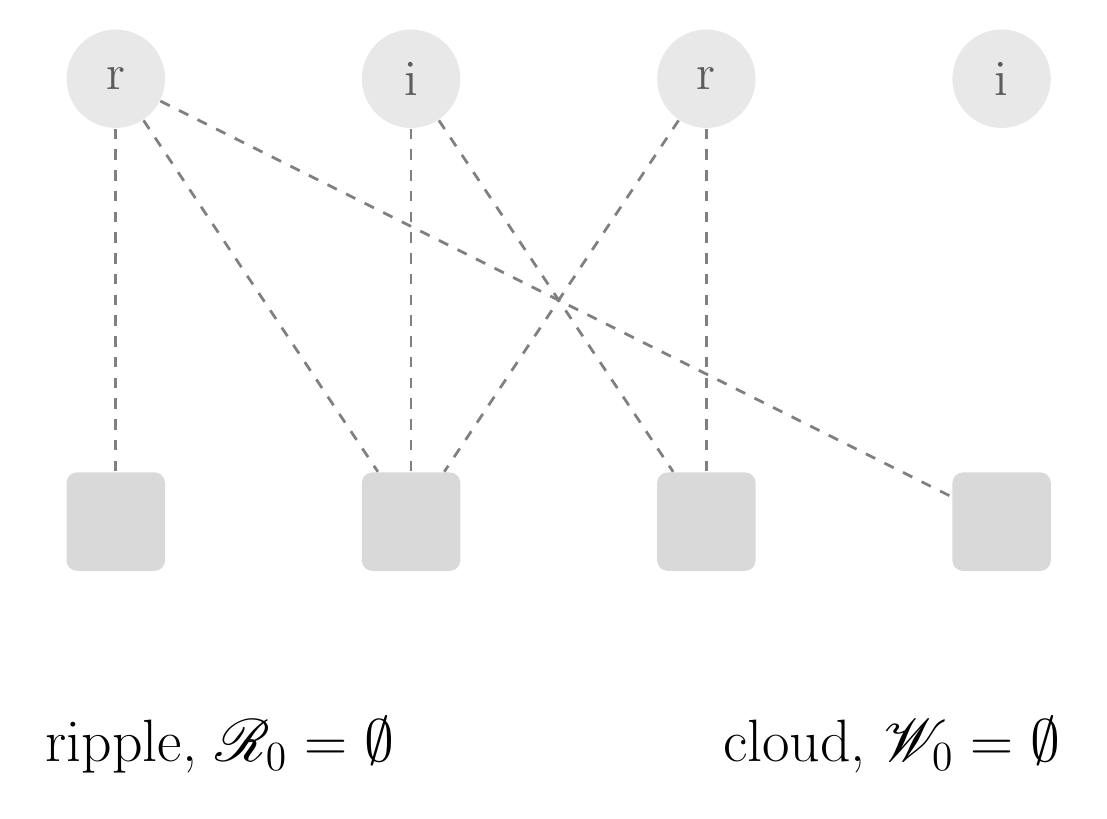}
        \centering \caption{Triangulation procedure example, $u=0$.\label{fig:example_0}}
        \end{center}
        \end{figure}
\FloatBarrier
\end{enumerate}

\fran{For illustration we also show the effect of triangulation on the generator matrix $\Grx$ in Figure~\ref{fig:example_mat_0}. Concretely, on the left hand side we can see matrix $\Grx$ before the triangulation procedure starts, and on the right hand side we can see matrix $\Grx$ after the triangulation procedure ends. We can see that after triangulation the upper left corner of matrix $\Grx$ has an upper triangular shape. We can also see that triangulation reorders the rows and columns of matrix $\Grx$ in order to create an upper diagonal matrix. Concretely, the first row corresponds to the first input symbol that was marked as resolvable, $v_1$. The second row corresponds to the second input symbol that was marked  as resolvable, $v_3$. The last row corresponds to the first inactivated input symbol, $v_2$, and the second to last (third) row to the second symbol that was marked inactive $v_4$. In the following we will stick to a bipartite graph representation of the triangulation procedure, since it allows us to ignore the row and column reordering, making inactivation conceptually simpler.}

\begin{figure}[t]
        \centering
        \begin{subfigure}[b]{0.49\textwidth}
        \centering
            \includegraphics[width=0.9\columnwidth]{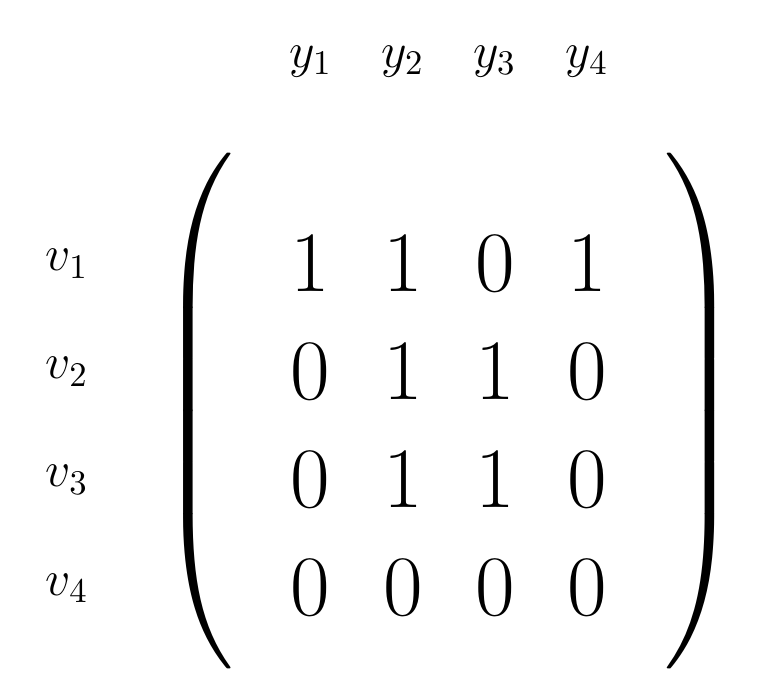}
            \subcaption{Before triangulation}
        \end{subfigure}
         \begin{subfigure}[b]{0.49\textwidth}
        \centering
            \includegraphics[width=0.9\columnwidth]{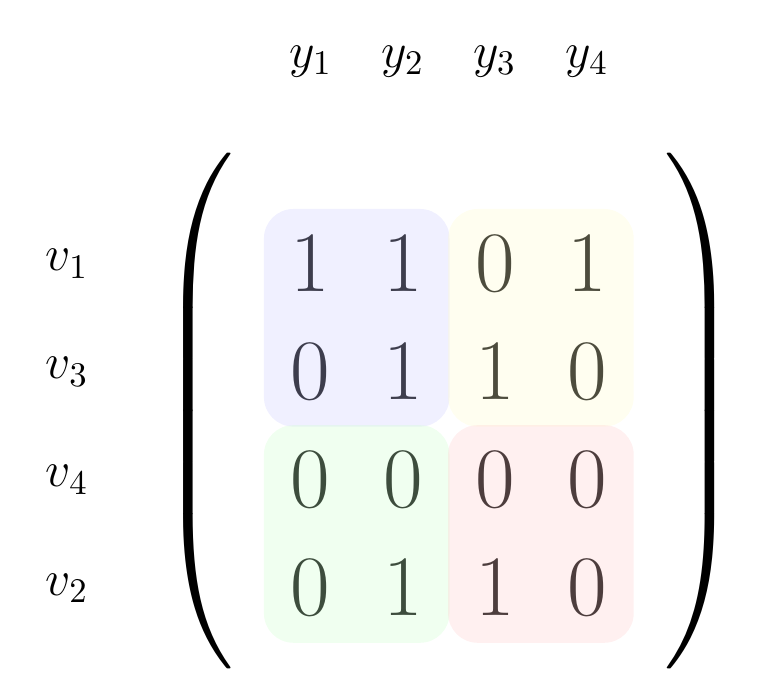}
            \subcaption{After triangulation}
        \end{subfigure}
         \caption{Generator matrix $\Grx$ before and after the triangulation procedure.} .\label{fig:example_mat_0}
\end{figure}

\FloatBarrier
\subsection{Inactivation Strategies}\label{sec:inact_strat}

In this thesis we will focus mostly on a specific inactivation strategy, namely, random inactivation. This inactivation strategy is chosen in order to render the analysis trackable. However, other inactivation strategies exist that lead to a lower number of inactivations, and, therefore, to a decreased decoding complexity. In this Section we will give a short overview of the different inactivation strategies that can be found in literature.

An inactivation consists simply of marking one of the active input symbols as inactive. Therefore, when performing an inactivation the decoder will be presented with as many choices as active symbols are present at that decoding stage.
For illustration, we provide a reduced decoding graph in Figure~\ref{fig:example_app}. The reduced decoding graph represents an \ac{LT} code with $k=9$ source nodes and $m=10$ output nodes. It can be observed how source nodes $v_6$  and $v_7$ have been previously marked as resolvable, whereas source nodes $v_8$ and $v_9$ have been previously marked as inactive. Source nodes $v_1$, $v_2$, $v_3$, $v_4$ and $v_5$ are still active. We can also see how $3$ output symbols have left the reduced graph, since their reduced degree is zero. Output symbols $\rosymb_1$, $\rosymb_2$, ..., $\rosymb_7$ are still in the reduced graph.

\begin{figure}[t]
        \begin{center}
        \includegraphics[width=\textwidth]{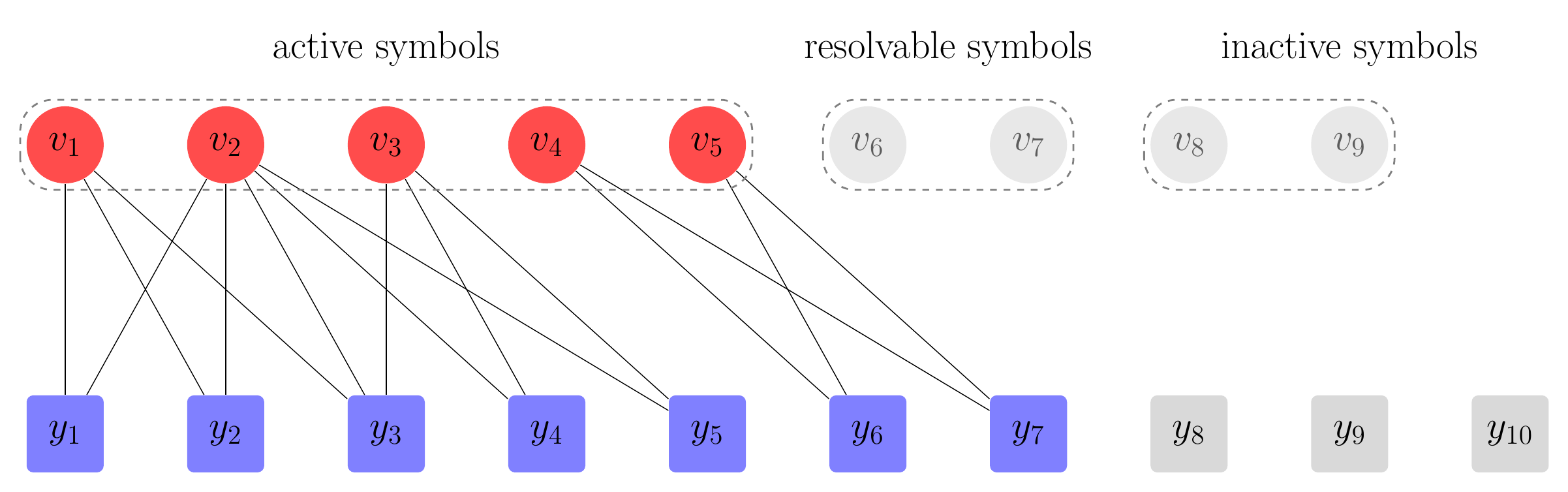}
        \centering \caption{Example of decoding graph of an \ac{LT} code under inactivation decoding \label{fig:example_app}}
        \end{center}
\end{figure}

In the following we will present several algorithms that correspond to different inactivation strategies.

\begin{algo}[Random inactivation]
One of the active input symbols, $v$, is chosen uniformly at random.  This input symbol is marked as inactive and leaves the reduced graph\footnote{This algorithm is a sub-algorithm of Algorithm~\ref{alg:triang}. It is provided for the sake of completeness.}.
\end{algo}

\fran{Note that random inactivation does not depend on the bipartite graph formed by the active symbols. The main advantages of this strategy are its simplicity, and the fact that it renders the analysis on inactivation decoding trackable.}
If random inactivation is applied to the decoding graph in Figure~\ref{fig:example_app} one of the five active input symbols is chosen at uniformly at random.

\begin{algo}[Maximum reduced degree inactivation]
The decoder marks as inactive the source symbol with maximum reduced degree. In case there are several source symbols with the same reduced degree\footnote{As for output symbols, the reduced degree of an input symbol is its degree in the reduced graph.}, one of the source symbols with maximum reduced degree is selected uniformly at random and it is marked as inactive.
\end{algo}

\fran{Maximum reduced degree inactivation aims at making the bipartite graph as sparse as possible by inactivating the input symbol with the most edges attached.} For illustration, we will apply maximum reduced degree inactivation to the decoding graph in Figure~\ref{fig:example_app}. We can see how among the four input symbols, $v_2$ has reduced degree $5$, $v_1$ and $v_3$ have reduced degree $3$, and $v_4$ and $v_5$ have reduced degree $2$. Hence, the decoder will inactivate $v_2$.

In order to introduce the next inactivation strategy we first need to introduce a definition

\begin{mydef}[Accumulated reduced degree]
The accumulated reduced degree of an output symbol is defined as the sum of the reduced degrees of all its neighbors (input symbols). Formally, for a output symbol of degree $d$, and denoting its neighbors as $v_1, v_2, \hdots, v_d$, its accumulated reduced degree is defined as
\[
\sum_{i=1}^{d} \deg_r{(v_i)}.
\]
\end{mydef}

Making use of this definition we can introduce the next inactivation strategy.

\begin{algo}[Maximum accumulated reduced degree inactivation]
The decoder selects the output symbol with minimum reduced degree in the reduced graph. In case there are several output symbols with the same minimum reduced degree, the decoder computes the accumulated reduced degree of each of them and selects the one with maximum accumulated reduced degree. In case there are several output symbols with same minimum reduced degree and same maximum accumulated reduced degree, one of them is chosen at random. The decoder then marks as inactive one of the neighbors of the selected output symbol.
\end{algo}

\fran{The rational behind maximum accumulated reduced degree inactivation is trying to make an input symbol enter the ripple as soon as possible.} Let us also apply maximum accumulated weight inactivation to the decoding graph in Figure~\ref{fig:example_app}. It is easy to see how we have $6$ output symbols of minimum degree $2$, $\rosymb_1$, $\rosymb_2$, $\rosymb_4$,  $\rosymb_5$, $\rosymb_6$ and $\rosymb_7$, with respective accumulated weights $8$, $8$, $8$, $8$, $4$ and $4$. Thus the decoder selects one output symbol at random among those with minimum reduced degree and maximum accumulated reduced degree,i.e., among $\rosymb_1$, $\rosymb_2$, $\rosymb_4$ and  $\rosymb_5$. Let us assume that $\rosymb_1$ is selected. Next, one of its two neighbors is selected at random, for example $v_1$ \fran{and is inactivated. As a consequence $\rosymb_1$ and  $\rosymb_2$ enter the ripple.}

In order to introduce the last inactivation strategy, it is necessary to introduce some concepts dealing with graphs. This concepts are valid for generic graphs (there is no implicit assumption of the graph being a bipartite graph).

\begin{mydef}[Path]
In a graph, a path is a sequence of edges that connects a sequences of vertices (nodes).
\end{mydef}

\begin{mydef}[Connected component]
A connected component of a graph is a subgraph in which any two vertices (nodes) are connected to each other by a path, and in which no vertex (node) is connected to any vertices outside the connected component.
\end{mydef}

Commonly, connected components are referred simply as components. Note that if two nodes, $A$ and $B$ belong to the same component there exists at least one path to go from $A$ to $B$. This does not imply that $A$ and $B$ are neighbors, since the path can be composed of any number of edges. However, if $A$ and $B$ do not belong to the same component there is no path to go from $A$ to $B$.

\begin{algo}[Maximum component inactivation]
{~}
\begin{itemize}
\item The decoder searches for all degree 2 output symbols.
\item If there are no degree 2 output symbols the decoder inactivates one input symbol at random.
\item Otherwise (if there are 1 or more output symbols of degree 2.)
\begin {itemize}
\item The decoder computes the unipartite graph induced by the degree 2 output symbols on the input symbols so that
\begin{itemize}
\item each vertex in this graph corresponds to a degree 2 output symbol
\item two vertices of the induced unipartite graph are only connected to each other if the corresponding degree 2 output symbols have a neighbor (input symbol) in common
\end{itemize}
\item The decoder searches the components in the unipartite graph and computes its size (number of nodes).
\item The decoder inactivates one input symbol that is connected to the maximum component (that of maximum size) in the unipartite graph.
\end{itemize}
\end{itemize}
\end{algo}

The rational behind maximum component inactivation is that when we inactivate  any of the input symbols connected to a component, in the subsequent steps of inactivation decoding all the input symbols connected to the component are marked as resolvable, and, thus, no inactivation happens.

Let us apply maximum component inactivation to the decoding graph in Figure~\ref{fig:example_app}. The degree $2$ output symbols induce the graph shown in Figure~\ref{fig:components}. We can observe how there are two connected components. The first is formed by $\rosymb_1$, $\rosymb_2$, $\rosymb_4$, $\rosymb_5$ and the second by $\rosymb_6$ and $\rosymb_7$.  The decoder hence marks as inactive one input symbol connected to the largest component. Hence, the decoder inactivates either, $v_1$, $v_2$ or $v_3$ (see Figure~\ref{fig:example_app}).

\begin{figure}[t]
        \begin{center}
        \includegraphics[width=\textwidth]{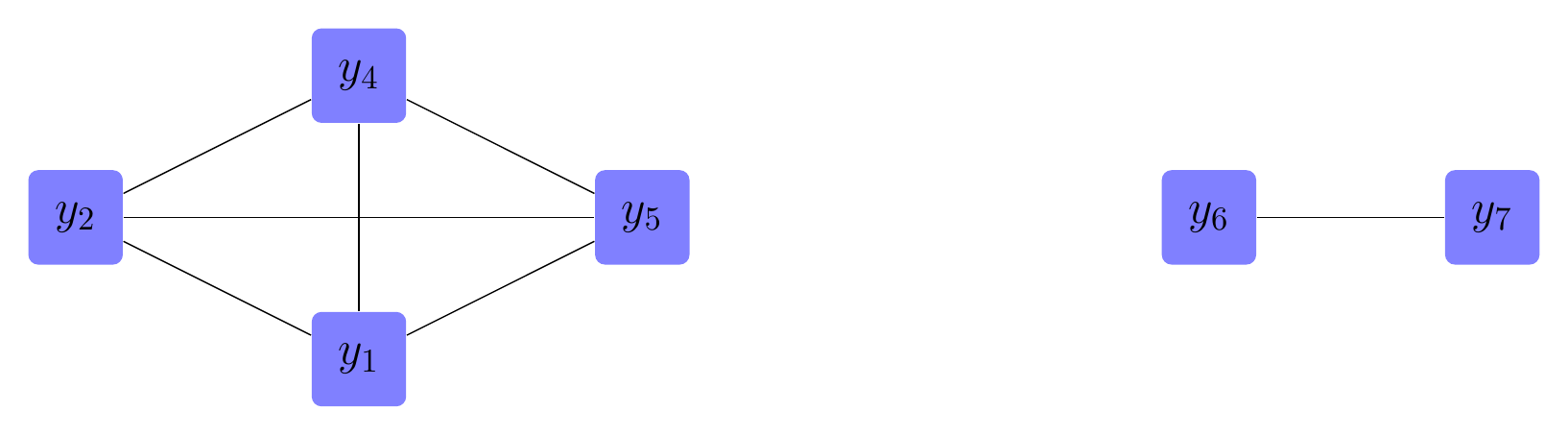}
        \centering \caption{Connected components of the decoding example.\label{fig:components}}
        \end{center}
\end{figure}

In practice different inactivation strategies lead to a different number of inactivations. Among the different inactivation strategies presented in this section, maximum component inactivation usually provides the least number of inactivations, followed closely by maximum accumulated weight inactivation, maximum weight inactivation and at last random inactivation. A detailed comparison of the performance of the different inactivation strategies for R10 Raptor codes can be found in Appendix~\ref{app:practical}.

\section{Analysis under Random Inactivation} \label{chap:inact_analysis}

In the following sections we present novel finite length analysis methods for \ac{LT} codes under inactivation decoding. The goal of these methods is obtaining the number of input symbols that are inactivated after triangulation is over (or an estimation thereof).

\subsection{First Order Finite Length Analysis} \label{chap:inact_first_order}
In \cite{Karp2004,shokrollahi2009theoryraptor,Maatouk:2012} the iterative decoder of \ac{LT} codes was analyzed using a dynamic programming approach. This analysis models the iterative decoder as a finite state machine and it can be used to derive the probability of decoding failure (under iterative decoding). In this section we extend the analysis of the iterative decoder performed in \cite{Karp2004,shokrollahi2009theoryraptor,Maatouk:2012} to the inactivation decoder. The analysis we present in this section is similar to the analysis of batched sparse codes under inactivation decoding presented in \cite{Chingbats}.

As in \cite{Karp2004,shokrollahi2009theoryraptor,Maatouk:2012}, we model the decoder as a finite state machine with state
\[
\S{u}:=(\Cu, \Ru )
\]
where we recall, $\Cu$  and $\Ru$ are respectively the number of output symbols in the cloud and ripple when $u$ output symbols are still active.
In this section a recursion is derived that allows to obtain  $\Pr \{ \S{u-1}=(\c_{u-1}, \r_{u-1}) \}$ as a function of $\Pr \{ \S{u}=(\cu, \ru )\}$ .

Let us first analyze how the ripple and cloud change in the transition from $u$ to $u-1$ active source symbols. In the transition exactly one active source symbol is marked as either resolvable or inactive and all its attached edges are removed. Whenever edges are erased in the graph the degree of one or more output symbols gets reduced.  Consequently, some of the cloud symbols may enter the ripple and some of the ripple symbols may become of reduced degree zero and leave the reduced graph.
We first focus on the symbols that leave the cloud and enter the ripple in the transition given that $\S{u}=(\c_u, \r_u)$.  Since in an \ac{LT} code the neighbors of all output symbols are selected independently and uniformly at random, in a transition each output symbol will leave the cloud and enter the ripple independently from other output symbols. Thus,
the number of cloud symbols which leave $\cloud{u}$ and enter $\ripple{u-1}$ is binomially distributed with parameters $\c_u$ and $p_u$, being $p_u$ the probability of a symbol leaving $\cloud{u}$ to enter $\ripple{u-1}$. Using Bayes' theorem $p_u$ can be written as:
\begin{align}
p_u := \Pr \{ \rosymb \in \ripple{u-1} | \rosymb \in \cloud{u} \} = \frac { \Pr \{ \rosymb \in \ripple{u-1}\, , \, \rosymb \in \cloud{u} \} }  { \Pr \{ \rosymb \in \cloud{u} \}}.
\label{eq:pu_prob}
\end{align}


Let us first consider the numerator of \eqref{eq:pu_prob} assuming that output symbol $\rosymb$ has (original) degree $d$:
\[
\Pr \{ \rosymb \in \ripple{u-1}\, , \, \rosymb \in \cloud{u} | \deg(\rosymb)= d\}.
\]
   This corresponds to the probability that
   \begin{itemize}
   \item one of the $d$ edges of output symbol $\rosymb$ is connected to the symbol being marked as inactive or resolvable at the transition,
   \item another edge is connected to one of the $u-1$ active symbols after the transition,
   \item the remaining $d-2$ edges connected to the $k-u$ not active input symbols (inactive or resolvable).
   \end{itemize}
   In other words, the symbol must have \emph{reduced} degree $2$ \emph{before} the transition and \emph{reduced} degree $1$ \emph{after} the transition.
\begin{prop}
The probability that a symbol $\rosymb$ belongs  to the cloud at step $u$  and  enters the ripple at step $u-1$,  given its original degree $d$ is given by
\begin{align}
 \Pr \{ \rosymb \in \ripple{u-1}\, , \,  \rosymb \in \cloud{u} | \deg(\rosymb)= d \} =   \begin{cases}
  \frac{d}{k} (d-1)\frac{u-1}{k-1}  \frac{\binom{k-u}{d-2}}{\binom{k-2}{d-2}} & \text{if }  d\geq 2 \\
  0 & \text{if }  d < 2
 \end{cases}.
\label{eq:z_and_l_d}
\end{align}
\end{prop}
\begin{proof}
First, the probability  that one edge is connected to the symbol being marked as inactive or resolvable at the transition is $1/k$, and there are $d$ distinct choices for the edge connected to it. This accounts for the term $d/k$ in \eqref{eq:z_and_l_d}.

Second, there are  $d-1$ choices for the edge going to the $u-1$ active symbols after the transition, and the probability of an edge being connected to the set of $u-1$ active symbols is $(u-1)/(k-1)$. This is reflected in the term $(d-1)(u-1)/(k-1)$ in \eqref{eq:z_and_l_d}.

Finally, the last term corresponds to the probability of having exactly $d-2$ edges going to the $k-u$ not active input symbols:
\[
\frac{\binom{k-u}{d-2}}{\binom{k-2}{d-2}}. \qedhere
\]
\end{proof}
If the conditioning on $d$ in \eqref{eq:z_and_l_d} is removed we obtain
\begin{align}
\Pr \{ \rosymb  \in \ripple{u-1}\, , \,  \rosymb  \in \cloud{u} \} =
\mathlarger {\sum}_{d=2}^{\dmax}   \Omega_d \frac{d}{k} (d-1)  \frac{u-1}{k-1}  \frac{\binom{k-u}{d-2}}{\binom{k-2}{d-2}}.
\label{eq:z_and_l}
\end{align}

\medskip

The denominator of \eqref{eq:pu_prob} is given by the probability that the randomly chosen output symbol $\rosymb$ is in the cloud when $u$ input symbols are still active. This is equivalent to the probability of not being in the ripple or having reduced degree zero (all edges are going to symbols marked as inactive or resolvable) as provided by the following Proposition.
\begin{prop}
The probability that the randomly chosen output symbol $\rosymb$ is in the cloud when $u$ input symbols are still active corresponds to
\begin{align}
\Pr & \{ \rosymb \in \cloud{u}\}=  1 -  \mathlarger{\sum}_{d=1}^{\dmax}  \Omega_d  \left[ u\frac{\binom{k-u}{d-1}}{\binom{k}{d}} + \frac{\binom{k-u}{d}}{\binom{k}{d}}\right].
\label{eq:z}
 \end{align}
\end{prop}
\begin{proof}
The probability of $\rosymb$ not being in the cloud is given by the probability of $\rosymb$ having reduced degree $0$ or being in the ripple. Since the two events are mutually exclusive, we can compute such probability as the sum of two probabilities, the probability of $\rosymb$ being in the ripple (i.e., having reduced degree $1$) and the probability of $\rosymb$ having reduced degree $0$.

We will first focus on the probability of $\rosymb$ being in the ripple. Let us assume $\rosymb$ is of degree $d$. The probability that $\rosymb$ has reduced degree $1$ equals the probability of $\rosymb$ having exactly one neighbor among the $u$ active source symbols and the remaining $d-1$ neighbors among the $k-u$ non-active ones. This is given by
\[
d ~\frac{u}{k} ~\frac{\binom{k-u}{d-1}}{\binom{k-1}{d-1}} = u\frac{\binom{k-u}{d-1}}{\binom{k}{d}}
\]
that corresponds to the first term in \eqref{eq:z}.

The probability of $\rosymb$ having reduced degree $0$ is the probability that all $d$ neighbors of $\rosymb$ are in the $k-u$ non-active symbols. This leads to the term

\[
\frac{\binom{k-u}{d}}{\binom{k}{d}}
\]
in \eqref{eq:z}.
\end{proof}

The probability $p_u$ can be finally obtained through \eqref{eq:pu_prob}, making use of \eqref{eq:z_and_l} and of \eqref{eq:z} and corresponds to:
\begin{align}
p_u =  \frac{ \mathlarger {\sum}\limits_{d=2}^{\min(\dmax,k-u+2 )}   \Omega_d \,d \,(d-1)\frac{1}{k}  \frac{u-1}{k-1}  \frac{\binom{k-u}{d-2}}{\binom{k-2}{d-2}} }
{  1 -  \mathlarger{\sum}\limits_{d=1}^{\myop{\min (\dmax, k-u+1)}}    \Omega_d \,u \frac{\binom{k-u}{d-1}}{\binom{k}{d}d}
-  \mathlarger{\sum}\limits_{d=1}^{\myop{\min (\dmax, k-u)}}   \Omega_d \frac{\binom{k-u}{d}}{\binom{k}{d}}    }
\end{align}
where  the upper limits of the summations have been adjusted to take into account that an output symbol cannot have more than $d$ edges going to a set of $d$ input symbols.

Let us now focus on the number of symbols leaving the ripple during the transition from $u$ to $u-1$ active symbols, which we shall denote by $\erv_u$. We denote by $\Erv_u$ the random variable associated with $\erv_u$. We distinguish two cases.
In the first case, the ripple is not empty and no inactivation takes place. Hence, an output symbol $\rosymb$ is chosen at random from the ripple and its only neighbor $v$ is marked as resolvable and it is removed from the graph. Any other output symbol in the ripple which is connected to the input symbol $v$ leaves the ripple during the transition. Hence, for $\ru>0$ we have
\begin{align}
\Pr\{\Erv_u=\erv_u & |\Ru=\ru\} = \binom{\ru-1}{\erv_u-1} \left(\frac{1}{u}\right)^{\erv_u-1} \left( 1- \frac{1}{u} \right)^{\ru-\erv_u}
\end{align}

\vspace{4pt}

\noindent
with $1\leq \erv_u \leq \ru$.

In the second case, the ripple is empty ($\ru=0$). Since no output symbols can leave the ripple, we have
\begin{align}
\Pr\{\Erv_u=\erv_u|\Ru=0\} = \begin{cases} 1  & \text{if } \erv_u = 0 \\ 0  & \text{if } \erv_u \neq 0 \end{cases}.
\end{align}

\medskip

Now we are in the position to derive the transition probability
\[
\Pr\{\S{u-1}=(\c_{u-1},\r_{u-1})|\S{u}=(\c_{u},\r_{u})\}.
\]
Let us recall that, by definition, $\b_u$ denotes the variation of number of cloud elements in the transition from $u$ to $u-1$ active symbols. Formally,
\[\b_u:=\cu-\c_{u-1}.
\]
Let us also observe that  the variation of the ripple size is subject to the following equilibrium constraint
 \[
 \erv_u-\bu=\ru-\r_{u-1}
 \]
 which follows from the definition $\erv_u$ and $\bu$.
 The transition probability is given by:
\begin{align}
\Pr\{\S{u-1}&=(\cu-\bu,\ru-\erv_u+\bu) | \S{u}=(\cu,\ru)\} = \\[2mm]
& \binom{\cu}{\bu} {p_u}^{\bu} (1-p_u)^{\cu-\bu} \binom{\ru-1}{\erv_u-1} \,
\left(\frac{1}{u}\right)^{\erv_u-1} \left( 1- \frac{1}{u} \right)^{\ru-\erv_u}
\label{eq:prob_transition}
\end{align}
for $\ru>0$, while for $\ru=0$ we have
\begin{align}
\Pr\{\S{u-1}&=(\cu-\bu,\bu) | \S{u}=(\cu,0)\} =
 \binom{\cu}{\bu} {p_u}^{\bu} (1-p_u)^{\cu-\bu}.
\label{eq:prob_transition_r_0}
\end{align}

\vspace{4pt}

\noindent
Finally, the probability of the decoder being in state $\S{u-1} =(\c_{u-1},\r_{u-1})$ can be computed in a recursive manner via \eqref{eq:prob_transition}, \eqref{eq:prob_transition_r_0}. The decoder state is initialized as
\begin{align}
 \Pr\{\S{k} =(\c_{k},\r_{k}) \} = \binom{m}{\r_{k}} \Omega_1^{\r_{k}} \left( 1-\Omega_1\right)^{\c_{k}}
\end{align}
for all non-negative $\c_{k},\r_{k}$ such that $\c_{k}+\r_{k} = m$
where $m$ is the number of  output symbols.

Let us denote by $\Y$ the random variable that corresponds to the cumulative number of inactivations after the $k$ steps. \fran{The expected value of $\Y$  is given by}
\begin{align}
\Exp\left[\Y\right]= \sum_{u=1}^{k} \sum_{ \cu }  \Pr\{\S{u} =(\cu,0) \}. \label{eq:avg_inact}
\end{align}

Figure ~\ref{fig:mbms_k_1000} shows the expected number of inactivations for an \ac{LT} code with $k=1000$ and the output degree distribution used in standardized Raptor codes, $\Omegarten$. 
The chart compares the average number of inactivations obtained through  Monte Carlo simulation and by \eqref{eq:avg_inact}. It can be observed how there is a tight match between the analysis and the simulation results.

\begin{figure}[t]
\begin{center}
\includegraphics[width=\figwbigger]{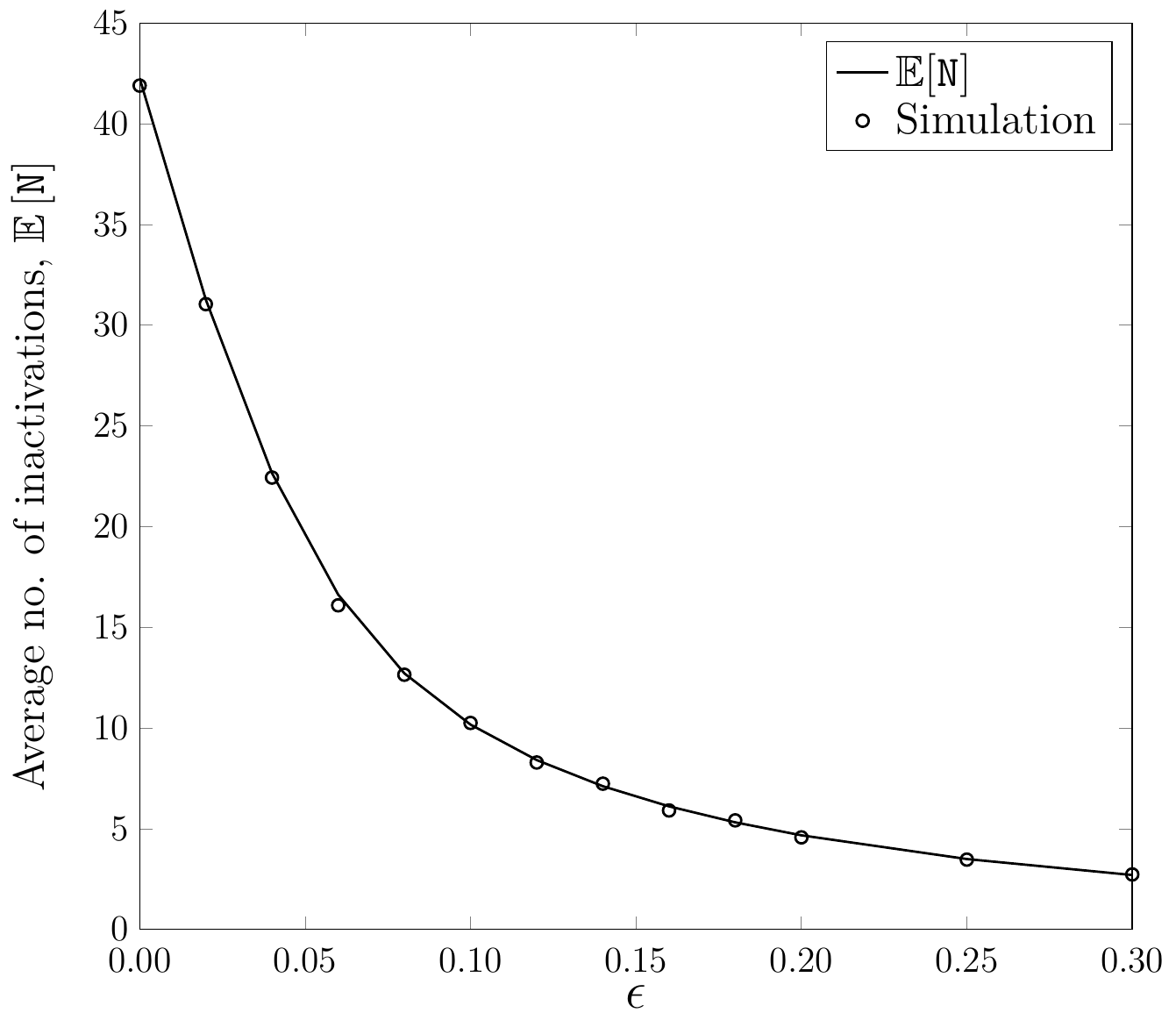}
\centering \caption[Average number of inactivations vs.\  relative overhead for an \ac{LT} code with $k=1000$ and  with  degree distribution $\Omegarten$]{Average number of inactivations vs.\  relative overhead $\reloverhead$ for an \ac{LT} code with $k=1000$ and  with  degree distribution $\Omegarten$.}
\label{fig:mbms_k_1000}
\end{center}
\end{figure}

\FloatBarrier
\subsection{Complete Finite Length Analysis}\label{chap:inact_distribution}
The analysis presented in Section~\ref{chap:inact_first_order} is able to provide the expected number of inactivations (first moment). In this section we shall see that the model can be easily modified to obtain also the complete probability distribution of the number of inactivations.
For this purpose, we first need to include in the state definition the number of inactive input symbols. Hence the state is given by
\[
\S{u}=(\Cu, \Ru, \Nu )
\]
 with $\Nu$ being the random variable that corresponds to the number of inactivations at step $u$. Again, we  proceed by deriving a recursion that allows deriving ${\Pr\{ \S{u-1}=(\c_{u-1}, \r_{u-1}, \n_{u-1} )\}}$ as a function of $\Pr\{ \S{u}=(\cu, \ru, \nu )\}$.
Let us first look at the transition from $u$ to $u-1$ active symbols when $\r_{u}\geq1$, that is, when no inactivation takes place. In this case the number of inactivations stays the same and we have $\n_{u-1} = \nu$. Therefore, we have
\begin{align}
\Pr\{ \S{u-1}=(\c_{u} & -\b_{u}, \r_{u}-\erv_{u}+\b_{u}, \nu) | \S{u}=(\c_u,\r_u,\nu)\} = \\
& \binom{\c_u}{\b_u} {p_u}^{b_u} (1-p_u)^{\c_u-\b_u}\,
 \binom{\r_u-1}{\erv_u-1} \left(\frac{1}{u}\right)^{\erv_u-1} \left( 1- \frac{1}{u} \right)^{\r_u-\erv_u}.
\label{eq:prob_transition_full}
\end{align}
Let us now look at the transition from $u$ to $u-1$ active symbols when $\r_{u}=0$, that is, when an inactivation takes place. In this case the number of inactivations is increased by one yielding
\begin{align}
\Pr\{ \S{u-1}=(\c_{u}-& \b_{u},\b_{u}, \nu+1) | \S{u}=(\c_u,0,\nu)\}  = \binom{\c_u}{\b_u} {p_u}^{\b_u} (1-p_u)^{\c_u-\b_u}.
\label{eq:prob_transition_r_0_full}
\end{align}
The probability of the decoder being in state $\S{u-1} =(\c_{u-1},\r_{u-1}, \n_{u-1})$ can be computed recursively via \eqref{eq:prob_transition_full}, \eqref{eq:prob_transition_r_0_full} starting with the initial condition
\begin{align}
 \Pr\{\S{k}=(\c_k,\r_k, \n_u) \} = \binom{m}{\r} \Omega_1^\r \left( 1-\Omega_1\right)^{\c_k}
 \end{align}
for all non-negative $\c_k, \r_k$ such that $\c_k+\r_k=m$ and $\n_k=0$.

The distribution of the number of inactivations needed to complete the decoding process is finally given by
\begin{align}
f_{\Y}(\y) = \sum_{ \c_0} \sum_{ \r_0} \Pr\{\S{0} =(\c_0,\r_0, \y) \}.\label{eq:distribution}
\end{align}

From \eqref{eq:distribution} we may obtain the cumulative distribution $F_{\Y}(\y)$ which corresponds to the probability of performing at most $\y$ inactivations during the decoding process. The cumulative distribution of the number of inactivations has practical implications. Let us assume the fountain decoder runs on a platform with limited computational capability. For example, the decoder may be able to perform a maximum number of inactivations (recall that the complexity of inactivation decoding is cubic in the number of inactivations, $\y$). Suppose the maximum number of inactivations that the decoder can handle is $\y^*$. For such a decoder, the probability of decoding failure will be lower bounded by ${1- F_{\Y}(\y^*)}$\footnote{The probability of decoding failure is actually higher than $1- F_{\Y}(\y^*)$ since the system of equations to be solved in the \acf{GE} step of inactivation decoding might be rank deficient.}.

Figure~\ref{fig:mbms_dist_inact} shows the distribution of the number of inactivations, for an \ac{LT} code with degree distribution $\Omegarten$ given in \eqref{eq:dist_mbms} and source block size $k = 500$. The chart shows the distribution of the number of inactivations obtained through  both Monte Carlo simulation and by \eqref{eq:distribution}. Again, we can observe how there is a very tight match between the analysis and the simulation results.

\begin{figure}[t]
\begin{center}
\includegraphics[width=\figwbigger]{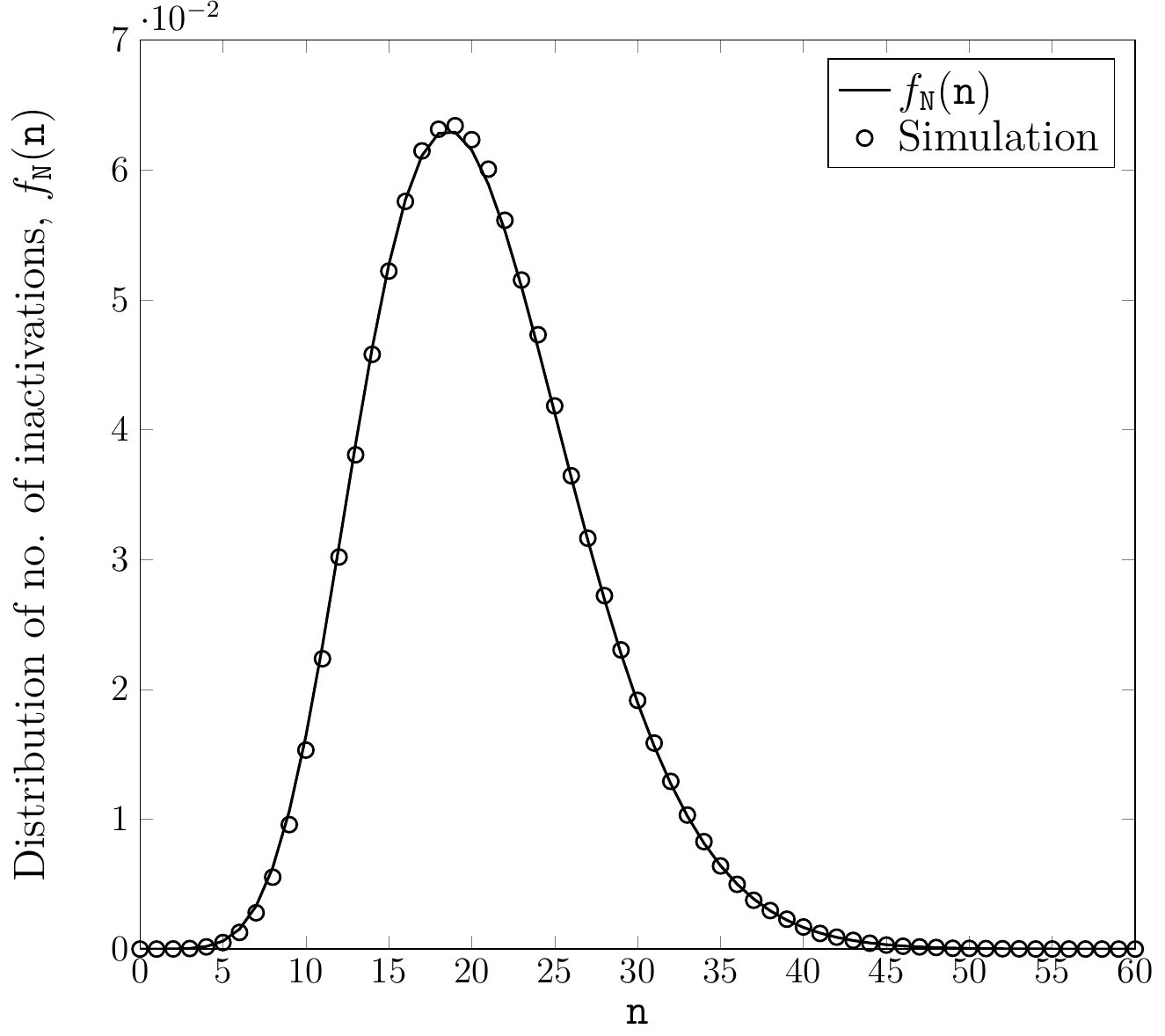}
\centering \caption[Distribution of the number of inactivations for an \acs{LT} code with $k=500$, $m=510$ and degree distribution $\Omegarten$]{Distribution of the number of inactivations for an \ac{LT} code with $k=300$ and  degree distribution $\Omegarten$ given in \eqref{eq:dist_mbms}. }
\label{fig:mbms_dist_inact}
\end{center}
\end{figure}

\subsection{Binomial Approximation}\label{chap:inact_low_complex}

In Sections~\ref{chap:inact_first_order} and \ref{chap:inact_distribution} we have derived recursive methods that can compute the expected number of inactivations and the distribution of the number of inactivations needed to complete the decoding process. The disadvantage of these methods is that their evaluation is computationally complex for large $k$. In this section we propose another recursive method that can provide a reasonably accurate estimation of the number of inactivations with lower computational complexity.

\fran{Let us start by introducing the following definition.
\begin{mydef}[Reduced degree-$d$ set] The reduced degree-$d$ set is defined as the set of output symbols of reduced degree $d$, and we denote it by $\Rijsets{d}$.
\end{mydef}
We shall denote the cardinality of $\Rijsets{d}$ by $\Rijcards{d}$ and its associated random variable by $\Rijreals{d}$.
Obviously, the set $\Rijsets{1}$ corresponds to the ripple. Moreover, the cloud $\cloudset$ corresponds to the union of the sets of output symbols of reduced degree higher than $1$, that is,
\[
\cloudset = \bigcup_{d=2}^{\dmax} \Rijsets{d}.
\]
Furthermore, since the sets $\Rijsets{d}$ are obviously disjoint, we have
\[
\Cloud = \sum_{d=2}^{\dmax} \Rijreals{d}.
\]
}

\fran{As we did in the previous section, we shall add a temporal dimension through subscript $u$, where $u$ corresponds to the number of active input symbols in the graph. Thus, $\Rijset{u}{d}$ is the set of reduced degree $d$ output symbols when $u$ input symbols are still active. Furthermore, $\Rijreal{u}{d}$ and $\Rijcard{u}{d}$ are respectively the random variable associated to the number of reduced degree $d$ output symbols when $u$ input symbols are still active and its realization.
The triangulation process can be modelled by means of a finite state machine with state
\[
\S{u}:=\left(   \Rijreal{u}{1}, \Rijreal{u}{2},\ldots, \Rijreal{u}{\dmax}  \right).
\]
}

\fran{
Let us analyze next how the change of the number of reduced degree $d$ output symbols in the transition from $u$ to $u-1$ active input symbols.
We shall consider first a randomly chosen output symbol $\rosymb$ with reduced degree $d \geq 2$. Denoting by $\probtx{u}{d}$ the probability that the degree of  $\rosymb$ decreases to $d-1$ after the transition from $u$ to $u-1$, that is,
\begin{align}
	\probtx{u}{d} := \Pr \{ \rosymb \in \Rijset{u-1}{d-1} | \rosymb \in \Rijset{u}{d} \}
\end{align}
we have,
\begin{prop}
	\label{prop:approx}
	The probability that a randomly chosen output symbol $\rosymb$, that has reduced degree $d\geq 2$ when $u$ input symbols are active, has reduced degree $d-1$ when $u-1$ input symbols are active is
	\begin{align}
		\probtx{u}{d} = \frac{d}{u}.
	\end{align}
\end{prop}
\begin{proof}
Before the transition we have $u$ active input symbols and output symbol $\rosymb$ has exactly $d$ neighbors among the $u$ active input symbols. In the transition from $u$ to $u-1$ active symbols, $1$ input symbol is selected at random and marked as either resolvable or inactive. The probability that the degree of $\rosymb$ gets reduced corresponds to the probability that one of its $d$ neighbors is marked as resolvable or inactive.
\end{proof}
Since all output symbols choose their neighbors independently, we have that the random variable associated with the number of output symbols of reduced degree  $d$ that become of reduced degree $d-1$ in the transition from $u$ to $u-1$ output symbols, $\ntx{u}{d}$, conditioned to $\Rijreal{u}{d} = \Rijcard{u}{d}$,   is binomially distributed with parameters $\Rijcard{u}{d}$ and $\probtx{u}{d} = d/u$.
}

\fran{Since output symbols select their neighbors without replacement, they cannot have two edges going to the same input symbols. Thus, we have,
\begin{align}
	\Rijreal{u-1}{d}  &=   \Rijreal{u}{d} + \ntx{u}{d+1}  - \ntx{u}{d}
\label{eq:equilibrium}
\end{align}
}

Using a dynamical programming approach similar to the one in Section~\ref{chap:inact_first_order} it is possible to derive a recursion to determine the decoder state probability, which also provides the expected number of inactivations. However, such a recursion would be much more complex to evaluate than the one in Section~\ref{chap:inact_first_order} because the number of possible states is now much larger. Nevertheless, the complexity can be dramatically reduced by introducing a simplifying assumption. \fran{Concretely,
the method presented in this section is based on the assumption that $\Rijreal{u}{d}$ can be approximated by a binomially distributed random variable with parameters $\m$ and $\rij{u}{d}$. Formally, we assume that the distribution of the decoder state at step $u$ can be approximated as a product of binomial distributions,
\begin{equation}
\Pr \{ \S{u}= \Rijcardv{u} \} \approx \mathlarger{\prod}_{d=1}^{\dmax} \binom{\m}{ \Rijcard{u}{d} } {\left(\rij{u}{d}\right)}^{ \Rijcard{u}{d}  }  \left(1- \rij{u}{d}\right)^{\m - \Rijcard{u}{d} }.
\label{eq:assumption}
\end{equation}
where
\[
	\Rijcardv{u}=\left(   \Rijcard{u}{1}, \Rijcard{u}{2},\ldots, \Rijcard{u}{\dmax}  \right)	
\]  }
This binomial distribution assumption is made for the sake of simplicity but it was shown to be reasonably accurate through Monte Carlo simulations.
This assumption greatly simplifies the finite state machine, since the finite length analysis now reduces to deriving a recursion to obtain $\rij{u}{d}$.

\fran{In order to derive a recursion to compute $\rij{u}{d}$, we shall distinguish two cases. First, we  consider the output symbols of reduced degree $d \geq 2$ .
From Proposition~\ref{prop:approx}, for $d \geq 2$ we have that $\ntx{u}{d}$ conditioned to ${\Rijreal{u}{d} = \Rijcard{u}{d}}$ is a binomial random variable with parameters $\Rijcard{u}{d}$ and $\probtx{u}{d} = d/u$.
Thus, we have
\begin{equation}
 \Exp \left[ \Exp \left[\ntx{u}{d} | \Rijreal{u}{d} \right] \right]= \frac{d}{u} \, \Exp \left[  \Rijreal{u}{d} \right].
 \label{eq:ntx_d_2}
\end{equation}
If we now take the expectation on both sides of \eqref{eq:equilibrium}, we have
\begin{align}
	\Exp \left[ \Rijreal{u-1}{d} \right] &=  \Exp \left[  \Rijreal{u}{d}\right] + \Exp \left[ \ntx{u}{d+1}\right]  - \Exp \left[ \ntx{u}{d}\right].
\label{eq:equilibrium_exp}
\end{align}
Substituting \eqref{eq:ntx_d_2} in \eqref{eq:equilibrium_exp} , and we can write
\begin{align}
 \Exp \left[ \Rijreal{u-1}{d} \right] &=  \Exp \left[  \Rijreal{u}{d}\right] + \frac{d+1}{u} \, \Exp \left[  \Rijreal{u}{d+1} \right]  - \frac{d}{u} \, \Exp \left[  \Rijreal{u}{d} \right].
\end{align}
If we now make use of our assumption and consider
\[
\Exp \left[ \Rijreal{u-1}{d} \right] = m \, \rij{u}{d}
\]
we can write
\[
\rij{u-1}{d} = \left(1-\frac{d}{u} \right) \rij{u}{d} + \frac{d+1}{u}  \rij{u}{d+1},
\]
for $d \geq 2$. }

\fran{
We shall now consider the output symbols of reduced degree 1. We are interested in $\ntx{u}{1}$, the random variable associated to the output symbols of reduced degree $d=1$ that become of reduced degree $0$ in the transition from $u$ to $u-1$ active input symbols. If we assume that the ripple is not empty, $\Rijcard{u}{1}  \geq 1$ , an output symbol $\rosymb$ is chosen at random from the ripple and its only neighbor $v$ is marked as resolvable and it is removed from the graph. Moreover, any other output symbol in the ripple being connected to input symbol $v$ also leaves the ripple during the transition. Thus, for $\Rijcard{u}{1}  \geq 1$ we have
\begin{align}
 \Exp \left[\ntx{u}{1} | \Rijreal{u}{1} = \Rijcard{u}{1} \geq 1 \right] &= 1 + \frac{1}{u}  \left( \Rijcard{u}{1} -1 \right) = 1 - \frac{1}{u} + \frac{1}{u}  \Rijcard{u}{1}
 \label{eq:deg1_ripple_not_empty}
\end{align}
whereas when the ripple is empty, $\Rijcard{u}{1}  =0$, we have
\begin{align}
 \Exp \left[\ntx{u}{1} | \Rijreal{u}{1} = 0 \right] = 0.
 \label{eq:deg1_ripple_empty}
\end{align}
If we put together \eqref{eq:deg1_ripple_not_empty} and \eqref{eq:deg1_ripple_empty} we have
\begin{align}
 \Exp \left[\ntx{u}{1}\right] &= \left( 1 - \frac{1}{u} \right) \, \Pr\{  \Rijreal{u}{1} \geq 1\} + \frac{1}{u}  \sum_{\Rijcard{u}{1}=1}^{m} \Rijcard{u}{1} \Pr\{  \Rijreal{u}{1} = \Rijcard{u}{1}\} \\
 &= \left( 1 - \frac{1}{u} \right) \, \left( 1 - \Pr\{  \Rijreal{u}{1} =0\} \right) +  \frac{1}{u} \Exp \left[ \Rijreal{u}{1} \right] \\
 &= \left( 1 - \frac{1}{u} \right) \left( 1 - \left(1- \rij{u}{1}\right)^m\right) +  \frac{1}{u} \, m \, \rij{u}{1}
 \label{eq:ntx_d_1}
\end{align}
If we now replace \eqref{eq:ntx_d_1} into \eqref{eq:equilibrium_exp} we obtain the following recursion for $\rij{u}{1}$,
\[
\rij{u-1}{1} = \left(1-\frac{1}{u} \right) \rij{u}{1} + \frac{2}{u} \rij{u}{2} -
\frac{(1-1/u) \left( 1- (1- \rij{u}{1}  )^m \right)}{m}
\]
}

In order to initialize the finite state machine we shall assume that before triangulation starts $\Rijreal{k}{d}$ follows a binomial distribution $\mathcal B (m, \Omega_d)$, i.e., we assume
\[
\rij{k}{d}= \Omega_d.
\]

The probability of an inactivation occurring in the transition from $u$ to $u-1$ active symbols corresponds simply to the probability of the ripple being empty
\begin{align}
\Pr \{ \Rijreal{u}{1} = 0 \} =  \left( 1-\rij{u}{1} \right)^ {\m}.
\end{align}
Let us recall that $\Y$ is the random variable associated with the cumulative number of inactivations after the $k$ steps of triangulation, we have
\begin{align}
\Exp \left[ \Y \right] = \sum_{u=1}^{k} \Pr \{ \Rijreal{u}{1} = 0 \}.
\end{align}

\begin{figure}[t]
\begin{center}
\includegraphics[width=\figwbigger ]{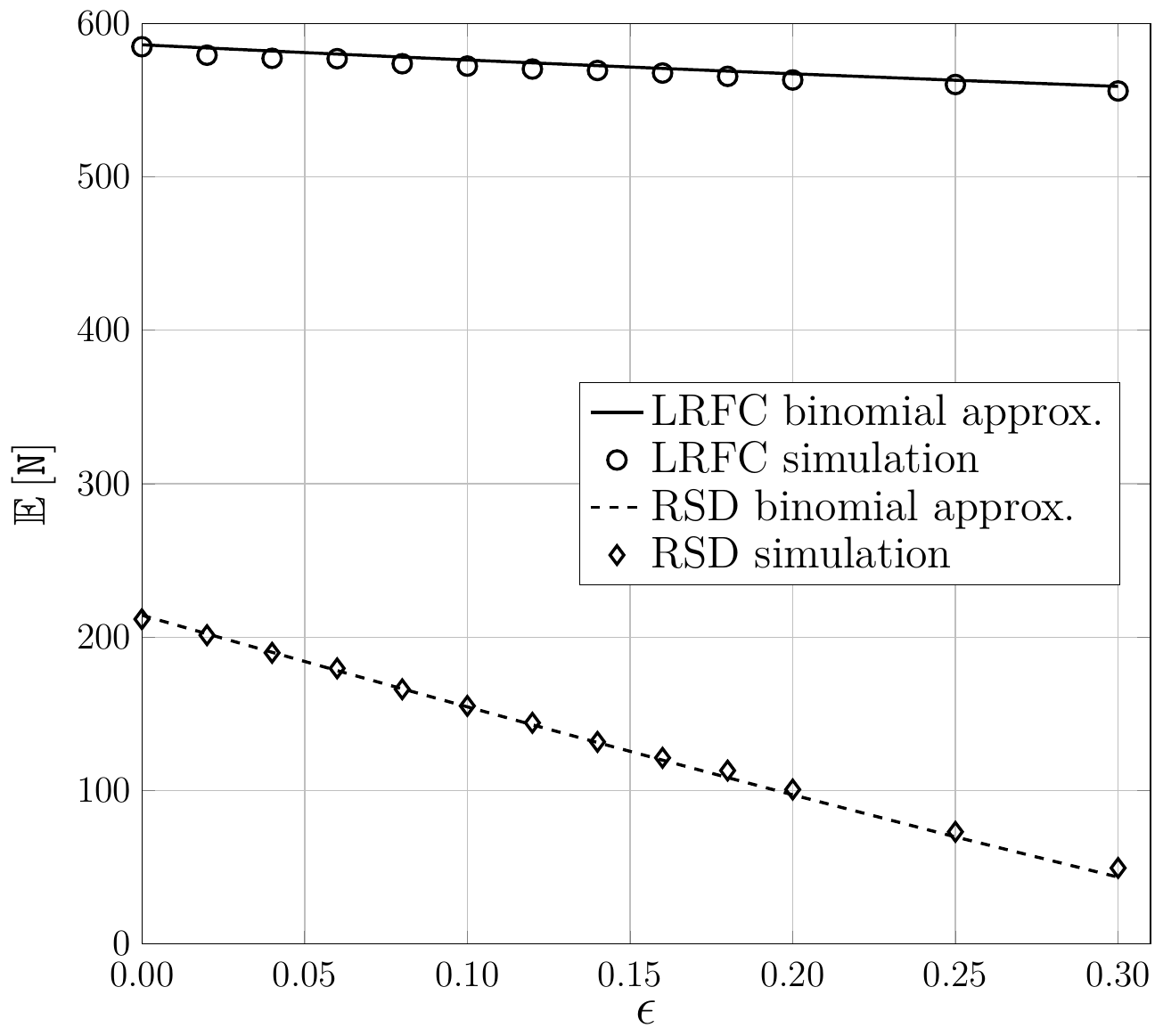}
\centering \caption[Average number of inactivations for an \acs{LRFC} and a \acs{RSD}]{Average number of inactivations needed to decode a \acl{LRFC} and a \ac{RSD} for $k=1000$ and average output degree $\bar \Omega =12$. The markers represent simulation results and the lines represent the predicted number of inactivations for random inactivation using the binomial approximation.}
\label{fig:k_1000_d_10}
\end{center}
\end{figure}

Figure~\ref{fig:k_1000_d_10} shows the average number of inactivations needed to complete decoding for  a \ac{LRFC}\footnote{The degree distribution of a \ac{LRFC} follows a binomial distribution.} and a \ac{RSD} with parameters $\partwo=0.09266$ and $\parone= 0.001993$, both with average output degree $\bar \Omega =12$ and $k=1000$. The figure shows results obtained by Monte Carlo simulation and also the estimation obtained under our binomial approximation. A tight match between simulation results and the estimation can be observed.

In Figure~\ref{fig:ripple} we shows the evolution of the ripple size $\Ripple_u = \Rijreal{u}{1}$ and the cumulative number of inactivations when $u$ input symbols are active. The output degree distribution is again a \ac{RSD} with parameters $\partwo=0.09266$ and $\parone= 0.001993$ and the relative receiver overhead is $\reloverhead = 0.2$. The figure shows the result of Monte Carlo simulations and the estimation obtained with our binomial approximation. It can be observed how the match between simulation results and the outcome of our approximation is tight.

\begin{figure}[t]
\begin{center}
\includegraphics[width=\figwbigger]{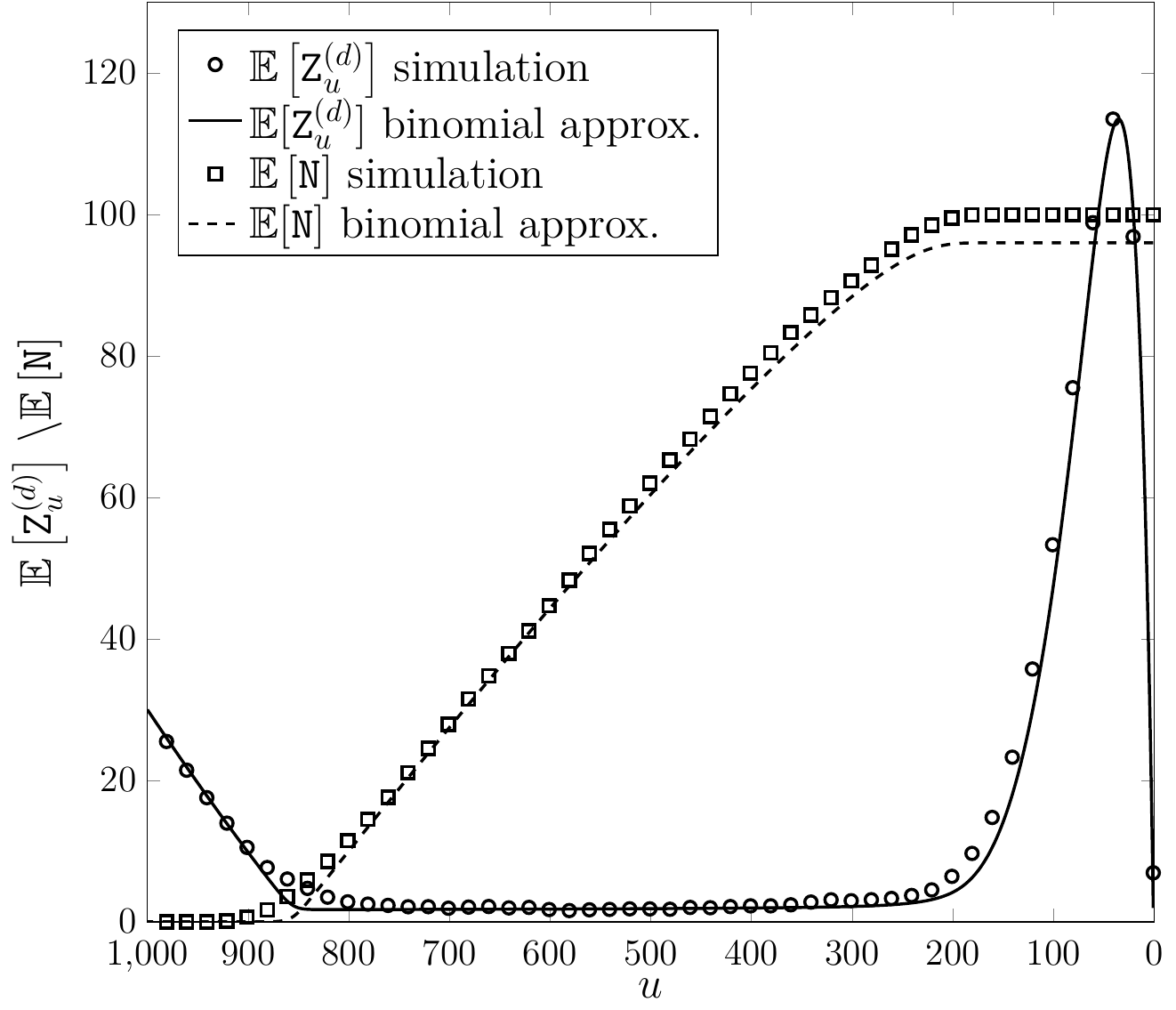}
\centering \caption[Evolution of the ripple and the cumulative number of inactivations for a \ac{RSD} with $k=1000$ and $\reloverhead = 0.2$]{Evolution of $\Ripple_u$ and the cumulative number of inactivations with respect to the number of active input symbols $u$. The output degree distribution is a \ac{RSD} with $k=1000$ and $\reloverhead = 0.2$. The lines represent the prediction obtained under the binomial distribution assumption. The markers represent the average obtained through Monte Carlo simulations.}
\label{fig:ripple}
\end{center}
\end{figure}

\FloatBarrier
\section{Code Design}\label{sec:lt_code_design}
This section focuses on the design of \ac{LT} codes optimized for inactivation decoding. Most of the works on \ac{LT} codes assume iterative decoding rather than inactivation decoding. An exception is the work in  \cite{mahdaviani2012raptor} where the authors derived a analytically family of degree distributions optimized for inactivation decoding. The authors in  \cite{mahdaviani2012raptor} found out that  in the unconstrained case, the optimal degree distribution under inactivation decoding corresponds to the ideal soliton distribution. One of the main shortcomings of this work is that there is no direct control on the average output degree of the degree distributions, \fran{which implies one has no control on the encoding complexity}.

In this section we present method to design \ac{LT} codes for inactivation decoding.
In this chapter three different methods have been presented to obtain the number of inactivations needed to complete decoding (or an estimation thereof). Any of these three methods can be used to perform a numerical optimization of the output degree distribution which leads to a \emph{low} number of inactivations. We will use the approximate method presented in Section~\ref{chap:inact_low_complex} since its evaluation is much faster than that of the other two methods. Nevertheless, it is still feasible to use the other methods for numerical code design provided that the source block size $k$ is not too large (up to several hundreds).

The numerical optimization algorithm we will use is \ac{SA} \cite{kirkpatrick1983optimization}, a fast meta-heuristic method for global optimization that is inspired in the process of annealing in metallurgy in which a material is cooled down slowly to obtain a crystal structure.
The starting point of \ac{SA} corresponds to an initial state $s_{\mathrm{init}}$ plus an initial temperature $T_{\mathrm{init}}$. At every step a number of candidate successive states for the system are generated as a slight variation of the previous state and the temperature decreases.  For high temperatures \ac{SA} allows moving the system to higher energy states with some probability that becomes smaller as the temperature of the system decreases. This process is repeated until the system reaches a target energy or until a maximum number of steps are carried out. In our case the states correspond to degree distributions. The energy of a state (degree distribution) has to be defined so that it is a  monotonically decreasing function of the number of inactivations $\Exp \left[ \Y \right]$ and the probability of decoding failure $\Pf$.  

Concretely, we consider a source block size $k=10000$ and we set the following constraints:
 \begin{itemize}
   \item {A target probability of decoding failure $ {P}_f^*= 10^{-2}$ at $\reloverhead= 0$.}
   \item {Maximum average output degree $\bar \Omega \leq 12$.}
   \item {Maximum output degree $d_{\mathrm{max}} = 150$.}
 \end{itemize}
The second and third constraint are introduced to limit the average and worst case encoding cost of an output symbol.
In order to embed these constraints into \ac{SA} the objective function to be minimized (energy) is defined as
\[
\Upsilon = \Exp \left[ \hat \Y \right] + f_p (\Pflow)
\]
where $f_p$ is defined as
\begin{align}\label{eq:obj_f}
f_p (\Pflow) =
    \begin{cases}
        0,  & \Pflow<{\Pf}^* \\
        b~(1- {\Pf}^* / \Pflow),& \mathrm{otherwise}
    \end{cases}
\end{align}
being ${\Pf}^*$ the target probability of decoding failure, $\Pflow$ the tight lower bound to it given in equation \eqref{eq:low_bound_schotsch} and $b$ a large positive number ($b= 1000$ was used in the example). The large $b$ factor ensures that degree distributions which do not comply with the target probability of decoding failure are discarded.

The use of $\Pflow$ in place of the actual $\Pf$ in the objective functions stems from the need of having a fast (though, approximate) performance estimation to be used within the \ac{SA} recursion (note in fact that the evaluation of the actual $\Pf$ may present a prohibitive complexity). This allows evaluating the energy of a state (i.e., degree distribution) very quickly. Although the lower bound in eq.~\eqref{eq:low_bound_schotsch} may not be tight for $\reloverhead=0$, it is very tight for values of $\reloverhead$ slightly larger than $0$. This means that in practice we will need $\reloverhead$ slightly larger than $0$ to comply with our requirements.

For the sake of illustration we will perform two different optimizations. In the first one we will constrain the degree distribution to be a (truncated) \ac{RSD} and in the second one we will impose no constraints to the degree distribution.

Let $\rsd$  be a \ac{RSD} distribution. For a given maximum degree $\dmax$ we define the truncated \ac{RSD} distribution, $\rsdt$, as

\begin{align}
\rsdt_i =
\begin{cases}
    \rsd_i,  & i < \dmax \\
    \sum_{j=\dmax}^{k}\rsd_j,  & i = \dmax \\
    0,& i > \dmax
\end{cases}.
\end{align}
Therefore, in this first optimization we aim at finding the \ac{RSD} parameters $\parone$ and $\partwo$ (see \eqref{eq:dist_rsd} in Section~\ref{sec:rsd} for the definition of $\parone$ and $\partwo$) that minimize our objective function in \eqref{eq:obj_f}. We shall denote the degree distribution obtained from this optimization process by $\Omega^{(1)}$.

In the second optimization we carry out  we set no constraints at all on the degree distribution, except for the design constraints on the average and maximum output degree. We refer to the distribution obtained by this optimization method as $\Omega^{(2)}$.

Figure~\ref{fig:inact} shows  the number of inactivations needed for decoding as a function of $\reloverhead$ for $\Omega^{(2)}$ and $\Omega^{(1)}$, which has parameters $\partwo = 0.05642$ and $\parone = 0.0317$.  If we compare the number of inactivations needed by $\Omega^{(2)}$ and $\Omega^{(1)}$ we can observe how the $\Omega^{(2)}$, the result of the unconstrained optimization, requires slightly less inactivations. In the figure we can also observe how the estimation of the number of inactivations $\Exp \left[ \hat \Y \right]$ lies slightly below $\Exp \left[ \Y \right]$, which is an effect that had already been observed in Section~\ref{chap:inact_low_complex}.

\begin{figure}[t!]
\begin{center}
\includegraphics[width=\figwbigger]{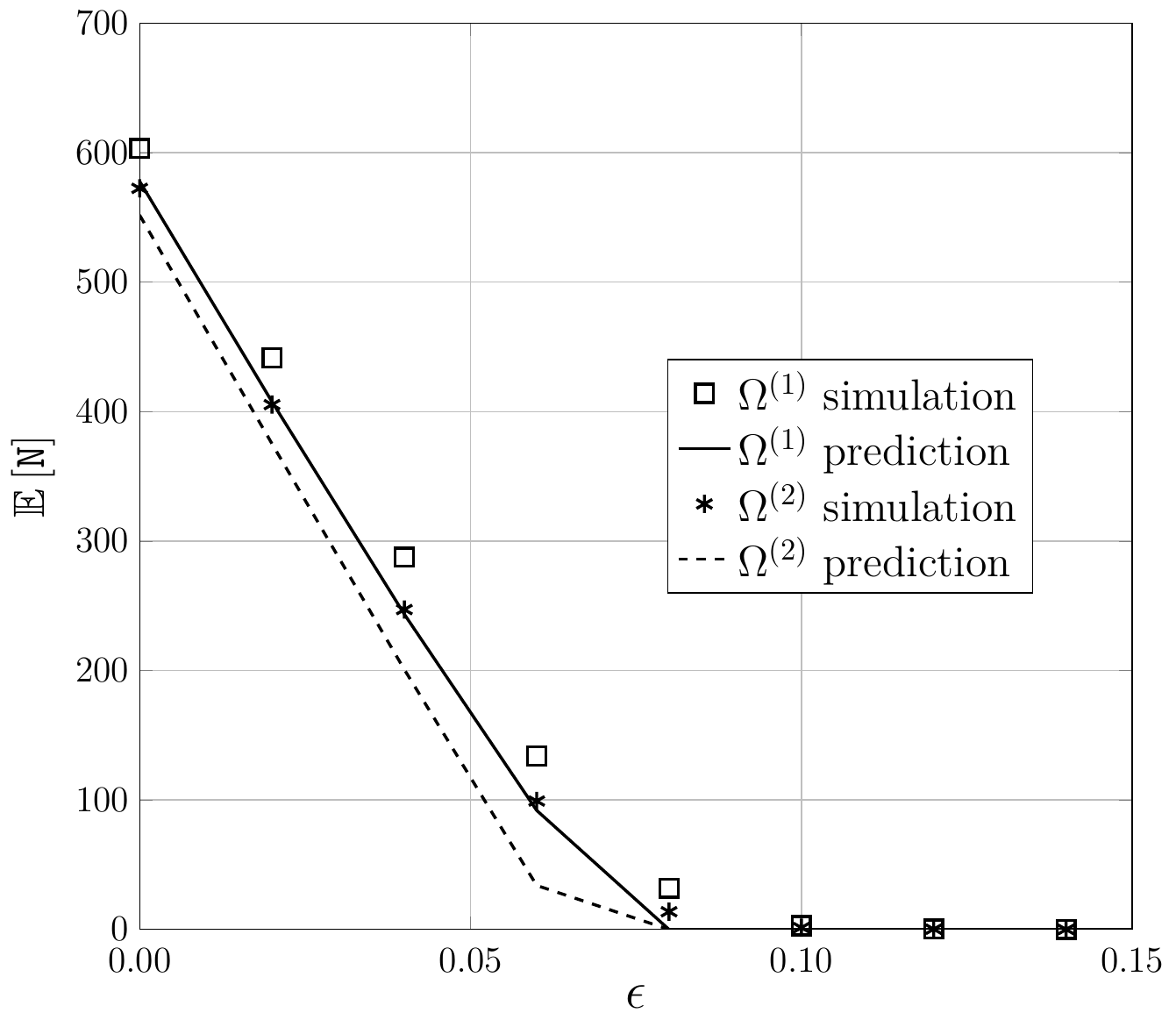}
\centering
\caption[Average number of inactivations vs.\ $\reloverhead$ for $\Omega^{(1)}$ and $\Omega^{(2)}$]{ Average number of inactivations needed for decoding vs.\ $\reloverhead$. The solid and dashed lines represent the predicted number of inactivations, $\Exp [ \hat{\Y} ] $ for $\Omega^{(1)}$ and $\Omega^{(2)}$, respectively.  The markers denote the average number of inactivations $\Exp [ \Y ] $ obtained by Monte Carlo simulations.}
\label{fig:inact}
\end{center}
\end{figure}

For the sake of completeness, the  probability of decoding failure for $\Omega^{(1)}$  and $\Omega^{(2)}$ is provided in Figure~\ref{fig:per}. It can be observed how $\Pflow$ is below the target value ${P}_f^*=10^{-2}$ in both cases, being the probability of decoding failure lower for the truncated \ac{RSD}.
\begin{figure}[t!]
\begin{center}
\includegraphics[width=\figwbigger]{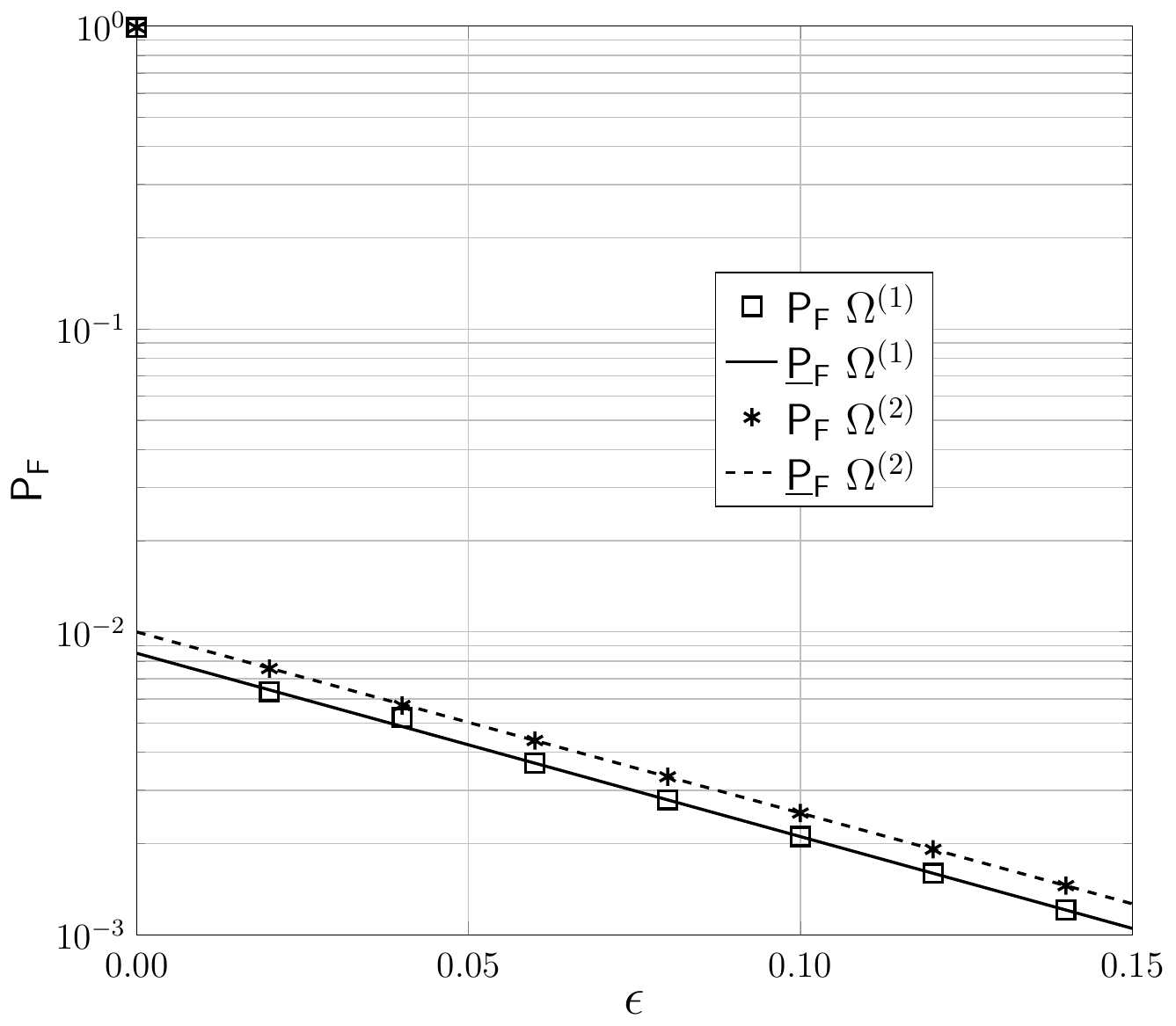}
\centering
\caption[Probability of decoding failure, $\Pf$ vs.\ $\reloverhead$ for $\Omega^{(1)}$ and $\Omega^{(2)}$]{Probability of decoding failure, $\Pf$ vs.\ $\reloverhead$ for $\Omega^{(1)}$ and $\Omega^{(2)}$. The lines represent the lower bound $\Pflow$ and  markers denote simulation results. \label{fig:per} }
\end{center}
\end{figure}

\FloatBarrier


\section{Summary}\label{chap:lt_summary}

In this chapter we have considered \ac{LT} codes under inactivation decoding, an \ac{ML} decoding algorithm belonging to the family of the structured or intelligent Gaussian elimination algorithms. Inactivation decoding aims at reducing the size of the system of equations that needs to be solved with standard Gaussian elimination (number of inactivations). The focus of the chapter is on the decoding complexity under inactivation decoding. In Section~\ref{chap:inact_first_order} we presented a first order analysis of \ac{LT} codes under inactivation decoding that provides the expected number of inactivations given an output degree distribution. The analysis is based on a dynamic programming approach that models the decoder as a finite state machine. This model was extended in Section~\ref{chap:inact_distribution} in order to obtain the probability distribution of the number of inactivations. Section~\ref{chap:inact_low_complex} presented an approximate low complexity method to estimate the expected number of inactivations. Finally in Section~\ref{sec:lt_code_design} we showed how these methods can be used to numerically design degree distributions to minimize the number of inactivations while fulfilling a set of design constraints.



%% file: Chapter4/chapter4.tex
\chapter{Raptor Codes under Inactivation Decoding} \label{chap:Raptor}
\ifpdf
    \graphicspath{{Chapter4/Chapter4Figs/PNG/}{Chapter4/Chapter4Figs/PDF/}{Chapter4/Chapter4Figs/}}
\else
    \graphicspath{{Chapter4/Chapter4Figs/EPS/}{Chapter4/Chapter4Figs/}}
\fi

Within this chapter we consider Raptor codes under inactivation decoding. In Section \ref{chap:raptor_rateless_ml} we develop upper bounds to the probability of decoding failure of $q$-ary Raptor codes under \ac{ML} decoding using the weight enumerator of the outer code (precode), or its expected weight enumerator in case the outer code is not deterministic but drawn at random from an ensemble of codes. The bounds are shown to be tight, specially in the error floor region, by means of simulations.
In Section \ref{chap:raptor_inactivation_decoding} we \fran{consider the decoding complexity of Raptor codes under inactivation decoding, which is an \ac{ML} decoding algorithm. More concretely, we provide a heuristic method to approximate the number of inactivations needed for decoding.} In Section~\ref{chap:rapt_code_design} we show how the results in this chapter can be used for Raptor code design by means of an example. Finally in Section~\ref{chap:raptor_summary} we summarize our contributions.

\section{Performance under ML Decoding}\label{chap:raptor_rateless_ml}

The probability of decoding failure of Raptor codes under \ac{ML} decoding\footnote{Inactivation decoding is a \acl{ML} decoding algorithm.} has been subject of study in several works. In \cite{Rahnavard:07}  upper and lower bounds to the intermediate symbol erasure rate were derived for Raptor codes with outer codes in which the elements of the parity check matrix are \ac{i.i.d.} Bernoulli random variables. This work was extended in \cite{Schotsch:14} by deriving an approximation to the performance of Raptor codes under \ac{ML} decoding  under the assumption that the number of erasures correctable by the outer code is small. Thus, this approximation holds only when the rate of the outer code is sufficiently high.
In \cite{Liva10:fountain} it was shown by means of simulations how the error probability of $q$-ary Raptor codes is very close to that of $q$-ary linear random fountain codes.
In \cite{wang:2015} upper and lower bounds to the probability of decoding failure of Raptor codes where derived. The outer codes considered in \cite{wang:2015} are systematic binary linear random codes.

Although a number of works has studied the probability of decoding failure of Raptor codes, to the best of the knowledge of the author, the available results hold only for specific binary outer codes (see \cite{Rahnavard:07,wang:2015,lazaro:ISIT2015}). In this section we derive an upper bound to the probability of decoding failure of $q$-ary Raptor codes using the weight enumerator of the outer code (precode), or its expected weight enumerator in case the outer code is not deterministic but drawn at random from an ensemble of codes. Thus, the bounds derived in this chapter are generic and can be applied to any Raptor code, provided that the weight enumerator of the outer code is known.

In this section we consider ensembles of $q$-ary Raptor codes under \ac{ML} decoding. More specifically,
Raptor codes constructed over $\mathbb {F}_{q}$ with an $(h,k)$ outer linear block code $\code$ are considered. Hence, the $k$ input (or source) symbols,  ${\vecu=(\Raptorinput_1,~\Raptorinput_2,~\ldots, \Raptorinput_k)}$,  belong to $\mathbb{F}_{q}$.

Denoting by $\Gp$ the employed generator matrix of the outer code, of dimension $(k \times h)$ and with elements in $\mathbb {F}_{q}$, the intermediate symbols can be expressed as
\[
\vecv = \vecu \Gp.
\]
Note that, by definition,  ${\vecv=(\Rintermsymbol_1,~\Rintermsymbol_2,~\ldots, \Rintermsymbol_h)}  \in \code$.
The intermediate symbols serve as input to a $q$-ary LT encoder, that generates an unlimited number of output symbols, ${\mathbf{\Rosymb}=(\Rosymb_1, \Rosymb_2, \ldots, \Rosymb_n)}$, where $n$ can grow unbounded. Again, the elements of $\mathbf{\Rosymb}$ belong to $\mathbb{F}_{q}$.
For any $n$ the output symbols can be expressed as
\[
\mathbf{\Rosymb} = \vecv \GLT = \vecu \Gp \GLT
\]
where $\GLT$ is an $(h \times n)$ with elements in  $\mathbb {F}_{q}$. Each column of $\GLT$ is associated with $\Rosymb_i$.
More specifically, each column of $\GLT$ is generated by first selecting an output  degree $d$ according to the degree distribution $\Omega$, and then selecting $d$ different indices uniformly at random between $1$ and $h$. Finally, the elements of the column corresponding to these indices are drawn independently and uniformly at random from $\mathbb {F}_{q} \backslash \{0\}$, while all other elements of the column are set to zero.

We consider the transmission over a \ac{QEC} at the output of which each transmitted symbol is either correctly received or erased.
Denoting by $m$ the number of output symbols collected by the receiver of interest, and expressing it as $m=k+\absoverhead$, \fran{where $\absoverhead$ is the absolute receiver overhead}. Let us denote by ${\mathbf{\Rrosymb}=(\Rrosymb_1, \Rrosymb_2, \ldots, \Rrosymb_m)}$ the $m$ received output symbols. Let us denote by $\mathcal{I} = \{i_1, i_2, \hdots, i_m \}$ the set of indices corresponding to the $m$ non-erased symbols, we have
\[
\Rrosymb_j = \Rosymb_{i_j}.
\]
In this section we will assume that \ac{ML} Raptor decoding is performed by solving the system of equations\footnote{In practice Raptor decoding is performed by solving the system of equations in \eqref{eq:raptor_sys_eq} that involves the constraint matrix. The two systems of equations are equivalent. \fran{In this section we take the system of equations involving the generator matrix of the Raptor code for convenience.} }
\[
\mathbf{\Rrosymb} = \vecv \GrxR
\]
where
\begin{align}
\GrxR = \Gp \GrxLT
\label{eq:sys_eq2}
\end{align}
with $\GrxLT$ given by the $m$ columns of $\GLT$ with indices in $\mathcal{I}$.


\subsection{Upper Bounds on the Error Probability}\label{sec:perf_bound}

An upper bound on the probability of failure $\Pf$ of a Raptor code constructed over $\mathbb {F}_{q}$ as a function of the receiver overhead $\absoverhead$ is established in the next Theorem.

\begin{theorem}\label{theorem:rateless}
Consider a Raptor code constructed over $\mathbb {F}_{q}$ with an $(h,k)$ outer code $\code$ characterized by a weight enumerator $\weo$, and an inner \ac{LT} code with output degree distribution $\Omega$.
The probability of decoding failure under optimum \fran{(\ac{ML})} erasure decoding given that ${m=k+\absoverhead}$ output symbols have been collected by the receiver can be upper bounded as
\[
\Pf  \leq \sum_{l=1}^h \weo_{\l} \pil^{k+\absoverhead}
\]
where  $\pil$ is the probability that a generic output symbol $Y$ is equal to $0$ given that the vector $\vecv$ of intermediate symbols has Hamming weight $l$. The expression of $\pil$ is
\begin{align}
\pil &= \frac{1}{q} +  \frac{q-1}{q} \sum_{j=1}^{\dmax} \Omega_j    \frac{\krawt_j(l; h,q)}{\krawt_j(0; h,q)}
\label{eq:pl}
\end{align}
where $\krawt_j(l; h,q)$ is the Krawtchouk polynomial of degree $j$ with parameters $h$ and $q$, defined as \cite{MacWillimas77:Book}
\[
\krawt_k(x;n,q) = \sum_{j=0}^k (-1)^j \binom{x}{j} \binom{n-x}{k-j} (q-1)^{k-j}.
\]
\end{theorem}
\begin{proof}
An optimum (\ac{ML}) decoder solves the linear system of equations in \eqref{eq:sys_eq2}. Decoding will fail whenever the system does not admit a unique solution, that is, if and only if  $\rank(\GrxR)<k$, i.e. if
${\exists\,  \vecu \in \mathbb {F}_q^k \backslash \{ \textbf{0}\} \,\, \text{s.t.} \,\, \vecu \GrxR = \textbf{0}}$.
Let us consider two  vectors $\vecu \in \mathbb {F}_q^k, \vecv \in \mathbb {F}_q^h$. Define  $E_{\vecu}$ as the event $\vecu \Gp \GrxLT = \mathbf{0}$.
Similarly, Define $E_{\vecv}$ as the event  $\vecv \GrxLT = \mathbf{0}$. We have
\begin{align}
\Pf & =   \Pr\left\{ \bigcup_{\vecu \in \mathbb {F}_q^k \backslash \{ \textbf{0}\}} E_{\vecu}  \right\}  = \Pr\left\{ \bigcup_{\vecv \in \code \backslash \{ \textbf{0}\} } E_{\vecv} \right\}
\label{eq:existence}
\end{align}
where we made use of the fact that due to linearity, the all zero intermediate word is only generated by the all zero input vector.

If we develop \eqref{eq:existence}, we have
\begin{align}
\Pf & = \Pr \left\{ \bigcup_{l=1}^h \,\,  \bigcup_{\vecv \in  \mathbb \code_l }  E_{\vecv}  \right\}
\label{eq:existence2}
\end{align}
where, $\code_l$ is the set of codewords in $\code$ of Hamming weight $l$, formally,
\[
\code_l = \left\{ \vecv \in \code : w_H(\vecv) = l \right\}.
\]

Let $L$ be a discrete random variable representing the Hamming weight of vector $\vecv \in \code$. Moreover, let $J$ and $I$ be discrete random variables representing the degree of the generic output symbol $\Rrosymb$, and the number of non-zero neighbors of such intermediate symbol, respectively. Note that $I \leq L$. By making use of Boole's inequality (also known as the union bound) it is possible to upper bound \eqref{eq:existence2} as
\begin{align}
\Pf & \leq \sum_{l=1}^{h} \Pr \left\{  \bigcup_{\vecv \in  \mathbb \code_l }  E_{\vecv}  \right\} \notag \\
&\leq \sum_{l=1}^{h} \weo_{\l} \Pr \left\{ E_{\vecv} | L=l  \right\} \, .
\label{eq:existence3}
\end{align}
Since output symbols are independent of each other, we have
\[
\Pr \left\{ E_{\vecv} | L=l \right\} = \pil^{k+\absoverhead}
\]
where $\pil = \Pr \{ y=0 | L=l\}$ is the conditional probability that the generic output symbol $\Rrosymb$ is equal to $0 \in \mathbb F_q$ given that $\vecv \in \code_l$. An expression for $\pil$ may be obtained observing that
\begin{align*}
\pil &= \sum_{j=1}^{\dmax} \Pr \{ y=0 | L=l,J=j \} \Pr \{ J=j | L=l \} \\
      &\stackrel{(\mathrm{a})}{=} \sum_{j=1}^{\dmax} \Omega_j \Pr \{ y=0 | L=l,J=j \} \\
      &\stackrel{(\mathrm{b})}{=} \sum_{j=1}^{\dmax} \Omega_j \sum_{i=0}^{\min\{j,l\}} \Pr \{ y=0 | I=i \} \! \Pr \{ I=i | L=l, J=j \}
\end{align*}
where equality $(\mathrm{a})$ is due to \mbox{$\Pr \{ J=j | L=l \} = \Pr \{ J=j \}$} $= \Omega_j$ and equality $(\mathrm{b})$ to $\Pr \{ y=0 | L=l, J=j, I=i \} = \Pr \{ y=0 | I=i \}$. Letting \mbox{$\pifroml = \Pr \{ I=i | L=l, J=j \}$}, since the $j$ intermediate symbols are chosen uniformly at random by the LT encoder we have
\begin{align}\label{eq:neighbors}
\pifroml = \frac{ \binom{\l}{i} \binom{h-\l}{j-i} } { \binom{h}{j}} \, .
\end{align}
Denoting $\Pr \{ y=0 | I=i \}$ by $\qi$ and observing that, due to the elements of $\GrxR$ being \ac{i.i.d.} and uniformly drawn in $\mathbb {F}_{q} \setminus \{0\}$, on invoking Lemma~\ref{lemma:galois} in Appendix~\ref{app:proofs}\footnote{Lemma~\ref{lemma:galois} is only valid for $q=2^m$, the case of most interest for practical purposes. The proof of the general case is a trivial extension of Lemma~\ref{lemma:galois}. The result in Lemma~\ref{lemma:galois} can also be found in \cite{schotsch:2013}, where the proof was derived using a different technique from the one used in Appendix~\ref{app:proofs} (see Appendix~\ref{app:proofs} for more details).} we have
\begin{align}\label{eq:sum}
\qi =\frac{1}{q} \left( 1 + \frac{(-1)^i}{(q-1)^{i-1}}\right).
\end{align}
Hence, $\pil$ is given by
\begin{align}
\pil &=  \sum_{j=1}^{\dmax} \Omega_j  \sum_{i=0}^{\min\{j,l\}} \pifroml  \, \qi
\end{align}
where $\pifroml$ and $\qi$ are given by \eqref{eq:neighbors} and \eqref{eq:sum}, respectively.

Expanding this expression and rewriting it using Krawtchouk polynomials and making use of the Chu-Vandermonde identity\footnote{The Chu-Vandermonde identity, or Vandermonde's convolution formula is given by
\[
\binom{x+a}{n} = \sum_{k=0}^{n}\binom{x}{k} \binom{a}{n-k}
\] }, one obtains \eqref{eq:pl}. 
\end{proof}

The following theorem makes the bound in Theorem~\ref{theorem:rateless} tighter for $q>2$. For $q=2$ the following theorem is equivalent to Theorem~\ref{theorem:rateless} .
\begin{theorem}\label{lemma:bound_tight}
Consider a Raptor code constructed over $\mathbb {F}_{q}$ with an $(h,k)$ outer code $\code$ characterized by a weight enumerator $\weo$, and an inner \ac{LT} with output degree distribution $\Omega$.
The probability of decoding failure under optimum erasure decoding given that ${m=k+\absoverhead}$ output symbols have been collected by the receiver can be upper bounded as
\[
\Pf  \leq \sum_{l=1}^h \frac{\weo_{\l}}{q-1} \pil^{k+\absoverhead}
\]
\end{theorem}
\begin{proof}
The bound \eqref{eq:existence3} can be tightened by a factor $q-1$ exploiting the fact that for a linear block code $\code$ constructed over $\mathbb {F}_{q}$, if $\mathbf{\Rosymb}$ is a codeword, $\alpha \mathbf{\Rosymb}$ is also a codeword, $\forall \alpha \in \mathbb F_{q} \backslash \{0\}$, \cite{Liva2013}.
\end{proof}

\begin{remark}
The upper bound in Theorem~\ref{lemma:bound_tight} can also be applied to LT codes. In that case, $\weo_{\l}$ needs to be replaced by the total number of sequences of Hamming weight $l$ and length $k$,
\[
\weo_{\l}= \binom{k}{l} (q-1)^{l-1}.
\]
The upper bound obtained for LT codes coincides with the bound in \cite{schotsch:2013} (Theorem 1), where only LT codes are considered. Thus, we can regard Theorem~\ref{lemma:bound_tight} as an extension of the work in  \cite{schotsch:2013} to Raptor codes.
\end{remark}

\subsection{Random Outer Codes from Linear Parity-Check Based Ensembles}\label{sec:ensemble}

Both Theorem~\ref{theorem:rateless} and Theorem~\ref{lemma:bound_tight} apply to the case of a deterministic outer code. In this section we extend these results to the case of a random outer code drawn from an ensemble of codes. Specifically, a parity-check based ensemble of outer codes is considered, denoted by $\outercodeensemble$, defined by a random matrix of size ${(h - k) \times h}$ whose elements belong to $\mathbb F_q$. A linear block code of length $h$ belongs to $\outercodeensemble$ if and only if at least one of the instances of the random matrix is a valid parity-check matrix for it. Moreover, the probability measure of each code in the ensemble is the sum of the probabilities of all instances of the random matrix which are valid parity-check matrices for that code. Note that all codes in $\outercodeensemble$ are linear, have length $h$, and have dimension $k' \geq k$\footnote{This comes from the fact that the parity-check matrix can be rank deficient.}.

In the following the expression ``Raptor code ensemble'' is used to refer to the set of Raptor codes obtained by concatenating an outer code belonging to the parity-check based ensemble $\outercodeensemble$ with an \ac{LT} encoder having distribution $\Omega$. This ensemble shall be denoted as $(\outercodeensemble, \Omega)$. The following theorem extends the result in Theorem~\ref{lemma:bound_tight} to ensembles of Raptor codes.

\begin{theorem}\label{corollary:rateless}
Consider a Raptor code ensemble $(\outercodeensemble, \Omega)$ and let $\weoensemble= \{ \weoensemble_0,\weoensemble_1,\dots,\weoensemble_h \}$ be the expected weight enumerator
of a code $\code$ that is randomly drawn from ensemble $\outercodeensemble$, i.e., let $\weoensemble_{l} = \Exp_{\code}[A_l(\code)]$ for all $l \in \{0,1,\dots,h\}$. Let $\barPf$
be the average probability of decoding failure of the Raptor code ensemble obtained by concatenating a code $\code$  randomly drawn from ensemble $\outercodeensemble$ with the \ac{LT} encoder with degree distribution $\Omega$, under \fran{\acf{ML}} erasure decoding and given that ${m=k+\absoverhead}$ output symbols have been collected by the receiver. Then
\[
\barPf  \leq  \sum_{l=1}^h \frac{\weoensemble_{\l}}{q-1}  \pil^{k+\absoverhead} \, .
\]
\end{theorem}
\begin{proof}
We can express the average probability of decoding failure as
\begin{align}
\barPf = \Exp_{\outercodeensemble} [ \Pf(\code)]
\label{eq:ensemble}
\end{align}
where $\Pf(\code)$ is the probability of decoding failure when outer code $\code$ is selected as outer code, and the expectation is taken over all the outer codes $\code$ in the ensemble $\outercodeensemble$.

From Theorem~\ref{lemma:bound_tight}  we have
\begin{align}
\barPf \leq \Exp \left[ \sum_{l=1}^h \frac{\weo_{\l}(\code) }{q-1} \pil^{k_\code+\absoverhead} \right]
\label{eq:ensemble2}
\end{align}
where $\weo_{\l}(\code)$ is the number of codewords of weight $l$ in $\code$ and $k_\code$ is the dimension of $\code$.

For all codes $\code$ in the ensemble $\outercodeensemble$ we have $k_\code \geq k$. Furthermore, since $\pil$ is a probability we have $\pil \leq 1$ and we can write
\[
\pil^{k_\code+\absoverhead} \leq \pil^{k+\absoverhead}
\]
which allows us to upper bound \eqref{eq:ensemble2} as
\[
\barPf \leq \Exp \left[ \sum_{l=1}^h \frac{\weo_{\l}(\code) }{q-1} \pil^{k+\absoverhead} \right]
\]
which by linearity of the expectation becomes
\[
\barPf \leq  \sum_{l=1}^h \frac{ \Exp \left[ \weo_{\l}(\code) \right] }{q-1} \pil^{k+\absoverhead} = \sum_{l=1}^h \frac{\weoensemble_{\l}}{q-1}  \pil^{k+\absoverhead}
\]
\end{proof}
Theorem~\ref{theorem:rateless} can be extended in the same way as Theorem~\ref{lemma:bound_tight} to consider the case when the outer code is drawn from an ensemble of codes.

\subsection{Numerical Results}

All the results presented in this section use the \ac{LT} output degree distribution from standard R10 Raptor codes,  \cite{MBMS12:raptor,luby2007rfc} with degree generator polynomial:
\begin{align}
\Omega(\x) &= 0.0098\x + 0.4590\x^2+ 0.2110\x^3+0.1134\x^4 \\
&+ 0.1113\x^{10} + 0.0799\x^{11} + 0.0156\x^{40}.
\label{eq:dist_mbms_ch4}
\end{align}

The bound used in the different figures in this section is that given in Theorem~\ref{lemma:bound_tight} for deterministic outer codes.
When considering outer codes drawn for an ensemble the bound in Theorem~\ref{corollary:rateless} is used.

\subsubsection{Hamming outer code}
In this section we consider Raptor codes with binary Hamming outer codes. The weight enumerator of a Hamming code can be derived easily using the following recursion,
\[
(i+1)\, A_{i+1} + A_i + (n-i+1)\, A_{i-1}= \binom{n}{i}
\]
with $A_0=1$ and $A_1=0$ \cite{MacWillimas77:Book}. The expression obtained from this recursion can then be used together with Theorem~\ref{theorem:rateless} to derive an upper bound on the probability of decoding failure.

In Figure~\ref{fig:Hamming_sim} we show the probability of decoding failure for a Raptor code with a $(63,57)$ binary Hamming outer code as a function of the absolute overhead, $\absoverhead$.
In order to obtain the average probability of failure, Monte Carlo simulations were run until $200$ errors were collected. It can be observed how the upper bound is tight, specially in the error floor region (when $\absoverhead$ is large).
\begin{figure}[t]
    \centering
    \includegraphics[width=\figwbigger]{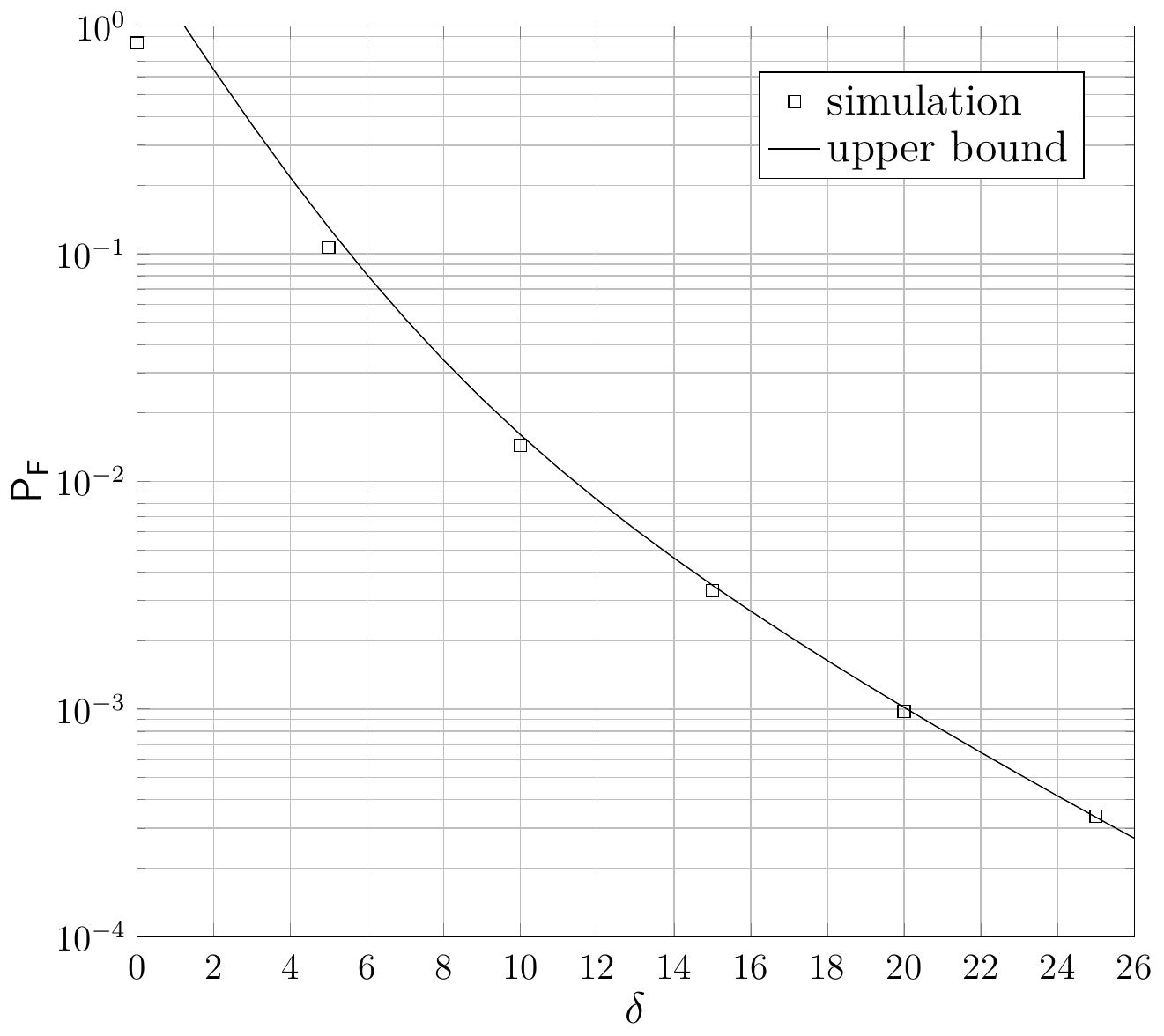}
    \caption[Probability of decoding failure $\Pf$ vs.\ $\absoverhead$ for a Raptor with a $(63,57)$ Hamming outer code]{Probability of decoding failure $\Pf$ vs.\ absolute overhead for a Raptor code with a $(63,57)$ Hamming outer code \fran{and \ac{LT} output degree distribution given in \eqref{eq:dist_mbms_ch4}}. The solid line denotes the upper bound on the probability of decoding failure \fran{ in Theorem~\ref{lemma:bound_tight}}. The markers denote simulation results.}\label{fig:Hamming_sim}
\end{figure}

\subsubsection{Linear random outer code}
In this section an ensemble of Raptor codes is consider where the outer code is selected from the $q$-ary linear random ensemble. The average weight enumerator of the  linear random ensemble was first derived in \cite{Gallager63} for the binary case and then in \cite{ashikhmin:98} for the $q$-ary case and  has the expression
\begin{align}
\weoensemble_{\l} = \binom{h}{\l} q^{-h (1-\ro)}  (q-1)^l.
\end{align}

In order to simulate the average probability of decoding failure of the ensemble, $6000$ different outer codes were generated. For each outer code and overhead value $10^3$ decoding attempts were carried out. The average probability of decoding failure was obtained averaging the probabilities of decoding failure obtained with the different outer codes. Note that the objective of the simulation was not characterizing the performance of every individual code but to characterize the average performance of the ensemble.
In order to select the outer code an $(h-k)\times h$ parity check matrix was selected at random by generating each of its elements according to a uniform distribution in $\mathbb F_{q}$.

\begin{figure}[t]
    \centering
    \includegraphics[width=\figwbigger]{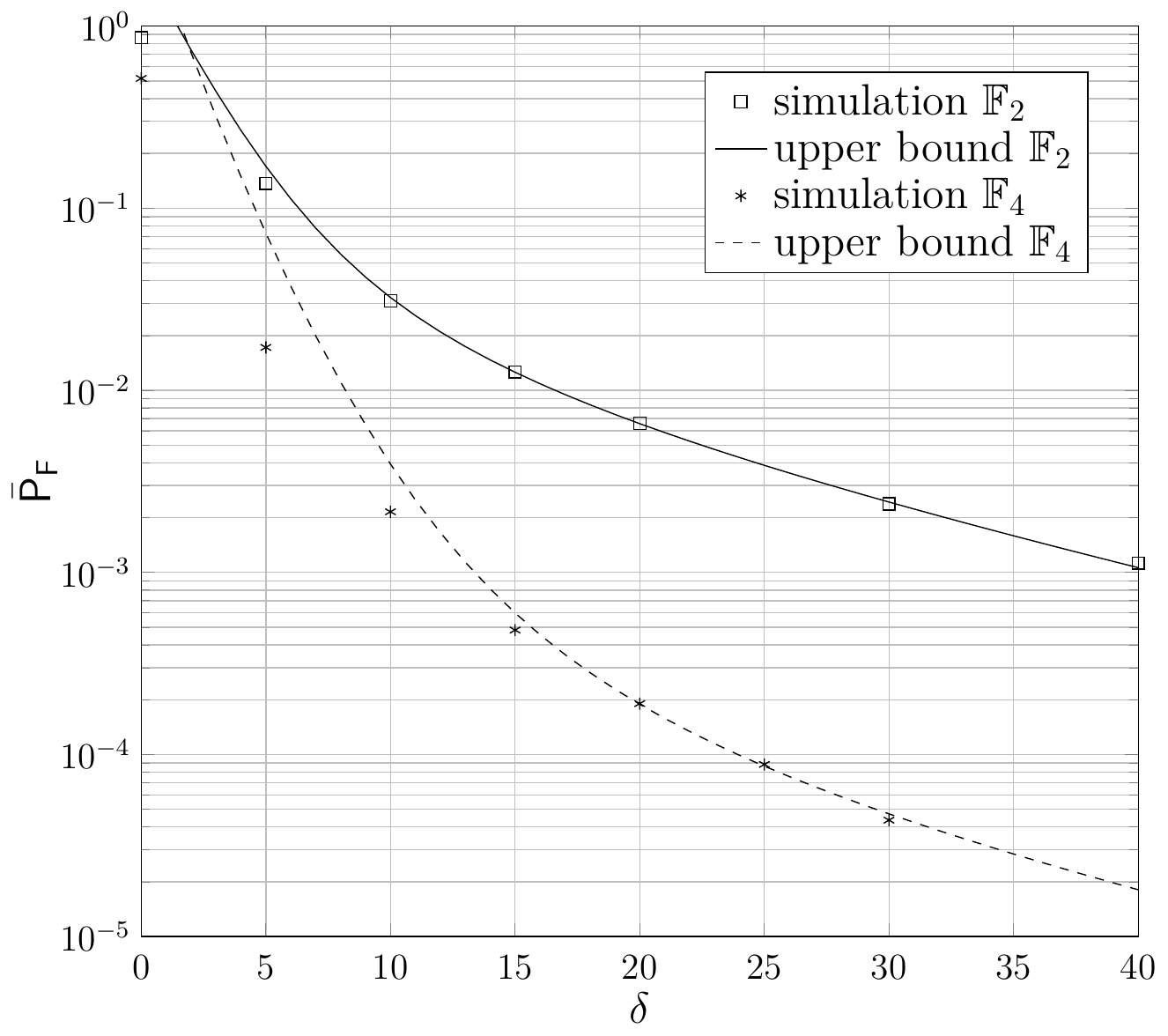}
    \caption[Expected probability of decoding failure $\barPf$ vs.\ $\absoverhead$ for two Raptor code ensembles where the outer code is selected from the (70,64) linear random ensemble  constructed over $\mathbb{F}_2$ and $\mathbb{F}_4$]{Expected probability of decoding failure $\barPf$ vs.\ absolute overhead for two Raptor code ensembles where the outer code is selected from the (70,64) linear random ensemble  constructed over $\mathbb{F}_2$ and $\mathbb{F}_4$ \fran{and with \ac{LT} output degree distribution given in \eqref{eq:dist_mbms_ch4}}. The solid and dashed lines denote the upper bound on the probability of decoding failure for
    the ensembles constructed over $\mathbb{F}_2$ and $\mathbb{F}_4$ respectively \fran{(see Theorem~\ref{corollary:rateless})}.
    The square and asterisk markers denote simulation results for $\mathbb{F}_2$ and  $\mathbb{F}_4$ respectively.}\label{fig:pf_k_64_h70}
\end{figure}

Figure~\ref{fig:pf_k_64_h70} shows the simulation results for Raptor codes with a linear random outer code with $k=64$ input symbols and $h=70$ intermediate symbols and the degree distribution in \eqref{eq:dist_mbms_ch4}. Two different ensembles were considered, a binary one and one constructed over $\mathbb{F}_4$. We can observe how in both cases the bounds hold and are tight except for very small values of $\absoverhead$.

\FloatBarrier

\FloatBarrier
\section{Inactivation Decoding Analysis}\label{chap:raptor_inactivation_decoding}

In this section we will consider Raptor codes under inactivation decoding, the efficient \ac{ML} decoding algorithm that is described in Section~\ref{chap:LT_inact} within the context of \ac{LT} codes. For simplicity, in this section we will consider only binary Raptor codes, being the extension to non-binary Raptor codes straightforward.

Let us recall, that \ac{ML} decoding of \ac{LT} codes consists of solving the system of equations in \eqref{eq:ml_eq_sys}:
\[
\mathbf{\rosymb}^T = \Grx^T \mathbf{v}^T
\]
where $\mathbf{\rosymb}$ is the (row) vector of received output symbols, $\mathbf{v}$ is the (row) vector of source symbols and $\Grx^T$ is the transposed generator matrix of the \ac{LT} code after removing the rows associated to output symbols that were erased by the channel. The matrix $\Grx^T$ has dimensions $m \times k$, $m$ being the number of output symbols collected and $k$ the number of source symbols.

\ac{ML} decoding of Raptor codes consists of solving the system of equations in \eqref{eq:raptor_sys_eq}. That is,
\[
\begin{bmatrix}
\zeros \\
\mathbf{\Rrosymb^T}
\end{bmatrix} =
\constmatrix \, \vecv^T
\]
where $\vecv$ is a row vector representing the intermediate symbols, that are the input to the \ac{LT} encoder, $\mathbf{\Rrosymb}$ is a column vector representing the received output symbols and $\constmatrix$ is the constraint matrix of the Raptor code with dimension $\left( \left(h-k+m \right) \times h\right)$. Vector $\zeros$ is a $(h-k) \times 1$ column zero vector (see Section~\ref{chap:raptor_under_ML} for more details).

If we compare the two systems of equations, we can see how Raptor \ac{ML} decoding is very similar to \ac{ML} decoding of an \ac{LT} code with $h$ source symbols and $h-k+m$ output symbols. The role of matrix $\Grx$ for \ac{LT} codes is played by the constraint matrix of the Raptor code, $\constmatrix$. The main difference is that while for \ac{LT} codes all the rows of $\Grx$ are independent and identically distributed according to the degree distribution $\Omega$, this is no longer true for matrix $\constmatrix$ (for a generic Raptor code). Let us recall that the constraint matrix of a Raptor code is defined as:
\[
\constmatrix =
\begin{bmatrix}
 \hmatrixpre \\
 \GrxLT^T
\end{bmatrix}
\]
where $\hmatrixpre$ is the parity check matrix of the outer code, and $\GrxLT$ is a binary matrix that  corresponds to the generator matrix of the \ac{LT} code after removing the columns associated with the output symbols that were erased by the channel. Thus, while the rows in $\GrxLT$  are independent and identically distributed this is not true for matrix $\constmatrix$.

Let us introduce the following definition.

\begin{mydef}[Surrogate  \ac{LT} code]
Consider a Raptor code with outer code parity-check matrix $\hmatrixpre$ and an inner \ac{LT} code with degree distribution $\Omega$ and assume the receiver has collected $m=k + \absoverhead$ output symbols.
The surrogate \ac{LT} code of the Raptor code is defined as an \ac{LT} code with $h$ input symbols, $m$ output symbols, and degree distribution $\Omegaeq$ given by the expected Hamming weight distribution of  the rows of the constraint matrix $M$ of the Raptor code. Formally
\[
\Omegaeq= \frac{h-k}{h-k+m} \Theta + \frac{m}{h-k+m} \Omega = \Omegaeq= \frac{h-k}{h+\absoverhead} \Theta + \frac{m}{h+\absoverhead} \Omega
\]
where $\Theta=\{ \Theta_1,\Theta_2,\hdots,\Theta_h \}$, being  $\Theta_i$ the fraction of rows of Hamming weight $i$ in $\hmatrixpre$.
\end{mydef}
\begin{remark} The degree distribution of the surrogate \ac{LT} code, $\Omegaeq$ depends on the receiver overhead $\absoverhead$.
\end{remark}

Given the similarity between inactivation decoding of Raptor and \ac{LT} codes, it is possible to approximate inactivation decoding of a Raptor code as inactivation decoding of its surrogate \ac{LT} code. Using this heuristic approximation we will show how the approaches derived in Sections~\ref{chap:inact_first_order}, \ref{chap:inact_distribution} and \ref{chap:inact_low_complex} for \ac{LT} codes can be adapted to approximate inactivation decoding of Raptor codes with a reasonable accuracy. 


\subsection{Raptor Codes with Linear Random Outer Codes}\label{sec:raptor_inact_rand}

In the case of Raptor codes with a linear Random outer code the constraint matrix corresponds to:
\[
\constmatrix =
\begin{bmatrix}
 \hmatrixpre \\
 \GrxLT^T
\end{bmatrix}
\]
where
\begin{itemize}
\item $\hmatrixpre$ is the parity check matrix of a linear random code with size $\left( \left(h-k \right) \times h\right)$. The Hamming weight of each row corresponds to a binomial random variable with parameters $h$ and $1/2$.
\item $\GrxLT$ is a $\left( h \times m\right)$ binary matrix which defines the relation between the intermediate symbols and the output symbols due to the LT encoding. The Hamming weight of each column corresponds to the output degree distribution of the inner \ac{LT} code $\Omega$.
\end{itemize}
Thus the Hamming weight distribution of $\hmatrixpre$ corresponds to
\[
\Theta_i =  \mathcal B (h, 1/2)
\]
where $\mathcal B (h, 1/2)$ is a binomial distribution with parameters $h$ and $1/2$. Therefore, the degree distribution of the surrogate \ac{LT} code corresponds to
\[
\Omegaeq = \frac{m}{h-k+m} \Omega +  \frac{h-k}{h-k+m} \mathcal B (h, 1/2)
\]

For illustration, in Figure~\ref{fig:omega_equiv} we provide an example of degree distribution of the  surrogate \ac{LT} code for a  Raptor code with a $(106,80)$ linear random outer code with $m=80$ and degree distribution $\Omegarten$. The contribution of the outer code can be clearly distinguished, it corresponds to the bell shaped curve around degree $d=53$.
\begin{figure}
        \centering
        \includegraphics[width=\figwbigger]{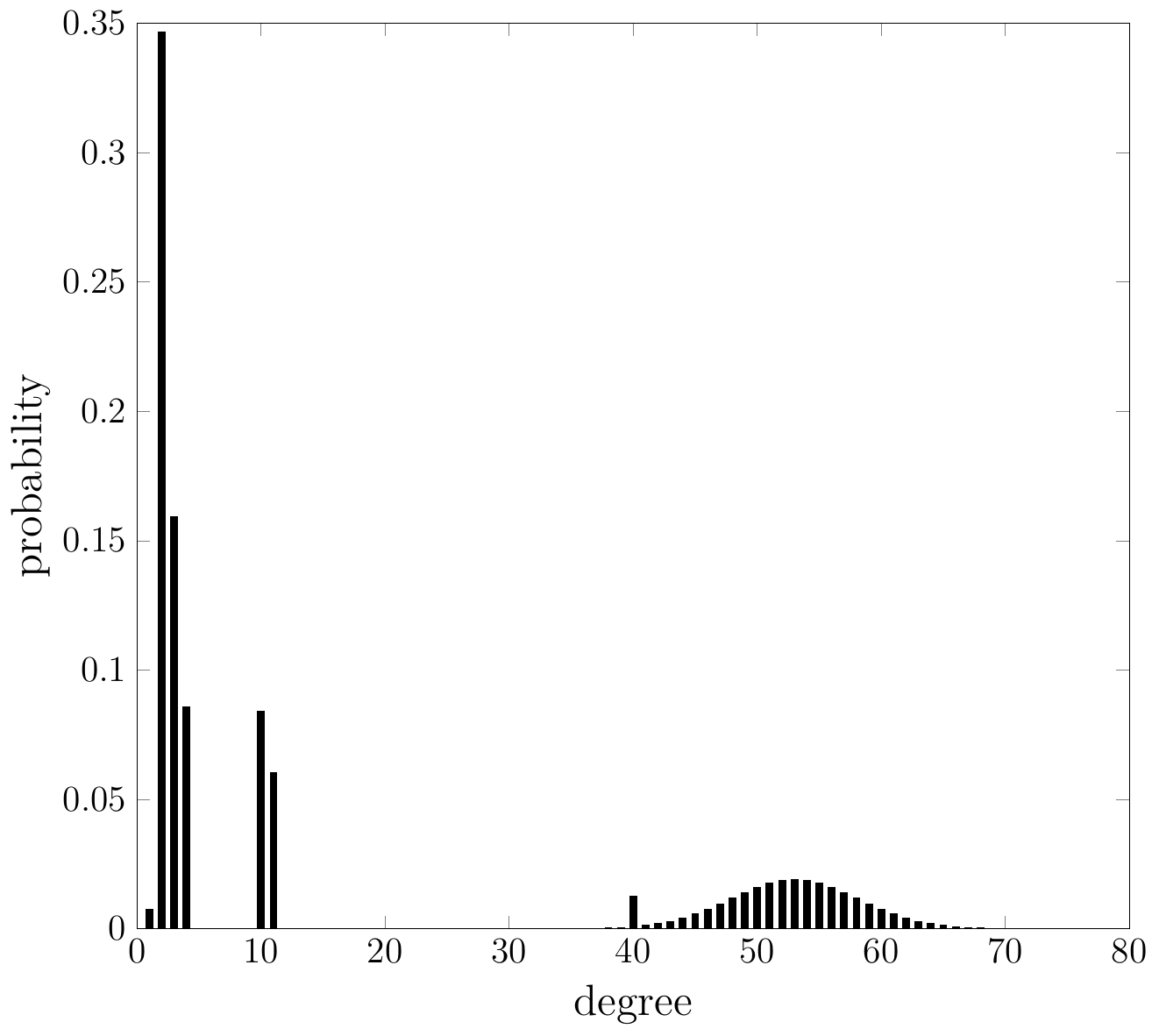}
        \caption[Surrogate \ac{LT} output degree distribution for a Raptor code with a linear random precode.]{Surrogate \ac{LT} degree distribution $\Omegaeq$ for a Raptor code with a $(106,80)$ linear random outer code with $m=80$ and degree distribution $\Omegarten$. }
\label{fig:omega_equiv}
\end{figure}

Figure~\ref{fig:raptor_inact} shows the average number of inactivations vs.\ the absolute receiver overhead $\absoverhead$ for a Raptor code with a $(233,200)$ linear random precode with degree distribution $\Omegarten$. The figure shows results obtained by Monte Carlo simulations and the approximations obtained using the methods in Sections~\ref{chap:inact_first_order} and \ref{chap:inact_low_complex} for the surrogate \ac{LT} code. The match between the simulation results and the approximation is good.

\begin{figure}[t]
        \centering
        \includegraphics[width=\figwbigger]{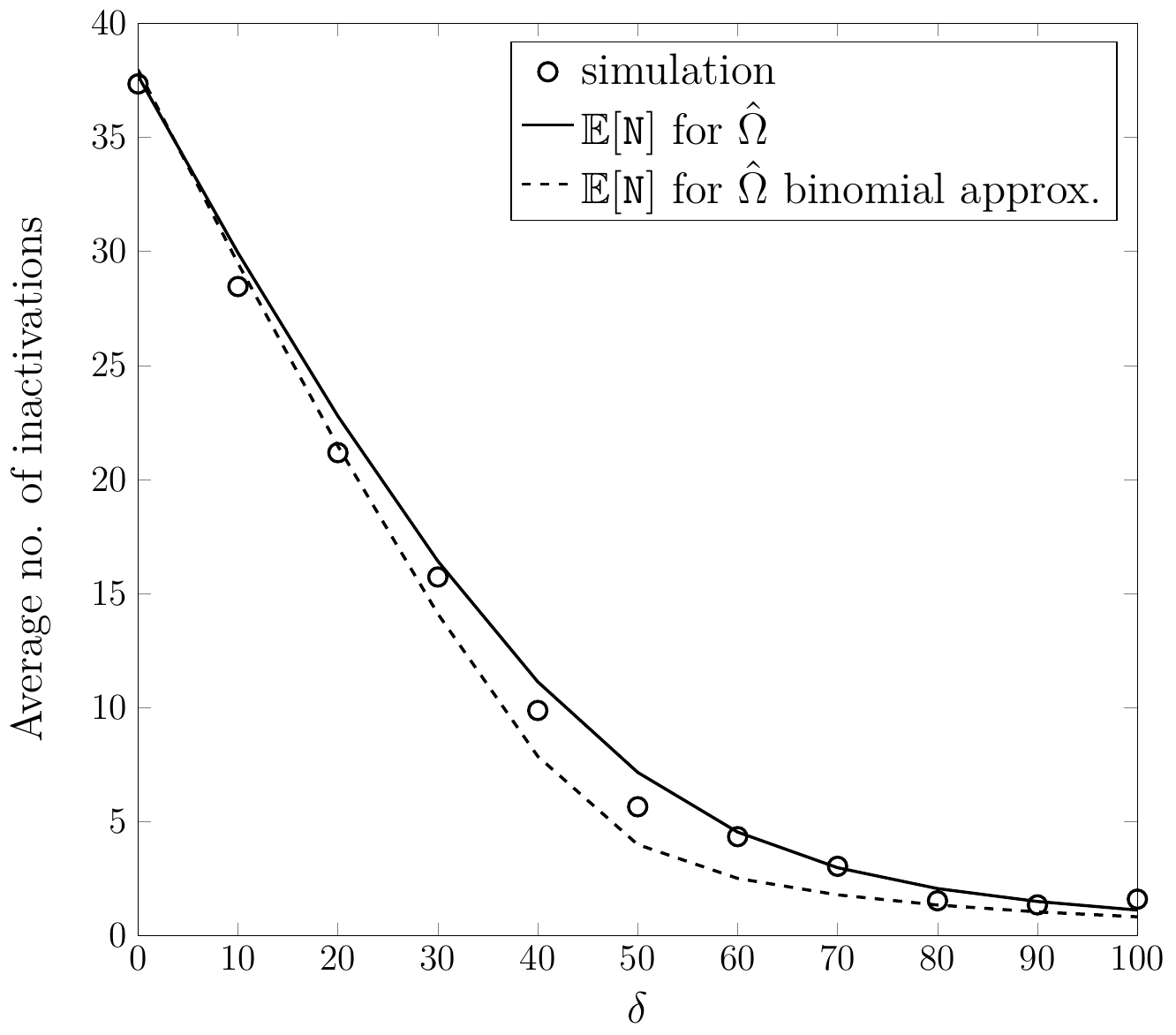}
        \caption[Average number of inactivations vs.\ $\absoverhead$ for a Raptor code with a $(233,200)$ linear random outer code with degree distribution $\Omegarten$]{Average number of inactivations vs.\ absolute receiver overhead $\absoverhead$ for a Raptor code with a $(233,200)$ linear random outer code with degree distribution $\Omegarten$. The markers represent the results of Monte Carlo simulations. The solid and dashed lines represent the number of inactivations for the surrogate \ac{LT} code using the methods Sections~\ref{chap:inact_first_order} and \ref{chap:inact_low_complex} respectively.}
\label{fig:raptor_inact}
\end{figure}

Finally, Figure~\ref{fig:raptor_inact_dist} shows the distribution of the number of inactivations for a Raptor code with a $(233,200)$ linear random outer code with $m=200$ and degree distribution $\Omegarten$ and the approximation obtained using the method in Section~\ref{chap:inact_distribution} for the surrogate \ac{LT} code. We can observe how the estimation of the distribution of the number of inactivations is not very accurate. While the average value is estimated correctly, the actual distribution of the number of inactivations is more concentrated around the mean than its estimation. \fran{A possible explanation for this effect is that the Raptor code has a constant number large Hamming weight rows in its constraint matrix, which correspond to the parity check matrix of the outer code. However, its surrogate \ac{LT} code implicitly makes the assumption that the number of large weight rows is random. Thus, it also considers realizations with too many/few large Hamming weight rows (output symbols), leading to a higher dispersion (less concentration) of the number of inactivations around the mean value.}
\begin{figure}[t]
        \centering
        \includegraphics[width=\figwbigger]{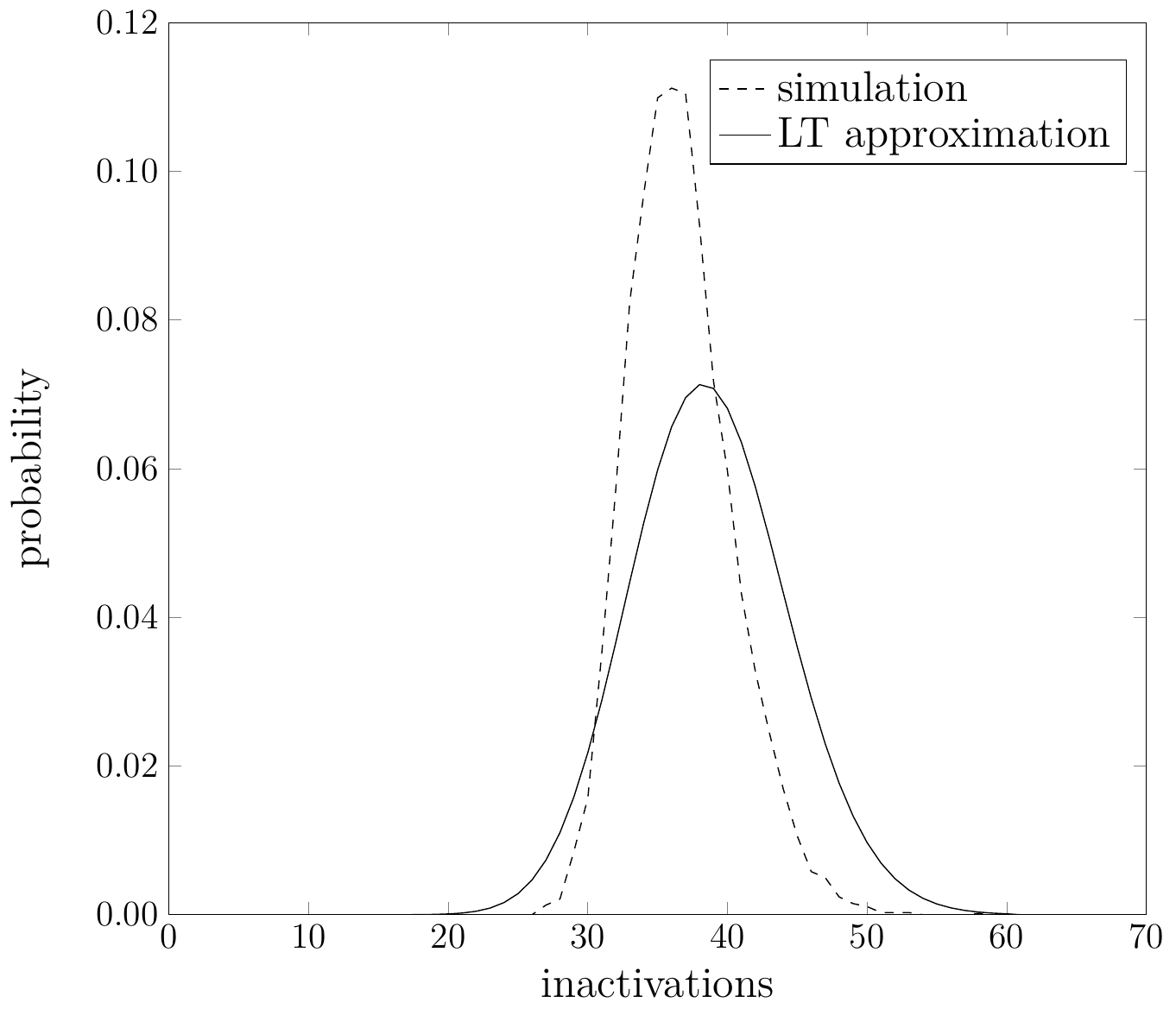}
        \caption[Distribution of number of inactivations for a Raptor code with a $(233,200)$ linear random outer code with $m=200$ and degree distribution $\Omegarten$]{Distribution of the number of inactivations for a Raptor code with a $(233,200)$ linear random outer code with $m=200$ and degree distribution $\Omegarten$. The dashed line represents the results of Monte Carlo simulations. The solid line represents the estimated number of inactivations using the method in Section~\ref{chap:inact_distribution} for the surrogate \ac{LT} code.}
\label{fig:raptor_inact_dist}
\end{figure}

\FloatBarrier

\subsection{R10 Raptor Codes}\label{sec:raptor_inact_r10}

In the case of R10 Raptor codes, the precode is a concatenation of two systematic codes, an \ac{LDPC} code and an \ac{HDPC} code (see Section~\ref{chap:raptor_r10}). Hence, two different parts can be distinguished in the parity check matrix of the precode:
\begin{itemize}
\item \ac{LDPC} part. There are $\srten$ rows associated to \ac{LDPC} redundant symbols. The Hamming weight of each row is approximately  $3 \lfloor k/\srten\rceil +1$, where $\lfloor x \rceil$ denotes the closest integer to $x$.
\item \ac{HDPC} part. There are $\hrten$ rows associated to \ac{HDPC} redundant symbols. The Hamming weight of each row is approximately $ \lfloor(k+\srten)/2\rceil + 1$
\end{itemize}
Thus, the distribution of the Hamming weight of the rows of $\hmatrixpre$ is approximately
\[
\Theta_i \approx \frac{\srten}{\srten+\hrten}  \mathcal{D}\left(3 \left\lfloor \frac{k}{\srten} \right\rceil +1\right)+  \frac{\hrten}{\srten+\hrten}  \mathcal{D}\left(\left\lfloor \frac{k+\srten}{2}\right\rceil + 1\right),
\]
where $\mathcal{D}(i)$ denotes a (discrete) Kronecker delta at $i$. Therefore, the degree distribution of the surrogate \ac{LT} code is approximately
\begin{align}
\Omegaeq &\approx \frac{m}{\srten+\hrten-k+m} \Omega +  \frac{\srten}{\srten+\hrten+m}  \mathcal{D}\left(3 \left \lfloor  \frac{k}{\srten} \right \rceil +1 \right) \\
&+ \frac{\hrten}{\srten+\hrten+m}  \mathcal{D}\left( \left\lfloor\frac{k+\srten}{2} \right\rceil + 1\right).
\label{eq:omega_eq_r10}
\end{align}

For illustration, in Figure~\ref{fig:omega_equiv_r10} we provide the surrogate \ac{LT} code degree distribution for a R10 Raptor code with $k=80$ and $m=80$.  In this case the  redundant \ac{LDPC} symbols are modeled as degree $16$ output symbols and the \ac{HDPC} symbols by degree $50$ output symbols.
\begin{figure}
        \centering
        \includegraphics[width=\figwbigger]{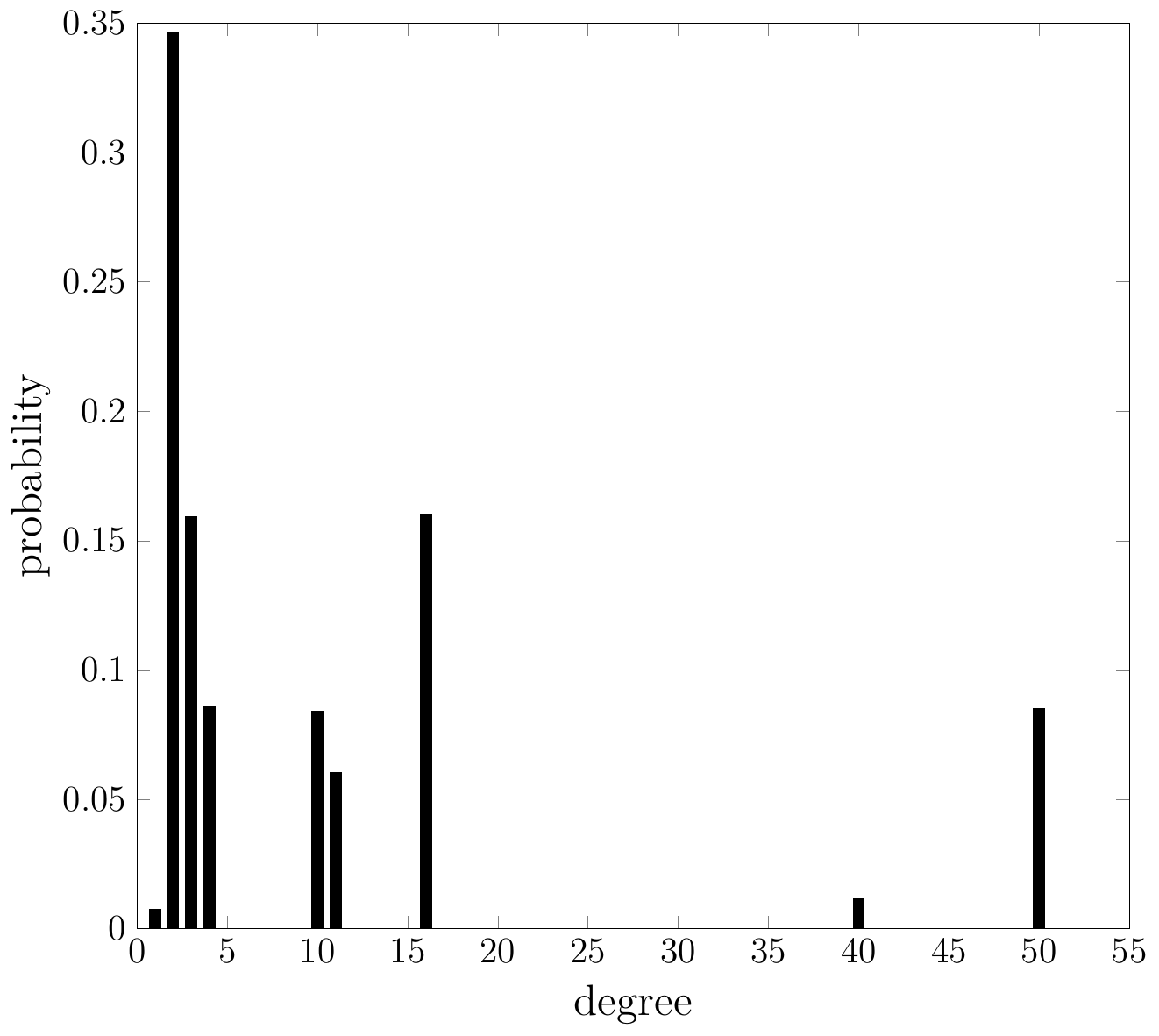}
        \caption[Surrogate \ac{LT} output degree distribution for an R10 Raptor code.]{Surrogate \ac{LT} degree distribution $\Omegaeq$ for an R10 Raptor code with $k=80$.}
\label{fig:omega_equiv_r10}
\end{figure}

Using the surrogate \ac{LT} code approximation, the methods presented in Sections~\ref{chap:inact_first_order}, \ref{chap:inact_distribution} and \ref{chap:inact_low_complex} can be used to estimate the number of inactivations needed to complete Raptor decoding.

\begin{figure}
        \centering
        \includegraphics[width=0.579\columnwidth]{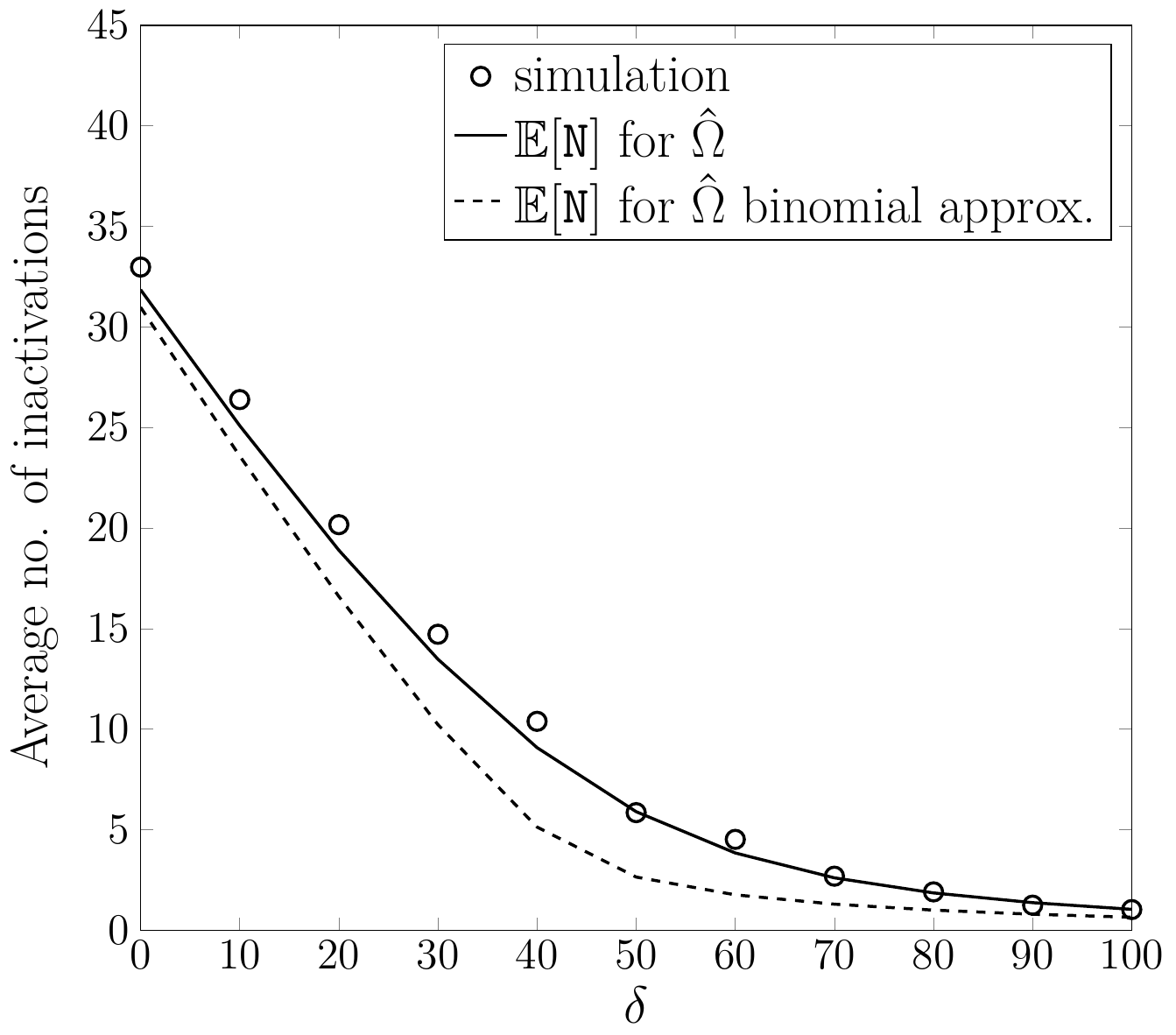}
        \caption[Average number of inactivations vs.\ $\absoverhead$ for a R10 Raptor code with $k=200$]{Average number of inactivations vs.\ absolute receiver overhead $\absoverhead$ for a R10 Raptor code with $k=200$. The markers represent the results of Monte Carlo simulations. The solid and dashed lines represent the estimated number of inactivations for the surrogate \ac{LT} code using the methods in Sections~\ref{chap:inact_first_order} and \ref{chap:inact_low_complex} respectively.}
\label{fig:raptor_inact_r10}
        \hspace{-4mm}\includegraphics[width=\figw]{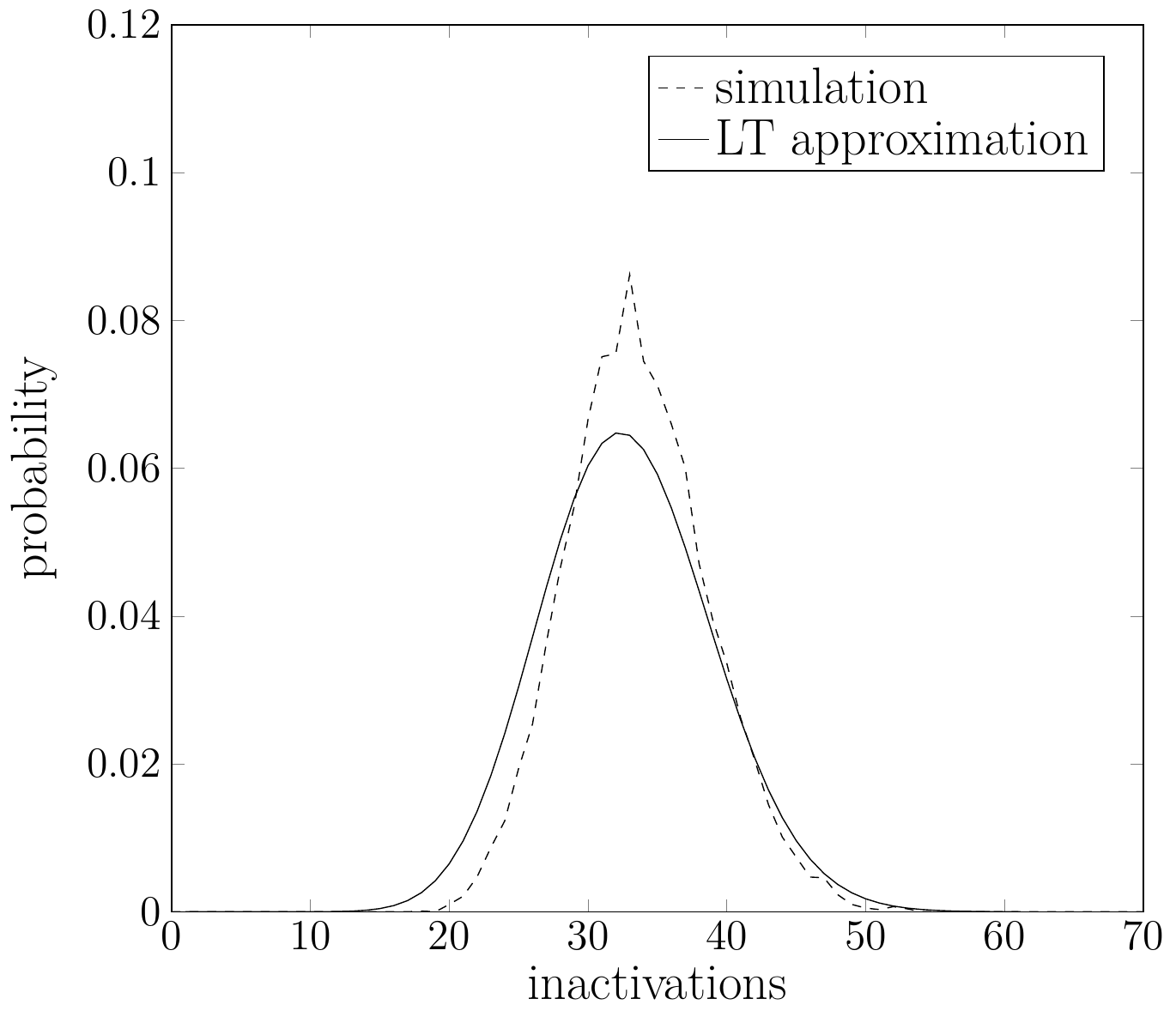}
        \caption[Distribution of number of inactivations for a R10 Raptor code with $k=200$ and $m=200$]{Distribution of the number of inactivations for a R10 Raptor code with $k=200$ and $m=200$. The dashed line represents the results of Monte Carlo simulations. The solid line represents the estimated number of inactivations using the surrogate \ac{LT} code approximation and the method in Section~\ref{chap:inact_distribution}.}
\label{fig:raptor_r10_inact_dist}
\end{figure}

Figure~\ref{fig:raptor_inact_r10} shows the average number of inactivations vs.\ the absolute receiver overhead $\absoverhead$ for a  R10 Raptor code with $k=200$. The figure shows results obtained by Monte Carlo simulations and the approximations obtained using the methods in Sections~\ref{chap:inact_first_order} and \ref{chap:inact_low_complex} for the surrogate \ac{LT} code. The match between the simulation results and the method in Section~\ref{chap:inact_low_complex} is good.

Finally, Figure~\ref{fig:raptor_r10_inact_dist} shows the distribution of the number of inactivations for a R10 Raptor code with a $k=200$ and $m=200$ and its estimation using the surrogate \ac{LT} code approximation and the method in Section~\ref{chap:inact_distribution}.
\fran{If we compare this figure with Figure~\ref{fig:raptor_inact_dist} we can see how the estimation of the distribution of the number of inactivations works better for R10 Raptor codes than for codes with a linear random outer code.}

\subsection{Discussion}
In this section we have proposed an approximate analysis of Raptor codes by introducing the concept of the surrogate \ac{LT} code of a Raptor code. The simulation results presented in Sections~\ref{sec:raptor_inact_rand} and \ref{sec:raptor_inact_r10} show how the approximation is reasonably good in order to estimate the expected number of inactivations. However, the approximation is not accurate enough to estimate the distribution of the number of inactivations.  The reason for this deviation is the fact that the surrogate \ac{LT} approximation makes implicitly the assumption that the rows of the constraint matrix $\constmatrix$ are independent and identically distributed. \fran{More concretely, it assumes there is a random number of large Hamming weight rows, which correspond to the rows in the parity check matrix of the outer code. However, for a Raptor code, the number of rows of $\constmatrix$ that correspond to the parity check matrix of the outer code is fixed.}

\section{Code Design}\label{chap:rapt_code_design}

Within this section we provide a Raptor code design example. More concretely we design a degree distribution for the \ac{LT} component of a binary Raptor code with a $(63,57)$ outer Hamming code. The design goal is  achieving a target probability of decoding failure lower than ${\Pf}^*=10^{-3}$ for $\absoverhead=15$ while minimizing the number of inactivations needed for decoding at $\absoverhead=15$.
Moreover, we will constraint the output degree distribution to have exactly the same maximum and average output degree as standard R10 Raptor codes ($\avgd=4.6314$ and $\dmax=40$). \fran{Note that a constraint on the average output degree is equivalent to a constraint on the average encoding complexity /cost. Moreover, the constraint on the maximum output degree gives us control on the worst case encoding complexity.} Furthermore we will constraint the output degree distribution to have the same support as the degree distribution of R10 raptor codes, that is, only degrees $1,2,3,4,10,11$ and $40$ will be assigned a probability larger than $0$. 
These constraints are chosen to illustrate the fact that arbitrary constraints can be introduced in the code design.

The design of the \ac{LT} output degree distribution is  formulated as a numerical optimization problem. More concretely, the numerical optimization algorithm that is used is \acf{SA} (see Section~\ref{sec:lt_code_design} for more details). The objective function to be minimized is defined as:
\[
\Upsilon = \Exp \left[ \hat \Y \right] + f_p (\Pf^{\mathrm{up}})
\]
where $\Exp \left[ \hat \Y \right]$ is the estimated number of inactivations needed for decoding the surrogate \ac{LT} code and $f_p$ is defined as
\begin{align}
f_p (\Pf^{\mathrm{up}}) =
    \begin{cases}
        0,  & \Pf^{\mathrm{up}}<{\Pf}^* \\
        b~(1- {\Pf}^* / \Pf^{\mathrm{up}}),& \mathrm{else}
    \end{cases}
\end{align}
being ${\Pf}^*$ the target probability of decoding failure at $\absoverhead=15$ , $\Pf^{\mathrm{up}}$ its upper bound given in Theorem~\ref{theorem:rateless} and $b$ a large positive number ($b= 10000$ was used in the example). The large $b$ factor ensures that degree distributions which do not comply with the target probability of decoding failure are discarded.

The degree distribution obtained from our optimization is the following:
\begin{align}
\Omega^*(\x) &= 0.0490 ~\x^{1} + 0.3535 ~\x^{2} +  0.1135 ~\x^{3} +  0.2401 ~\x^{4}  \\
&+  0.1250 ~\x^{10} +  0.1183 ~\x^{11}  +  0.0006 ~\x^{40}.
\label{eq:omega_opt}
\end{align}

\begin{figure}[t]
        \begin{center}
        \includegraphics[width=\figwbigger]{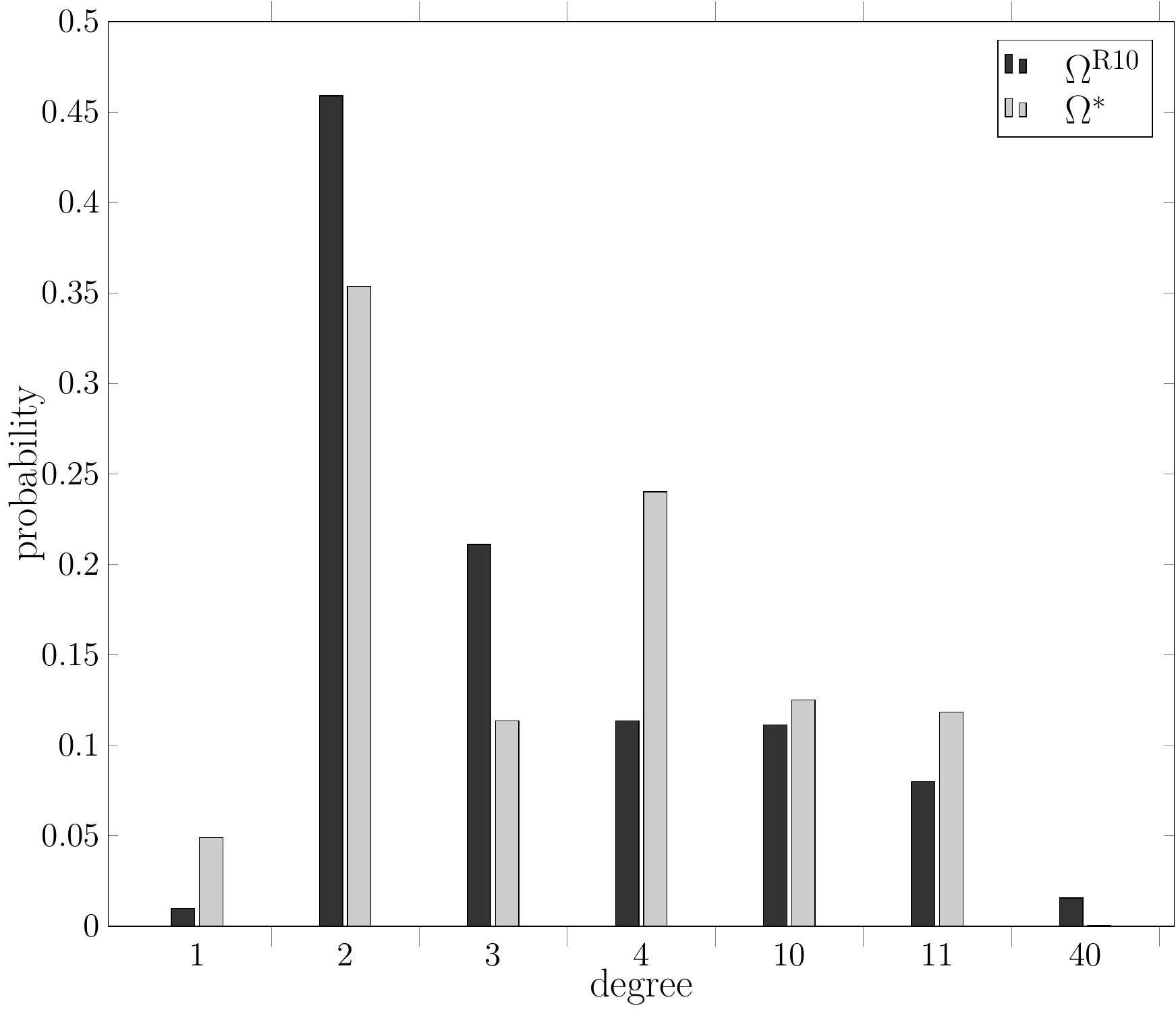}
        \centering
        \caption[Comparison of the output degree distribution of R10 Raptor codes, $\Omegarten$ with the output degree distribution obtained through optimization $\Omega^*$.]{Comparison of the output degree distribution of R10 Raptor codes, $\Omegarten$ with the output degree distribution obtained through optimization $\Omega^*$.}
        \label{fig:omega_comp}
        \end{center}
\end{figure}

\fran{Figure~\ref{fig:omega_comp} compares the output degree distribution of R10 Raptor codes $\Omegarten$,  given in \eqref{eq:dist_mbms_ch4}, with the degree distribution obtained from our optimization $\Omega^*$, given in \eqref{eq:omega_opt}.  Both distributions have the same average output degree, and in both cases the degree with maximum probability is $2$. However, the distributions are quite different.}

In order to compare the performance of the two Raptor codes considered Monte Carlo simulations were carried out. In order to derive the probability of decoding failure for each overhead value $\absoverhead$ simulations were run until 200 errors were collected. To obtain the average number of inactivations, 1000 decodings were carried out for each overhead value $\absoverhead$.

Figure~\ref{fig:overhead_perf} shows the probability of decoding failure $\Pf$ vs.\ the absolute receiver overhead $\absoverhead$ for the two binary Raptor codes with Hamming outer codes, with degree distributions $\Omegarten$ and $\Omega^*$. The upper bound to the probability of failure is also shown for both Raptor codes. We can observe how the Raptor code with degree distribution $\Omega^*$ meets the design goal, its probability of decoding failure at $\absoverhead=15$ is below $10^{-3}$. If we compare the two Raptor codes, we can see how the probability of decoding failure of the Raptor code with $\Omega^*$ is lower than that with $\Omegarten$. For $\Pf$ below $10^{-3}$,  the Raptor code with degree distribution $\Omega^*$ needs approximately $5$ less overhead symbols to achieve the same $\Pf$ as the Raptor code with degree distribution $\Omegarten$.

In Figure~\ref{fig:inact_Hamming_Design} the average number of inactivations is shown as a function of the absolute receiver overhead for the two binary Raptor codes with Hamming outer codes, with degree distributions $\Omegarten$ and $\Omega^*$.  We can observe how the degree distribution obtained from the optimization process, $\Omega^*$, leads to a higher number of inactivations, and, thus, to a higher decoding complexity.

\begin{figure}[t]
        \begin{center}
        \includegraphics[width=\figwbigger]{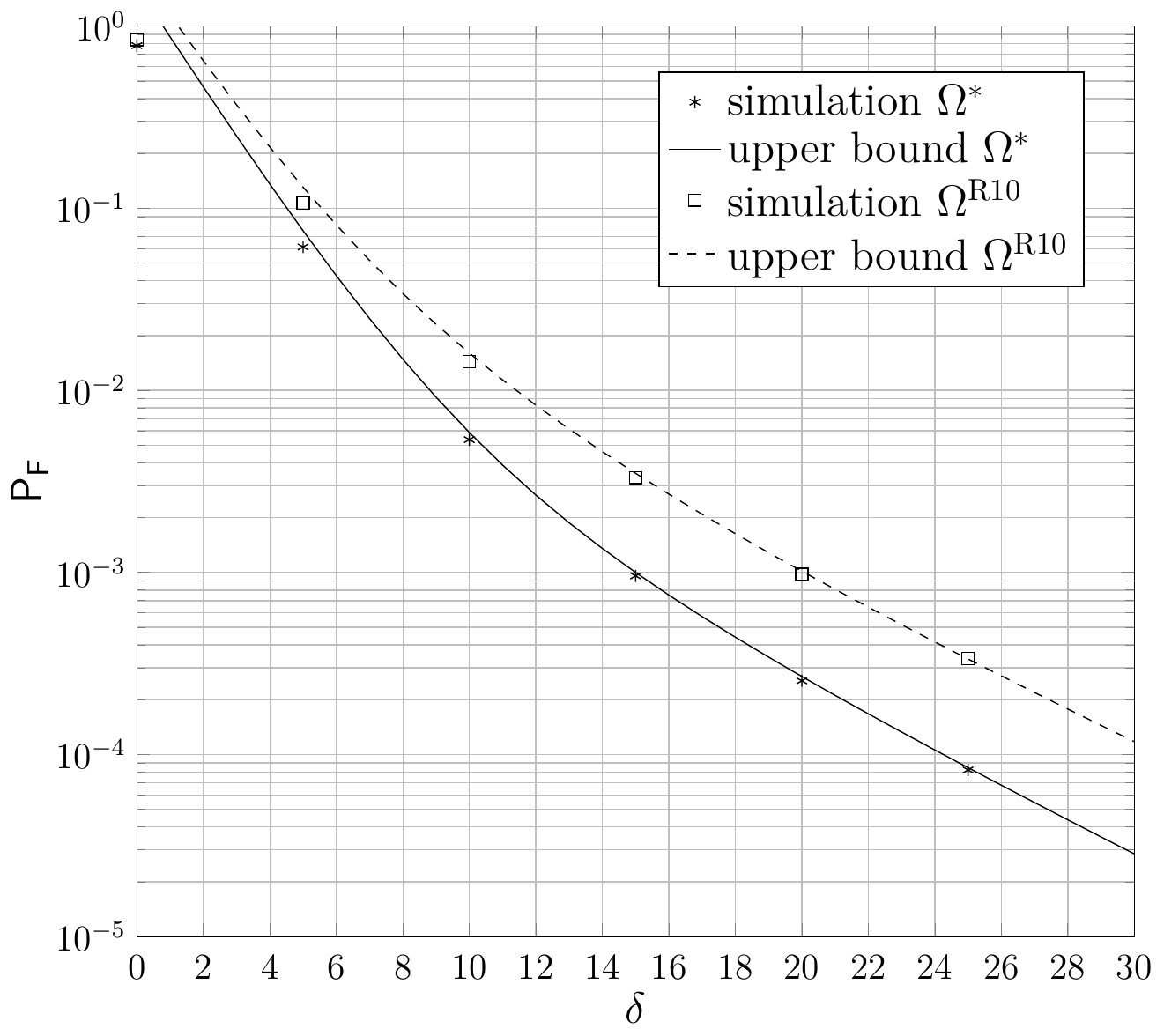}
        \centering
        \caption[Probability of decoding failure $\Pf$ vs.\ $\absoverhead$ for binary Raptor codes with a $(63,57)$ Hamming outer code and \acs{LT} degree distributions $\Omega^*$ and $\Omegarten$]{Probability of decoding failure $\Pf$ vs.\ absolute receiver overhead $\absoverhead$ for binary Raptor codes with a $(63,57)$ Hamming outer code and \acs{LT} degree distributions $\Omega^*$ and $\Omegarten$ The markers represent the result of simulations, while the lines represents the upper bound to the probability of decoding failure in Theorem~\ref{theorem:rateless}.}
        \label{fig:overhead_perf}
        \end{center}
\end{figure}
\begin{figure}[t]
        \begin{center}
        \includegraphics[width=\figwbigger]{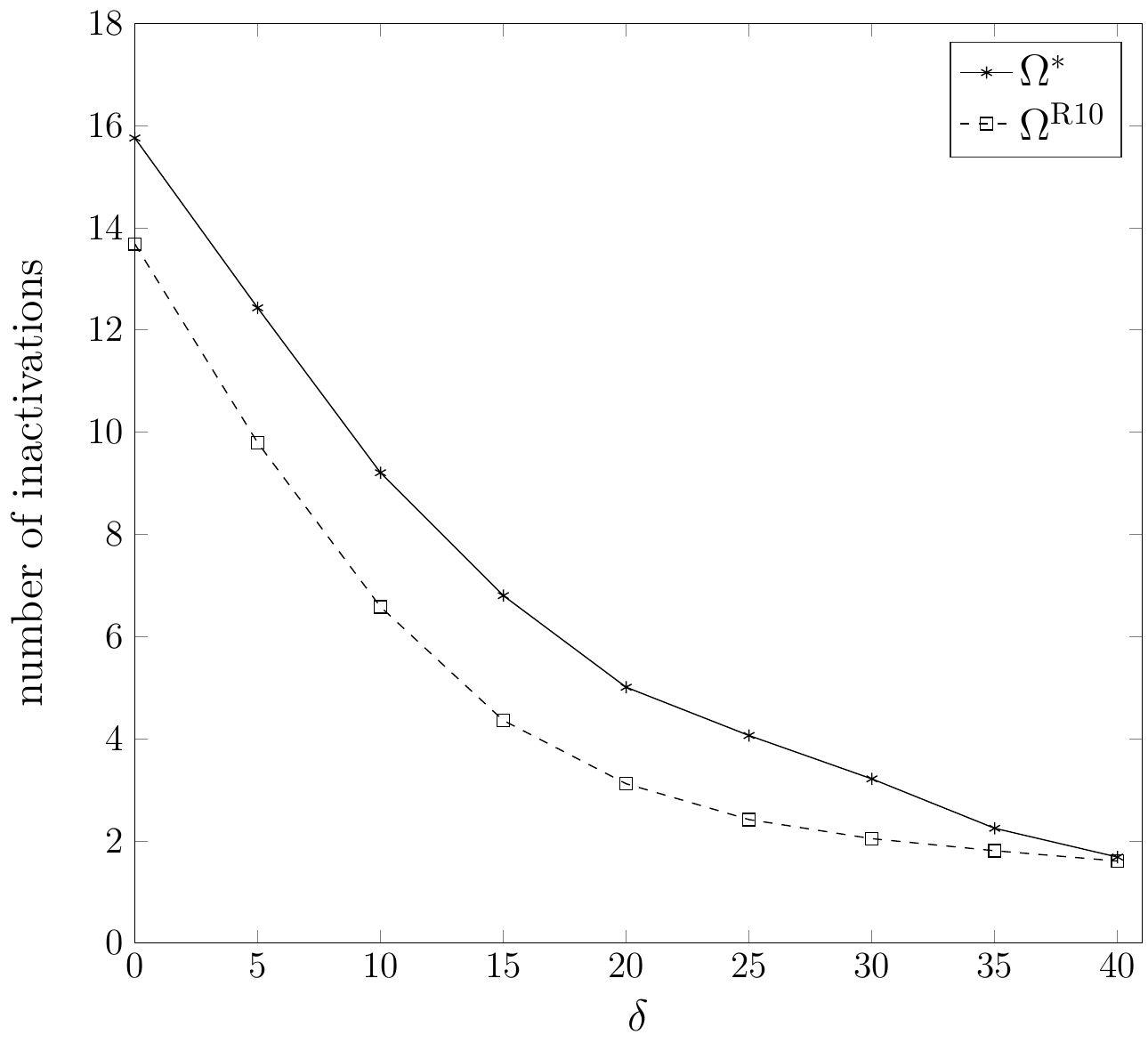}
        \centering
        \caption[Number of inactivations vs.\ $\absoverhead$ for binary Raptor codes with a $(63,57)$ Hamming outer code and \acs{LT} degree distributions $\Omega^*$ and $\Omegarten$]{Number of inactivations vs.\ absolute receiver overhead $\absoverhead$ for binary Raptor codes with a $(63,57)$ Hamming outer code and \acs{LT} degree distributions $\Omega^*$ and $\Omegarten$.}
        \label{fig:inact_Hamming_Design}
        \end{center}
\end{figure}

The results in Figures~\ref{fig:overhead_perf} and \ref{fig:inact_Hamming_Design} illustrate the tradeoff between probability of decoding failure and number of inactivations (decoding complexity). In general if one desires to improve the probability of decoding failure it is necessary to accept a higher decoding complexity.


\section{Summary}\label{chap:raptor_summary}
In this chapter we have focused on Raptor codes under inactivation decoding.
In Section \ref{chap:raptor_rateless_ml} an upper bound to the probability of decoding failure of $q$-ary Raptor codes under \ac{ML} decoding has been derived. This bound is based on the weight enumerator of the outer code, or its average weight enumerator when the outer code is randomly drawn from a code ensemble. The bounds derived are tight, specially in the error floor region, as it is shown by means of simulations.
In Section \ref{chap:raptor_inactivation_decoding} a heuristic method is presented that yields an approximate analysis of Raptor codes under inactivation decoding. The method is shown to be accurate for several examples.
Finally in Section~\ref{chap:rapt_code_design} a code design is presented based on the results presented in this chapter. More concretely, we  have designed the degree distribution of the \ac{LT} component of a binary Raptor code with a $(63,57)$  Hamming outer code. The design goal was obtaining a probability of decoding failure $\Pf<10^{-3}$ at $\absoverhead=15$ while minimizing the number of inactivations needed for decoding.



%% file: Chapter5/chapter5.tex
\chapter{Fixed-Rate Raptor Codes} \label{chap:Raptor_fixed_rate}
\ifpdf
    \graphicspath{{Chapter5/Chapter5Figs/PNG/}{Chapter5/Chapter5Figs/PDF/}{Chapter5/Chapter5Figs/}}
\else
    \graphicspath{{Chapter5/Chapter5Figs/EPS/}{Chapter5/Chapter5Figs/}}
\fi

Despite the fact that Raptor codes were originally designed for a rateless setting, they are sometimes used as fixed-rate codes due to their excellent performance and low complexity (see Section \ref{chap:raptor_r10}). In this chapter we focus on the performance of Raptor codes under \ac{ML} decoding in a fixed-rate setting. More concretely we analyze the distance properties of an ensemble of (fixed-rate) Raptor codes with linear random outer codes that resembles R10 Raptor codes. This chapter is organized as follows. In Section~\ref{sec:raptor_ensemble} we introduce the ensemble of raptor codes to be studied. In Section~\ref{sec:dist_spectrum_raptor} the average distance spectrum of the ensemble is derived. Section~\ref{sec:dist_region} presents sufficient and necessary conditions for the ensembles of Raptor codes to have a minimum distance growing linearly with the block length. In Section~\ref{sec:raptor_fr_sims} simulations are presented that validate the analytical results obtained in this chapter. Moreover, it is shown by means of simulations how the erasure correcting properties of the ensemble studied in this chapter resemble those of standard R10 Raptor codes as a first order approximation. Finally, the main contributions of this chapter are summarized in Section~\ref{chap:raptor_fixed_rate_summary}.

\section{Raptor Code Ensemble}\label{sec:raptor_ensemble}

A fixed-rate Raptor can be seen as the concatenation of a precode with a fixed-rate \ac{LT} code, as shown in Figure~\ref{fig:raptor_serial}.

\begin{figure}[t]
        \centering
        {\includegraphics[width=0.90\columnwidth]{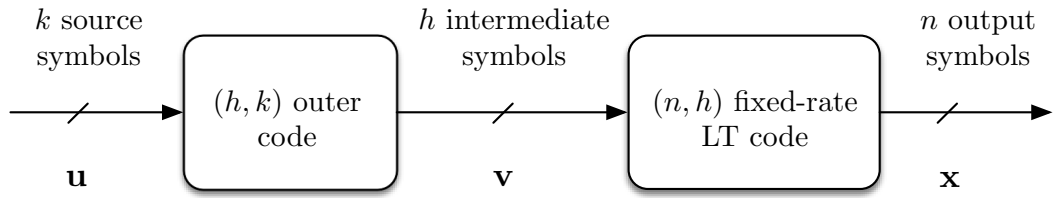}}
        \caption[Fixed-rate Raptor code as serially concatenated code]{A fixed-rate Raptor code consists of a serial concatenation of a linear block code (pre-code) with a fixed-rate \ac{LT} code.}
        \label{fig:raptor_serial}
\end{figure}

In general analyzing the distance properties of one particular code is very difficult. In \cite{berlekamp78:intractability} it was shown how the problem of finding the weights of a linear code is NP complete, that is, no fast solution to this problem is known. Instead of focusing on one particular Raptor code we will focus on an \emph{ensemble} of Raptor codes and derive average results for this ensemble. We focus on Raptor code ensembles where the outer code belongs to the linear random ensemble.
The choice of this ensemble is not arbitrary. The outer code used by the R10 Raptor code, the most widespread version of binary Raptor codes (see Section~\ref{chap:raptor_r10}),  is a concatenation of two systematic codes, the first being a high-rate regular \ac{LDPC} code and the second a pseudo-random code characterized by a dense parity check matrix. The outer codes of  R10 Raptor codes were designed to behave as codes drawn from the linear random ensemble in terms of rank properties, but allowing a fast algorithm for matrix-vector multiplication \cite{shokrollahi2011raptor}. Thus,  the ensemble we analyze may be seen as a simple model for practical Raptor codes with outer codes specifically designed to mimic the behavior of linear random codes. This model has the advantage to make the analytical investigation tractable. Moreover, \fran{in spite of its simplicity, this model provides us an insight into the behaviour of R10 Raptor codes in a fixed-rate setting, as illustrated by simulation results in this chapter.}

The ensemble of Raptor codes we will analyze is obtained by a serial concatenation of an outer code in the $\left(\ri n,\ro\ri n\right)$ binary linear random block code ensemble $\oensemble$\footnote{This ensemble was first analyzed by Gallager in his PhD Thesis \cite{Gallager63} and is sometimes known as the Gallager random code ensemble.}, with all possible realizations of an $\left(n,\ri n\right)$ fixed-rate \ac{LT} code with output degree distribution $\Omega= \{ \Omega_1, \Omega_2,\Omega_3, \ldots, \Omega_{\dmax}\}$.  We denote this ensemble as $\ensemble(\oensemble,\Omega, \ri, \ro, n)$.

In our analysis we often talk about expected properties of a code selected randomly in the ensemble $\ensemble(\oensemble,\Omega, \ri, \ro, n)$.
This random selection is performed first by randomly drawing the parity-check matrix of the linear random precode. This is achieved drawing $(h-k)h$ independent and identically distributed Bernoulli uniform random variables, each of which is associated to one element of the parity check matrix. Second, the \ac{LT} code is generated according to the usual \ac{LT} encoding process. Each output symbol  is generated independently from all other symbols by drawing a degree $d$ according to $\Omega$ and then choosing uniformly at random $d$ distinct symbols out of the $h$ intermediate ones.

For illustration in Figure.~\ref{fig:raptor_random_precode} we provide the constraint matrix for a Raptor code with a linear random precode, with $k=20$, $h=38$ and $m=30$ with the \ac{LT} degree distribution of R10 Raptor codes.  In the upper part, highlighted in blue, the parity check matrix of the precode code can be distinguished. It can be observed how this sub-matrix is dense. The lower part of the constraint matrix (highlighted in red) corresponds to the \ac{LT} symbols and is sparse. \fran{If we compare this constraint matrix with the constraint matrix of R10 Raptor codes in Figure~\ref{fig:raptor_r10_matrix}, we can see how the parity check matrix of the outer code is now considerably denser. Hence, if we were to use a Raptor code with a linear random precode in practice, in general decoding would be more complex. For example, if we would use inactivation decoding we would need in general more inactivations for decoding.}

\begin{figure}
\begin{center}
\includegraphics[width=0.5\columnwidth]{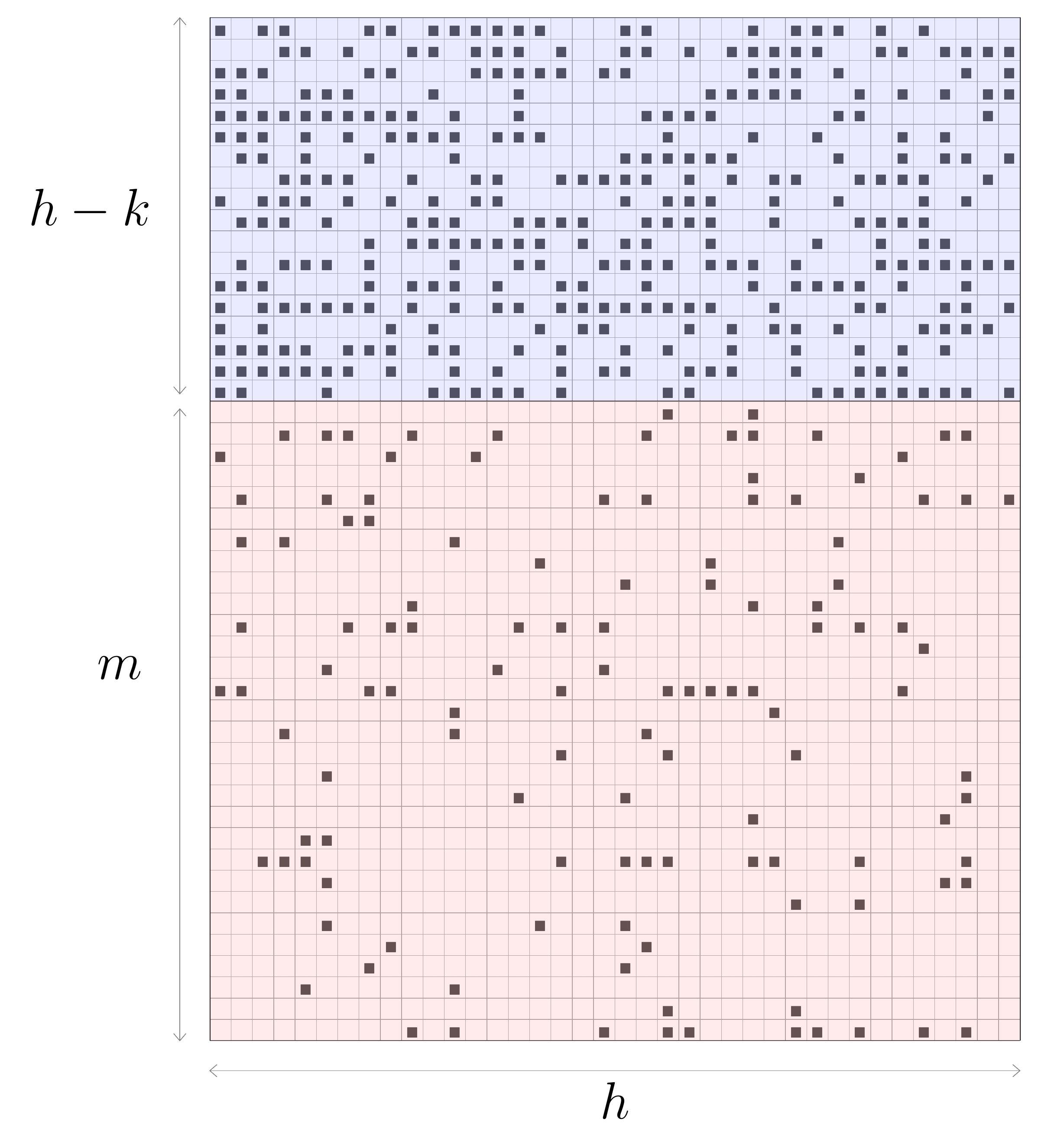}
\centering \caption[Constraint matrix of a Raptor code with a linear random precode, with $k=20$, $h=38$ and $m=30$]{Constraint matrix of a Raptor code with a linear random precode, with $k=20$, $h=38$ and $m=30$. The blue sub-matrix represents the parity check matrix of the precode. The red sub-matrix represents the transposed generator matrix of the \ac{LT} code.}
\label{fig:raptor_random_precode}
\end{center}
\end{figure}
\FloatBarrier

A related ensemble was analyzed in \cite{barg01:concat}, where lower bounds on the distance and error exponent are derived for a concatenated scheme with random outer code and a fixed inner code.

\section{Distance Spectrum}\label{sec:dist_spectrum_raptor}

In this section the expected \acf{WE} of a fixed-rate Raptor code picked randomly in the ensemble $\ensemble(\oensemble,\Omega, \ri, \ro, n)$ is characterized. We first obtain the expression for the expected \acl{WE}. Then, we analyze the asymptotic exponent of the \acl{WE}.

\begin{theorem}\label{theorem:we}
Let $A_\d$ be the expected multiplicity of codewords of weight $\d$ for a code picked randomly in the ensemble $\ensemble(\oensemble,\Omega, \ri, \ro, n)$. For $\d\geq1$ we have
\begin{align}\label{eq:WEF_Raptor}
A_\d = \binom {n}{\d} 2^{-h (1-\ro)} \sum_{\l=1}^h \binom{h}{\l}   \pl^\d (1-\pl)^{n-\d}
\end{align}
where
\begin{align}\label{eq:pl_finite}
\pl &= \sum_{j=1}^{\dmax} \Omega_j \sum_{\substack{i=\max(1,\l+j-h)\\ i~\mathrm{odd}}}^{ \min (\l,j)} \frac{ \binom{j}{i} \binom{h-j}{\l-i} } { \binom{h}{\l}} 
         = \sum_{j=1}^{\dmax} \Omega_j \sum_{\substack{i=\max(1,\l+j-h)\\ i~\mathrm{odd}}}^{ \min (\l,j)} \frac{ \binom{\l}{i} \binom{h-\l}{j-i} } { \binom{h}{j}} \, .
\end{align}
\end{theorem}
\begin{proof}
For serially concatenated codes we have
\begin{align}
A_\d = \sum_{\l=1}^{h} \frac{\weo_{\l} \wei_{\l,\d}}{ \binom {h} {\l}},
\label{eq:we_serial}
\end{align}
being $\weo_{\l}$  the average \acl{WE} of the outer code, and $\wei_{\l,\d}$ the average \acl{IO-WEF} of the inner (fixed-rate) \ac{LT} code.
For an $(h,k)$ linear random code, the average \acl{WE} is known to be \cite{Gallager63}
\begin{align}
\weo_{\l} = \binom{h}{\l} 2^{-h (1-\ro)}.
\label{eq:wef_random}
\end{align}
Let us now focus on the average \acl{IO-WEF} of the fixed-rate \ac{LT} code.  We denote by $\l$ the Hamming weight of the input word to the \ac{LT} encoder and by $\pjl$ the probability that a randomly chosen output symbol generated by the \ac{LT} encoder takes the value $1$ given that the intermediate word has Hamming weight $\l$ and the degree of the \ac{LT} code output symbol is $j$, i.e.,
\[
\pjl:=\Pr\{X_i=1|\hw(\vecV)=\l,\deg(X_i)=j\}
\]
for any $i\in \{1,\dots,n\}$. We may express this probability as
\begin{align}
\pjl =
\sum_{\substack{i=\max(1,\l+j-h)\\ i~\textrm{odd}}} ^{ \min (\l,j)} \frac{ \binom{j}{i} \binom{h-j}{\l-i} } { \binom{h}{\l} } \,= \sum_{\substack{i=\max(1,\l+j-h)\\ i~\textrm{odd}}} ^{ \min (\l,j)} \frac{ \binom{\l}{i} \binom{h-\l}{j-i} } { \binom{h}{j}}
\label{eq:p_j_l}
\end{align}
By removing the conditioning on $j$, $\pl$ is obtained, that is the probability of any of the $n$ output bits of the fixed-rate \ac{LT} encoder taking value $1$ given a Hamming weight $\l$ for the intermediate word, i.e.,
\[
\pl:=\Pr\{X_i=1|\hw(\vecV)=l\}
\]
for any $i\in \{1,\dots,n\}$.
We have
\begin{align}
\pl = \sum_{j=1}^{\dmax} \Omega_j \pjl.
\label{eq:p_l}
\end{align}
Given the fact that every output bit is generated independently, the Hamming weight of the \ac{LT} codeword conditioned to an intermediate word of weight $\l$ is a binomially distributed random variable with parameters $n$ and $\pl$. Hence,
\begin{align}
\Pr\{\hw(\vecX) = \d | \hw(\vecV) = \l\}
=\binom {n}{\d} \pl^\d (1-\pl)^{n-\d}.\label{eq:distr_weight_LT}
\end{align}
We are now in the position of calculating the average \acl{IO-WEF} of a \ac{LT} code by multiplying \eqref{eq:distr_weight_LT} by the number of weight-$\l$ intermediate words, yielding
\begin{align}
\wei_{\l,\d}= \binom {h}{\l} \binom {n}{\d} \pl^\d (1-\pl)^{n-\d}.
\label{eq:iowef_lt}
\end{align}
Finally, by making use of \eqref{eq:we_serial}, \eqref{eq:wef_random} and \eqref{eq:iowef_lt}, we obtain \eqref{eq:WEF_Raptor}.
\end{proof}
\begin{corollary}\label{corollary:A_0}
As opposed to $A_d$ with $d \geq 1$, whose expression is given in Theorem~\ref{theorem:we}, the expected number of codewords of weight $0$, $A_0$, is given by
\begin{align}
A_0 &= 1 + \sum_{l=1}^h \frac{\weo_{\l} \wei_{\l,0}}{ \binom {h} {\l}} \\
       &= 1 + 2^{-n \ri (1 - \ro) } \sum_{l=1}^h \binom {h} {\l} (1 - \pl)^n \, .\label{eq:A_0}
\end{align}
\end{corollary}
\begin{remark}
An expected number of weight-$0$ codewords larger than one implies that different input messages can be mapped into the same codeword. The fact that  $A_0>1$ is because we are drawing both the outer and inner code at random. Hence, we do not ensure by construction that different input messages are mapped into different codewords. In Section~\ref{sec:rate_reg}, Theorem~\ref{theorem:zero_codeword}, it will be shown that if the $(\ri,\ro)$ pair belongs to a region that is refereed to as ``positive normalized typical minimum distance region'', the expected number $A_0$ of zero weight codewords approaches $1$ (exponentially) as $n$ increases.


\end{remark}

So far we have considered finite length Raptor codes. Often, when dealing with ensembles of codes, their distance properties can be captured in a very compact form by letting the block length $n$ tend to infinity, while keeping the code rate constant. Such analysis of \ac{LDPC} codes was performed by Gallager in his Ph.D. Thesis \cite{Gallager63}.
Hereafter, we denote the normalized output weight of the fixed-rate Raptor code by $\nd = \d/n$ and the normalized output weight of the outer code (input weight to the \ac{LT} encoder) by $\nl = \l/h$. The asymptotic exponent of the weight distribution of an ensemble is defined as
\begin{align}
G(\nd) = \lim_{n \to \infty} \frac{1}{n} \log_2 \we_{\nd n} \, . \label{eq:growth_rate_def}
\end{align}
Commonly, $G(\nd)$ is also referred to as growth rate. \fran{The growth rate of a code or code ensemble is a compact representation of the properties of the code when its block length is asymptotically large. In particular, if for a given normalized output weight $\nd$, we have $G(\nd)>0$, we expect to have asymptotically many codewords with normalized weight $\nd$. On the other hand, if for a given $\nd$ we have $G(\nd)<0$, we expect to have asymptotically few codewords with weight $\nd$.}

Next we compute the growth rate of the weight distribution for the ensemble $\msr{C}_{\infty}(\oensemble,\Omega, \ri, \ro)$, that is the ensemble $\msr{C}(\oensemble,\Omega, \ri, \ro, n)$ in the limit where $n$ tends to infinity for constant $\ri$ and $\ro$.

\begin{theorem}\label{theorem:growth_rate}
The asymptotic exponent of the weight distribution of the fixed-rate Raptor code ensemble $\msr{C}_{\infty}(\oensemble,\Omega, \ri, \ro)$ is given by
\begin{align}\label{eq:growth_rate}
G(\nd) = \Hb(\nd) - \ri  (1-\ro) +  \fmax(\nd)
\end{align}
where $\Hb$ is the binary entropy function and
\begin{align}\label{eq:max}
\fmax(\nd) := \max_{ \nl \in \mathscr D_{\nl}} \f(\nd, \nl),
\end{align}
being $\f(\nd, \nl)$ and $\mathscr D_{\nl}$ defined as follows,
\begin{align}\label{eq:f}
\f(\nd, \nl) := \ri \Hb(\nl) + \nd \log_2 \npnl + (1- \nd) \log_2 \left(1 - \npnl\right) \, ,
\end{align}
\begin{align}
\mathscr D_{\nl} = \left\{ \begin{array}{cl} (0,1) & \textrm{if  } \Omega_j = 0 \textrm{  for all even } j\\
(0,1] & \textrm{otherwise} \, , \end{array} \right.
\end{align}
and with $\npnl$ defined as
\begin{align}
\npnl := \frac{1}{2} \sum_{j=1}^{\dmax} \Omega_j  \left[  1-\left( 1-2\nl\right)^j \right].
\label{eq_npnl}
\end{align}
\end{theorem}
\begin{proof}
Let us define $\mathbb N^*_h = \{1,2,\dots,h\}$. From \eqref{eq:WEF_Raptor} we have
\begin{align} \label{eq:proof_G}
                                              & \frac{1}{n} \log_2 A_{\nd n} \notag \\
                                              &=\frac{1}{n} \log_2 {n \choose \nd n} -  \ri (1- \ro) + \frac{1}{n} \log_2  \sum_{l=1}^h {h \choose l} \pl^d (1-\pl)^{n-d}  \notag \\
                                              &\stackrel{\mathrm{(a)}}{\leq} \Hb(\nd) -\frac{1}{2n} \log_2 \left(2 \pi n \nd (1-\nd)\right) -  \ri (1- \ro) + \frac{1}{n} \log_2 \sum_{l=1}^h {h \choose l} \pl^d (1-\pl)^{n-d} \notag \\
                                              &\stackrel{\mathrm{(b)}}{\leq} \Hb(\nd) -\frac{1}{2n} \log_2 \left(2 \pi n \nd (1-\nd)\right) -  \ri (1- \ro) + \frac{1}{n}\log_2 ( \ri n) \notag \\
                                              &+ \frac{1}{n} \log_2 \max_{l \in \mathbb N^*_{h-1}} \left\{ {h \choose l} \pl^d (1-\pl)^{n-d} \right\} \notag \\
                                              &\stackrel{\mathrm{(c)}}{\leq} \Hb(\nd) -\frac{1}{2n} \log_2 (2 \pi n \nd (1-\nd)) -  \ri (1- \ro) + \frac{1}{n}\log_2 ( \ri n) \notag \\
                                             &+ \max_{l \in \mathbb N^*_{h-1}} \left\{  \ri \Hb\left(\frac{l}{h}\right) - \frac{1}{2n}\log_2\left(2 \pi  \ri n\frac{l}{h}\left(1-\frac{l}{h}\right)\right) \right. \notag \\
                                             &  + \nd \log_2 \pl + (1-\nd) \log_2(1-\pl) \bigg\} \notag \\
                                             &= \Hb(\nd) -\frac{1}{2n} \log_2 (2 \pi n \nd (1-\nd)) -  \ri (1- \ro) + \frac{1}{n}\log_2 ( \ri n) \notag \\
                                             &+ \max_{\nl \in \left\{\frac{1}{\ri n},\dots,\frac{\ri n -1}{\ri n}\right\}} \left\{  \ri \Hb\left(\nl\right) - \frac{1}{2n}\log_2\left(2 \pi  \ri n \nl \left(1- \nl\right)\right) + \nd \log_2 p_{\ri n \lambda} \right. \notag \\
                                             & + (1-\nd) \log_2(1-p_{\ri n \lambda}) \bigg\}
\end{align}
Inequality $\mathrm{(a)}$ follows from the well-known tight bound \cite{Gallager63}
\begin{align}
{n \choose \sigma n} \leq \frac{2^{n \Hb(\sigma)}}{\sqrt{2 \pi n \sigma (1-\sigma)}}, \qquad 0<\sigma<1
\label{eq:gallagher_upper}
\end{align}
where $\Hb$ is the binary entropy function, while $\mathrm{(b)}$ follows from
\begin{align}
\sum_{l=1}^h {h \choose l} \pl^d (1-\pl)^{n-d} \leq h \max_{l \in \mathbb N^*_h} {h \choose l} \pl^d (1-\pl)^{n-d}
\label{eq:proof_G_summ}
\end{align}
and from the fact that the maximum cannot be taken for $l=h$ for large enough $n$ (as shown next). Inequality $\mathrm{(c)}$ is due again to \eqref{eq:gallagher_upper}, to $\log_2(\cdot)$ being a monotonically increasing function, and to $1/n$ being a scaling factor not altering the result of the maximization with respect to $l$.

We may prove the fact that the maximum is not taken for $l=h$, for large enough $h$, as follows. By calculating directly \eqref{eq:p_l} for $l=h$ and $l=h-1$ it is easy to show that we have
\begin{align*}
\p_h = \sum_{\substack{j=1\\ j~\textrm{odd}}}^{d_{\max}} \Omega_j \, ,
\end{align*}
and
\begin{align*}
\p_{h-1} = \sum_{\substack{j=1\\ j~\textrm{odd}}}^{d_{\max}} \frac{h-j}{h} \Omega_j + \sum_{\substack{j=1\\ j~\textrm{even}}}^{d_{\max}} \frac{j}{h} \Omega_j \, .
\end{align*}
%
For increasing $h$ we have $\p_{h-1}/\p_h \rightarrow 1$. Hence, there exists $h_0(\Omega)$ such that
\begin{align}
h\, \p_{h-1}^d (1-\p_{h-1})^{n-d} > \p_h^d (1-\p_h)^{n-d}
\end{align}
for all $h>h_0(\Omega)$. Hence, for all such values of $h$ the maximum cannot be taken at $l=h$.

Next, by defining
\begin{align}
\hat{\lambda}_n = \mathop{\mathrm{argmax}}_{\lambda \in \left\{\frac{1}{\ri n}, \frac{2}{\ri n}, \dots, \frac{\ri n-1}{\ri n} \right\} } \Big\{ &  \ri H_b(\lambda)  - \frac{1}{2n}\log_2(2 \pi  \ri n\lambda(1-\lambda)) \\
&  + \nd \log_2 \p_{\ri n \lambda} + (1-\nd) \log_2(1-\p_{\ri n \lambda}) \Big\}
\label{eq:lambda_hat}
\end{align}
the right-hand side of \eqref{eq:proof_G} may be recast as
\begin{align}
                                             &\Hb(\nd) -\frac{1}{2n} \log_2 \left(2 \pi n \nd (1-\nd)\right) -  \ri (1- \ro) + \frac{1}{n}\log_2 ( \ri n) \notag \\
                                             & +  \ri \Hb(\hat{\lambda}_n) - \frac{1}{{2n}}\log_2(2 \pi  \ri n\hat{\lambda}_n(1-\hat{\lambda}_n)) + \nd \log_2 p_{\ri n \hat{\lambda}_n} + (1-\nd) \log_2(1-p_{\ri n \hat{\lambda}_n}) \, .
\end{align}
The two terms $\frac{1}{2n} \log_2 (2 \pi n \nd (1-\nd))$ and $\frac{1}{n} \log_2 (\ri n)$ in the last expression converge to zero as $n \rightarrow \infty$. Moreover, also the term $\frac{1}{{2n}}\log_2 (2 \pi  \ri n\hat{\nl}_n(1-\hat{\nl}_n))$ converges to zero regardless of the behavior of the sequence $\hat{\lambda}_n$. In fact, it is easy to check that the term $\frac{1}{{2n}}\log_2(2 \pi  \ri n\hat{\lambda}_n(1-\hat{\lambda}_n))$ converges to zero in the limiting cases $\hat{\lambda}_n=\frac{1}{\ri n}$ $\forall n$ and $\hat{\lambda}_n = \frac{\ri n-1}{\ri n}$ $\forall n$, so it does in all other cases. 

If we now develop the right hand side of \eqref{eq:proof_G} further, for large enough $n$, we have
\begin{align}\label{eq:proof_G_2}
                                             & \Hb(\nd) -\frac{1}{2n} \log_2 (2 \pi n \nd (1-\nd)) -  \ri (1- \ro) + \frac{1}{n}\log_2 ( \ri n) \notag \\
                                             &+ \max_{\nl \in \left\{\frac{1}{\ri n},\dots,\frac{\ri n -1}{\ri n}\right\}} \left\{  \ri \Hb\left(\nl\right) - \frac{1}{2n}\log_2\left(2 \pi  \ri n \nl \left(1- \nl\right)\right) + \nd \log_2 p_{\ri n \lambda} \right.\\
                                             &+ (1-\nd) \log_2(1-p_{\ri n \lambda}) \bigg\} \notag \\
& \stackrel{\mathrm{(d)}}{\leq} \Hb(\nd) -\frac{1}{2n} \log_2 (2 \pi n \nd (1-\nd)) -  \ri (1- \ro) + \frac{1}{n}\log_2 ( \ri n) \notag \\
& + \sup_{\nl \in \mathbb Q \cap (0,1)} \bigg\{  \ri \Hb\left(\nl\right) - \frac{1}{2n}\log_2\left(2 \pi  \ri n \nl \left(1- \nl\right)\right) + \nd \log_2 \left(\npnl + \frac{K}{n} \right) \notag \\
& \qquad \qquad \quad \,\, + (1-\nd) \log_2 \left(1-\npnl + \frac{K}{n} \right) \bigg\} \notag \\
& \stackrel{\mathrm{(e)}}{=} \Hb(\nd) -\frac{1}{2n} \log_2 (2 \pi n \nd (1-\nd)) -  \ri (1- \ro) + \frac{1}{n}\log_2 ( \ri n) \notag \\
& + \sup_{\nl \in (0,1)} \bigg\{  \ri \Hb\left(\nl\right) - \frac{1}{2n}\log_2\left(2 \pi  \ri n \nl \left(1- \nl\right)\right) + \nd \log_2 \left(\npnl + \frac{K}{n} \right) \notag \\
& \qquad \qquad \quad \,\, + (1-\nd) \log_2 \left(1-\npnl + \frac{K}{n} \right) \bigg\} \\
&:= \Gamma_n (\nd).
\end{align}
being $\mathbb Q$ the set of rational numbers. Inequality $\mathrm{(d)}$ follows from the fact that, as it can be shown, $|\npnl - p_{\ri n \lambda}|<K/n$ (uniformly in $\lambda$) for large enough $n$ and from the fact that the supremum over $\mathbb Q \cap (0,1)$ upper bounds the maximum over the finite set $\left\{ \frac{1}{\ri n}, \dots, \frac{\ri n - 1}{\ri n} \right\}$. Equality $\mathrm{(e)}$ is due to the density of $\mathbb Q$. In equality $\mathrm{(e)}$, the function of $\nl$ being maximized is regarded as a function over the real interval $(0,1)$ (i.e., $\lambda$ is regarded as a real parameter).

The upper bound \eqref{eq:proof_G_2} on $\frac{1}{n} \log_2 A_{\nd n}$ is valid for any finite but large enough $n$. If we now let $n$ tend to infinity, all inequalities $\mathrm{(a)}$--$\mathrm{(d)}$ are satisfied with equality. In particular: for $\mathrm{(a)}$ this follows from the well-known exponential equivalence ${n \choose \nd n} \doteq 2^{n \Hb(\nd)}$; for $\mathrm{(b)}$ from the exponential equivalence $\sum_l 2^{n f(l)} \doteq \max_l 2^{n f(l)}$; for $\mathrm{(c)}$ from ${\ri n \choose \hat{\nl}_n \ri n} \doteq 2^{n \Hb(\hat{\nl}_n)}$ (due to $\frac{1}{{2n}}\log_2 (2 \pi  \ri n\hat{\lambda}_n(1-\hat{\lambda}_n))$ vanishing for large $n$); for $\mathrm{(d)}$ from the fact that, asymptotically in $n$, applying the definition of limit we can show that the maximum over the set $\left\{\frac{1}{\ri n}, \dots, \frac{\ri n -1}{\ri n} \right\}$ upper bounds the supremum over $\mathbb Q \cap (0,1)$ (while at the same time being upper bounded by it for any $n$). The expression of $\npnl$ is obtained by assuming $n$ tending to $\infty$ using the expression of $\pl$. Alternatively, the same expression is obtained by assuming $n$ tending to $\infty$ and letting an output symbol of degree $i$ choose its $i$ neighbors \emph{with} replacement.

By letting $n$ tend to infinity and by cancelling all vanishing terms, we finally obtain the statement. Note that we can replace the supremum by a maximum over $\mathscr D_{\nl}$ as this maximum is always well-defined.\footnote{In fact, for any $\nd \in [0,1]$ the function $\f(\nd, \nl)$ diverges to $-\infty$ as $\nl \rightarrow 0^+$. Moreover, it diverges to $-\infty$ as $\nl \rightarrow 1^-$ if $\Omega_j=0$ for all even $j$ and converges as $\nl \rightarrow 1^-$ otherwise. Finally, for all $\nd \in [0,1]$ it is continuous for all $\nl \in \mathscr D_{\nl}$.}
\end{proof}

In the next two lemmas, which will be useful in the sequel,  the derivative of the growth rate function is characterized. For the sake of clarity, we introduce the notation $\np(\nl)$ instead of $\np_{\nl}$. Hence, we stress the fact that $\np(\nl)$ is a function of $\nl$.

\begin{lemma} \label{lemma:growth_rate_derivative}
The derivative of the growth rate of the weight distribution of a fixed-rate Raptor code ensemble $\msr{C}_{\infty}(\oensemble,\Omega, \ri, \ro)$ is given by
\begin{equation}
G'(\nd) = \log_2 \frac{1-\nd}{\nd} + \log_2 \frac{\np(\nlo)}{1-\np(\nlo)}  \,  \nonumber
\end{equation}
where
\begin{align}\label{eq:lo_def}
\nlo(\nd) := \argmax_{\nl \in \mathscr D_{\nl}} \left\{ \f(\nd, \nl) \right\} \, .
\end{align}
\end{lemma}
\begin{proof}
Let us rewrite the expression of $G(\nd)$ in \eqref{eq:growth_rate} as
\[
{G(\nd)=\Hb(\nd) - \ri(1-\ro) + \f(\nd,\nlo(\nd))} \, .
\]
We must have
\begin{equation}\label{eq:critical_point}
\frac{\partial \f}{\partial \nl} (\nd, \nlo) = 0 \, .
\end{equation}
Where $\nlo$ is actually a function of $\nd$, $\nlo(\nd)$.
Taking the derivative with respect to $\nd$, and after elementary algebraic manipulation we obtain
\begin{align*}
G'(\nd) &= \log_2\frac{1-\nd}{\nd} + \log_2 \frac{\np(\nlo)}{1-\np(\nlo)} + \frac{\partial \f}{\partial \nl} (\nd, \nlo) \, \frac{\mathrm d \nlo}{\mathrm d \nd}
\end{align*}
which, applying \eqref{eq:critical_point}, yields the statement.
\end{proof}
\begin{lemma}\label{corollary:der}
For all $0< \nd < 1/2$, the derivative of the growth rate of the weight distribution of a fixed-rate Raptor code ensemble $\msr{C}_{\infty}(\oensemble,\Omega, \ri, \ro)$ fulfills
\[
G'(\nd)>0.
\]
\end{lemma}
\begin{proof}
If in the expression for $G'(\nd)$ in Lemma~\ref{lemma:growth_rate_derivative} we impose $G'(\nd)=0$, we obtain
$$
\frac{1-\nd}{\nd} = \frac{1-\varrho(\nlo)}{\varrho(\nlo)}.
$$
Since the function $(1-x)/x$ is monotonically decreasing for $x \in (0,1)$, this implies $\nd = \varrho(\nlo)$. Next, observing \eqref{eq:lo_def}, by the definition of $\nlo$,  its partial derivative $\partial \mathsf f(\nd,\lambda) / \partial \lambda$ must be zero when calculated for $\lambda=\nlo$. The expression of this partial derivative is
$$
\frac{\partial \mathsf f}{\partial \lambda}(\nd,\lambda) = \mathsf r_{\mathsf i} \log_2 \frac{1 - \lambda}{\lambda} + \frac{\varrho'(\lambda)}{\log 2} \cdot \frac{\nd - \varrho(\lambda)}{\varrho(\lambda)(1-\varrho(\lambda))} \,
$$
so we obtain
$$
\mathsf r_{\mathsf i} \log_2 \frac{1 - \nlo}{\nlo} + \frac{\varrho'(\nlo)}{\log 2} \cdot \frac{\nd - \varrho(\nlo)}{\varrho(\nlo)(1-\varrho(\nlo))} = 0 \, .
$$
As shown above, for any $\nd$ such that $G'(\nd)=0$ we have $\nd=\varrho(\nlo)$. Substituting in the latter equation we obtain $\nlo=1/2$ which implies $\nd=\varrho(1/2)=1/2$. Therefore, the only value of $\nd$ such that $G'(\nd)=0$ is $\nd=1/2$. Due to continuity of $G'(\nd)$ and to the fact that $G'(\nd) \rightarrow +\infty$ as $\nd \rightarrow 0^+$ (as shown in  Appendix~\ref{sec:necessity}). Therefore, we conclude that $G'(\nd)>0$ for all $0 < \nd < 1/2$.
\end{proof}

\fran{A useful concept when characterizing the distance properties of an ensemble is the (normalized) typical minimum distance, which we define formally as follows.}

\begin{mydef}The normalized typical minimum distance of a fixed-rate Raptor code ensemble ${\msr{C}_{\infty}(\oensemble,\Omega, \ri, \ro)}$  is the real number
%
\begin{align*}
\dmint := \begin{cases}
0 & \text{if } \lim_{\nd \to 0^+} G(\nd) \geq 0 \\
\inf \{ \nd>0 : G(\nd) > 0 \} & \text{otherwise.}
\end{cases}
\end{align*}
\end{mydef}

\begin{example}
Figure~\ref{fig:growth} shows the growth rate $G(\nd)$ for three different ensembles $\msr{C}_{\infty}(\oensemble,\Omegarten, \ri, \ro)$, where $\Omegarten$ is the output degree distribution used in the standards \cite{MBMS12:raptor}, \cite{luby2007rfc} (see details in Table~\ref{table:dist}) and $\ro=0.99$ for three different  $\ri$ values. The growth rate of a linear random code ensemble with rate $\rate=0.99$ is also shown. It can be observed how the curve for $\ri = 0.95$ does not  cross the $x$-axis, the curve for $\ri = 0.88$ has $\dmint=0$ and the curve for $\ri=0.8$ has $\dmint=0.0005$.
\newline
\fran{This example highlights an important fact, if we fix the outer code rate to a very high value, concretely $\ro=0.99$, our ensemble still can achieve a (normalized) typical minimum distance larger than 0 when the rate of the inner code is low enough, (in this case $\ri < 0.88$).
}
\begin{figure}[!t]
\begin{center}
\includegraphics[width=\figwBigger]{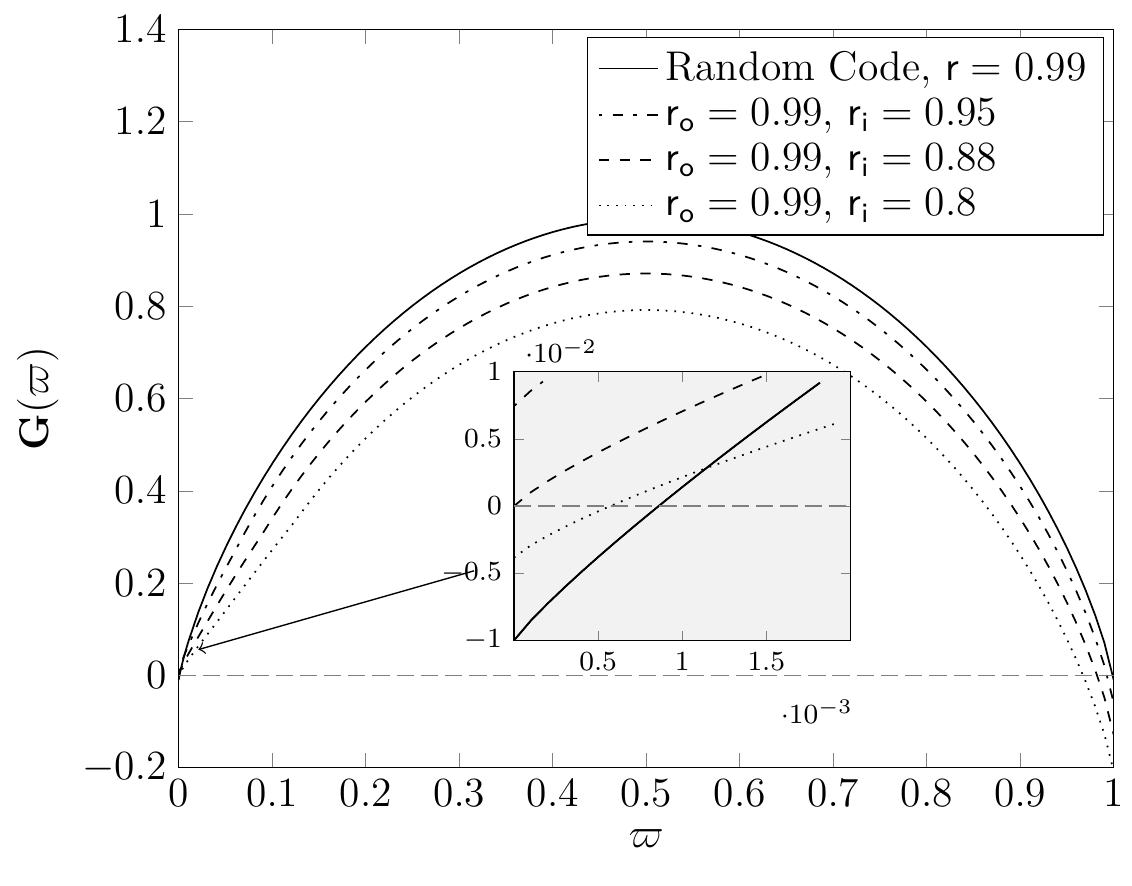}
\centering \caption[Growth rate vs.\ normalized output weight $\nd$ for a linear random code and $3$ different ensembles of Raptor codes]{Growth rate vs.\ normalized output weight $\nd$. The solid line shows the growth rate of a linear random code with rate $\rate=0.99$. The dot-dashed, dashed, and dotted lines show the growth rates $G(\nd)$ of the ensemble $\msr{C}_{\infty}(\oensemble,\Omegarten, \ri, \ro=0.99)$ for $\ri=0.95$, $0.88$ and $0.8$, respectively.}
\label{fig:growth}
\end{center}
\end{figure}
\end{example}

\begin{example}
Figure~\ref{fig:gilbert} shows the overall rate $\rate$ of the fixed-rate Raptor code ensemble $\msr{C}_{\infty}(\oensemble,\Omegarten, \ri~=~\rate/\ro, \ro)$ versus the normalized typical minimum distance $\dmint$. It can be observed how, for constant overall rate $\rate$, $\dmint$ increases as the outer code rate $\ro$ decreases. We can also observe how decreasing $\ro$ allows to get closer to the asymptotic Gilbert-Varshamov bound. Thus, if the overall rate $\rate$ is kept constant, the distance properties of the fixed-rate Raptor code ensemble improve
as the rate of the outer code decreases. Note that by decreasing the outer code rate the decoding complexity will generally increase.
\begin{figure}[t]
\begin{center}
\includegraphics[width=\figwBigger]{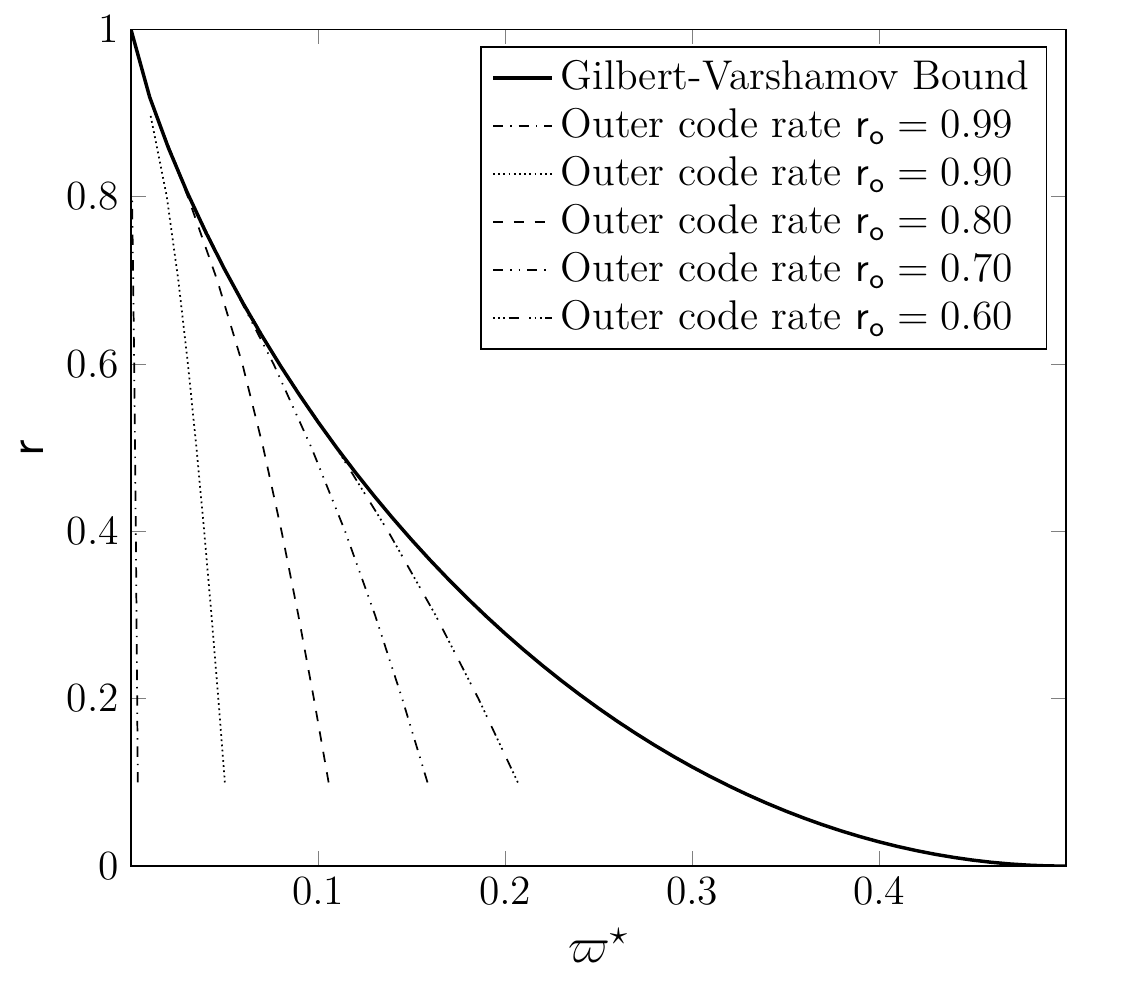}
\centering \caption[Overall rate $\rate$  vs.\ the normalized typical minimum distance $\dmint$ for different ensembles of Raptor codes]{Overall rate $\rate$  vs.\ the normalized typical minimum distance $\dmint$. The solid line represents the asymptotic Gilbert-Varshamov bound. The markers represent Raptor codes ensembles $\msr{C}_{\infty}(\oensemble,\Omegatwo, \ri=\rate/ \ro, \ro)$ with different outer code rates, $\ro$.}
\label{fig:gilbert}
\end{center}
\end{figure}
\end{example}

\section{Positive Distance Region} \label{sec:dist_region}

\label{sec:rate_reg}
The objective of this section is determining the conditions that a fixed-rate Raptor ensemble $\msr{C}_{\infty}(\oensemble,\Omega,\ri,\ro)$ needs to fulfil in order to
exhibit good normalized typical distance properties. More specifically, given a degree distribution $\Omega$ and an overall rate $\rate$ for the Raptor code, we are interested in the allocation of the rate between the outer code and the fixed-rate LT code to achieve a strictly positive normalized typical minimum distance.
\begin{mydef}[Positive normalized typical minimum distance region] The positive normalized typical minimum distance region of an ensemble $\msr{C}_{\infty}(\oensemble,\Omega,\ri,\ro)$ is defined as the set $\region$ of code rate pairs $\left( \ri, \ro \right)$ for which the ensemble possesses a positive normalized typical minimum distance. Formally:
\begin{align}
\region:=\left\{(\ri,\ro) \succeq (0,0) | \dmint(\Omega, \ri,\ro)> 0 \right\}
\nonumber
\end{align}
where we have used the notation $\dmint= \dmint(\Omega, \ri,\ro)$ to emphasize the dependence on $\Omega$, $\ri$ and $\ro$.
\end{mydef}
In the following theorem the positive normalized typical distance region for an LT output degree distribution $\Omega$ is developed .
\begin{theorem} The region  $\region$ is given by
\label{theorem_inner}
\begin{align}
\region:=\left\{\left( \ri, \ro \right) \succeq (0,0) | \ri (1-\ro) >
\max_{ \nl \in \mathscr D_{\nl}} \left\{\ri \Hb(\nl) + \log_2 \left(1 - \npnl\right)\right\}
\right\}\, .
\label{eq:theorem_region}
\end{align}
\end{theorem}
\begin{proof}
See Appendix~\ref{sec:proof_inversion}.
\end{proof}
The next two theorems characterize the distance properties of a fixed-rate Raptor code with linear random outer code picked randomly in the ensemble $\ensemble(\oensemble,\Omega, \ri, \ro, n)$ with $(\ri, \ro)$ belonging to $\region$.
\begin{theorem} \label{theorem:min_dist}
Let the random variable $\D$ be the minimum nonzero Hamming weight in the code book of a fixed-rate Raptor code picked randomly in an ensemble $\ensemble(\oensemble,\Omega, \ri, \ro, n)$. If $(\ri, \ro) \in \region$ then
\begin{align*}
\lim_{n \rightarrow \infty} \Pr \{ \D \leq \nd n \} = 0
\end{align*}
exponentially in $n$, for all $0 < \nd < \dmint$.
\end{theorem}
\begin{proof}
We can upper bound this probability via union bound as
\begin{align}
\label{eq:pr_dmin}
\Pr\{ \D \leq \nd n\} & \leq   \sum_{w=1}^{\nd n} A_w.
\end{align}
We will start by proving that the sequence $A_{\d}$ is non-decreasing for $\d < n/2$ and sufficiently large $n$.  As $n \rightarrow \infty$, the expression $\frac{1}{n} \log_2 \frac{A_{\nd n}}{A_{\nd n -1}}$ converges to $\Gamma_n (\nd) - \Gamma_n (\nd -\frac{1}{n})$, being $\Gamma_n(\nd)$ given in \eqref{eq:proof_G_2}. From Lemma~\ref{corollary:der}  we know that  $G'(\nd)> 0$ for $0 < \nd < 1/2$. Moreover, as $n \rightarrow \infty$, from Theorem~\ref{theorem:growth_rate} we have $\Gamma_n(\nd) \rightarrow G(\nd)$. Hence, for sufficiently large $n$, $\Gamma_n (\nd) \geq \Gamma_n (\nd -\frac{1}{n})$. This implies that $A_{\d}$ is non decreasing.

We can now write
\begin{align}
\Pr\{ \D \leq \nd n\}  \leq  \nd n A_{\nd n} \leq\nd n 2^{n \Gamma_n(\nd) }\,  ,
\end{align}
where we have used $A_{\nd n}\leq 2^{n \Gamma_n(\nd) }$, being $\Gamma_n(\nd)$ given in \eqref{eq:proof_G_2}.

As $n \rightarrow \infty$ we have $\Gamma_n(\nd) \rightarrow G(\nd)$. Moreover, $G(\nd)<0$ for all $0 < \nd < \dmint$, provided $(\ri,\ro) \in \region$. Hence, $\Pr\{ \D \leq \nd n\}$ tends to $0$ exponentially on $n$.
\end{proof}
\begin{remark}
From Theorem~\ref{theorem:min_dist}, we have that when a rate point $(\ri, \ro)$ belongs to the region $\region$, there is  an exponential decay of the probability to find codewords with weight less than $\dmint n$, \fran{which is a very favorable property for code ensembles.} This exponential decay shall be attributed to the presence of the linear random outer code that is characterized by a dense parity-check matrix, and makes the growth rate function monotonically increasing for $\nd$ for which it is negative, $0 < \nd < \dmint$.

As a comparison, for \ac{LDPC} code ensembles characterized by a positive normalized typical minimum distance, the growth rate function starts from $G(0)=0$ with negative derivative, reaches a minimum, and then increases to cross the $x$-axis. In this case, for $\nd < \dmint$ the sum in the upper bound is dominated by those terms corresponding to small values of $w$, yielding either a polynomial decay (as for Gallager's codes \cite{Gallager63} ) or even  $\Pr \{ D \leq \nd n \}$ tending to a constant (as it is for irregular unstructured \ac{LDPC} ensembles \cite{orlitsky05:stopping,di06:weight}). \fran{This is in general worse for the distance properties of the code compared to the exponential decay observed in our ensemble}

\fran{However, one should remark that this exponential decay of the probability of having codewords with weight less than  $\dmint n$ comes at the cost of complexity, since the outer code is dense, and hence complex to encode and decode.}
\end{remark}
\begin{theorem}\label{theorem:zero_codeword}
Let the random variable $\Z$ be the multiplicity of codewords of weight zero in the code book of a fixed-rate Raptor code picked randomly in the ensemble $\ensemble(\oensemble,\Omega, \ri, \ro, n)$. If $(\ri, \ro) \in \region$ then
\begin{align*}
\Pr \{ \Z > 1 \} \rightarrow 0 \quad \textrm{as } n\rightarrow \infty \, .
\end{align*}
\end{theorem}
\begin{proof}
In order to prove the statement we have to show that the probability measure of any event $\{ \Z = t \}$ with $t \in \mathbb N \setminus \{0, 1\}$ vanishes as $n \rightarrow \infty$. We start by analyzing the behavior of $\mathbb E[\Z]=A_0$, whose expression is $\mathbb E[\Z] = 1 + 2^{-n \ri (1 - \ro) } \sum_{l=1}^h \binom {h} {\l} (1 - \pl)^n$. Using an argument analogous to the one adopted in the proof of Theorem~~\ref{theorem:growth_rate}, for large enough $n$ we have
\begin{align*}
\frac{1}{n} \log_2 \left( 2^{-n \ri (1 - \ro) } \sum_{l=1}^h \binom {h} {\l} (1 - \pl)^n \right) \leq \Xi_n
\end{align*}
where
\begin{align*}
\Xi_n := - \ri (1 - \ro) + \frac{1}{n}\log_2(\ri n) + \sup_{\nl \in (0,1)} \Big\{ & \ri \Hb(\nl) - \frac{1}{2n} \log_2 \left( 2 \pi \ri n \nl ( 1 - \nl) \right) \\
&+ \log_2 (1 - \np_{\nl} + K/n) \Big\} \, .
\end{align*}
Therefore we can upper bound $\mathbb E[\Z]$ as $\mathbb E [\Z] \leq 1 + 2^{n \Xi_n}$ which, if $(\ri,\ro) \in \region$, implies $\mathbb E[\Z] \rightarrow 1$ exponentially as $n \rightarrow \infty$ due to $\Xi_n \rightarrow G(0)$ and $G(0)<0$.\footnote{It is worth noting that $G(\nd)$ is right-continuous at $\nd=0$. This follows from the expression of $G(\nd)$ proved in Theorem~\ref{theorem:growth_rate} and from the fact that $\fmax(\nd)$ is right-continuous at $\nd=0$ as shown in the proof of Theorem~\ref{theorem_inner}.} Next, it is easy to show that $\mathbb E[\Z] \geq 1$ and, via linear programming, that the minimum is attained if and only if $\Pr \{\Z=1\}=1$ and $\Pr\{\Z=t\}=0$ for all $t \in \mathbb N \setminus \{0,1\}$. Since in the limit as $n \rightarrow \infty$ of $\mathbb E[\Z] \rightarrow 1$, we necessarily have a vanishing probability measure for any event $\{ \Z = t \}$ with $t \in \mathbb N \setminus \{0, 1\}$.
\end{proof}
\begin{remark}
From Theorem~\ref{theorem:min_dist} and Theorem~\ref{theorem:zero_codeword}, a fixed-rate Raptor code picked randomly in the ensemble $\ensemble(\oensemble,\Omega, \ri, \ro, n)$  is characterized first by a minimum distance that is exponentially concentrated around $\dmint n$ and second by  an encoding function whose kernel only includes the all-zero length $k$ message (hence bijective), with probability approaching $1$ as $n \rightarrow \infty$. \fran{In other words, for rate pairs within region $\region$ the probability of having more than 1 weight-0 codeword tends to 0 exponentially with $n$. This means that the dimension of the code is $k$ with high probability. Furthermore, for rate pairs within region $\region$ the minimum distance grows linearly with the block length $n$, which is a very favorable property.}
\end{remark}

In the following an outer region $\outer$ to region $\region$ is introduced. The outer region $\outer$ only depends on the average output degree of the inner \ac{LT} code.
\begin{theorem}
\label{pro:outer}
The positive normalized typical minimum distance region $\region$ of a fixed-rate Raptor code ensemble $\msr{C}_{\infty}(\oensemble,\Omega, \ri, \ro)$ fulfills $\region \subseteq \outer$, where
\begin{equation}
\outer := \left\{(\ri,\ro) \succeq (0,0) | \ri \leq \min \left( \phi(\ro), \frac{1}{\ro}\right) \right\}
 \label{eq:outer_bound_2}
\end{equation}
with

\[
\phi(\ro)=
\begin{cases}
\frac{\bar \Omega \log_2 (1/\ro)}{\Hb(1-\ro) -(1-\ro)} \qquad  \, \ro> \ro^* \\
1/\ro \qquad \qquad \qquad \ro \leq \ro^*
\end{cases}
\]
being $\ro^*$ the only root of $\Hb(1-\ro) -(1-\ro)$ in  $\ro \in (0,1)$, numerically $\ro^* \approx 0.22709$.
\end{theorem}

\begin{proof}
See Appendix~\ref{sec:proof_outer}.
\end{proof}

\begin{example}
\label{example_region}
Figures~\ref{fig:regiona} and \ref{fig:regionb} show the positive normalized typical minimum distance regions, $\region$ for $\Omegarten$ and $\Omegatwo$ (see Table~\ref{table:dist}) together with their outer bounds $\outer$. We can observe how the outer bound is tight in both cases except for inner codes rates close to $\ri=1$. In the figures several isorate curves are also shown, along which the rate of the Raptor code $\rate=\ri~\ro$ is constant. For example, in order to have a positive normalized typical minimum distance and an overall rate $\rate=0.95$, the figures show that the rate of the outer code must lay below $\ro<0.978$ for both distributions. Let us assume for a moment we would like to design a fixed-rate Raptor code, with degree distribution $\Omegarten$ or $\Omegatwo$, overall rate $\rate=0.95$ and for a given length $n$, which is assumed to be large. Different choices for $\ri$ and $\ro$ are possible. If $\ro$ is not chosen as $\ro<0.978$ the average minimum distance of the ensemble will not grow linearly on $n$. Hence, many codes in the ensemble will exhibit high error floors even under \ac{ML} erasure decoding. \fran{This example illustrates how the concepts introduced in this chapter can be used to design a fixed-rate Raptor code. More concretely, for  a constant overall rate of the Raptor code $\rate$ we can distribute the rate among the outer code and the inner \ac{LT} code such that the typical minimum distance of the Raptor code is positive (i.e., in order to have a low error floor).}

\end{example}

\begin{figure}[t]
        \centering
        \includegraphics[width=\figw]{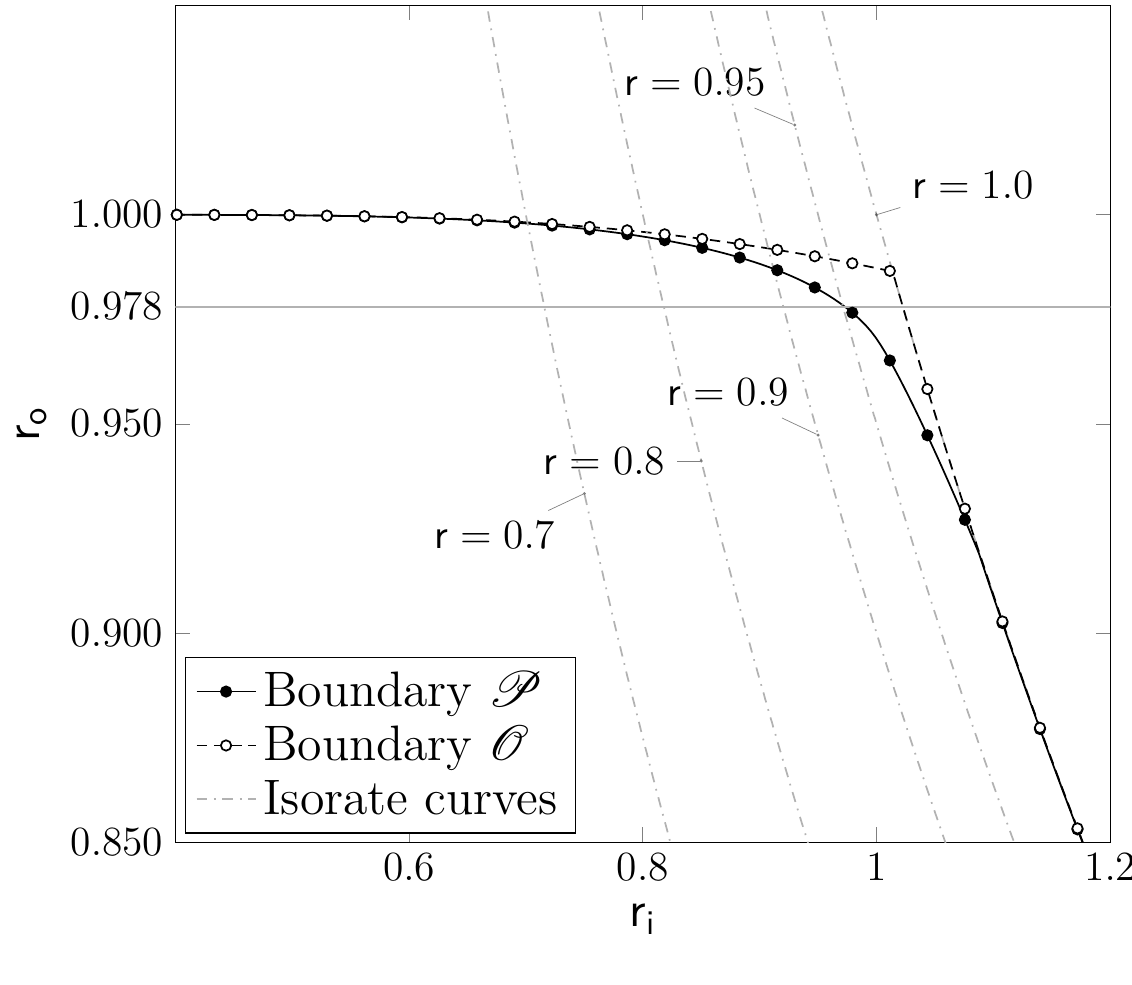}
        \caption[Positive growth rate regions of a fixed-rate Raptor code ensemble with degree distribution $\Omegarten$.]{Positive growth rate region of a fixed-rate Raptor code ensemble with degree distribution $\Omegarten$. The solid lines with black markers represent the positive growth-rate $\region$ and the dashed lines with white markers represents its outer bound $\outer$. The gray dashed lines represent isorate curves for different rates $\rate$.}
        \label{fig:regiona}
        \centering
        \includegraphics[width=\figw]{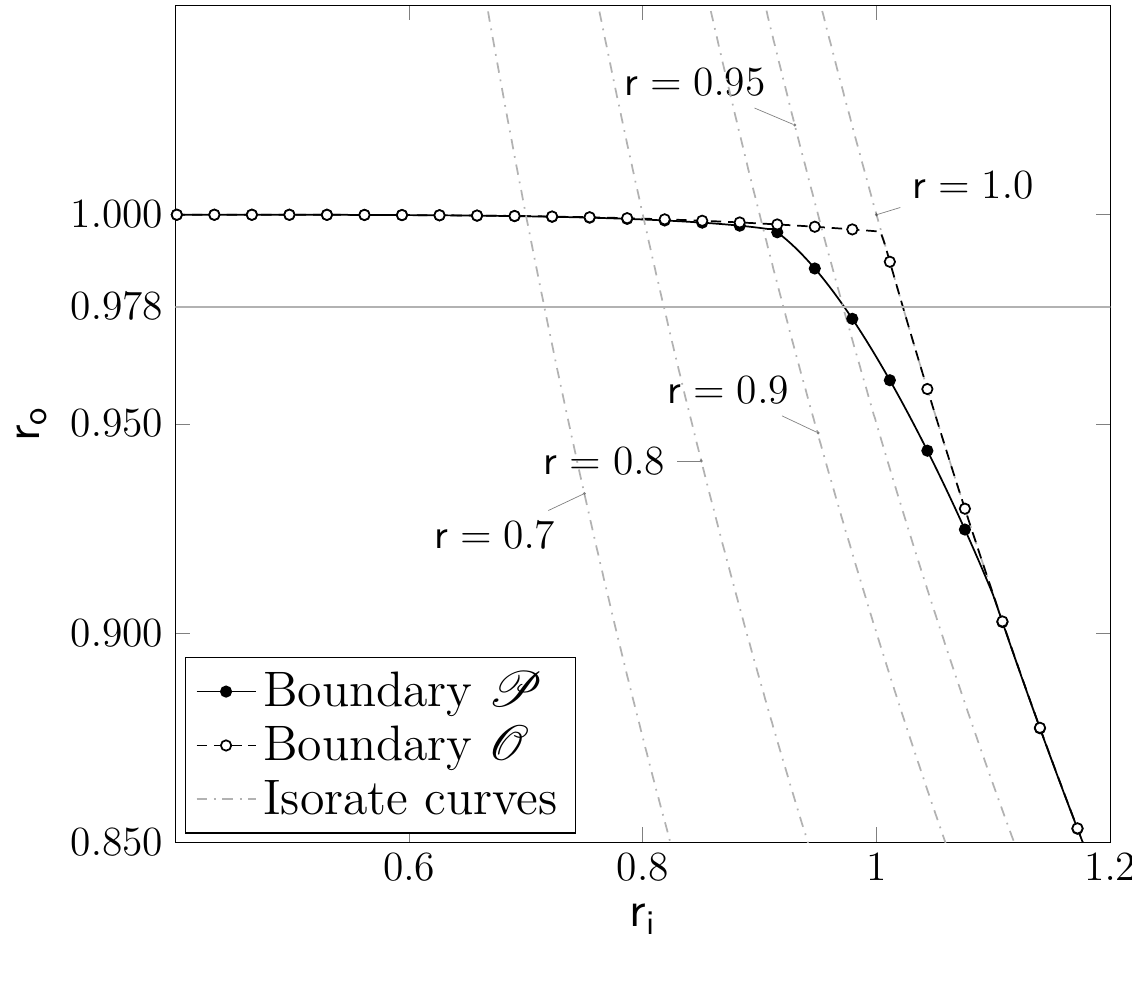}
        \caption[Positive growth rate regions of a fixed-rate Raptor code ensemble with degree distribution $\Omegatwo$.]{Positive growth rate region of a fixed-rate Raptor code ensemble with degree distribution $\Omegatwo$. The solid lines with black markers represent the positive growth-rate $\region$ and the dashed lines with white markers represents its outer bound $\outer$. The gray dashed lines represent isorate curves for different rates $\rate$.}
        \label{fig:regionb}
\end{figure}

\begin{table}[t]
\caption[Degree distributions $\Omegarten$ and $\Omegatwo$ ]{Degree distributions $\Omegarten$, defined in \cite{MBMS12:raptor,luby2007rfc} and $\Omegatwo$, defined in \cite{shokrollahi06:raptor}}
\begin{center}
\begin{tabular}{|c|c|c|c|c|c|c|c|c|}
\hline
  Degree  &         $1$    & $2$    &   $3$&    $4$ & $5$   &  $8$  & $9$ & \\ \hline
  $\Omegarten$ \rule{0pt}{3ex}    &  0.0098 & 0.459 & 0.211 & 0.1134 &       &      &      &\\ \hline
  $\Omegatwo$ \rule{0pt}{3ex}    &  0.0048 &0.4965 & 0.1669& 0.0734 &0.0822 &0.0575&0.036 & \\ \hline \hline
  Degree         & $10$    & $11$  & $18$  & $19$   & $40$  & $65$ &  $66$ & $\bar \Omega$ \\ \hline
  $\Omegarten$ \rule{0pt}{3ex}    & 0.1113  & 0.0799&       &        & 0.0156&      &       & 4.6314      \\ \hline
  $\Omegatwo$ \rule{0pt}{3ex}    &         &       &0.0012 &0.0543  & 0.0156&0.0182&0.0091 & 5.825 \\ \hline
\end{tabular}
\end{center}\label{table:dist}
\end{table}

\section{Experimental Finite Length Results}\label{sec:raptor_fr_sims}

In this section we present experimental results in order to validate the analytical results obtained in the previous sections. First, by means of examples it is illustrated how the developed results can be used to make accurate statements about the performance of fixed-rate Raptor code ensembles in the finite length regime.
Next, we provide some results that show that a tradeoff exists between performance and decoding complexity.
Finally,  some simulation results are presented that illustrate that the distance properties obtained for linear random outer codes are a fair approximation for the distance properties obtained with the standard R10 Raptor outer code (see \cite{MBMS12:raptor,luby2007rfc}).

\subsection{Results for Linear Random Outer Codes}\label{sec:results_CER}

In this section we consider Raptor code ensembles $\ensemble(\oensemble,\Omegarten, \ri, \ro, n)$ for different values of $\ri$, $\ro$, and $n$ but keeping the overall rate of the Raptor code constant to $\rate=0.9014$. In Figure~\ref{fig:region_results} the boundary of $\region$ and $\outer$ is shown for the \ac{LT} degree distribution $\Omegarten$ together with an isorate curve for $\rate=0.9014$. The markers placed along the isorate curve in the figure represent the two different $\ri$ and $\ro$ combinations that are considered in this section. The first point ($\ri=0.9155$, $\ro=0.9846$), marked with an asterisk, is inside but very close to the boundary of $\region$ for $\Omegarten$. We will refer to ensembles corresponding to this point as \emph{bad} ensembles. The second point, ($\ri=0.9718$, $\ro=0.9275$) marked with a triangle, is inside but quite far from the boundary of $\region$ for $\Omegarten$. We will refer to ensembles corresponding to this point as \emph{good} ensembles. In general, one would expect the good ensembles to have better distance properties than the bad ensembles, and hence, better erasure correcting properties.

\begin{figure}[t]
        \centering
        \includegraphics[width=\figwbigger]{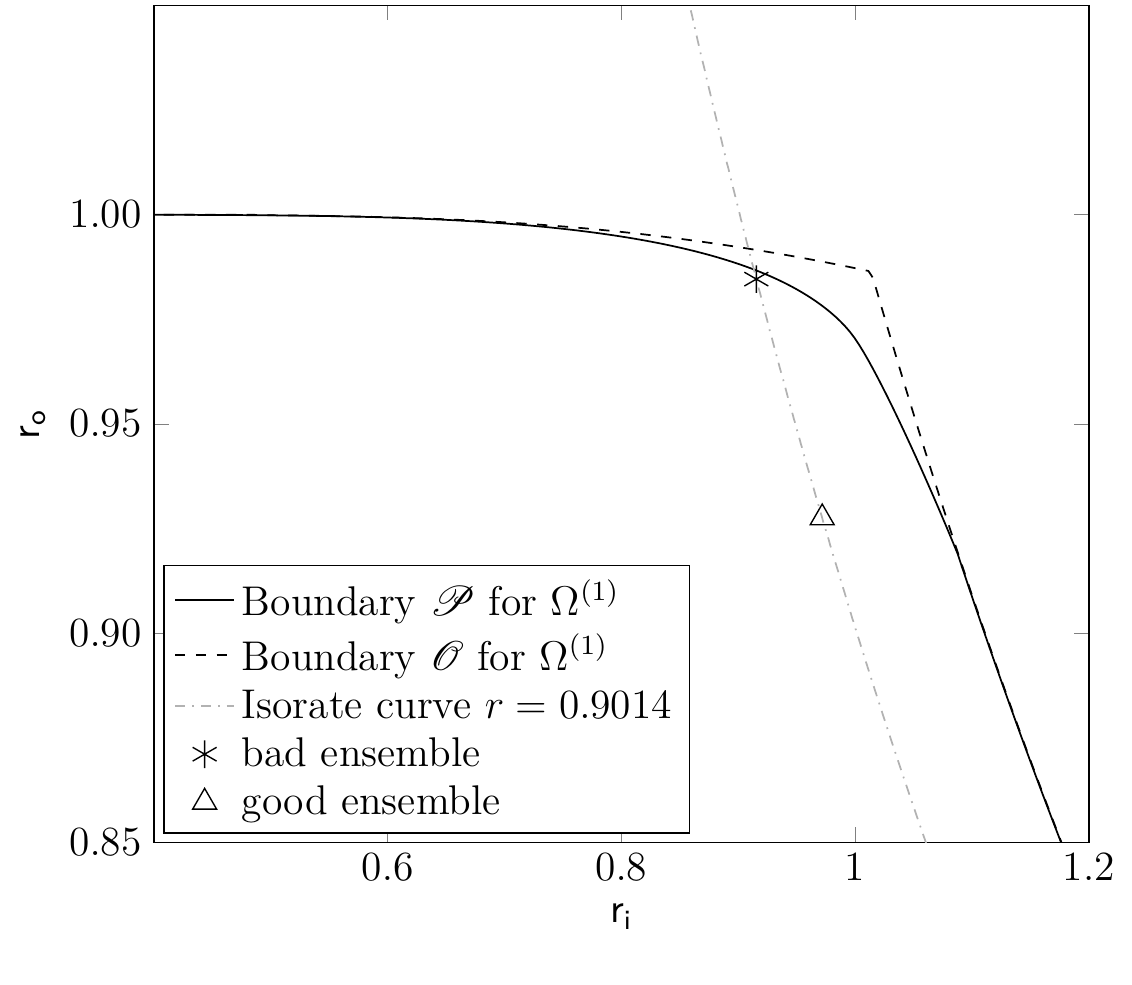}
        \caption[Positive growth rate region of for $\Omegarten$]{Positive growth rate region. The solid and dashed lines represent the positive growth-rate region of $\Omegarten$ its outer bound. The dashed-dotted line represents the isorate curve for $\rate=0.9014$ and the markers represent two different points along the isorate curve that correspond to two different code configurations with the same rate $\rate$ but different values of $\ri$ and $\ro$. The asterisk marker represents the bad ensemble, whereas the triangle marker represents the good ensemble.}
\label{fig:region_results}
\end{figure}

\fran{Following  \cite{Gallager63} we introduce the notion of typical minimum distance for finite length, which is useful when considering ensembles of finite length Raptor codes.}
\begin{mydef}
The typical minimum distance, $\dmintt$  of an ensemble $\ensemble(\oensemble,\Omega, \ri, \ro, n)$ is defined as the integer number
\begin{align*}
\dmintt := \begin{cases}
0 & \text{if } A_0 > 1 + 1/2 \\
\max \{ d\geq0 : \sum_{i=0}^{d} \we_i -1 < 1/2 \} & \text{otherwise.}
\end{cases}
\end{align*}
\end{mydef}

According to this definition, at least half of the codes in the ensemble will have a minimum distance of $\dmintt$ or larger. The equivalent of $\dmintt$ in the asymptotic regime is $\dmint$, the (\emph{asymptotic}) normalized minimum distance of the ensemble $\msr{C}_{\infty}(\oensemble,\Omega, \ri, \ro)$.
For sufficiently large $n$ it is expected that $\dmintt$ converges to $\dmint n$.
In Figure~\ref{fig:d_min_delta_star} $\dmintt$ and $\dmint n$ are shown as a function of the block length $n$. We can observe how the good ensemble has a larger typical minimum distance than the bad ensemble. In fact for all values of $n$ shown in Figure~\ref{fig:d_min_delta_star} the bad ensemble has typical minimum distance $\dmintt=0$\footnote{Since the bad ensemble is inside the positive growth rate region its minimum distance increases linearly with $n$. Thus for large enough $n$ its typical minimum distance will be strictly positive.}. It can also be observed how already for small values of $n$ the $\dmintt$ and $\dmint n$ are very similar. Therefore, we can say that at least for this example the result of our asymptotic analysis of the minimum distance holds already for small values of $n$.
\begin{figure}[t]
        \centering
        \includegraphics[width=\figwbigger]{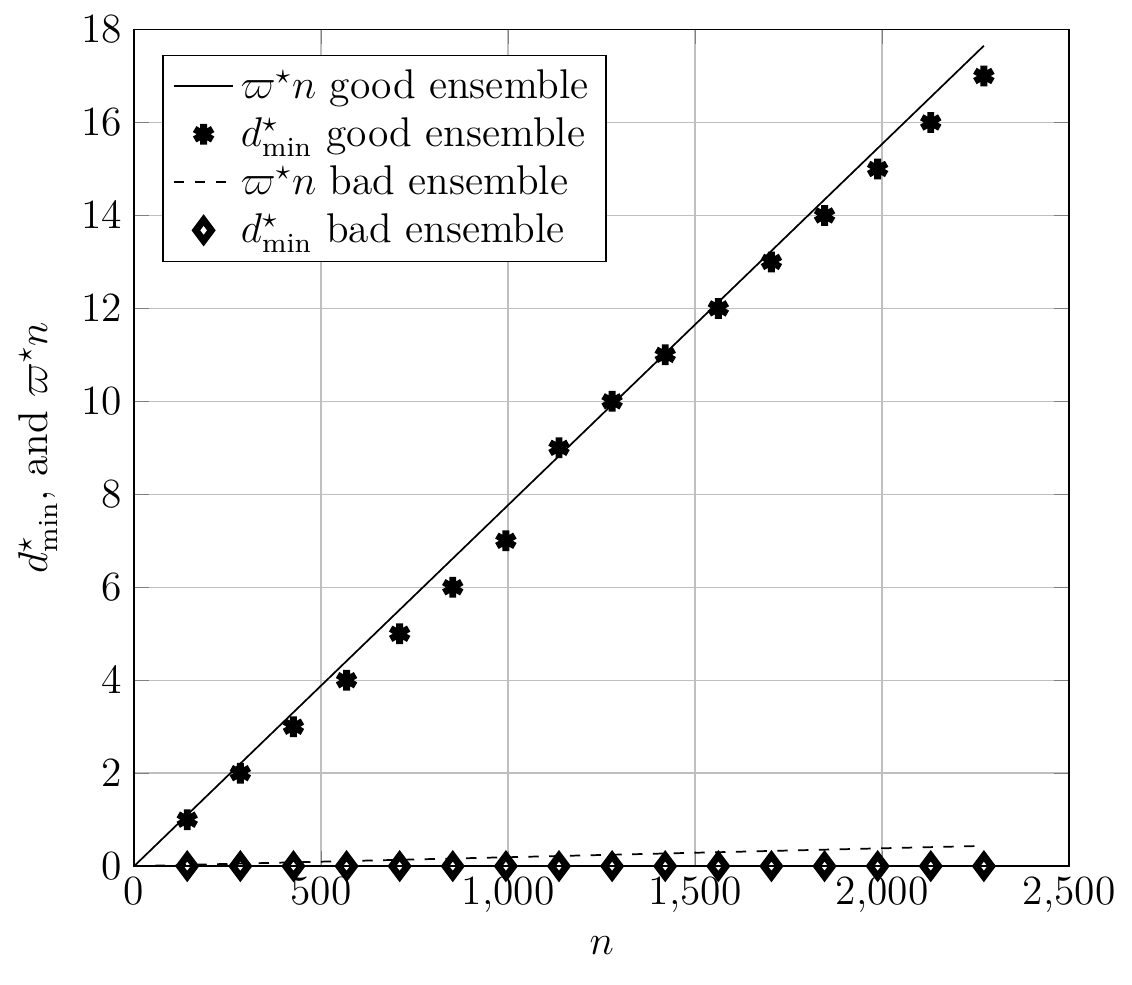}
        \caption[Typical minimum distance $\dmintt$ as a function of the block length $n$]{Typical minimum distance $\dmintt$ as a function of the block length $n$ for ensembles with $\ro= 0.9275$ and $\ro=0.9846$ and $\rate=0.9014$. The markers represent $\dmintt$ whereas the lines represent $\dmint n$.}
\label{fig:d_min_delta_star}
\end{figure}

The expression of the average weight enumerator in Theorem~\ref{theorem:we} can be used in order to upper bound the average \acf{CER}  over a \ac{BEC} with erasure probability $\erasprob$ according to \eqref{eq:bound_Gavg}, \cite{CDi2001:Finite}. However, the upper bound  in \eqref{eq:bound_Gavg} needs to be slightly modified to take into account codewords of weight $0$. We have

\begin{align}
\label{eq:bound_Gavg_modif}
&\Exp_{\ensemble(\oensemble,\Omega, \ri, \ro, n)} \left[P_B(\erasprob)\right]\leq
P^{(\mathsf S)}_{B}(n,k,\erasprob) \nonumber \\
& + \sum_{e=1}^{n-k} {n \choose e} \erasprob^e (1-\erasprob)^{n-e} \min \left\{1, \sum_{w=1}^e {e \choose w} \frac{\we_w}{{n \choose w}}\right\} + \we_0-1
\end{align}
where $P^{(\mathsf S)}_{B}(n,k,\erasprob)$ is the Singleton bound given by \eqref{eq:bound_S}.

Since we consider Raptor codes in a fixed-rate setting, it is possible to expurgate Raptor code ensembles as it was done by Gallager in his PhD thesis for \ac{LDPC} code ensembles \cite{Gallager63}. Let us consider an integer
$\ds \geq 0$ so that
\begin{align}
\label{eq:pr_d_min_ex}
\Pr\{ \dmin \leq \ds\} & \leq   \sum_{w=0}^{\ds} A_w -1 = \theta < 1/2.
\end{align}
We define the expurgated ensemble $\ensemble^{\text{ex}}(\oensemble,\Omega, \ri, \ro, n, \ds)$ as the ensemble composed of the codes in the ensemble $\ensemble(\oensemble,\Omega, \ri, \ro, n)$  whose minimum distance is $\dmin >\ds$. The expurgated ensemble contains at least a fraction $1 - \theta>1/2$ of the codes in the original ensemble. From \cite{Gallager63} it is known that the average \acl{WE} of the expurgated ensemble can be upper bounded by:
\begin{align*}
\we^{\text{ex}}_d
\begin{cases}
\leq  2 \we_d  & \text{if } d > \ds  \\
= 0             & \text{if } 1 \leq d \leq\ds \\
\end{cases}
\end{align*}

In order to obtain the simulation results of this section, in each ensemble  $6000$ codes\footnote{For clarity of presentation only 300 codes are shown in the figures.} were selected randomly from the ensemble. For each code Monte Carlo simulations over a \ac{BEC} were performed  until $40$ errors were collected or a maximum of $10^5$ codewords were simulated. We remark that the objective here was characterizing the average performance of the ensemble, and not so much the performance of every single code (this would have required more codewords to be simulated for each code).

In Figure~\ref{fig:UB_128} we show the \ac{CER} vs.\ the erasure probability $\erasprob$ for two ensembles with $\rate=0.9014$ and $k=128$ that have different outer code rates, $\ro=0.9275$ (good ensemble) and $\ro=0.9846$ (bad ensemble). According to Figure~\ref{fig:d_min_delta_star} the good ensemble is characterized by a typical minimum distance $\dmintt=2$ whereas the bad ensemble is characterized by $\dmintt=0$.
For both ensembles the upper bound in \eqref{eq:bound_Gavg_modif} holds for the average \ac{CER}. However, it can be observed how the performance of the codes in the ensemble shows a high dispersion due to the short block length ($n=142$),i.e., the \ac{CER} curves of different codes from the ensemble can be quite far apart. In fact, in both ensembles a fraction of the codes has a minimum distance equal to zero, that leads to \ac{CER}$=1$ for all values of $\erasprob$ (around $1\%$ for the good ensemble and $30\%$ for the bad ensemble).
If one compares Figure~\ref{fig:UB_good_128} and Figure~\ref{fig:UB_bad_128} it can be seen how  the fraction of codes performing close to the random coding bound is larger in the good ensemble than in the bad ensemble. For the good ensemble Figure~\ref{fig:UB_good_128} shows also an upper bound on the average \ac{CER} for the expurgated ensemble with $\ds=1$. It can be observed how the expurgated ensemble has a lower error floor.
For the bad ensemble no expurgated ensemble can be defined (no $\ds\geq 0$ exists that leads to $\theta<1/2$ in \eqref{eq:pr_d_min_ex}).

\begin{figure}[t]
        \centering
        \begin{subfigure}[b]{0.99\textwidth}
        \centering
            \includegraphics[width=\figw]{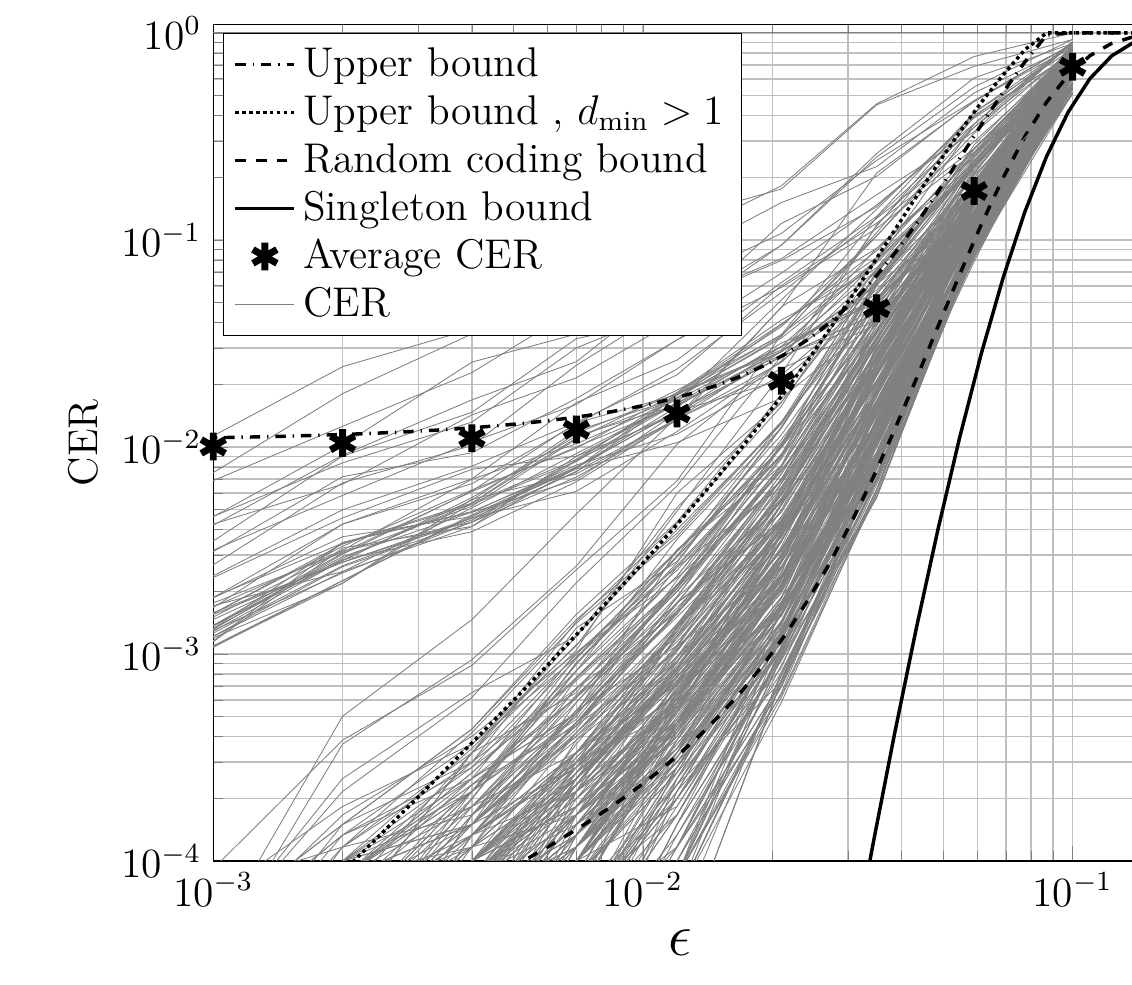}
            \subcaption{good ensemble, $\ro=0.9275$, $\rate=0.9014$} \label{fig:UB_good_128}
        \end{subfigure}
        \begin{subfigure}[b]{0.99\textwidth}
        \centering
            \includegraphics[width=\figw]{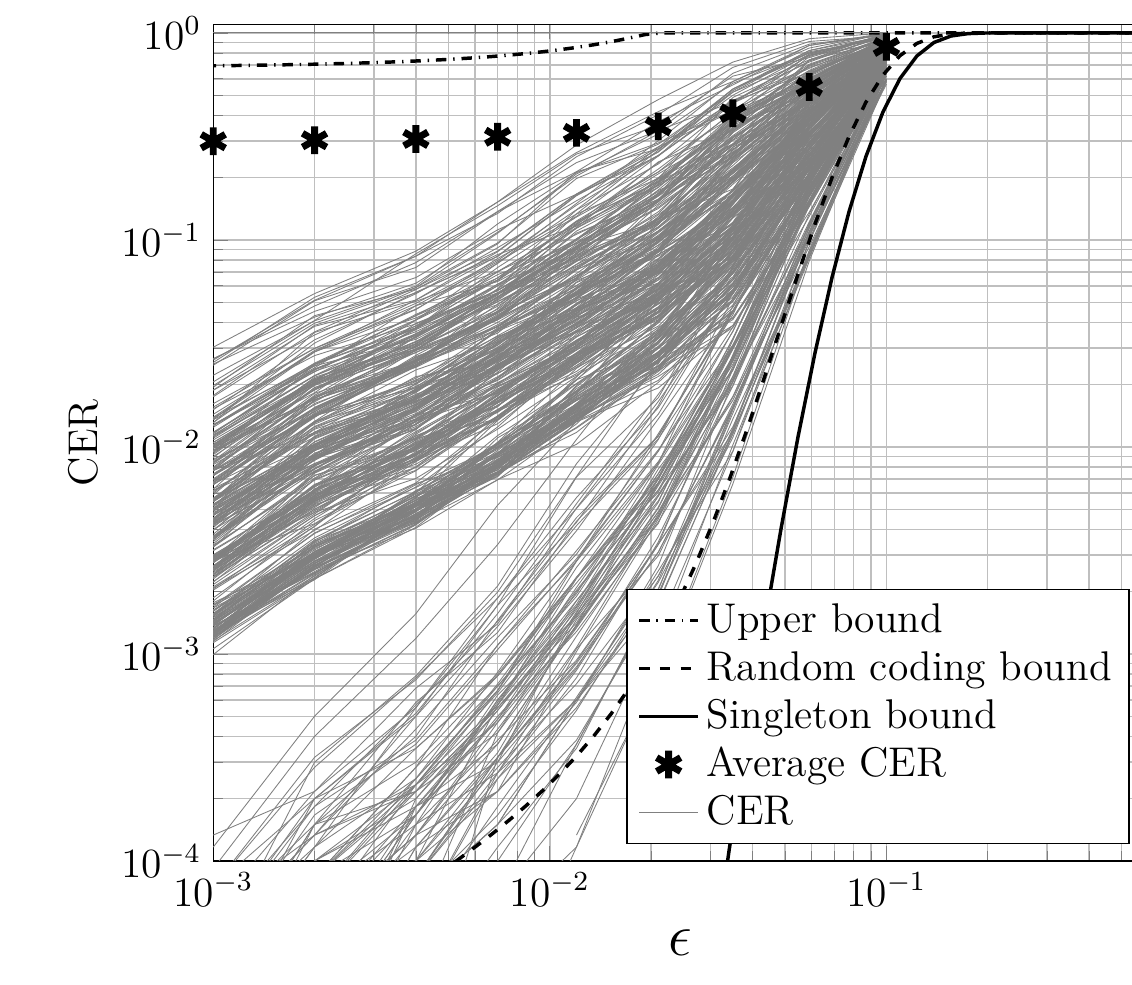}
            \subcaption{bad ensemble, $\ro=0.9846$, $\rate=0.9014$}\label{fig:UB_bad_128}
        \end{subfigure}
        \caption[Codeword error rate \ac{CER} vs.\ erasure probability $\erasprob$ for two ensembles with $\rate=0.9014$ and $k=128$ but different values of $\ro$.]{Codeword error rate \ac{CER} vs.\ erasure probability $\erasprob$ for two ensembles with $\rate=0.9014$ and $k=128$ but different values of $\ro$.  The solid, dashed and dot-dashed lines represent respectively the Singleton bound, the Berlekamp random coding bound and the upper bound in \eqref{eq:bound_Gavg_modif}. The dotted line represents the upper bound for the expurgated ensemble for $\ds=1$. The markers represent the average \ac{CER} of the ensemble and the thin gray curves represent the performance of the different codes in the ensemble, both obtained through Monte Carlo simulations.}\label{fig:UB_128}
\end{figure}

Figure~\ref{fig:UB_256} shows the \ac{CER} vs.\ $\erasprob$ for two ensembles using the same outer code rates as in Figure~\ref{fig:UB_128} but this time for $k=256$. We can observe how the \ac{CER} shows somewhat less dispersion compared to $k=128$. Moreover, comparing Figure~\ref{fig:UB_good_256} and Figure~\ref{fig:UB_good_128} it can be seen how for the good ensemble ($\ro=0.9275$) the error floor is much lower for $k=256$ than for $k=128$, due to a larger  typical minimum distance. Actually, whereas for $k=128$ there were some codes with minimum distance zero in the good ensemble, for $256$ we did not find any code with minimum distance zero out of the $6000$ codes that were simulated. For the good ensemble it is possible again to considerably lower the error floor by expurgation.
However, for the bad ensemble, comparing Figure~\ref{fig:UB_bad_256} and Figure~\ref{fig:UB_bad_128} we can see how the error floor is approximately the same for $k=128$ and $k=256$. In fact, in both cases the typical minimum distance is zero.

\begin{figure}[t]
        \centering
        \begin{subfigure}[b]{0.99\textwidth}
        \centering
            \includegraphics[width=\figw]{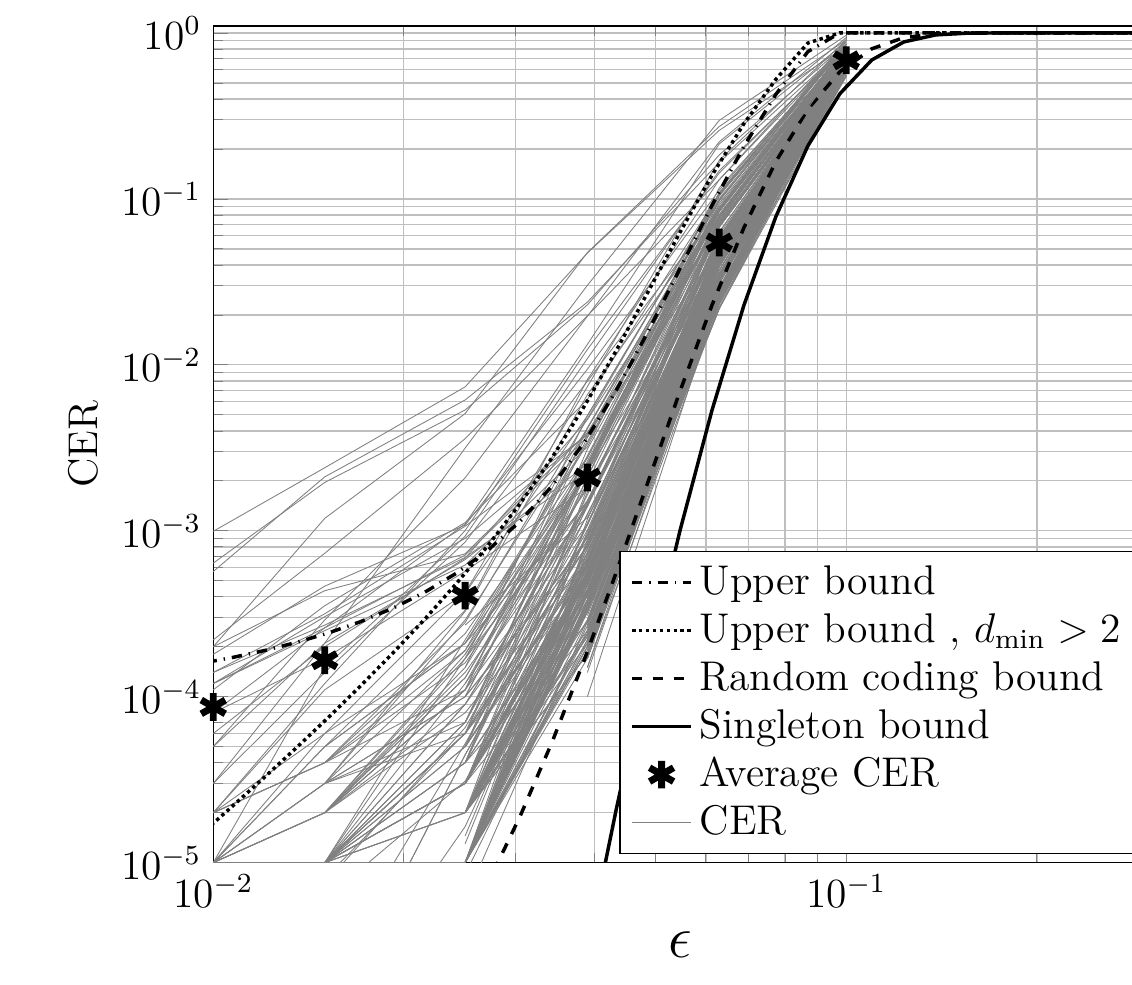}
            \subcaption{good ensemble, $\ro=0.9275$, $\rate=0.9014$} \label{fig:UB_good_256}
        \end{subfigure}
        \begin{subfigure}[b]{0.99\textwidth}
        \centering
            \includegraphics[width=\figw]{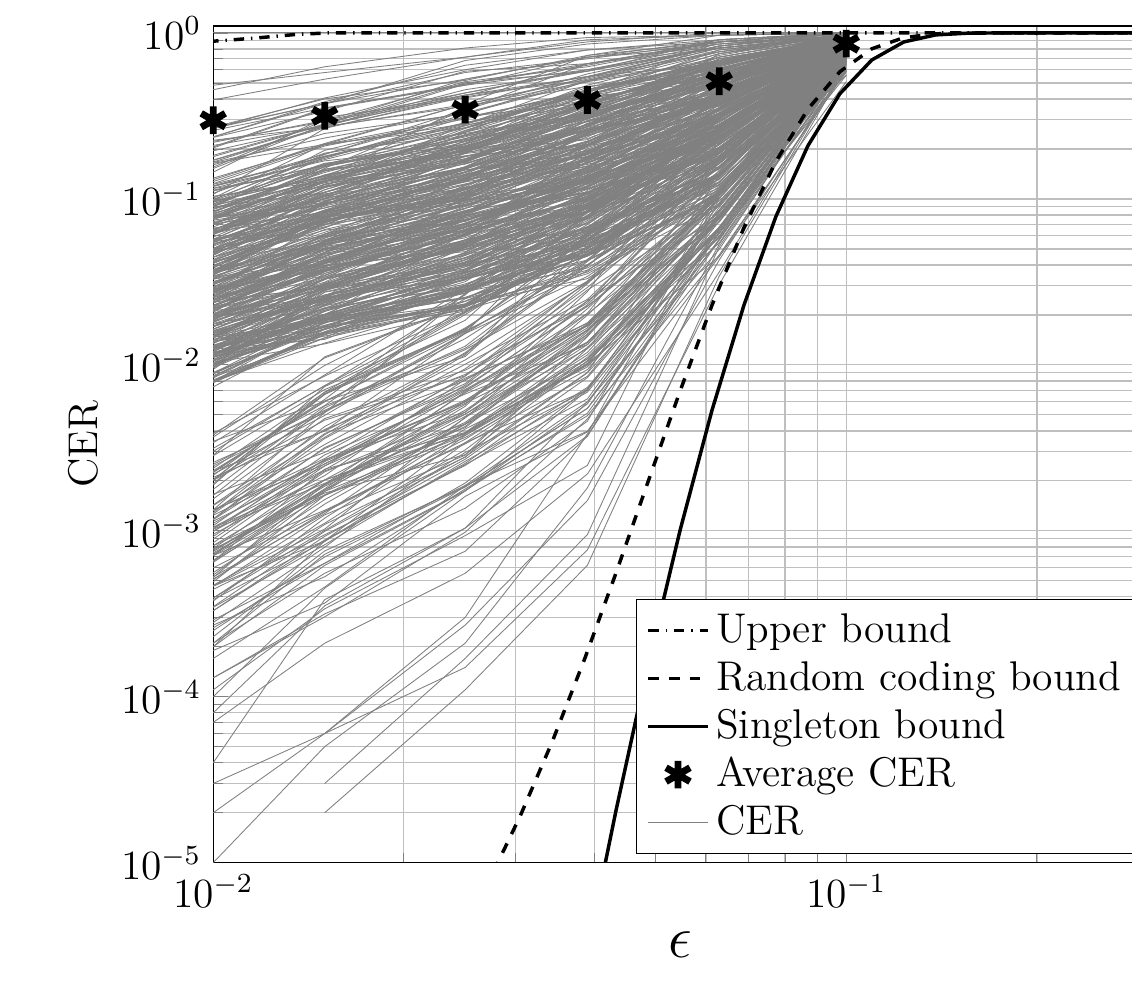}
            \subcaption{bad ensemble, $\ro=0.9846$, $\rate=0.9014$}\label{fig:UB_bad_256}
        \end{subfigure}
        \caption[Codeword error rate \ac{CER} vs.\ erasure probability $\erasprob$ for two ensembles with $\rate=0.9014$ and $k=256$ but different values of $\ro$.]{Codeword error rate \ac{CER} vs.\ erasure probability $\erasprob$ for two ensembles with $\rate=0.9014$ and $k=256$ but different values of $\ro$.  The solid, dashed and dot-dashed lines represent respectively the Singleton bound, the Berlekamp random coding bound and the upper bound in \eqref{eq:bound_Gavg_modif}. The dotted line represents the upper bound for the expurgated ensemble for $\ds=2$. The markers represent the average \ac{CER} of the ensemble and the thin gray curves represent the performance of the different codes in the ensemble, both obtained through Monte Carlo simulations.}\label{fig:UB_256}

\end{figure}

%
%

\subsection{Comparison with Raptor Codes with Standard R10 Outer Codes}
In this section we provide a numerical example that illustrates how the results obtained for linear random outer codes closely approximate the results with the standard R10 Raptor outer code  (c.f. Section~\ref{chap:raptor_r10}). Specifically, we consider Raptor codes with an \ac{LT} degree distribution $\Omega(x) = 0.0098 x   + 0.4590  x^2 +  0.2110  x^3 +  0.1134 x^4 +   0.2068 x^5$. Figure~\ref{fig:region_compare} shows the positive growth rate region for such a degree distribution (assuming a linear random outer code) and three different rate points, two of which are inside the region $\region$ while the third one lies outside. The $(\ri,\ro)$ rate pairs for the three points are specified in the figure caption.

In Figure~\ref{fig:CER_compare} we show the average \ac{CER} obtained through Monte Carlo simulations for the ensembles of Raptor codes with $k=1024$, output degree distribution $\Omega(x)$ and two different outer codes, the standard R10 outer code and a linear random outer code. The three different rate points given in Figure~\ref{fig:region_compare} are considered. For each rate point the average CER is given for the ensemble using the standard R10 outer code and for the a linear random outer code.
The \acp{CER} obtained with both precodes are very similar in all cases. Furthermore, the error floor behavior of the Raptor code ensemble with R10 outer code is in agreement with the  position of the corresponding point on the $(\ri,\ro)$ plane with respect to the $\region$ region, although this region is obtained using the simple linear random outer code model. Concretely, for rate points inside $\region$ the error floor is low, and it tends to become lower the further the point is from the boundary of $\region$. However, for rate points outside region $\region$, we have a very high error floor. \fran{Thus, our analysis, which is done for linear random outer codes, can be used to make accurate predictions on the behaviour of Raptor codes employing the R10 outer code.}

\begin{figure}
\begin{center}
    \includegraphics[width=\figw]{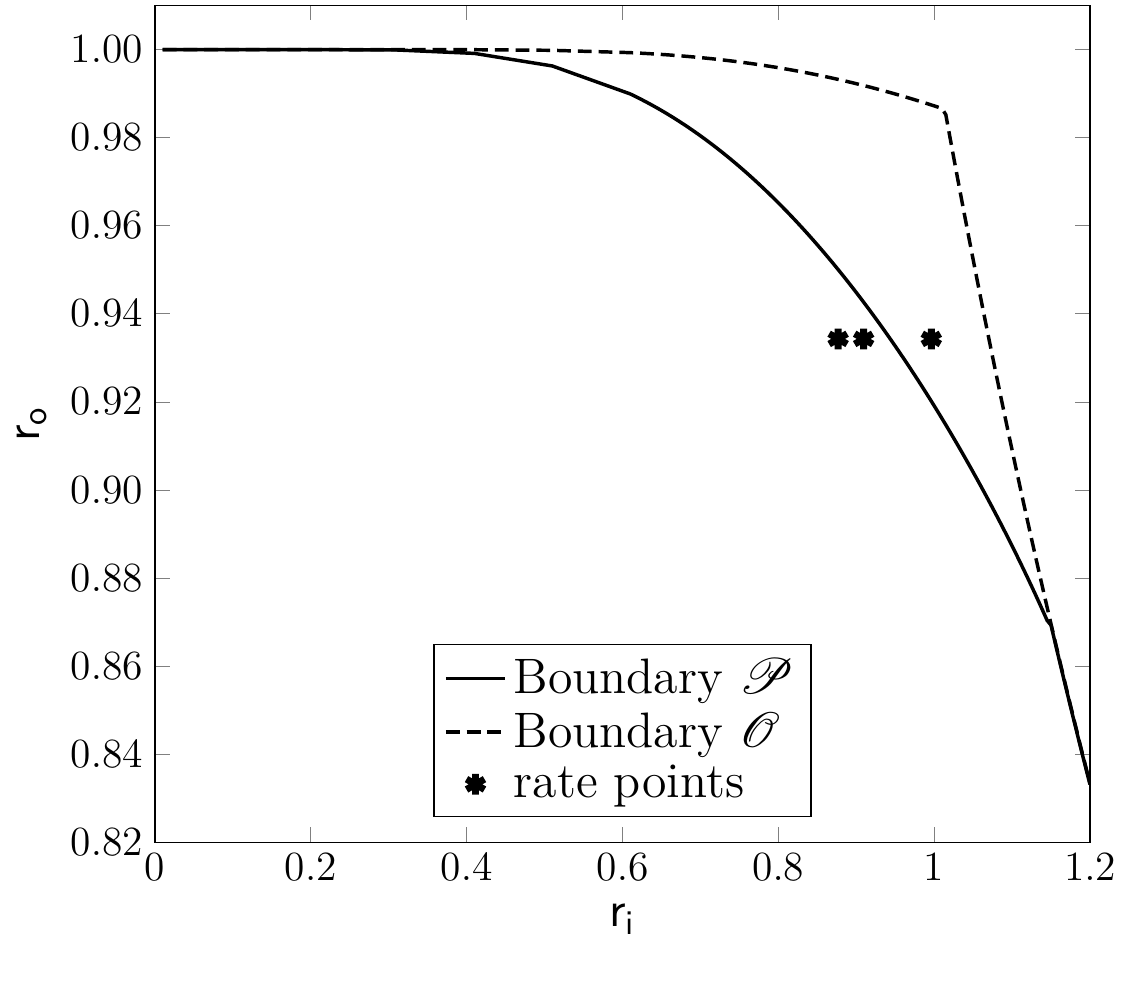}
    \centering \caption[Positive growth rate region for the degree distribution  $\Omega(x) = 0.0098 x   + 0.4590  x^2 +  0.2110  x^3 +  0.1134 x^4 +   0.2068 x^5$]{Positive growth rate region for the degree distribution  $\Omega(x) = 0.0098 x   + 0.4590  x^2 +  0.2110  x^3 +  0.1134 x^4 +   0.2068 x^5$.  The markers represent three different rate points all of them with $\ro=1024/1096$ but with different inner code rates,    $\ri=1096/1100$, $\ri=1096/1205$ and $\ri=1096/1250$.} \label{fig:region_compare}
    \end{center}
\end{figure}
\begin{figure}
    \begin{centering}
    \hspace{2.5cm}
        \includegraphics[width=0.78\columnwidth]{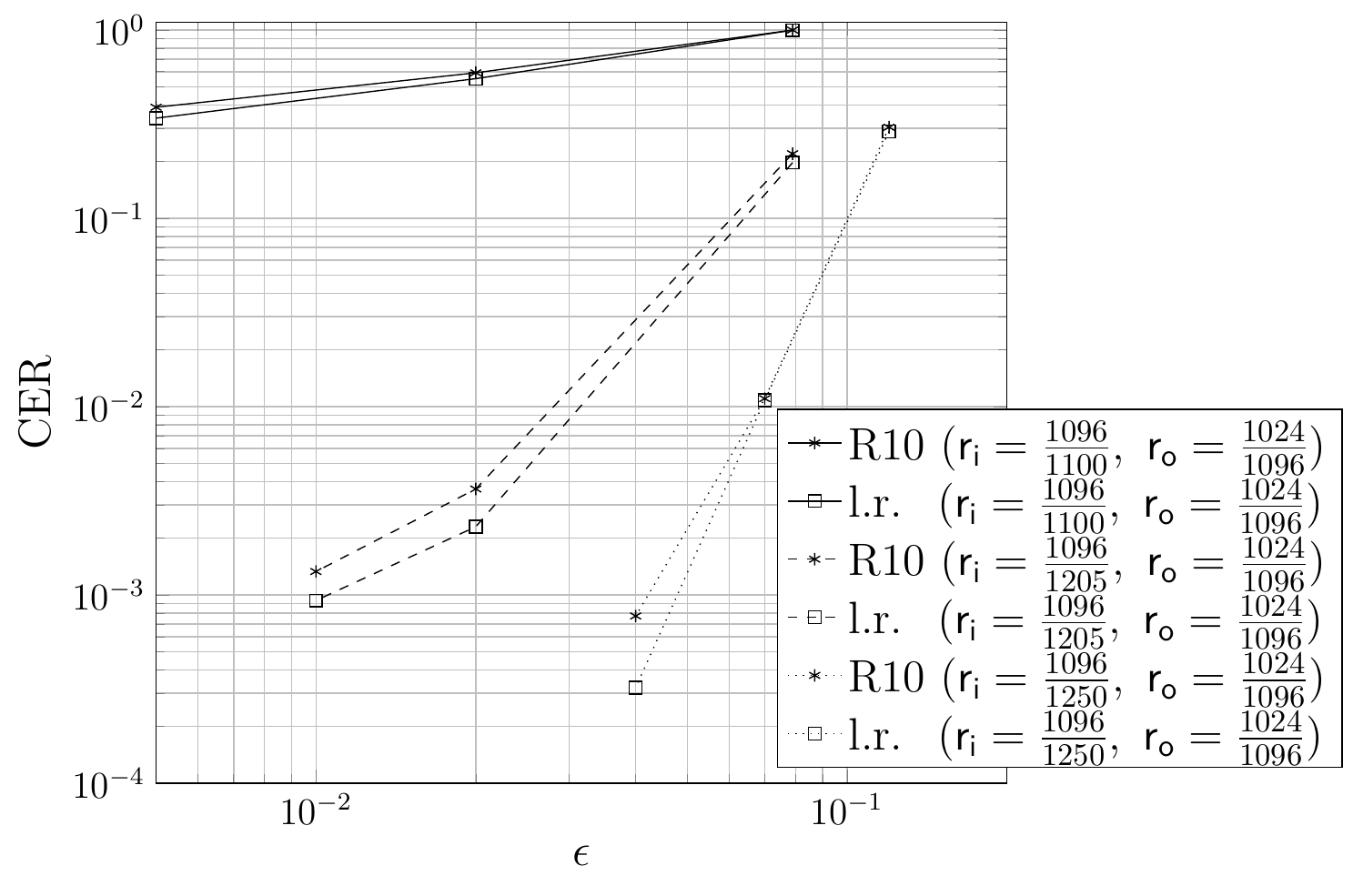}
        \caption[CER  vs.\ $\reloverhead$ for two Raptor codes, one with linear random outer code and the other with the standard outer code of R10 Raptor codes]{Average \ac{CER} for Raptor code ensembles using $\Omega(x) = 0.0098 x   + 0.4590  x^2 +  0.2110  x^3 +  0.1134 x^4 +   0.2068 x^5$ as output degree distribution and two different outer codes, the standard outer code of R10 Raptor codes and a linear random outer code, (l.r.) in the legend.}\label{fig:CER_compare}
        \end{centering}
\end{figure}


\section{Summary}\label{chap:raptor_fixed_rate_summary}
In this chapter we have analyzed the distance spectrum of fixed-rate Raptor codes with outer codes from the linear random ensemble. The expression of the average \acl{WE} and the growth rate of the \acl{WE} as functions of the rate of the outer code and the rate and degree distribution of the inner \ac{LT} code have been derived. Based on these expressions necessary and sufficient conditions to have Raptor code ensembles with a positive typical minimum distance were derived. \fran{These conditions lead to a region $\region$ defined in the $(\ri,\ro)$ plane, where $\ri$ is the rate of the inner \ac{LT} code and $\ro$ is the rate of the outer code. Points inside region $\region$ correspond to fixed-rate Raptor code ensembles with a positive typical minimum distance.} Moreover, a simple necessary condition has been developed too, that only requires  (besides the inner and outer code rates) the knowledge of the average output degree. \fran{This condition leads to a region $\outer$  in the $(\ri,\ro)$ plane, which provides an outer bound to $\region$ and holds for all degree distributions having the same average output degree.} The applicability of the theoretical results has been demonstrated by means of simulation results. Furthermore, simulation results have been presented that show that the performance of Raptor codes with  linear random outer codes is close to that of Raptor codes with the standard outer code of R10 Raptor codes. Thus, we speculate that the results obtained for Raptor codes with linear random outer codes hold as first approximation for standard R10 Raptor codes.



%% file: Chapter6/chapter6.tex
\chapter{Parallel Concatenated Fountain Codes} \label{chap:parallel}
\ifpdf
    \graphicspath{{Chapter6/Chapter6Figs/PNG/}{Chapter6/Chapter6Figs/PDF/}{Chapter6/Chapter6Figs/}}
\else
    \graphicspath{{Chapter6/Chapter6Figs/EPS/}{Chapter6/Chapter6Figs/EPS}}
\fi

In this chapter we present a novel fountain coding scheme that is specially suited for small values of $k$.
The proposed scheme consists of a parallel concatenation of  a $(\nc,k)$ block code with a \acf{LRFC}.
 The scheme is specially interesting when the block code is \acf{MDS}. The remainder of this chapter is organized as follows. In Section~\ref{sec:concatenation} the proposed concatenated scheme is described. In Section \ref{sec:concat_mds_pre} of the scheme is analyzed for the case in which the block code is \ac{MDS}. In
 Section~\ref{sec:generic_precode} the performance of the scheme is analyzed for a generic block code in the fixed-rate setting.  In Section~\ref{sec:results} numerical results are presented for a multicasting system making use of the proposed fountain coding scheme, and the performance is compared with that of \ac{LRFC} codes. Finally, a summary of the results in this chapter is presented in Section~\ref{sec:conc}.


\section{Scheme Description}\label{sec:concatenation}

Let us define the source block $\mathbf{\isymbc}=(\isymbc_1, \isymbc_2, \ldots, \isymbc_k)$ as a vector of source symbols belonging to a finite field of order $q$, i.e., $\mathbf{\isymbc}\in \mathbb {F}_q^k$. In the proposed scheme, the source block is first encoded via a $(\nc,k)$ linear block code $\code'$ over $\mathbb {F}_q$ with generator matrix $\Gpre$. We will make use of Raptor code terminology and call this block code also precode. The encoded block is hence given by
\[
\mathbf{\osymbpre}=\mathbf{\isymbc}\Gpre=(\osymbpre_1,\osymbpre_2,\ldots,\osymbpre_{\nc}).
\]
Additional redundancy symbols can be obtained using an \ac{LRFC}, that is,  by computing linear random combinations of the $k$ source symbols as
\begin{align}
\osymbconcat_i=\osymblrfc_{i-\nc}=\sum_{j=1}^{k}g_{j,i}\isymbc_j, \qquad i=\nc+1,\ldots, \lc
\label{eq:encoding}
\end{align}
where the coefficients $g_{j,i}$ in \eqref{eq:encoding} are selected from $\mathbb {F}_q$ uniformly at random.

Thus, the  encoded sequence corresponds to:
\[
\mathbf{\osymbconcat}=(\mathbf{\osymbpre}|\mathbf{\osymblrfc}).
\]
Where, $\mathbf{\osymbpre}$ and $\mathbf{\osymblrfc}$ are respectively the output of the block code and the \ac{LRFC}. The generator matrix of the concatenated code has the form
\begin{align}
\mathbf{G}=
\underbrace{\left(\begin{array}{cccc}
  g_{1,1} & g_{1,2} & \ldots & g_{1,\nc} \\
  g_{2,1} & g_{2,2} & \ldots & g_{2,\nc} \\
  \vdots  & \vdots  & \ddots & \vdots  \\
  g_{k,1} & g_{k,2} & \ldots & g_{k,\nc}
\end{array}\right|}_{\Gpre}  \underbrace{\left|\begin{array}{cccc}
  g_{1,\nc+1} & g_{1,\nc+2} & \ldots  & g_{1,\lc} \\
  g_{2,\nc+1} & g_{2,\nc+2} & \ldots  & g_{2,\lc} \\
  \vdots    & \vdots    & \ddots  & \vdots  \\
  g_{k,\nc+1} & g_{k,\nc+2} & \ldots  & g_{k,\lc}
\end{array}\right)}_{\Glrfc}
\end{align}
where $\Gpre$ and $\Glrfc$ are the generator matrices of the precode and the \ac{LRFC} respectively. The encoded sequence can be written as:
\[
\mathbf{\osymbconcat}=\mathbf{\isymbc}\mathbf{G}=(\osymbconcat_1,\osymbconcat_2,\ldots,\osymbconcat_\lc).
\]

This scheme can actually be seen as a parallel concatenation of the linear block code $\code '$ and of an \ac{LRFC} (Figure~\ref{fig:par}),
where the first $\nc$ output symbols are the codeword symbols of the block code.\footnote{This represents a difference with Raptor codes, for which the output of the precode is further encoded by a \ac{LT} Code. Hence the first $n$ output symbols of a Raptor encoder do not coincide with the output of the precode.}

We remark that, being the \ac{LRFC} rateless, the number of  output symbols $\lc$ can grow indefinitely. Thus, the proposed scheme is also rateless.
The encoder may be seen as a modified fountain  encoder, whose first $\nc$ output symbols $(\osymbconcat_1,\osymbconcat_2,\ldots,\osymbconcat_\nc)$ correspond to the codeword output by the encoder of $\code'$, whereas the following $\lc-\nc$ symbols are the output of the  \ac{LRFC} encoder. A related rateless construction was proposed in \cite{kasai:fountain}, where a mother non-binary \ac{LDPC} code was modified by replicating the codeword symbols (prior  multiplication by a non-zero field element) and thus by (arbitrarily) lowering the code rate. In our work, the mother code corresponds to the block code, and the additional redundant symbols are produced by the \ac{LRFC}.

\begin{figure}
\begin{center}
\includegraphics[width=0.65\columnwidth]{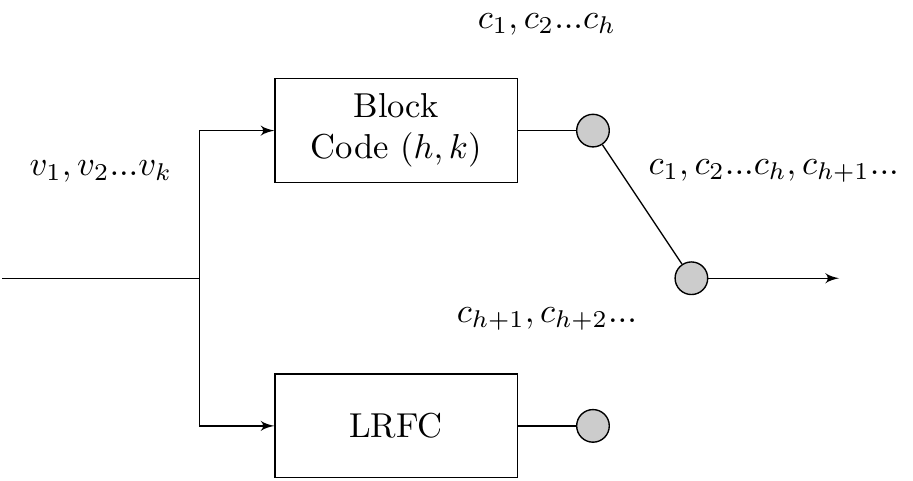}
\centering \caption{Novel fountain coding scheme seen as a parallel concatenation of a $(\nc,k)$ linear block code and a \ac{LRFC}.}\label{fig:par}
\end{center}
\end{figure}

We will assume that the output symbols $\mathbf{\osymbconcat}$ are transmitted over an erasure channel with erasure probability $\erasprob$. Let us assume that at the receiver side $m=k+\absoverhead$ output symbols are collected, where $\absoverhead$ is the (absolute) receiver overhead. Let us denote by $\msr J=\{j_1, j_2, \ldots, j_{m}\}$ the set of the indices of the output symbols of $\mathbf{\osymbconcat}$ that have been collected by a specific receiver. The received vector $\mathbf{\rosymb}$ is hence given by
\[
\mathbf{\rosymb}=(\rosymb_1, \rosymb_2, \ldots,
y_{m})=(\osymbconcat_{j_1},\osymbconcat_{j_2},\ldots,\osymbconcat_{j_{m}})
\]
and it can be related to the source block $\mathbf{\isymbc}$ as
\[
\mathbf{\rosymb}=\mathbf{\isymbc}\Grx
\]
 Here, $\Grx$ denotes the $k\times m$ matrix made by the columns of $\mathbf{G}$ with indices in $\msr J$, i.e.,
\[
\Grx= \left(\begin{array}{cccc}
  g_{1,j_1} & g_{1,j_2} & \ldots & g_{1,j_{m}} \\
  g_{2,j_1} & g_{2,j_2} & \ldots & g_{2,j_{m}} \\
  \vdots  & \vdots  & \ddots & \vdots  \\
  g_{k,j_1} & g_{k,j_2} & \ldots & g_{k,j_{m}}
\end{array}\right).
\]
The recovery of $\mathbf{\isymbc}$ reduces to solving the system of $m=k+\absoverhead$ linear equations in $k$ unknowns
\begin{align}
\Grx^T\mathbf{\isymbc}^T=\mathbf{y}^T.
\label{eq:solve}
\end{align}
The solution of \eqref{eq:solve} can be obtained \fran{by means of a \acf{ML} decoding algorithm (e.g., via Gaussian elimination or via inactivation decoding)}  if and only if $\textrm{rank} (\tilde{\mathbf{G}})=k$.

\section{Maximum Distance Separable Precode} \label{sec:concat_mds_pre}
\label{sec:bounds}

In this section we consider the case in which the precode is \acf{MDS}. \fran{The reasons to consider \ac{MDS} codes are twofold. First, \ac{MDS} codes meet the Singleton bound with equality, which means that over an erasure channel, decoding succeeds with probability one if the receiver is able to collect at least $k$ symbols. Second, the use of \ac{MDS} precodes leads to a very simple performance model, as it will be shown in this section. In particular}, when binary codes are used, we assume $(k+1,k)$ \acf{SPC} codes. When operating on higher order finite fields, we consider  \ac{GRS} codes.

Based on the bounds in {\eqref{eq:tightbounds}}, we will derive tight upper and lower bounds for the decoding failure probability $\Pf$ of our parallel concatenated fountain coding scheme for memoryless erasure channels.  In our analysis we will assume that an encoded sequence $\mathbf{\osymbconcat}$ composed of $\lc \ge \nc$ symbols is transmitted over a memoryless erasure channel with erasure probability of $\erasprob$.\footnote{The case $l < n$ is not considered since it is equivalent to shortening the linear block code.}


In our analysis we will consider two different cases. In the first case among the $m$ received symbols, at least $k$ have indices in $\{1, 2, \ldots, \nc\}$. That is, at least $m'\geq k$ symbols produced by the linear block encoder have been received. In this case, since the precode $\code'$ is \ac{MDS}, the system of equations in \eqref{eq:solve} will be solvable with probability $1$.
The probability of this event (collecting at least $k$ output symbols out of the first $\nc$) is given by:
\begin{align}
Q(\erasprob)=\sum_{i=k}^\nc {\nc \choose i} (1-\erasprob)^i
\erasprob^{\nc-i}.
\end{align}

In the second, less trivial case,  $m'<k$ among the $m$ received symbols have indices in $\{1, 2, \ldots, \nc\}$. That is, less than $k$ output symbols from the precode are collected. This second case is complementary to the first one and will occur with probability
\[
P(\erasprob)=1-Q(\erasprob).
\]
In this case, matrix $\Grx^T$ can be partitioned as
\begin{align}
\Grx^T=\left(\begin{array}{c} \Grxpre^T\\
\Grxlrfc^T \end{array}\right)=\left(\begin{array}{cccc}
  g_{1,j_1} & g_{2,j_1} & \ldots & g_{k,j_1} \\
  g_{1,j_2} & g_{2,j_2} & \ldots & g_{k,j_2} \\
  \vdots  & \vdots  & \ddots & \vdots  \\
  g_{1,j_{m'}} & g_{2,j_{m'}} & \ldots & g_{k,j_{m'}}\\ \hline
  g_{1,j_{m'+1}} & g_{2,j_{m'+1}} & \ldots & g_{k,j_{m'+1}} \\
  g_{1,j_{m'+2}} & g_{2,j_{m'+2}} & \ldots & g_{k,j_{m'+2}} \\
  \vdots  & \vdots  & \ddots & \vdots  \\
  g_{1,j_{m}} & g_{2,j_{m}} & \ldots & g_{k,j_{m}}
\end{array}\right). \label{eq:G_partition}
\end{align}
The fact that the precode $\code'$ is  \ac{MDS} assures that $\textrm{rank} (\Grxpre)=m'$, i.e., the first $m'$ rows of
$\tilde{\mathbf{G}}^T$ are linearly independent. The remaining rows of $\Grx^T$ correspond to $\Grxlrfc^T$ that has size
$m''\times k$, with $m''=m-m'$. The elements in $\Grxlrfc^T$  are uniformly distributed  in $\mathbb F _q$. It follows that the matrix in
\eqref{eq:G_partition} can be put (via column permutations over $\Grx^T$ and row permutations/combinations over $\Grxpre^T$) in the form
\begin{align}
\hat{\mathbf{G}}^T=\left(\begin{array}{ccc} \mathbf{I} & \vline &
\mathbf{A} \\\hline
\mathbf{0} & \vline & \mathbf{B}\\
\end{array}\right), \label{eq:G_partition_manipulation}
\end{align}
where $\mathbf I$ is the $m' \times m'$ identity matrix, $\mathbf{0}$ is a $m'' \times m'$ all-$0$ matrix, and  $\mathbf{A}$, $\mathbf{B}$ have respective sizes $m' \times (k-m')$ and $m''
\times (k-m')$. The lower part of $\hat{\mathbf{G}}^T$, given by $\left(\mathbf{0} | \mathbf{B}\right)$, is obtained by adding to each row of $\Grxlrfc^T$ a linear  combination of rows from $\Grxpre^T$, in a way that the $m'$ leftmost columns of $\Grxlrfc^T$ are all set to zero.
\fran{Thus, the elements of submatrix $\mathbf{B}$ are obtained by adding a deterministic symbol of $\mathbb F_q$ to an element of $\Grxlrfc^T$, which is uniformly distributed in $\mathbb F_q$. It follows that}
the statistical properties of $\Grxlrfc^T$ are inherited by the $m'' \times (k-m')$ submatrix $\mathbf{B}$, whose elements are, hence, uniformly distributed in $\mathbb F_q$. It follows that \eqref{eq:solve} is solvable if and only if $\mathbf{B}$ is full rank, i.e., if and only if $\textrm{rank}(\mathbf{B})=k-m'$.
\noindent
Let us denote the decoding failure event as $\fevent$. The conditional decoding failure probability can be expressed as
\begin{align}
\Pr\{\fevent|m',m'<k,\absoverhead\}=\Pr\{\textrm{rank}(\mathbf{B})<k-m'\}.\label{eq:cond_F_prob_1}
\end{align}
Matrix $\mathbf{B}$ is a $m'' \times (k-m') = (k+\absoverhead-m') \times (k-m')$ random matrix having $\absoverhead$ rows in excess with respect to the number of columns. Hence, we can replace  \eqref{eq:cond_F_prob_1} in \eqref{eq:tightbounds}, obtaining the bounds
\begin{align}
q^{-\absoverhead-1}\leq\Pr\{F|m',m'<k,\absoverhead\}<\frac{1}{q-1}q^{-\absoverhead}.
\label{eq:cond_F_prob_2_bounds}
\end{align}
The bounds in  \eqref{eq:tightbounds} are independent from the size of the matrix, they
depend only on the overhead. Therefore, we can remove the conditioning on $m'$ from \eqref{eq:cond_F_prob_2_bounds}, leaving
\[
q^{-\absoverhead-1}\leq\textrm{Pr}\{F|m'<k,\absoverhead\}<\frac{1}{q-1}q^{-\absoverhead}.
\]
The failure probability can now be expressed as
\begin{align}
\begin{array}{cc}
\Pf=& \Pr\{F|m'<k,\absoverhead\}\Pr\{m'<k\}\\
&+\Pr\{F|m'\geq k,\absoverhead\}\Pr\{m'\geq k\}
\label{eq:general_bound}
\end{array}
\end{align}
where $\Pr\{F|m'\geq k,\absoverhead\}=0$ (since at least $k$ of the symbols produced by the \ac{MDS} encoder have been collected) and
$\Pr\{m'<k\}=P(\erasprob)$. It results that
\begin{align}
P(\erasprob) q^{-\absoverhead-1}\leq
\Pf<P(\erasprob)\frac{1}{q-1}q^{-\absoverhead}.
\label{eq:final_bounds}
\end{align}
If one inspects \eqref{eq:tightbounds} and {\eqref{eq:final_bounds}}, one can see how the bounds on the failure probability of the concatenated scheme are scaled down by a factor
$P(\erasprob)$, which is a monotonically increasing function of $\erasprob$. Therefore, when the channel conditions are \emph{bad} (i.e., for large $\erasprob$) $P(\erasprob)\rightarrow
1$, and the bounds in {\eqref{eq:final_bounds}} tend to coincide with the bounds in \eqref{eq:tightbounds}. On the other hand, if the channel conditions are \emph{good} (i.e., for small $\erasprob$), most of the time $m'\geq k$ symbols produced by the linear block encoder are received and decoding succeeds (recall the assumption of \ac{MDS} code). In these conditions, $P(\erasprob)\ll 1$, and according to the bounds in  {\eqref{eq:final_bounds}} the failure probability may decrease by  several orders of magnitude.

Given the fact that the probability of decoding failure of the concatenated scheme is a function of the erasure probability, the scheme is not universal anymore in a strict sense\footnote{In this concatenated fountain coding scheme the output symbols are not statistically identical and independent from each other. As a consequence its performance depends on the channel erasure rate and its performance will also vary if the channel is not memoryless.}. At low channel erasure probabilities the proposed scheme will outperform \acp{LRFC}, whereas for large erasure probabilities it will perform as \acp{LRFC}. Hence, the performance of our scheme is lower bounded by that of \ac{LRFC}, which are universal codes (their performance \fran{depends only on the number of output symbols received and not on the erasure probability of the channel}). Therefore, one could argue that the proposed scheme is universal in a broad sense, although its probability of decoding failure does depend on the erasure probability of the channel

Figure~\ref{GF_2} shows the probability of decoding failure $\Pf$ as a function of the number of overhead symbols $\absoverhead$ for a concatenated code built using a $(11,10)$ \ac{SPC} code {over} $\mathbb {F}_2$. We can observe how, for lower erasure probabilities, the gain in performance of the concatenated code with respect to a \ac{LRFC}  is larger. For $\erasprob=0.01$ the decoding failure probability is more than $2$ orders of magnitude lower than that of a \ac{LRFC}.
\begin{figure}
\begin{center}
\includegraphics[width=\figwbigger,draft=false]{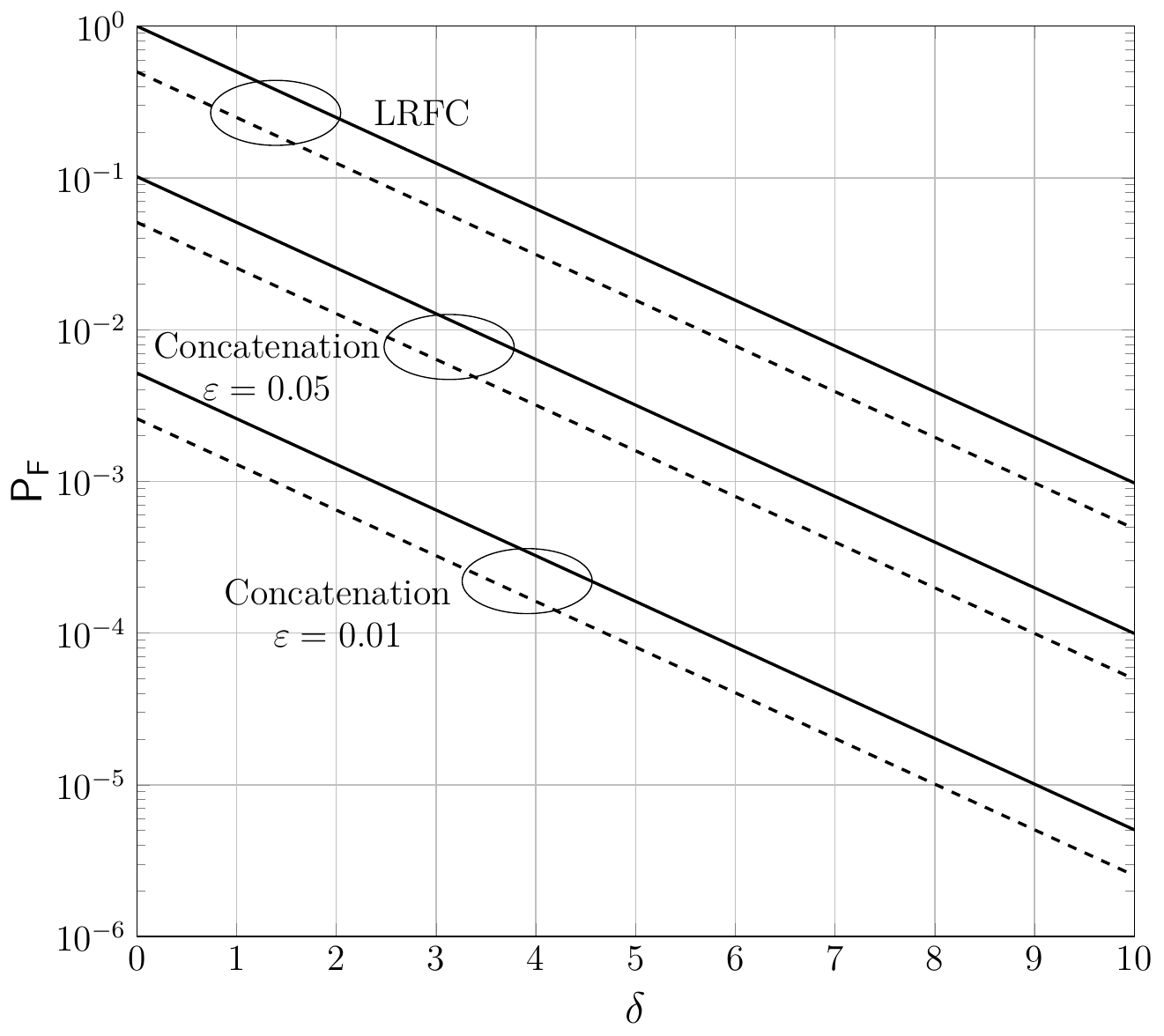}
\centering \caption[$\Pf$ vs.\ $\absoverhead$  for a concatenated code built using a $(11,10)$ \ac{SPC} over $\mathbb {F}_{2}$]{$\Pf$ vs.\ $\absoverhead$  for a concatenated code built using a $(11,10)$ \ac{SPC} code over $\mathbb {F}_{2}$ for  different values of erasure probability $\erasprob$.
Upper bounds are represented by solid lines and lower bounds are represented by dashed lines.} \label{GF_2}
\end{center}
\end{figure}

Figure~\ref{GF_16} shows the probability of decoding failure vs.\ the number of overhead symbols $\absoverhead$ for the concatenation of a $(15,10)$ \ac{RS} and a \ac{LRFC} over $\mathbb {F}_{16}$. The performance of the concatenated code is compared with that of the \ac{LRFC} built on the same field for different erasure probabilities. In this case the decrease in terms of probability of decoding failure is even more notable than in binary case. For a channel with an erasure probability $\erasprob=0.05$, the probability of decoding failure of the concatenated scheme is $4$ orders of magnitude lower than {that of} the \ac{LRFC}.  \fran{If we compare Figure~\ref{GF_16} with Figure~\ref{GF_2}, we can observe how the upper and lower bounds are closer to each other for the codes constructed over $\mathbb {F}_{16}$ compared to the binary codes. This effect stems from the fact that the bounds in \eqref{eq:tightbounds} become tighter as the Galois field order  $q$ increases.}
\begin{figure}
\begin{center}
\includegraphics[width=\figwbigger,draft=false]{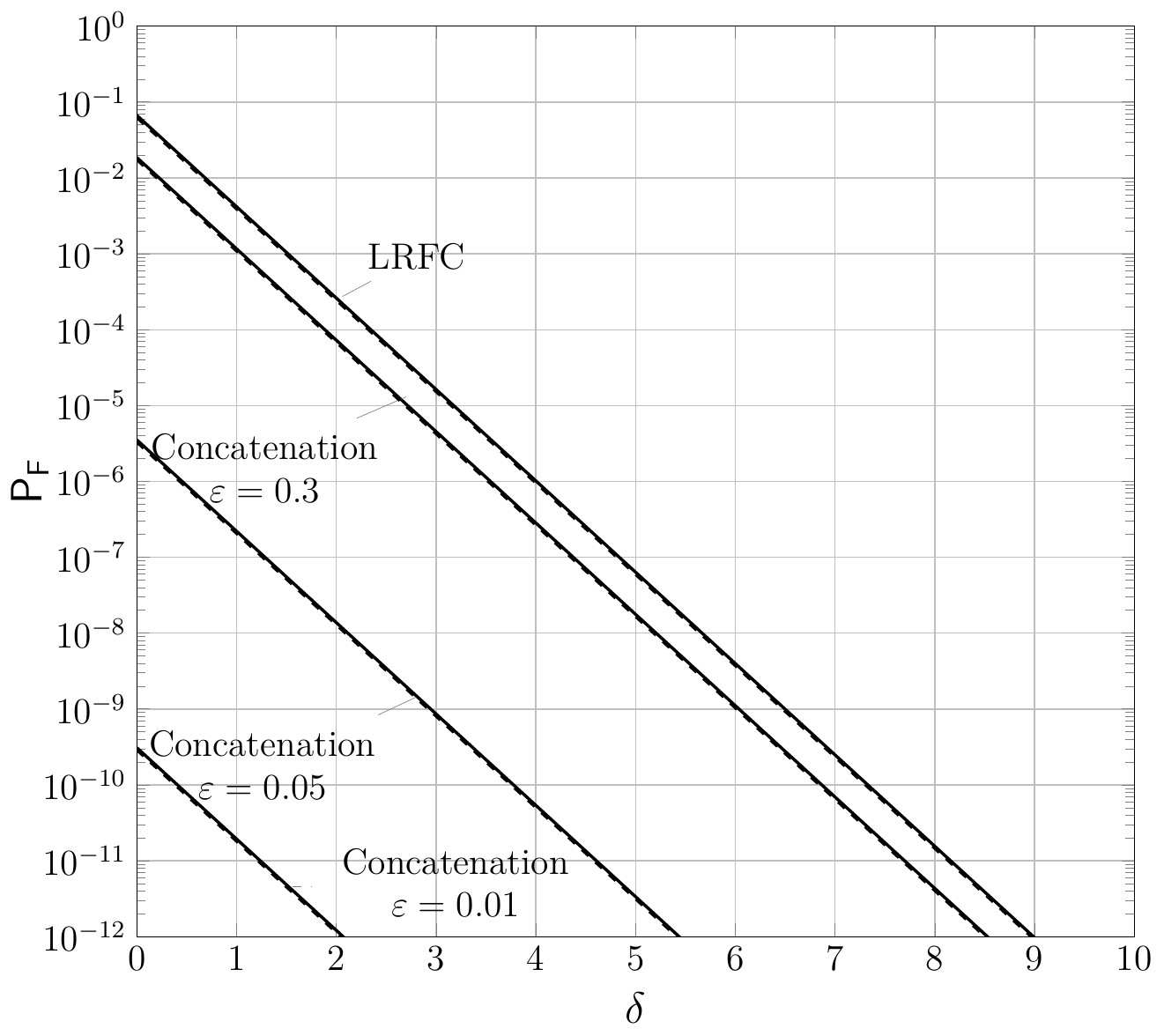}
\centering \caption[$\Pf$ vs.\  $\absoverhead$  for a concatenated code built using a $(15,10)$ \ac{RS} over $\mathbb {F}_{16}$]{$\Pf$ vs.\  $\absoverhead$  for a concatenated code built using a $(15,10)$ \ac{RS} over $\mathbb {F}_{16}$ for  different values of  of erasure probability $\erasprob$.
Upper bounds are represented by solid lines and lower bounds are represented by dashed lines.} \label{GF_16}
\end{center}
\end{figure}

Figure~\ref{GF_2_sim} shows the probability of decoding failure $\Pf$, as a function of the receiver overhead $\absoverhead$, obtained via Monte Carlo  simulations together with the bounds in \eqref{eq:final_bounds}. The results refer to a concatenation of a $(11,10)$ \ac{SPC} with an \ac{LRFC} over $\mathbb {F}_{2}$, and a channel with an erasure probability $\erasprob=0.1$. As expected, the simulation results tightly match the bounds.

Figure~\ref{GF_16_sim} shows similar simulation results for  a concatenation of a $(15,10)$ \ac{RS} code with an \ac{LRFC} over $\mathbb {F}_{16}$, for a channel erasure probability $\erasprob=0.1$. Also in this case, the results are very close to the bounds.

\begin{figure}
\begin{center}
\includegraphics[width=\figwbigger,draft=false]{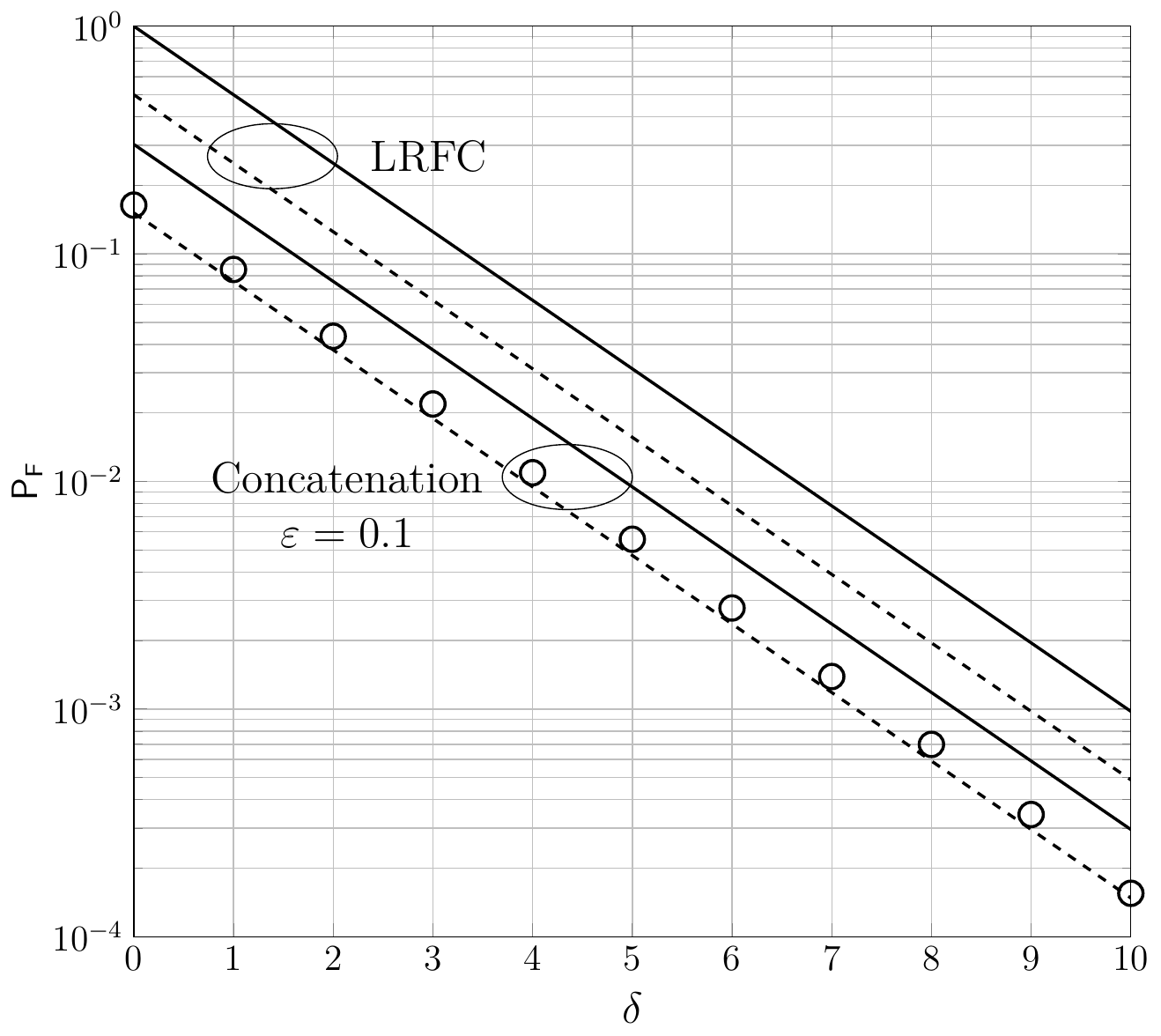}
\centering \caption[$\Pf$ vs.\ $\absoverhead$ symbols for
the concatenation of a $(11,10)$ \ac{SPC} code and a \ac{LRFC} over $\mathbb {F}_{2}$ and $\erasprob=0.1$]{$\Pf$ vs.\ $\absoverhead$ for
the concatenation of a $(11,10)$ \ac{SPC} code and a \ac{LRFC} over $\mathbb {F}_{2}$ and $\erasprob=0.1$. Upper bounds are represented by solid lines and lower bounds are represented by dashed lines. The points marked with '$\circ$' denote actual simulations.} \label{GF_2_sim}
\end{center}
\end{figure}
\begin{figure}
\begin{center}
\includegraphics[width=\figwbigger,draft=false]{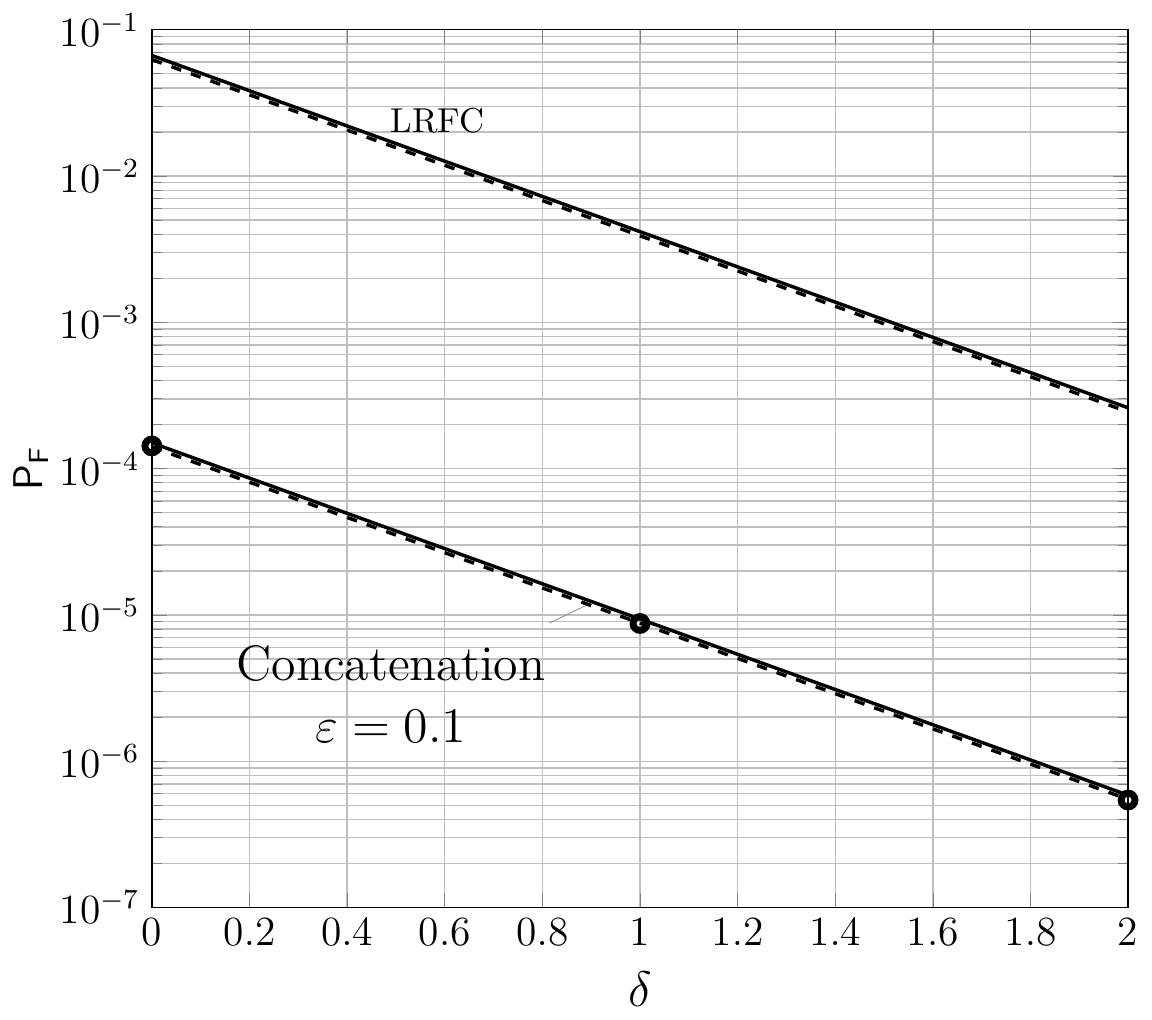}
\centering \caption[$\Pf$ vs.\ $\absoverhead$  for
the concatenation of a $(15,10)$ \ac{RS} and \ac{LRFC} over $\mathbb {F}_{16}$ and $\erasprob=0.1$]{$\Pf$ vs.\ $\absoverhead$  for
the concatenation of a $(15,10)$ \ac{RS} and \ac{LRFC} over $\mathbb {F}_{16}$ and $\erasprob=0.1$. Upper and lower bounds are represented by solid and dashed lines, respectively. The markers '$\circ$' denote simulations.} \label{GF_16_sim}
\end{center}
\end{figure}

\section{Generic Precode in a Fixed-Rate Setting}\label{sec:generic_precode}

Fountain codes are often used in a fixed-rate setting (see Chapter~\ref{chap:Raptor_fixed_rate}). In this context, the main advantage in the use of fountain codes with respect to block erasure correcting codes stems from the possibility of adapting code rate and block length to the transmission needs (e.g., channel conditions) in a flexible manner. In this section, we consider the concatenated scheme in the general case where the block code $\code'$ is not necessarily  \acf{MDS} in a fixed-rate setting. We derive the \acf{WE} of the concatenated code and use it to derive a tight upper bound on the block error probability of the code.

The coding scheme considered in this chapter is a parallel concatenation of a linear block code and a \ac{LRFC}, which for a finite rate setting is a random generator matrix code.  Let us denote as $\msr{C}(\code',k,\lc,q)$ the ensemble of codes obtained by a parallel-concatenation of a $(\nc,k)$ linear block code over $\mathbb{F}_q$, $\code'$, with all possible realizations of an \ac{LRFC}, where $k$ is the number of source symbols, $\lc$ is the total number of output symbols and $q$ is the finite field order. Note that the codes in the ensemble have fixed-rate $r=k/\lc$. We denote as $\wef_i(X)$ the \acf{CO-WEF} averaged over the ensemble $\msr{C}(\code',k,\lc,q)$  conditioned to the input source block having
weight $i$,
\begin{align}
\wef_i(X)= \sum_{w=1}^{\lc} \we_{i,w}X^w
\nonumber
\end{align}
where $\we_{i,w}$ is the average number of codewords of Hamming weight $w$ produced by Hamming weight-$i$ inputs, that is, $\we_{i,w}$ is the \acf{IO-WEF} of the code. For the ensemble of parallel-concatenated codes the average \acl{CO-WEF}  admits a very compact expression:
\begin{align}
\wef_i(X)= \frac{\wef_i^{\code'}(X) \wef_i^{\msr {L}(k,\hc,q)}(X)  }{{k \choose i}},
\label{eq:iowe}
\end{align}
where $\wef_i^{\code'}(X)$ is the \acl{CO-WEF} of the linear block code, and $\wef_i^{\msr {L}(k,\hc,q)}(X)$ is the average \acl{CO-WEF} of the ensemble $\msr {L}(k,\hc,q)$, being $\msr {L}(k,\hc,q)$ the ensemble of linear block codes over $\mathbb{F}_q$ with $k \times \hc$ generator matrix $\Glrfc$, with $\hc=\lc-\nc$.
 Let us assume that $\wef_i^{\code'}(X)$ is known.\footnote{In general, the derivation of the \acl{CO-WEF} $\wef_i^{\code'}(X)$ for a code is not trivial, unless the code $\code'$ (or its dual code) has small dimension \cite{MacWillimas77:Book}.} In this case, the derivation of  $\wef_{i,w}$ reduces to the calculation of  $\wef_i^{\msr {L}(k,\hc,q)}(X)$.

The average number of codewords of Hamming weight $w$ produced by Hamming weight-$i$ inputs for the ensemble ${\msr {L}(k,\hc,q)}$, ${\we_{i,w}^{\msr {L}(k,\hc,q)}}$,  is given by:
\begin{align}
\we_{i,w}^{\msr {L}(k,\hc,q)} = {{k \choose i}} {{\hc \choose w}} \probi^w \left( 1-\probi \right) ^{\hc-w},
\label{eq:iowe_LRFC_0}
\end{align}
where $\probi$ the probability of one of the  $\hc$ output symbols having a non-zero value conditioned to having an input of Hamming weight $i$. Given that  the coefficients of $\Glrfc$ are picked with uniform probability over $\mathbb{F}_q$, we have that{\footnote{{Note that when $i=0$ the encoder input is given by the all-zero word. Thus, the encoder output is zero with probability $1$ due to the linearity of the code ensemble $\msr {L}(k,\hc,q)$.}}}
\begin{align}
\begin {array}{lll}
&\probi = \frac{q-1}{q}&, \; i\neq 0 \\
&\probi = 0&, \; i = 0.
\end{array}
\label{eq:iowe_LRFC_1}
\end{align}
Using \eqref{eq:iowe}, \eqref{eq:iowe_LRFC_0} and \eqref{eq:iowe_LRFC_1} the \acl{CO-WEF} of our concatenated scheme is  obtained.

Once the \acl{CO-WEF} has been derived, the average \ac{WEF} $\wef(X)$  can be obtained from the average \acl{CO-WEF} by summing over all possible input weights
\[
\wef(X) = \sum_i \wef_i X^i.
\]
Finally, the average number of codewords of Hamming weight $w$ (average \acl{WE}) $\we_w$  is simply obtained as the coefficient of degree $w$ in the \ac{WEF}.
The average \acl{WE} of the concatenated ensemble can be used now to derive a tight upper bound on the expected block error probability for the codes of the ensemble using Di's upper bound, \eqref{eq:bound_Gavg} \cite{CDi2001:Finite}.

As an example, we consider a concatenated scheme where the block code is a binary $(63,57)$ Hamming code. The \acl{CO-WEF} $\wef_i(X)$ of a $(\nc=2^t-1,k=\nc-t)$ Hamming code is known from \cite{Chiaraluce:hamming} and corresponds to
\begin{align}
\wef(x,X) = &\frac {(1+x)^{2^{t-1}-t-1}} {2^t} \times  \Big( 2^t(1-x)^{2^{t-1}-t} (1-xX)^t \\
& - (1-x)^{2^{t-1}} (1+X)^t + (1+x)^{2^{t-1}} (1+X)^t \Big)
\end{align}
where $\wef(x,X)=\sum_i \wef_i(X)x^i$.

Figure~\ref{dist_spectrum} shows the average weight enumerator vs.\ the normalized weight, $\nd= w/\lc$ for the concatenated code for rates $r=1/2$ and $r=1/4$ and the weight enumerator of the precode alone (Hamming). The figure also shows the average weight enumerator of binary linear random generator matrix based ensembles with the same block length and rate. Codes in this ensemble are characterized by having a $k \times n$ generator matrix whose elements are picked uniformly at random in the binary field. Thus the ensemble is equivalent to the fixed-rate \ac{LRFC} ensemble. The average weight enumerator of this ensemble can be found in \cite{barg01:random}. In the figure we can observe how the weight spectrum of the concatenated ensemble is better than that of the binary linear random ensemble, in the sense that for same block length and rate the expected multiplicity of low weight codewords is lower. This will lead to a lower error floor.

Figure~\ref{sim_hamming} shows the upper bounds on the \acf{CER} of the ensemble, as a function of the channel erasure probability $\erasprob$ for different coding rates. The solid lines represent the upper bound on the \ac{CER} in \eqref{eq:bound_Gavg}, and the dashed  and red lines represent respectively the Berlekamp random coding bound \cite{berlekamp:bound},
which is an upper bound on the average block error probability of random codes, and the Singleton bound, which provides the block error probability of \ac{MDS} codes. The markers represent the results of Monte Carlo simulations. In order to obtain average results for the ensemble, the \ac{CER} was averaged over $1000$ different \ac{LRFC} realizations. As expected, the bound in \eqref{eq:bound_Gavg} is very tight in all cases. Results for three different rates are shown in the figure. The highest rate corresponds to the use of the Hamming code alone, and the other two rates are $r=0.8$ and $r=0.5$. While for the Hamming code the performance lies in between the one of random codes and the one of \ac{MDS} codes, as the code rate decreases the performance of the scheme gets closer to the Berlekamp random coding bound, which means that for low rates our scheme performs almost as a random code. However, for high rates the concatenated scheme performs substantially better than a random code, whose performance would be very close to the Berlekamp bound.

\begin{figure}
\begin{center}
\includegraphics[width=\figw,draft=false]{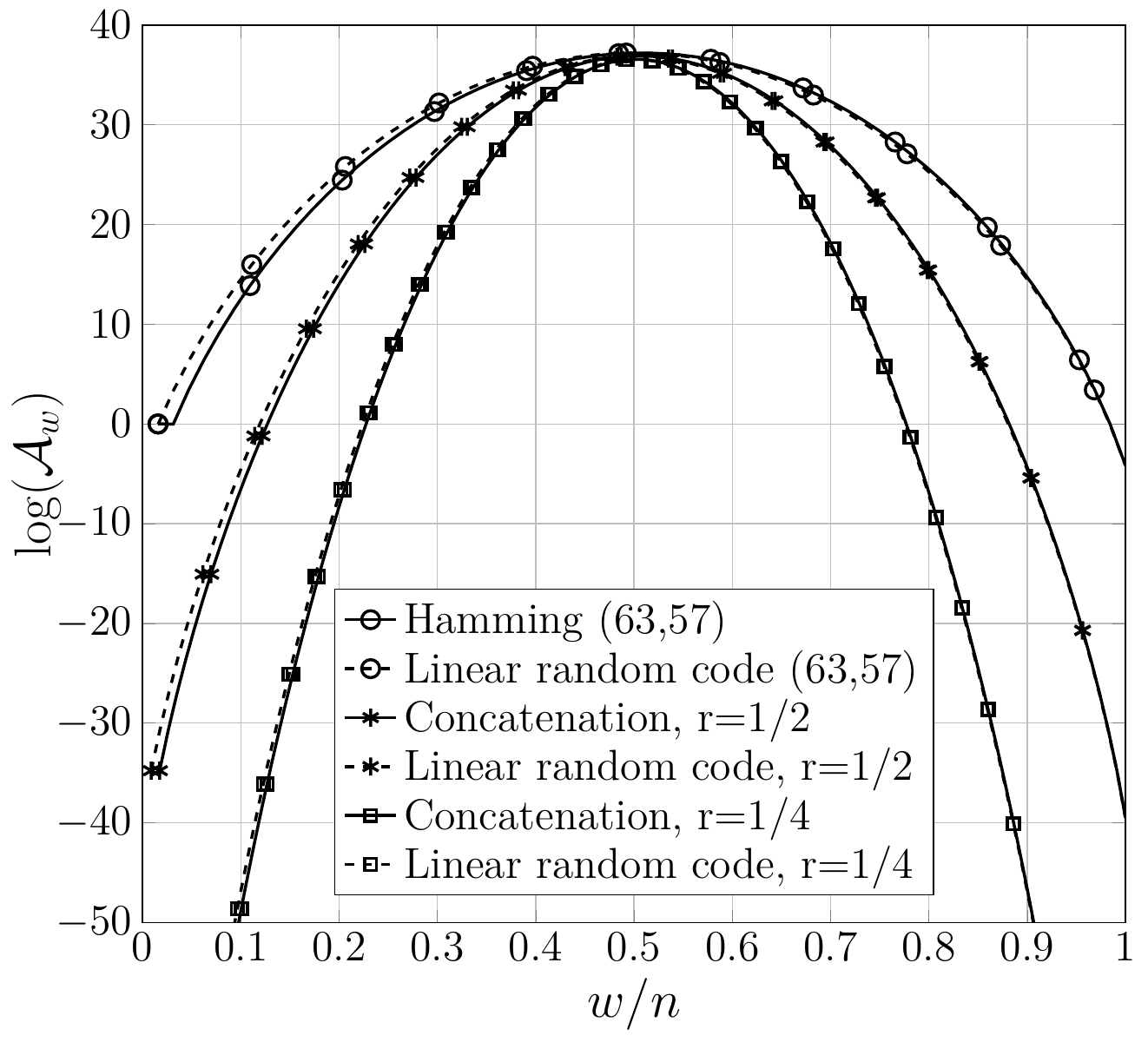}
\centering \caption[$\log (\we_w)$ vs.\ $w/\lc$ for the concatenation of a (63,57) Hamming code with a \ac{LRFC} code in $\mathbb {F}_{2}$]{$\log (\we_w)$ vs.\ $w/\lc$ for the concatenation of a (63,57) Hamming code with a \ac{LRFC} code in $\mathbb {F}_{2}$ and for the concatenated scheme with rates $r=\frac{1}{2}$ and $r=\frac{1}{4}$.}\label{dist_spectrum}
\end{center}
\end{figure}
\begin{figure}
\begin{center}
\includegraphics[width=\figw,draft=false]{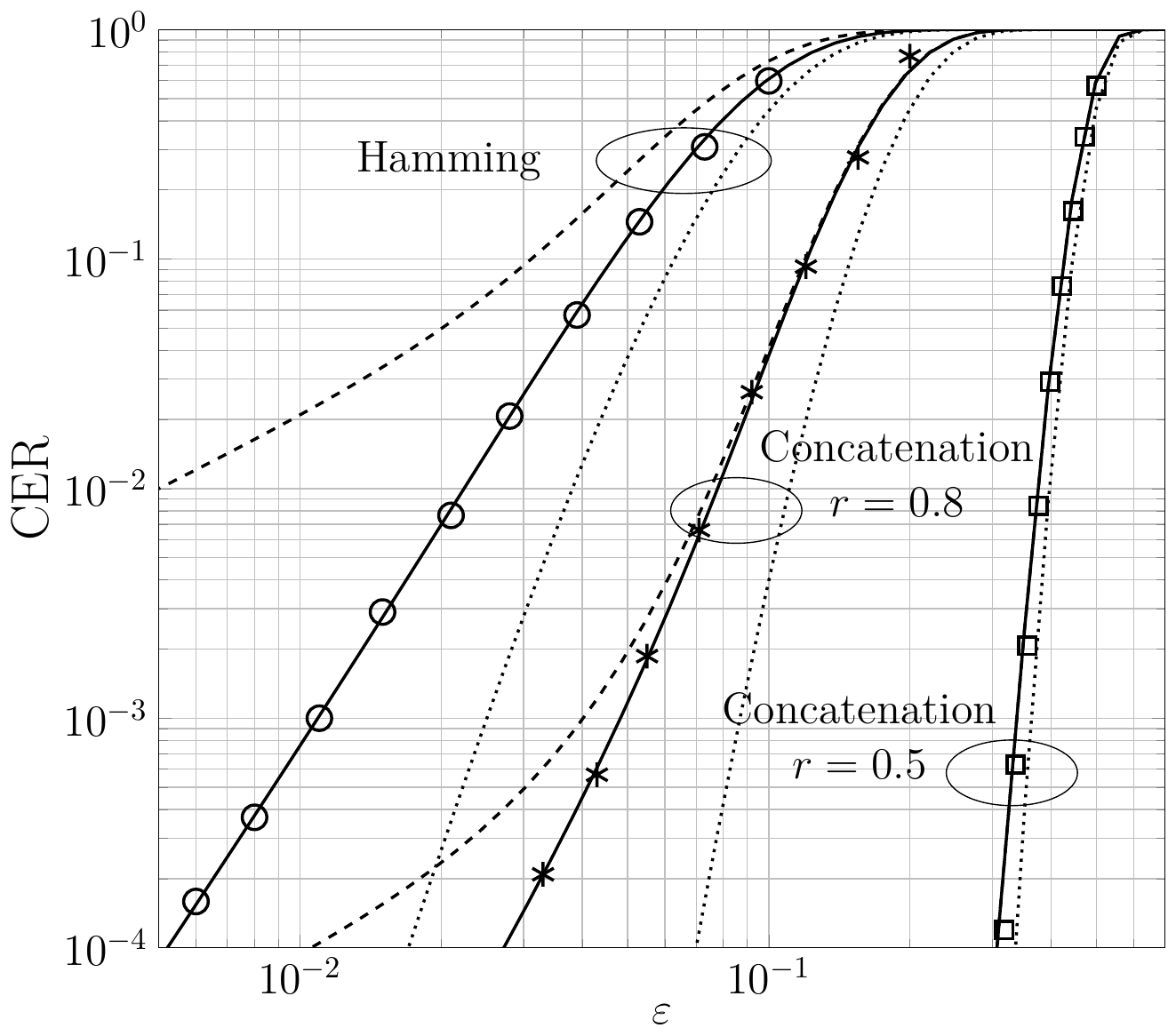}
\centering \caption[CER vs.\ erasure probability $\erasprob$ for the concatenation of a (63,57) Hamming code with a \ac{LRFC} code in $\mathbb {F}_{2}$]{CER vs.\ erasure probability $\erasprob$ for the concatenation of a (63,57) Hamming code with a \ac{LRFC} code in $\mathbb {F}_{2}$. The markers represent the result of Monte Carlo simulations. The solid line represents the upper bound in \cite{CDi2001:Finite}, and the dashed and  dotted lines represent the Berlekamp random coding bound and the Singleton bound respectively.}\label{sim_hamming}
\end{center}
\end{figure}

\section{Numerical Results}\label{sec:results}

In this section we investigate the performance of the concatenated scheme in a reliable multicasting scheme. Let us assume a transmitter wants to deliver a source block (a data file) to a set of $N$ of receivers. We will assume that the erasure channels from the transmitter to the different receivers are independent and have an identical erasure probability $\erasprob$. Furthermore, we assume that the receivers send an acknowledgement to the transmitter whenever they successfully decode the source block though an ideal (error- and delay-free) feedback channel.  After retrieving all the acknowledgments, the transmitter stops encoding additional symbols from the source block.

Let us denote by $\Delta$ the number of symbols in excess with respect to $k$ transmitted by the sender. We refer to $\Delta$ as the transmission overhead. When  $k+\Delta$ symbols have been transmitted, the probability that a specific receiver gathers exactly $m$ symbols is
\begin{align}
\ S\left(\Delta,m\right) = \binom{k+\Delta}{m}(1-\erasprob)^{m}\erasprob^{k+\Delta-m}.
\label{system_prob m}
\end{align}
The probability of decoding failure at the receiver given that the transmitter has sent $k+\Delta$ symbols is hence
\begin{align}
\Pse =& \sum_{m=0}^{k-1}\ S\left(\Delta,m\right)+\\
& + \sum_{m=k}^{k+\Delta}\ S\left(\Delta,m\right) \Pf|(\absoverhead=m-k,\erasprob).
\end{align}
Let us define the error probability in our system, $\Pe$, as he probability that at least one receiver is not able to decode the source block. This probability is given by
\begin{align}
\Pe (N,\Delta,\erasprob) = 1-(1-\Pse)^{N}
\label{system_failure_one user}
\end{align}
Observe that $\Pe(N,\Delta,\erasprob)$ can be easily bounded by means of \eqref{eq:final_bounds}. Following this approach, we compare the performance of the proposed concatenation to that of \acp{LRFC} and to that of an ideal fountain code. Let us recall that for an ideal fountain code the probability of decoding failure is zero whenever $k$ or more output symbols are collected (see Section~\ref{sec:fountain_bounds}).

We consider a system with $N=10^{4}$ receivers and a channel with an erasure probability $\erasprob=0.01$. The performance of \ac{LRFC} codes over $\mathbb {F}_{2}$ and $\mathbb {F}_{16}$ is depicted in Figure~\ref{sim_sender_side} together with that of two concatenated schemes: a concatenation of a $(11,10)$ \ac{SPC} code with a \ac{LRFC} code over $\mathbb {F}_{2}$, and a concatenation of a $(15,10)$ \ac{RS} code and a \ac{LRFC} code over $\mathbb {F}_{16}$.
We can observe how the binary concatenated scheme outperforms the binary \ac{LRFC}. For example, in order to achieve a target probability of error $\Pe = 10^{-4}$ the concatenated scheme needs only $\Delta=20$ overhead symbols whereas the
\ac{LRFC} requires a transmission overhead $\Delta=27$. In the binary case, both the \ac{LRFC} and the concatenated scheme are far from the performance of an ideal fountain code. If we now look at the non binary case, we can observe how the performance gap of the \ac{LRFC} w.r.t an ideal fountain code is much smaller than in the binary case. Furthermore, we can observe how the non-binary concatenated scheme is able to improve the performance of the \ac{LRFC} and almost completely close the performance gap w.r.t. an ideal fountain code.
\begin{figure}
\begin{center}
\includegraphics[width=\figwbigger,draft=false]{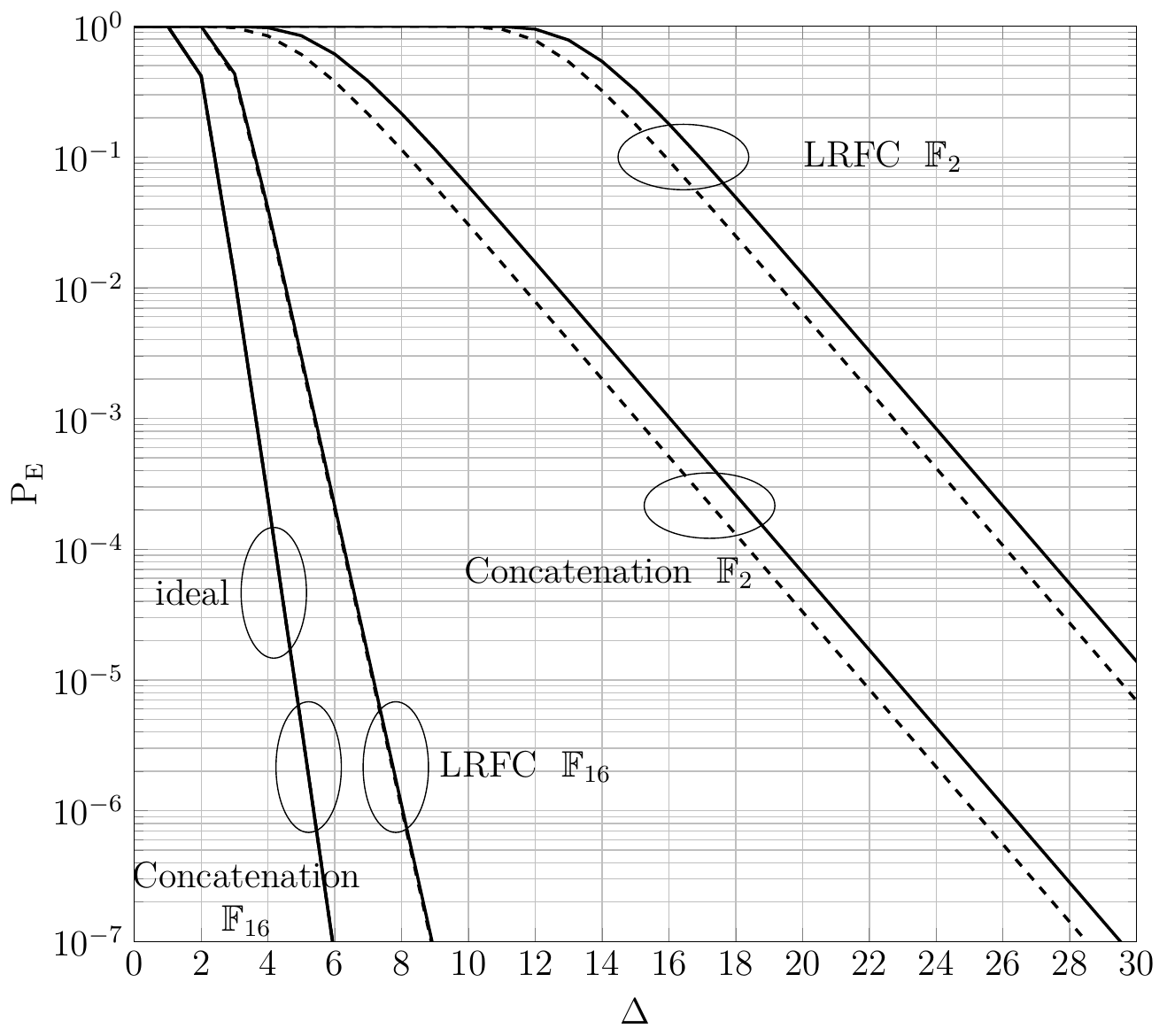}
\centering \caption[$\Pe$ vs.\ transmitter overhead $\Delta$ in a system with $N=10000$ receivers and $\erasprob=0.01$ for different fountain codes.]{ $\Pe$ vs.\ overhead at the transmitter in a system with $N=10000$ receivers and $\erasprob=0.01$. Results are shown for different fountain codes: \ac{LRFC} in $\mathbb {F}_{2}$, \ac{LRFC} in $\mathbb {F}_{16}$, concatenation of a (11,10) \ac{SPC} code with a \ac{LRFC} code in $\mathbb {F}_{2}$, and a concatenation of a {$(15,10)$} \ac{RS} code and a \ac{LRFC} code over $\mathbb {F}_{16}$.}\label{sim_sender_side}
\end{center}
\end{figure}

\section{Summary}\label{sec:conc}
In this chapter a novel coding scheme has been introduced. The scheme consists of a parallel concatenation of a block code with a \acf{LRFC} code, both constructed over the same finite field. The performance of the concatenated coding scheme under \ac{ML} decoding has been analyzed through the derivation of tight bounds on the probability of decoding failure.  This scheme can be seen as a way of turning any block code rateless (or rate flexible), so that additional output symbols can be generated on demand. The proposed scheme is in general only practical \fran{when the code dimension $k$ is small.}

Specially interesting is the case in which the block code is a \ac{MDS} code. In this case, the scheme can provide failure probabilities lower than those of \ac{LRFC} codes by several orders of magnitude, provided that the erasure probabilities of the channel is not too high. The general case in which the block code is not \ac{MDS} has also been analyzed. In this case the scheme has been analyzed in a fixed-rate setting, and it has been shown by means of examples how the concatenated scheme outperforms \acp{LRFC}. Given the fact that Raptor codes are essentially random codes, their performance would be at best as good as that of \acp{LRFC}. Thus, the proposed scheme also outperforms Raptor codes in terms of decoding failure probability. However, \fran{one should remark that the proposed  scheme is only practical when the code dimension $k$ is small}, since its decoding complexity is rather high.


The focus in this chapter has been exclusively on the performance under \ac{ML} decoding. However, it is possible to exploit the structure of the block code (precode) in order to decrease the decoding complexity of the scheme. For example, in \cite{lazaro2013parallel} an enhanced decoding algorithm  was proposed for the case in which the precode is a \acf{GRS} code whose generator matrix is in Vandermonde form.



%% file: Conclusions/conclusions.tex
\chapter{Conclusion}\label{chap:conclusions}
\ifpdf
    \graphicspath{{Conclusions/ConclusionsFigs/PNG/}{Conclusions/ConclusionsFigs/PDF/}{Conclusions/ConclusionsFigs/}}
\else
    \graphicspath{{Conclusions/ConclusionsFigs/EPS/}{Conclusions/ConclusionsFigs/}}
\fi

\def\baselinestretch{1.66}

In this dissertation we have investigated fountain codes under \acf{ML} \fran{erasure} decoding. In particular three types of fountain codes have been considered, \ac{LT} codes, Raptor codes and a new class of parallel concatenated fountain codes.

Regarding \ac{LT} codes, the main contribution of this thesis is a detailed analysis of a particular \ac{ML} decoding algorithm, inactivation decoding. More concretely, the focus has been on the decoding complexity of \ac{LT} codes under inactivation decoding in terms of the number of inactivations. Given an \ac{LT} degree distribution and $k$, the code dimension or equivalently the number of input symbols,  dynamic programming approaches have been used to derive the expected number of inactivations and its probability distribution. Furthermore, a low complexity algorithm has been proposed to estimate the number of inactivations. Additionally, we have shown by means of an example how the analysis of \ac{LT} codes presented can be used to numerically design \ac{LT} codes optimized for inactivation decoding.

Raptor codes have also been considered in this thesis. First, upper bounds to the probability of decoding failure of Raptor codes have been derived, using the weight enumerator of the outer code, or the average weight enumerator when the outer code is drawn at random from an ensemble. These bounds show that Raptor codes can be analyzed similarly to fixed-rate block codes. Furthermore, we have shown how the complexity of Raptor codes under inactivation decoding can be approximated introducing the concept of a surrogate \ac{LT} code. Moreover, we have shown by means of an example how the results obtained for Raptor codes can be used to design finite length Raptor codes with a good tradeoff between probability of decoding failure and complexity under inactivation decoding.

Additionally, an analysis of the distance spectrum of ensembles of fixed-rate Raptor codes has been presented, for the case in which the outer code is picked from the linear random ensemble. This ensemble of Raptor codes resembles standard R10 Raptor codes as a first order approximation. For this ensemble, the average weight enumerator and its growth rate have been derived. Furthermore, sufficient and necessary conditions for the ensemble to have a minimum distance growing linearly with the block length (positive typical minimum distance) have been derived. By means of simulations, it has been shown how the results obtained for the ensemble of Raptor codes studied can be extrapolated to Raptor codes using the standard R10 outer code as a first approximation.

The last contribution of the dissertation is the introduction of a novel class of fountain codes, that consists of a parallel concatenation of a block code with a \acf{LRFC}. This scheme is specially interesting when the block code is a \acf{MDS} code.  In this case, the scheme's performance can be tightly upper and lower bounded by means of very simple formulae. Furthermore, the scheme can provide failure probabilities lower than those of \ac{LRFC} codes by several orders of magnitude, provided that the erasure probabilities of the channel is not too high, \fran{which is usually the case in most of the applications of erasure codes}.  Thus, in this setting the proposed scheme outperforms Raptor codes in terms of probability of decoding failure. However, this novel scheme is in general only practical for when the code dimension $k$ is small due to its high decoding complexity.




%% file: Appendix1/appendix1.tex
\chapter{Comparison of Inactivation Strategies} \label{app:practical}
\ifpdf
    \graphicspath{{Appendix1/Appendix1/PNG/}{Appendix1/Appendix1Figs/PDF/}{Chapter4/Chapter4Figs/}}
\else
    \graphicspath{{Appendix1/Appendix1Figs/EPS/}{Chapter4/Appendix1Figs/EPS}}
\fi
In this appendix the performance of the different inactivation techniques presented in Section~\ref{sec:inact_strat} are compared by means of simulations.
More concretely, we simulated a (non-systematic) R10 Raptor code for source block sizes ranging from $k=128$ to $k=8192$ for different absolute overheads $\absoverhead$. For each different value of $\absoverhead$, 300 decodings were carried out and the average number of inactivations was obtained for random inactivation, maximum reduced degree inactivation, maximum accumulated degree inactivation and maximum component inactivation.

The simulation results can be observed in Figures~\ref{fig:inact_strategies_128} to \ref{fig:inact_strategies_8192}. Looking at these figures it can be observed how random inactivation leads to the largest number of inactivations, followed by  maximum reduced degree inactivation, then maximum accumulated degree and finally maximum component inactivation, that leads to the least inactivations (lowest decoding complexity). It is remarkable that this ordering holds for all values of $k$ and $\absoverhead$.

\begin{figure}
        \begin{center}
        \includegraphics[width=\figw]{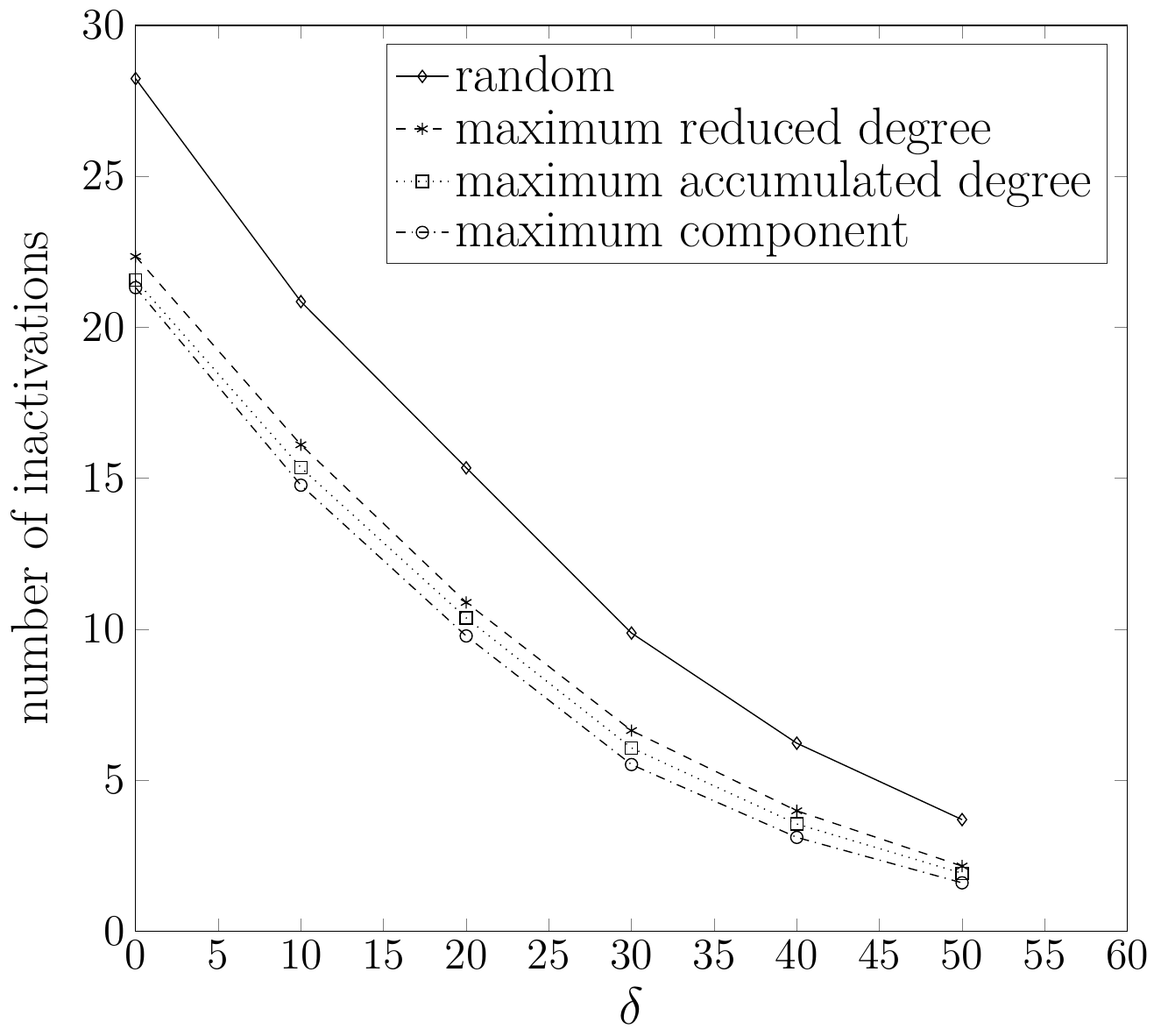}
        \centering
        \caption[Number of inactivations vs.\ $\absoverhead$ for a R10 Raptor code and $k=128$]{Number of inactivations vs.\ absolute receiver overhead $\absoverhead$ for a R10 Raptor code and $k=128$}
        \label{fig:inact_strategies_128}
        \end{center}
\end{figure}
\begin{figure}
        \begin{center}
        \includegraphics[width=\figw]{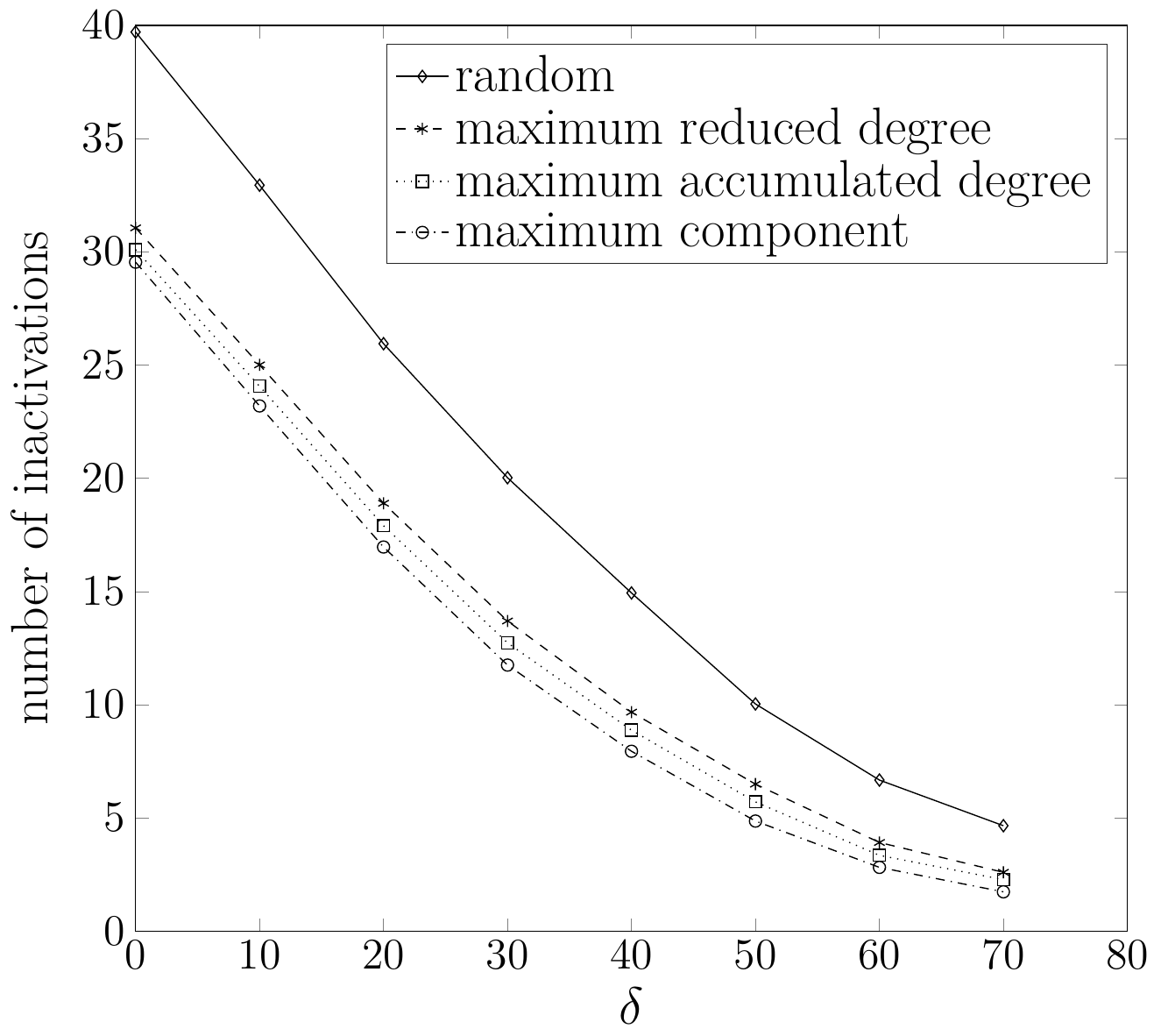}
        \centering
        \caption[Number of inactivations vs.\ $\absoverhead$ for a R10 Raptor code and $k=256$]{Number of inactivations vs.\ absolute receiver overhead $\absoverhead$ for a R10 Raptor code and $k=256$}
        \label{fig:inact_strategies_256}
        \end{center}
\end{figure}
\begin{figure}
        \begin{center}
        \includegraphics[width=\figw]{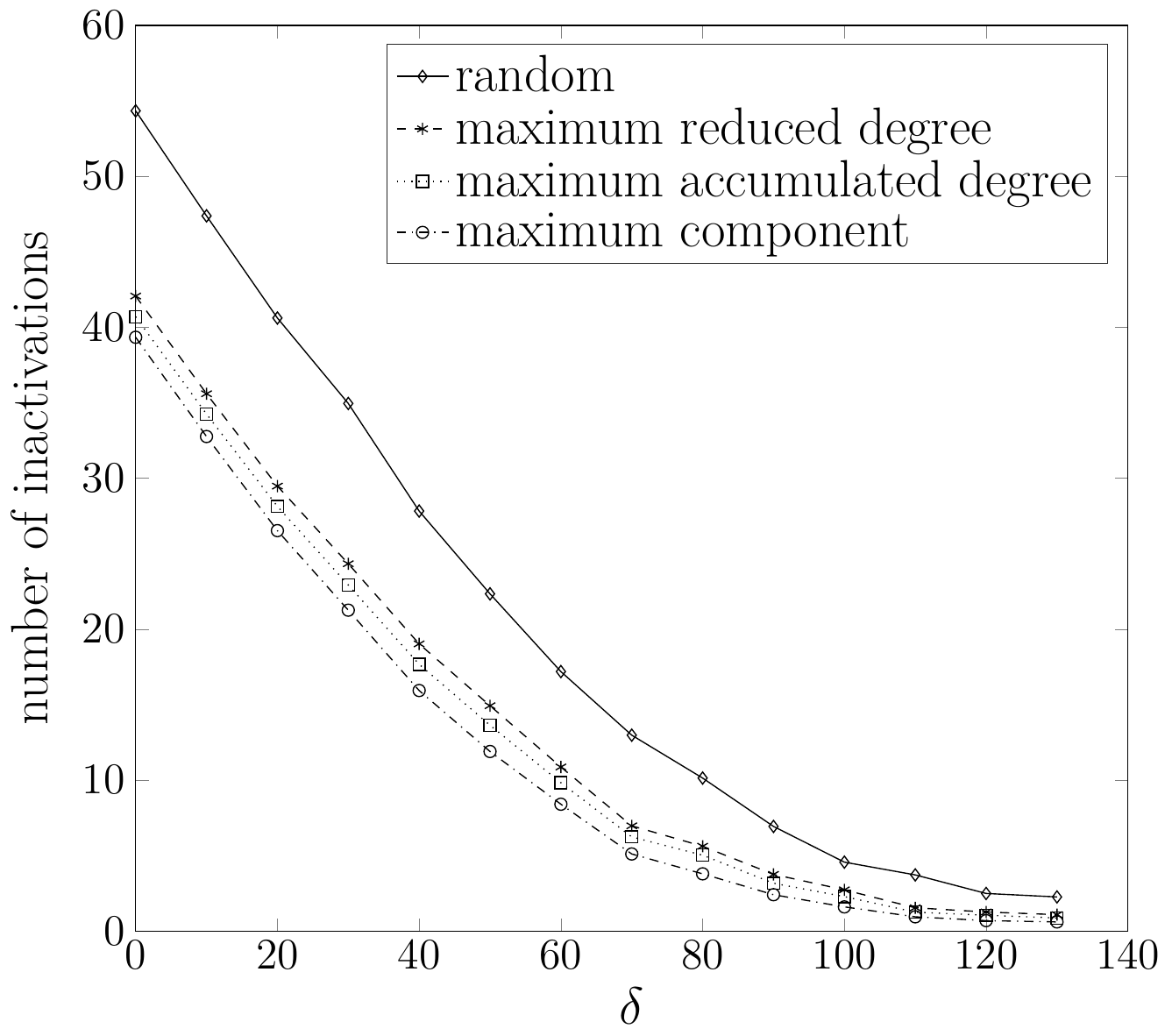}
        \centering
        \caption[Number of inactivations vs.\ $\absoverhead$ for a R10 Raptor code and $k=512$]{Number of inactivations vs.\ absolute receiver overhead $\absoverhead$ for a R10 Raptor code and $k=512$}
        \label{fig:inact_strategies_512}
        \end{center}
\end{figure}
\begin{figure}
        \begin{center}
        \includegraphics[width=\figw]{inact_strategies_k_1024}
        \centering
        \caption[Number of inactivations vs.\ $\absoverhead$ for a R10 Raptor code and $k=1024$]{Number of inactivations vs.\ absolute receiver overhead $\absoverhead$ for a R10 Raptor code and $k=1024$}
        \label{fig:inact_strategies_1024}
        \end{center}
\end{figure}
\begin{figure}
        \begin{center}
        \includegraphics[width=\figw]{inact_strategies_k_2048}
        \centering
        \caption[Number of inactivations vs.\ $\absoverhead$ for a R10 Raptor code and $k=2048$]{Number of inactivations vs.\ absolute receiver overhead $\absoverhead$ for a R10 Raptor code and $k=2048$}
        \label{fig:inact_strategies_2048}
        \end{center}
\end{figure}
\begin{figure}
        \begin{center}
        \includegraphics[width=\figw]{inact_strategies_k_4096}
        \centering
        \caption[Number of inactivations vs.\  $\absoverhead$ for a R10 Raptor code and $k=4096$]{Number of inactivations vs.\ absolute receiver overhead $\absoverhead$ for a R10 Raptor code and $k=4096$}
        \label{fig:inact_strategies_4096}
        \end{center}
\end{figure}
\begin{figure}[t]
        \begin{center}
        \includegraphics[width=\figw]{inact_strategies_k_8192}
        \centering
        \caption[Number of inactivations vs.\ $\absoverhead$ for a R10 Raptor code and $k=8192$]{Number of inactivations vs.\ absolute receiver overhead $\absoverhead$ for a R10 Raptor code and $k=8192$}
        \label{fig:inact_strategies_8192}
        \end{center}
\end{figure}

%% file: Appendix2/appendix2.tex
\chapter{Omitted Proofs} \label{app:proofs}
In this Appendix we provide some results that were omitted in the body of the thesis but are necessary for the proofs in Chapters~\ref{chap:Raptor} and \ref{chap:Raptor_fixed_rate}.

\section{Proof of Theorem~\ref{theorem:rateless}}
The following lemma is used in the proof of Theorem~\ref{theorem:rateless}.

\begin{lemma}\label{lemma:galois}
Let $X_1$, $X_2$ ... $X_l$  be i.i.d random variables with uniform distribution over $\mathbb F_{2^m} \backslash \{0\}$. Then
\[
\Pr \{X_1 + X_2+ \hdots + X_l = 0 \}= \frac{1}{q} \left( 1 + \frac{(-1)^i}{(q-1)^{i-1}}\right).
\]
\end{lemma}

\begin{proof}
The proof starts observing that the additive group of $\mathbb F_{2^m}$ is isomorphic to the vector space $\mathbb Z_2^m$. Thus, we consider
$X_1$, $X_2$ ... $X_l$ to be i.i.d random variables with uniform distribution over the vector space $\mathbb Z_2^m \backslash \{0\}$.

Let us introduce the auxiliary random variable
\[
W := X_1 + X_2+ \hdots + X_l
\]
Denote by $P_W(w)$ and by $P_X(x)$ the probability mass function of $W$ and $X_i$, with
\[
P_X(x) =
\begin{cases}
 0 & \text{if } x=0 \\
 \frac{1}{q-1} & \text{otherwise.}
 \end{cases}
\]
We have that
\[
P_W(w) = P_X(x) \ast  P_X(x) \ast \hdots \ast P_X(x)
\]
which can be re-stated via the m-dimensions 2-points DFT $\msr{J} \{\cdot\}$ as
\[
\msr J \{ P_W(w) \} =  \left( \msr J \{ P_X(x) \} \right)^l.
\]
We have that
\[
\hat P_X(t) := \msr J \{ P_X(x) \}=
\begin{cases}
 1 & \text{if } t=0 \\
 \frac{-1}{q-1} & \text{otherwise}
 \end{cases}
\]
Thus,
\[
\hat P_W(t) := \msr J \{ P_W(w) \}=
\begin{cases}
 1 & \text{if } t=0 \\
 \frac{(-1)^l}{(q-1)^l} & \text{otherwise.}
 \end{cases}
\]
We are interested in  $P_W(0)$ whose expression corresponds to
\[
 P_W(0) = \frac{1}{q} \sum_t \hat P_W(t) = \frac{1}{q} + \frac{1}{q} (q-1)  \frac{(-1)^l}{(q-1)^l}
\]
from which the statement follows.
\end{proof}

The result in this lemma can also be found in \cite{schotsch:2013}. However, the proof in \cite{schotsch:2013} uses a different approach based on a known result on the number of closed walks of length $l$ in a complete graph of size $q$ from a fixed but arbitrary vertex back to itself.

\section{Proof of Theorem \ref{theorem_inner}}\label{sec:proof_inversion}

We first prove that for all $(\ri,\ro)$ pairs in $\region$ we have a positive normalized typical minimum distance. Then, we prove that this is not possible for any other $(\ri,\ro)$ pair.

\subsection{Proof of Sufficiency}
A sufficient condition for a positive normalized typical minimum distance is
\begin{equation}
\lim_{\nd \to 0^+} G(\nd) < 0.
\end{equation}
From Theorem \ref{theorem:growth_rate} this is equivalent to
\begin{equation}
\ri  (1-\ro) >  \lim_{\nd \to 0^+} \max_{ \nl \in \mathscr D_{\nl}} \f(\nd, \nl).
\label{eq:lim_G}
\end{equation}
As we did  in Lemma~\ref{lemma:growth_rate_derivative} and Lemma~\ref{corollary:der}, let us  use the notation $\np(\nl)= \npnl$ to emphasize the dependence on $\nl$. We shall now show that
\begin{align} \label{eq:lim_max}
\lim_{\nd \to 0^+} \max_{ \nl \in \mathscr D_{\nl}} \f(\nd, \nl) &= \max_{ \nl \in \mathscr D_{\nl}} \lim_{\nd \to 0^+} \f(\nd, \nl) 
  = \max_{ \nl \in \mathscr D_{\nl}}  \left[ \ri \Hb(\nl) + \log_2 \left(1 - \np(\nl) \right)  \right]
\end{align}
that is, the maximization with respect to $\lambda$ and the limit as $\nd \rightarrow 0^+$ can be inverted, so that the region $\region$ in \eqref{eq:theorem_region} is obtained.

This fact is proved by simply showing that
\begin{align}
\lim_{\nd \rightarrow 0^+} \fmax(\nd) =  \fmax(0),
\end{align}
that is, the function $\fmax(\nd)=\max_{ \nl \in \mathscr D_{\nl}} \f(\nd,\nl)$ is right-continuous at $\nd = 0$. For this purpose it suffices to show
\begin{align}\label{eq:max_a_b}
\fmax(\nd) = \max_{\nl \in (a, b)} \f(\nd, \nl)
\end{align}
where $(a,b)$ is an interval independent of $\nd \in [0,\frac{1}{2})$ such that the function
\begin{align}
\log_2 \np(\nl)-\log_2(1-\np(\nl))
\end{align}
 is bounded over it, i.e.,
$$
\sup_{\nl \in (a,b)} \left| \log_2 \npnl - \log_2(1-\npnl) \right| = K \, .
$$
Under these conditions we have uniform convergence of $\f(\nd,\nl)$ to $\f(0,\nl)$ in the interval $(a,b)$ as $\nd \rightarrow 0^+$, namely,
\begin{align}\label{eq:uniform_convergence}
\f(0,\nl) - K \nd \leq \f(\nd,\nl) \leq \f(0,\nl) + K \nd, \qquad \forall \lambda \textrm{ s.t. } a < \nl < b \, .
\end{align}
The second inequality in \eqref{eq:uniform_convergence} implies $\fmax(\nd) \leq \fmax(0) + K \nd$. Furthermore, denoting by $\hat{\nl} \in (a,b)$ the maximizing $\nl$, we have
$$
\fmax(0) - K \nd = \f(0,\hat{\nl}) - K \nd \leq \f(\nd,\hat{\nl})
$$
which implies $\fmax(0) - K \nd \leq \fmax(\nd)$. Hence, we have
$$
\fmax(0) - K \nd \leq \fmax(\nd) \leq \fmax(0) + K \nd
$$
that yields $\lim_{\nd \rightarrow 0^+} \fmax(\nd) = \fmax(0)$, as desired.

Next, we shall prove \eqref{eq:max_a_b}. We start by observing that
in the case $\Omega_j=0$ for all even $j$ (in this case $\np(\nl)$ is strictly increasing) by direct computation we have $\partial\, \f(\nd,\nl) / \partial \nl < 0$ for all $0 \leq \nd < 1/2$ and for all $1/2 \leq \nl < 1$. Thus, in this case we can take $b=1/2$. In all of the other cases there exists $\xi$ such that $\np(\nl) \leq \xi < 1$ for all $0 < \nl < 1$ and we can take $b=1$. We prove the existence of $0 < a < 1/2$ (independent of $0 \leq \nd < 1/2$) such that the maximum is not taken for all $0 < \nl \leq a$  as follows. Denoting $c = \log_2 e$ and $\np'(\nl)=\mathrm d \np(\nl) / \mathrm d \nl$, we have
\begin{align*}
\frac{\partial\, \f(\nd,\nl)}{\partial \nl} = \ri \log_2(1-\nl) - \ri \log_2 \nl + c\, \nd \, \frac{\np'(\nl)}{\np(\nl)} - c\,(1-\nd) \frac{\np'(\nl)}{1-\np(\nl)} \, .
\end{align*}
Given that $0 < \np'(\nl) < +\infty$ for all $0 < \nl \leq 1/2$ and since
\begin{align*}
\lim_{\nl \rightarrow 0^+} \ri (1-\np(\nl)) (\log_2 (1-\nl) - \log_2 \nl) = + \infty \, ,
\end{align*}
there exists $a > 0$ such that
\begin{align*}
\ri(1-\np(\nl)) (\log_2 (1-\nl) -\log_2 \nl) > c\, \np'(\nl), \qquad \textrm{for all } 0 < \nl < a \, .
\end{align*}
This latter inequality implies
\begin{align*}
\ri(1-\np(\nl)) (\log_2 (1-\nl) -\log_2 \nl) > c\, \np'(\nl) - \nd \frac{c\, \np'(\nl)}{\np(\nl)}, \qquad \textrm{for all } 0 < \nl < a
\end{align*}
uniformly with respect to $\nd \in [0,1/2)$, that is equivalent to $\partial\, \f(\nd,\nl) /\partial \nl > 0 $ for all $0 < \nl < a$, independently of $\nd \in [0,1/2)$. Hence, the maximum cannot be taken between $0$ and $a$, with $a$ independent of $\nd \in [0,1/2)$.

\subsection{Proof of Necessity} \label{sec:necessity}
So far it has been proved that the condition on $(\ri,\ro)$ expressed by Theorem \ref{theorem_inner} is sufficient to have a positive normalized typical minimum distance. Now we need to show that this condition is also necessary.  Concretely, we need to prove that for the ensemble $\msr{C}_{\infty}(\oensemble,\Omega, \ri, \ro)$ all rate pairs $(\ri,\ro)$ such that $\lim_{\nd \rightarrow 0^+}G(\nd)=0$ (i.e., rate pairs on the boundary $\region$),  the derivative of the growth rate at $0$ is positive, $\lim_{\nd \rightarrow 0^+} G'(\nd) > 0$.

According to Lemma \ref{lemma:growth_rate_derivative} the expression of the derivative of the growth rate,  $G'(\nd)$ corresponds to
\[
G'(\nd) = \log_2 \frac{1-\nd}{\nd} + \log_2 \frac{\np(\nlo)}{1-\np(\nlo)}  \, .
\]
Therefore, since $G'(\nd)$ is the sum of two terms the first of which diverges to $+\infty$ as $\nd \rightarrow 0^+$, a necessary condition for the derivative to be negative is that the second term diverges to $-\infty$, i.e., $\lim_{\nd \rightarrow 0^+}\np(\nlo)=0$. This case is analyzed in the following lemma.
\begin{lemma}\label{lemma:limit_p}
If $\np (\nl)=0$ then $\nl \in \{ 0, 1 \}$ in case the LT distribution $\Omega$ is such that $\Omega_j=0$ for all odd $j$, and $\nl=0$ for any other LT distribution $\Omega$.
\end{lemma}
\begin{proof}
Let us recall that $\np(\nl)$ is the probability that the LT encoder picks an odd number of nonzero intermediate bits (with replacement) given that the intermediate codeword has Hamming weight $\nl h$. If $\Omega_j > 0$ for at least one odd $j$, then the only case in which a zero \ac{LT} encoded bit is generated with probability $1$ is the one in which the intermediate word is the all-zero sequence. If $\Omega_j=0$ for all odd $j$, there is also another case in which a nonzero bit is output by the LT encoder with probability $1$, i.e., the case in which the intermediate word is the all-one word.
\end{proof}

\medskip
Let us consider now a pair $(\ri,\ro)$ such that $\lim_{\nd \rightarrow 0^+}G(\nd)=0$. A fixed-rate Raptor code ensemble corresponding to this pair, has a positive typical minimum distance if and only if $\lim_{\nd \rightarrow 0^+} G'(\nd)<0$.
By Lemma~\ref{lemma:limit_p} this implies $\lim_{\nd \rightarrow 0^+} \nlo(\nd)=0$ when $\Omega_j>0$ for at least one odd $j$. It implies either $\lim_{\nd \rightarrow 0^+} \nlo(\nd)=0$ or $\lim_{\nd \rightarrow 0^+} \nlo(\nd)=1$ otherwise. That $\nlo(\nd)$ cannot converge to $0$ follows from the proof of sufficiency (as shown, the maximum for $\nd \in [0,1/2)$ is taken for $\lambda > a >0$). In order to complete the proof we now show that, in the case where $\Omega_j=0$ for all odd $j$, assuming $\lim_{\nd \rightarrow 0^+} \nlo(\nd)=1$ leads to a contradiction.

If $\Omega_j=0$ for all odd $j$, a Taylor series for $\np(\nl)$ around $\nl=1$ is $\np(\nl) = \bar \Omega (1-\nl) + o(\nl)$. Assuming $\lim_{\nd \rightarrow 0^+} \nlo(\nd)=1$, we consider the left-hand side of \eqref{eq:critical_point} and calculate its limit as $\nd \rightarrow 0^+$. We obtain
\begin{align}
& \lim_{\nd\rightarrow 0^+} \frac{\partial \f}{\partial \nl} (\nd, \nlo) 
= \ri \lim_{\nlo\rightarrow 1^-} \log_2 \frac{1-\nlo}{\nlo} + \lim_{\nd \rightarrow 0^+} \! \left( \frac{\nd}{\log 2} \, \frac{\np'(\nlo)}{\np(\nlo)} - \frac{1-\nd}{\log 2} \, \frac{\np'(\nlo)}{1-\np(\nlo)} \right) \\
&= \ri \lim_{\nlo\rightarrow 1^-} \log_2 \frac{1-\nlo}{\nlo} + \frac{1}{\log 2} \lim_{\nd \rightarrow 0^+} \frac{\np'(\nlo)(\nd-\np(\nlo))}{\np(\nlo)(1-\np(\nlo))} \\
&= \ri \lim_{\nlo\rightarrow 1^-} \log_2 \frac{1-\nlo}{\nlo} + \frac{1}{\log 2} \lim_{\nd \rightarrow 0^+} \frac{\bar \Omega (1- \nlo) - \nd}{1-\nlo}
\end{align}
where the last equality follows from the above-stated Taylor series expansion. According to \eqref{eq:critical_point}, the last expression must be equal to zero. This constraint requires the second limit to diverge to $+\infty$ (as the first limit diverges to $-\infty$). However, this cannot be fulfilled in any case when $\nd$ converges to zero and $\nlo$ to one. Actually, using standard Landau notation, when $1-\nlo = \Theta(\nd)$ or $\nd = \mathrm o(1-\nlo)$ the second limit converges, while when $1-\nlo = \mathrm o(\nd)$ it diverges to $-\infty$.

\section{Proof of Theorem \ref{pro:outer}} \label{sec:proof_outer}
In this proof we derive first a lower bound for $G(\nd)$ and then evaluate it for $\nd \to 0^+$. To obtain a lower bound for $G(\nd)$ we
first derive a lower bound for $\we_{\nd}$. Observing \eqref{eq:we_serial} it can be seen how $\we_{\nd}$ is obtained as a summation over all possible intermediate Hamming weights. A lower bound to $\we_{\nd}$  can be obtained by limiting the summation to the term $\nls=1-\ro$ yielding to
\begin{align}
\we_{\nd n} &\geq  \frac{\weo_{\nls h} \wei_{\nls h,\nd n}}{ \binom {h} {\nls h}} 
 = \weo_{\nls h} \Q_{\nd n,\nls h}\label{eq:appC:truncation}
\end{align}
where we have introduced
\[
\Q_{\nd n,\nl h}:=\frac{\wei_{\nl h,\nd n}}{ \binom {h} {\nl h}}
\]
that represents the probability that the inner encoder outputs a codeword with Hamming weight $\nd n$ given that the encoder input has weight $\nl h$.

Hence, we can now write
\begin{align}
G(\nd) &\geq \lim_{n \to \infty}  \frac{1}{n} \log_2 \weo_{\nls h} \Q_{\nd n,\nls h} 
 =  \lim_{n \to \infty}  \frac{1}{n} \log_2 \weo_{\nls h} + \lim_{n \to \infty}  \frac{1}{n} \log_2 \Q_{\nd n,\nls h} \\
&= \ri \left(\Hb(\nls) - (1-\ro) \right) + \lim_{n \to \infty}  \frac{1}{n} \log_2 \Q_{\nd n,\nls h}
\label{eq:app_upper_G}
\end{align}

We will now lower bound $\lim_{\nd \to 0^+} \Q_{\nd n,\nl h}$. We denote by

\[
q_{j,\l}:=\Pr\{X_i=0|\hw(\vecV)=\l,\deg(X_i)=j\}.
\]
Note that $q_{j,\l}=1-\pjl$. We have that

\begin{align}
\lim_{\nd \to 0^+} \Q_{\nd n,\nl h} & = \left(\sum_j \Omega_j q_{j,\nl h}\right)^n 
 \geq  \left(\sum_j \Omega_j \underline{q}_{j,\nl h}\right)^n
\end{align}
where $\underline{q}_{j,\l}\leq q_{j,\l}$ is the probability that the $j$ intermediate symbols selected to encoder $X_i$ are all zero.
For large $h$, we have
\[
\underline{q}_{j,\l}=\left( 1- \frac{\l}{h}\right)^j.
\]
Denoting by $\underline{q}_{\l}=\sum_j \Omega_j \underline{q}_{j,\l}$, by Jensen's inequality we have
\[
\underline{q}_{\l} \geq \left( 1- \frac{\l}{h}\right)^{\avgd}.
\]
Thus, we have that
\begin{equation}
\lim_{\nd \to 0^+} \Q_{\nd n,\nl h} \geq  \left( 1- \nl\right)^{n\avgd}. \label{eq:appC:Qbound}
\end{equation}

Replacing \eqref{eq:appC:Qbound} in \eqref{eq:app_upper_G} and recalling that $h=n\ri$ we obtain
\begin{align}
G(\nd) &\geq  \ri \left(\Hb(\nls) - (1-\ro) \right) + \lim_{n \to \infty}  \frac{1}{n} \log_2 \left( 1- \nls\right)^{n\avgd} \\
&=  \ri \left(\Hb(\nls) - (1-\ro) \right) + \avgd \log_2 \left( 1- \nls\right) \\
& = \ri \left(\Hb(1-\ro) - (1-\ro) \right) + \avgd \log_2 \ro.
\end{align}
By imposing $G(\nd)=0$ we obtain:
\[
\phi(\ro)= \frac{\bar \Omega \log_2 (1/\ro)}{\Hb(1-\ro) -(1-\ro)} \, .
\]
This expression is only valid when the denominator is negative, that is,  for $1>\ro>\ro^*$, being $\ro^*$ the only root of the denominator in $\ro \in (0,1)$, whose approximate numerical value is $\ro^* \approx 0.22709$.



%% file: Curriculum_Vitae/cv.tex
\chapter{Curriculum Vitae}
\begin{tabularx}{450pt}{lX}
  Last name: & L\'azaro Blasco \\
  First name:  & Francisco \\
  Nationality: & Spanish \\
  Date of birth: & 24.03.1983 \\
  Place of birth: & Zaragoza, Spain \\
  ~ & ~ \\
  09.1989 - 06.1997  &  Primary school in Alca\~niz, Teruel, Spain\\
  ~ & ~ \\
  09.1997 - 06.1999  &  Secondary school in Alca\~niz, Teruel, Spain\\
  ~ & ~ \\
  09.1999 - 06.2001  &  High school in Alca\~niz, Teruel, Spain\\
  ~ & ~ \\
  09.2001 - 12.2006  &  Studies of \emph{Telecommunication Engineering} at the University of  Zaragoza, Spain \newline Degree: Ingeniero Superior de Telecomunicaciones\\
  ~ & ~ \\
  11.2006 - 03.2007  &  Intern at Siemens Networks in Munich, Germany\\
  ~ & ~ \\
  04.2007 - 04.2008  &  Test Engineer at Rohde \& Schwarz in Munich, Germany\\
  ~ & ~ \\
  07.2008 - present  &  Scientific researcher in satellite and space communications at German Aerospace Center (DLR) in Oberpfaffenhofen, Germany\\
\end{tabularx}

%% file: thesis.bbl
\begin{thebibliography}{10}

\bibitem{Shannon1948}
C.~E. Shannon, ``A mathematical theory of communication,'' {\em Bell System
  Tech. J.}, vol.~27, pp.~379--423, 623--656, 1948.

\bibitem{hamming:1950}
R.~W. Hamming, ``Error detecting and error correcting codes,'' {\em Bell System
  Tech. J.}, vol.~29, no.~2, pp.~147--160, 1950.

\bibitem{golay:1949}
M.~J. Golay, ``Notes on digital coding,'' 1949.

\bibitem{muller:1954}
D.~E. Muller, ``Application of boolean algebra to switching circuit design and
  to error detection,'' {\em Electronic Computers, Transactions of the IRE
  Professional Group on}, no.~3, pp.~6--12, 1954.

\bibitem{reed:1954}
I.~Reed, ``A class of multiple-error-correcting codes and the decoding
  scheme,'' {\em Information Theory, Transactions of the IRE Professional Group
  on}, vol.~4, no.~4, pp.~38--49, 1954.

\bibitem{hocquenghem:1959}
A.~Hocquenghem, ``Codes correcteurs d’erreurs,'' {\em Chiffres (paris)},
  vol.~2, no.~147-156, p.~116, 1959.

\bibitem{bose:1960}
R.~C. Bose and D.~K. Ray-Chaudhuri, ``On a class of error correcting binary
  group codes,'' {\em Information and control}, vol.~3, no.~1, pp.~68--79,
  1960.

\bibitem{reed:RS}
I.~S. Reed and G.~Solomon, ``Polynomial codes over certain finite fields,''
  {\em Journal of the Society for Industrial \& Applied Mathematics}, vol.~8,
  no.~2, pp.~300--304, 1960.

\bibitem{Costello:history}
D.~J. Costello and G.~D. Forney, ``Channel coding: The road to channel
  capacity,'' {\em Proc. IEEE}, vol.~95, pp.~1150--1177, June 2007.

\bibitem{Elias55:2noisy}
P.~Elias, ``Coding for two noisy channels,'' in {\em In Proc. of Inf. Theory:
  Third London Symp.}, pp.~61--74, London: Butterworth Scientific, Ed. C.
  Cherry, 1955.

\bibitem{viterbi1967error}
A.~J. Viterbi, ``Error bounds for convolutional codes and an asymptotically
  optimum decoding algorithm,'' {\em {IEEE} Trans. Inform. Theory}, vol.~13,
  no.~2, pp.~260--269, 1967.

\bibitem{bahl1974optimal}
L.~Bahl, J.~Cocke, F.~Jelinek, and J.~Raviv, ``Optimal decoding of linear codes
  for minimizing symbol error rate,'' {\em {IEEE} Trans. Inform. Theory},
  vol.~20, no.~2, pp.~284--287, 1974.

\bibitem{forney1966concatenated}
G.~D. Forney~Jr, {\em Concatenated codes}.
\newblock Cambridge, MA: MIT Press, 1966.

\bibitem{ccsds:bluebookarticle}
{Consulative Committee for Space Data Systems (CCSDS)}, ``Telemetry channel
  coding,'' {\em Silver Book, {CCSDS 101.0-B-1-S}}, May 1984.

\bibitem{berrou1996near}
C.~Berrou and A.~Glavieux, ``Near optimum error correcting coding and decoding:
  Turbo-codes,'' {\em {IEEE} Trans. Commun.}, vol.~44, pp.~1261--1271, Oct.
  1996.

\bibitem{Gallager63}
R.~G. Gallager, {\em Low-Density Parity-Check Codes}.
\newblock PhD thesis, Dep. Electrical Eng., M.I.T, Cambridge, MA, July 1963.

\bibitem{MacKayldpc}
D.~J.~C. MacKay and R.~M. Neal, ``Near shannon limit performance of low density
  parity check codes,'' {\em Electronics Letters}, vol.~33, pp.~457--458, Mar.
  1997.

\bibitem{Metzer84:retransmission}
J.Metzner, ``An improved broadcast retransmission protocol,'' {\em {IEEE}
  Trans. Commun.}, vol.~32, pp.~679--683, June 1984.

\bibitem{luby97:PracticalLossRes}
M.~Luby, M.~Mitzenmacher, M.~A. Shokrollahi, D.~A. Spielman, and V.~Stemann,
  ``{Practical loss-resilient codes},'' in {\em Proc. 29th Symp. Theory
  Computing}, pp.~150--159, 1997.

\bibitem{luby2001efficient}
M.~Luby, M.~Mitzenmacher, A.~Shokrollahi, and D.~Spielman, ``Efficient erasure
  correcting codes,'' {\em {IEEE} Trans. Inform. Theory}, vol.~47,
  pp.~569--584, Feb. 2001.

\bibitem{oswald2002capacity}
P.~Oswald and A.~Shokrollahi, ``Capacity-achieving sequences for the erasure
  channel,'' {\em {IEEE} Trans. Inform. Theory}, vol.~48, pp.~3017--3028, Dec.
  2002.

\bibitem{miller04:bec}
D.~Burshtein and G.~Miller, ``An efficient maximum likelihood decoding of
  {LDPC} codes over the binary erasure channel,'' {\em {IEEE} Trans. Inform.
  Theory}, vol.~50, nov 2004.

\bibitem{paolini2012maximum}
E.~Paolini, G.~Liva, B.~Matuz, and M.~Chiani, ``Maximum likelihood erasure
  decoding of {LDPC} codes: Pivoting algorithms and code design,'' {\em {IEEE}
  Trans. Commun.}, vol.~60, pp.~3209--3220, Nov. 2012.

\bibitem{byers98:fountain}
J.~Byers, M.~Luby, M.~Mitzenmacher, and A.Rege, ``A digital fountain approach
  to reliable distribution of bulk data,'' in {\em Proc. of ACM SIGCOMM}, 1998.

\bibitem{luby02:LT}
M.~Luby, ``{LT} codes,'' in {\em Proc. 43rd Annual IEEE Symp. on Foundations of
  Computer Science}, (Vancouver, Canada), pp.~271--282, Nov. 2002.

\bibitem{shokrollahi2001raptor}
A.~Shokrollahi and M.~Lassen, S.~Luby, ``Multi-stage code generator and decoder
  for communication systems,'' Dec. 2001.
\newblock {US} Patent 7,068,729.

\bibitem{shokrollahi04:raptor}
A.~Shokrollahi, ``Raptor codes,'' in {\em {Proc. of the 2004 IEEE Int. Symp. on
  Inf. Theory}}, (Chicago, Illinois, US), p.~36, June 2004.

\bibitem{shokrollahi06:raptor}
M.~Shokrollahi, ``Raptor codes,'' {\em {IEEE} Trans. Inform. Theory}, vol.~52,
  pp.~2551--2567, June 2006.

\bibitem{maymounkov2002online}
P.~Maymounkov, ``Online codes,'' tech. rep., Technical report, New York
  University, 2002.

\bibitem{MBMS12:raptor}
{3GPP TS 26.346 V11.1.0}, ``{Technical Specification Group Services and System
  Aspects; Multimedia Broadcast/Multicast Service; Protocols and Codecs},''
  June 2012.

\bibitem{luby2007rfc}
M.~Luby, A.~Shokrollahi, M.~Watson, and T.~Stockhammer, ``{RFC} 5053: Raptor
  forward error correction scheme: Scheme for object delivery,'' tech. rep.,
  {IETF}, Oct. 2007.

\bibitem{shokrollahi2005systems}
M.~Shokrollahi, S.~Lassen, and R.~Karp, ``Systems and processes for decoding
  chain reaction codes through inactivation,'' Feb. 2005.
\newblock US Patent 6,856,263.

\bibitem{lazaro:ITW}
F.~L{\'a}zaro~Blasco, G.~Liva, and G.~Bauch, ``{LT} code design for
  inactivation decoding,'' in {\em Proc. 2014 IEEE Inf. Theory Workshop},
  (Hobart, Tasmania, Australia), pp.~441--445.

\bibitem{lazaro:scc2015}
F.~L{\'a}zaro~Blasco, G.~Liva, and G.~Bauch, ``Enhancing the {LT} component of
  {Raptor} codes,'' in {\em Proc. of the 10th Int. ITG Conf. on Systems,
  Commun. and Coding, SCC 2015}, (Hamburg, Germany).

\bibitem{lazaro:Allerton2015}
F.~L{\'a}zaro, G.~Liva, and G.~Bauch, ``Inactivation decoding analysis for {LT}
  codes,'' in {\em Proc. 52nd Annu. Allerton Conf. on Commun., Control, and
  Computing}, (Monticello, Illinois, USA), Oct. 2015.

\bibitem{lazaro2011concatenation}
F.~L{\'a}zaro~Blasco and G.~Liva, ``On the concatenation of non-binary random
  linear fountain codes with maximum distance separable codes,'' in {\em Proc.
  2011 IEEE Int. Conf. on Commun., (ICC)}, (Kyoto, Japan).

\bibitem{lazaro2013parallel}
F.~L\'azaro~Blasco, , G.~Garrammone, and G.~Liva, ``Parallel concatenation of
  non-binary linear random fountain codes with maximum distance separable
  codes,'' {\em {IEEE} Trans. Commun.}, vol.~61, pp.~4067--4075, Oct. 2013.

\bibitem{lazaro:ISIT2015}
F.~L\'azaro~Blasco, E.~Paolini, G.~Liva, and G.~Bauch, ``On the weight
  distribution of fixed-rate {Raptor} codes,'' in {\em {Proc. of the 2015 IEEE
  Int. Symp. on Inf. Theory}}, (Hong Kong, China), pp.~2880--2884, June 2015.

\bibitem{lazaro:JSAC}
F.~L{\'a}zaro, E.~Paolini, G.~Liva, and G.~Bauch, ``Distance spectrum of
  fixed-rate {Raptor} codes with linear random precoders,'' {\em {IEEE} J.
  Select. Areas Commun.}, Dec. 2015.

\bibitem{lazaro:Globecom2016}
F.~L{\'a}zaro, G.~Liva, E.~Paolini, and G.~Bauch, ``Bounds on the error
  probability of {Raptor} codes,'' (Washington DC, USA), Dec. 2016.
\newblock Proc. IEEE Globecom, to be published, availabe online under
  \emph{http://arxiv.org/abs/1604.07560}.

\bibitem{garrammone2013fragmentation}
G.~Garrammone and F.~L{\'a}zaro~Blasco, ``On fragmentation for fountain
  codes,'' in {\em Proc. of the 10th Int. ITG Conf. on Systems, Commun. and
  Coding, SCC 2013}, (Munich, Germany).

\bibitem{massey:81}
J.~L. Massey, ``Capacity, cutoff rate, and coding for a direct-detection
  optical channel,'' {\em {IEEE} Trans. Commun.}, vol.~29, pp.~1615--1621, Nov.
  1981.

\bibitem{singleton1964maximum}
R.~Singleton, ``Maximum distance q-nary codes,'' {\em {IEEE} Trans. Inform.
  Theory}, vol.~10, no.~2, pp.~116--118, 1964.

\bibitem{berlekamp:bound}
E.~Berlekamp, ``The technology of error-correcting codes,'' {\em IEEE
  Proceedings}, vol.~68, pp.~564--593, 1980.

\bibitem{berlekamp78:intractability}
E.~Berlekamp, R.~McEliece, and H.~Vantilborg, ``On the inherent intractability
  of certain coding problems,'' {\em {IEEE} Trans. Inform. Theory}, vol.~24,
  pp.~384--386, 1978.

\bibitem{CDi2001:Finite}
C.~Di, D.~Proietti, I.~Telatar, T.~Richardson, and R.~Urbanke, ``Finite-length
  analysis of low-density parity-check codes on the binary erasure channel,''
  {\em {IEEE} Trans. Inform. Theory}, vol.~48, pp.~1570 --1579, jun 2002.

\bibitem{Liva2013}
G.~Liva, E.~Paolini, and M.~Chiani, ``Bounds on the error probability of block
  codes over the q-ary erasure channel,'' {\em {IEEE} Trans. Commun.}, vol.~61,
  pp.~2156--2165, June 2013.

\bibitem{Graham:1994}
R.~L. Graham, D.~E. Knuth, and O.~Patashnik, {\em Concrete Mathematics: A
  Foundation for Computer Science}.
\newblock Boston, MA, USA: Addison-Wesley Longman Publishing Co., Inc.,
  2nd~ed., 1994.

\bibitem{CoverThomasBook}
T.~M. Cover and J.~A. Thomas, {\em Elements of information theory}.
\newblock New York: Wiley, 2nd~ed., 2006.
\newblock chapter 15.

\bibitem{Medard08:ARQ}
J.~K. Sundararajan, D.~Shah, and M.~M\'edard, ``{ARQ for Network Coding},'' in
  {\em ISIT'08: Proc. of the 2009 IEEE Int. Symp. on Inf. Theory},
  pp.~1651--1655, July 2008.

\bibitem{Liva10:fountain}
G.~Liva, E.~Paolini, and M.~Chiani, ``{Performance versus overhead for fountain
  codes over $\mathbb{F}_q$},'' {\em IEEE Comm. Letters.}, vol.~14, no.~2,
  pp.~178--180, 2010.

\bibitem{Karp2004}
R.~Karp, M.~Luby, and A.~Shokrollahi, ``Finite length analysis of {LT} codes,''
  in {\em Proc. 2004 IEEE Int. Symp. on Inf. Theory}, ({Chicago, Illinois,
  US}), June 2004.

\bibitem{maneva2006new}
E.~Maneva and A.~Shokrollahi, ``New model for rigorous analysis of
  {LT}-codes,'' in {\em Proc. 2006 IEEE Int. Symp. on Inf. Theory}, (Seattle,
  Washington, US), pp.~2677--2679, 2006.

\bibitem{Maatouk:2012}
G.~Maatouk and A.~Shokrollahi, ``Analysis of the second moment of the {LT}
  decoder,'' {\em {IEEE} Trans. Inform. Theory}, vol.~58, pp.~2558--2569, May
  2012.

\bibitem{shokrollahi2009theoryraptor}
A.~Shokrollahi, ``{Theory and applications of Raptor codes},'' {\em Mathknow},
  vol.~3, pp.~59--89, 2009.

\bibitem{schotsch:2013}
B.~Schotsch, G.~Garrammone, and P.~Vary, ``Analysis of {LT} codes over finite
  fields under optimal erasure decoding,'' {\em {IEEE} Commun. Lett.}, vol.~17,
  pp.~1826--1829, Sept. 2013.

\bibitem{wiedemann1986solving}
D.~H. Wiedemann, ``Solving sparse linear equations over finite fields,'' {\em
  {IEEE} Trans. Inform. Theory}, vol.~32, pp.~54--62, Jan. 1986.

\bibitem{lamacchia91:solving}
B.~A. LaMacchia and A.~M. Odlyzko, ``Solving large sparse linear systems over
  finite fields,'' {\em Advances in Cryptology-CRYPT0’90}, pp.~109--133,
  1991.

\bibitem{shokrollahi2003systematic}
M.~Shokrollahi and M.~Luby, ``Systematic encoding and decoding of chain
  reaction codes,'' June 2005.
\newblock US Patent 6,909,383.

\bibitem{shokrollahi2011raptor}
A.~Shokrollahi and M.~Luby, ``Raptor codes,'' {\em {Foundations and Trends in
  Commun. and Inf. Theory}}, vol.~6, no.~3-4, pp.~213--322, 2011.

\bibitem{RFC5053:raptor}
{RFC 5053}, ``{Network working group; Request for Comments: 5053; Raptor
  Forward Error Correction Scheme for Object Delivery},'' Oct. 2007.

\bibitem{DVB-SH:raptor}
{ETSI TR 102 993 V1.1.1}, ``{Digital Video Broadcasting (DVB); Upper Layer FEC
  for DVB Systems},'' Feb. 2011.

\bibitem{DVB-H:raptor}
{ETSI TS 102 472 V1.3.1}, ``{Digital Video Broadcasting (DVB); IP Datacast over
  DVB-H: Content Delivery Protocols},'' June 2009.

\bibitem{DVB-IPBased:raptor}
{ETSI TS 102 034 V1.5.1}, ``{Digital Video Broadcasting (DVB); Transport of
  MPEG-2 TS Based DVB Services over IP Based Networks},'' May 2014.

\bibitem{DVB:raptor_broadcast}
{ETSI EN 301 790 V1.5.1}, ``{Digital Video Broadcasting (DVB); Interaction
  channel for satellite distribution systems},'' May 209.

\bibitem{ITU:raptor}
{ITU-T H.701}, ``{International Telecommunication Union (ITU); Series H:
  Audiovisual and Multimedia Systems; IPTV multimedia services and applications
  for IPTV - General aspects; Content delivery error recovery for IPTV
  services},'' Mar. 2009.

\bibitem{Chingbats}
T.-C. Ng and S.~Yang, ``Finite-length analysis of {BATS} codes,'' in {\em Proc.
  of 2013 IEEE Int. Symp. on Network Coding, (NetCod)}, (Calgary, Alberta,
  Canada), June 2013.

\bibitem{mahdaviani2012raptor}
K.~Mahdaviani, M.~Ardakani, and C.~Tellambura, ``{On Raptor code design for
  inactivation decoding},'' {\em {IEEE} Commun. Lett.}, vol.~60,
  pp.~2377--2381, Sept. 2012.

\bibitem{kirkpatrick1983optimization}
S.~Kirkpatrick, D.~Gelatt, and M.~Vecchi, ``Optimization by simmulated
  annealing,'' {\em Science}, vol.~220, no.~4598, pp.~671--680, 1983.

\bibitem{Rahnavard:07}
N.~Rahnavard, B.~Vellambi, and F.~Fekri, ``Rateless codes with unequal error
  protection property,'' {\em {IEEE} Trans. Inform. Theory}, vol.~53,
  pp.~1521--1532, Apr. 2007.

\bibitem{Schotsch:14}
B.~E. Schotsch, {\em Rateless Coding in the Finite Length Regime}.
\newblock PhD thesis, Inst. of Commun. Systems and Data Proc., RWTH Aachen,
  Aachen, Germany, July 2014.

\bibitem{wang:2015}
P.~Wang, G.~Mao, Z.~Lin, M.~Ding, W.~Liang, X.~Ge, and Z.~Lin, ``Performance
  analysis of {R}aptor codes under maximum likelihood decoding,'' {\em {IEEE}
  Trans. Commun.}, vol.~64, pp.~906--917, Mar. 2016.

\bibitem{MacWillimas77:Book}
F.~{Mac Williams} and N.~Sloane, {\em The theory of error-correcting codes},
  vol.~16.
\newblock North Holland Mathematical Libray, 1977.

\bibitem{ashikhmin:98}
A.~Ashikhmin and A.~Barg, ``Minimal vectors in linear codes,'' {\em {IEEE}
  Trans. Inform. Theory}, vol.~44, pp.~2010--2017, Sept. 1998.

\bibitem{barg01:concat}
A.~Barg, J.~Justesen, and C.~Thommesen, ``Concatenated codes with fixed inner
  code and random outer code,'' {\em {IEEE} Trans. Inform. Theory}, vol.~47,
  pp.~361--365, Jan. 2001.

\bibitem{orlitsky05:stopping}
A.~Orlitsky, K.~Viswanathan, and J.~Zhang, ``Stopping set distribution of
  {LDPC} code ensembles,'' {\em {IEEE} Trans. Inform. Theory}, vol.~51,
  pp.~929--953, Mar. 2005.

\bibitem{di06:weight}
C.~Di, T.~Richardson, and R.~Urbanke, ``Weight distribution of low-density
  parity-check codes,'' {\em {IEEE} Trans. Inform. Theory}, vol.~52,
  pp.~4839--4855, Nov. 2006.

\bibitem{kasai:fountain}
K.~Kasai, D.~Declercq, and K.~Sakaniwa, ``Fountain coding via multiplicatively
  repeated non-binary {LDPC} codes,'' {\em {IEEE} Trans. Commun.}, vol.~60,
  pp.~2077--2083, Aug. 2012.

\bibitem{Chiaraluce:hamming}
F.~Chiaraluce and R.~Garello, ``{On the asymptotic performance of Hamming
  product codes},'' in {\em Proc. 6th Int. Symp. on Commun. Theory and
  Applications}, pp.~329--–334, 2001.

\bibitem{barg01:random}
A.~Barg and G.~D. Forney, ``Random codes: minimum distances and error
  exponents,'' {\em {IEEE} Trans. Inform. Theory}, vol.~48, pp.~2568--2573,
  Sept. 2002.

\end{thebibliography}
